\documentclass{article}
\usepackage{amssymb,amsmath,graphicx,ucs,hyperref}
\usepackage[utf8x]{inputenc}

\usepackage{tikz}
\usepgflibrary{shapes}

\newcommand*\fond[2]{\tikz[baseline=(char.base)]{\node[regular polygon, regular polygon sides=#1,draw,inner sep=1pt] (char) {#2};}}

\newtheorem{proposition}{Proposition}
\newtheorem{theorem}{Theorem}
\newtheorem{lemma}{Lemma}
\newtheorem{open}{Open question}

\newenvironment{proof}{\noindent \emph{Proof. }}{\hfill \hbox{\rlap{$\sqcap$}$\sqcup$}\\}

\title{Compact packings of the plane\\with three sizes of discs\thanks{
    The work of O. S was supported within frameworks of the state task for ICP RAS 0082-2014-0001 (state registration AAAA-A17-117040610310-6).
    The work of Th. F and A. H was supported by the Partenariat Hubert Curien (PHC) Gundishapur.
    }}
\author{
Thomas Fernique\footnote{Universit\'e Paris 13, CNRS, Sorbonne Paris Cit\'e, UMR 7030, 93430 Villetaneuse, France.}
\and Amir Hashemi\footnote{Department of Mathematical Sciences, Isfahan University of Technology, Isfahan, Iran.}
\and Olga Sizova\footnote{Faculty of Mathematics, Higher School of Economics, 119048 Moscow, Russia.}~\footnote{Semenov Institute of Chemical Physics, 119991 Moscow, Russia.}
}

\date{}

\begin{document}
\maketitle

\begin{abstract}
  A compact packing is a set of non-overlapping discs where all the holes between discs are curvilinear triangles.
  There is only one compact packing by discs of size $1$.
  There are exactly $9$ values of $r$ which allow a compact packing by discs of sizes $1$ and $r$.
  We prove here that there are exactly $164$ pairs $(r,s)$ allowing a compact packing by discs of sizes $1$, $r$ and $s$.
\end{abstract}

\section{Introduction}

A set of interior-disjoint discs is called a {\em packing}.
Packings are of special interest to model the structure of materials, {\em e.g.}, crystals or granular materials, and the goal in this context is to understand which typical or extremal properties have the packings (see, {\em e.g.}, \cite{LH93,OH11,HST12}).
In 1964, T\'oth coined the notion of {\em compact packing} \cite{FT64}: this is a packing whose {\em contact graph} (the graph which connects the center of mutually tangent discs) is triangulated.
Equivalently, all its holes are curvilinear triangles.

There is only one compact packing with with all the discs of the same size, called the {\em hexagonal compact packing}: the disc centers are located on the triangular grid.
In \cite{Ken06}, it is proven that there are exactly $9$ values of $r$ which allow a compact packing with discs of size (radius) $1$ and $r$.
Fig.~\ref{fig:2packings} depicts an example of compact packing for each case.
All these packings already appeared in \cite{FT64}, except $c_5$ which later appeared in \cite{LH93} and $c_2$ which was new at that time.

\begin{figure}[hbtp]
\centering
\begin{tabular}{lll}
  c1 (L)\hfill 1111r & c2 (S)\hfill 111rr & c3 (L)\hfill 11r1r\\
  \includegraphics[width=0.3\textwidth]{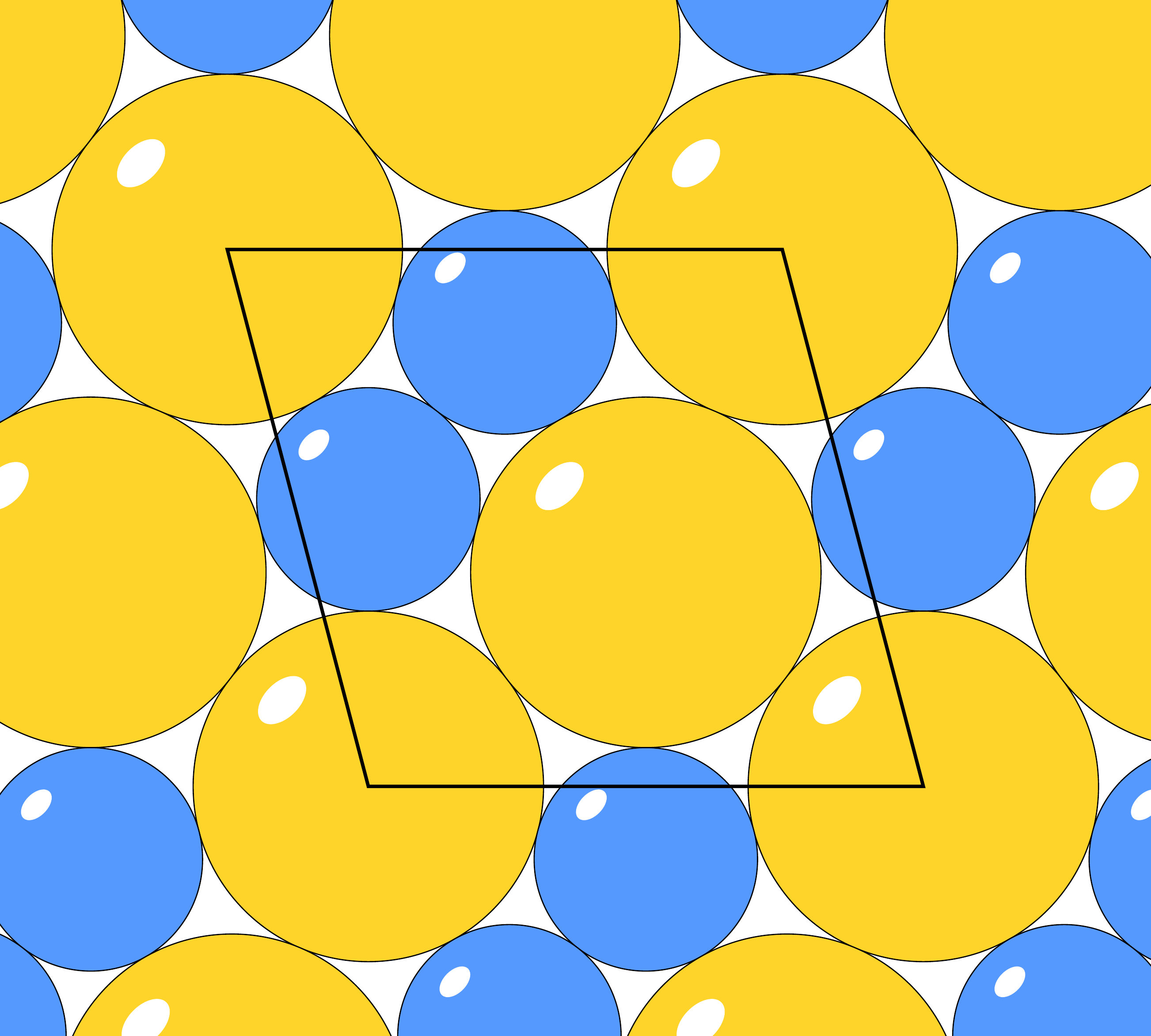} &
  \includegraphics[width=0.3\textwidth]{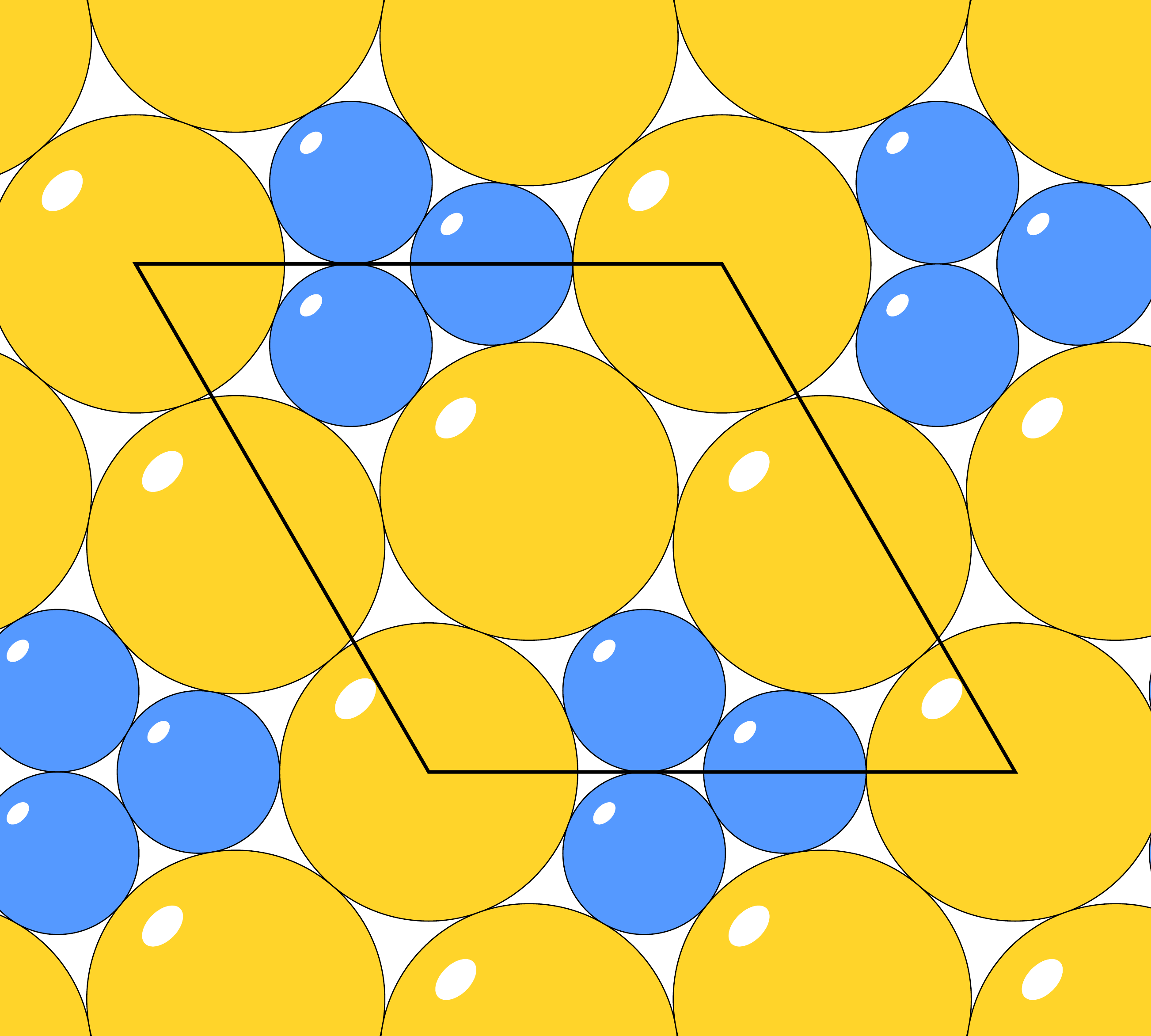} &
  \includegraphics[width=0.3\textwidth]{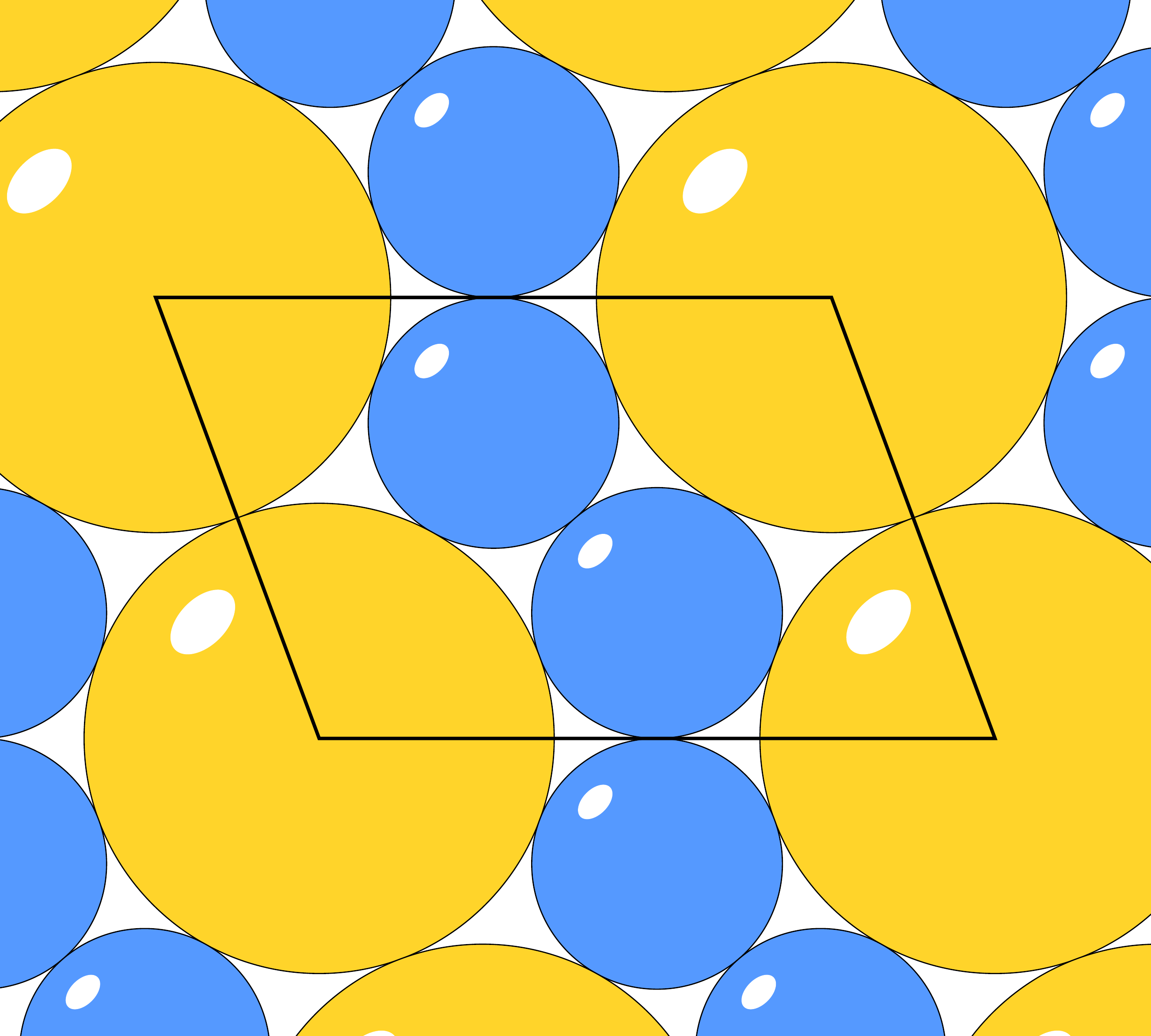}\\
  c4 (E)\hfill 1111 & c5 (H)\hfill 11rrr & c6 (H)\hfill 1r1rr\\
  \includegraphics[width=0.3\textwidth]{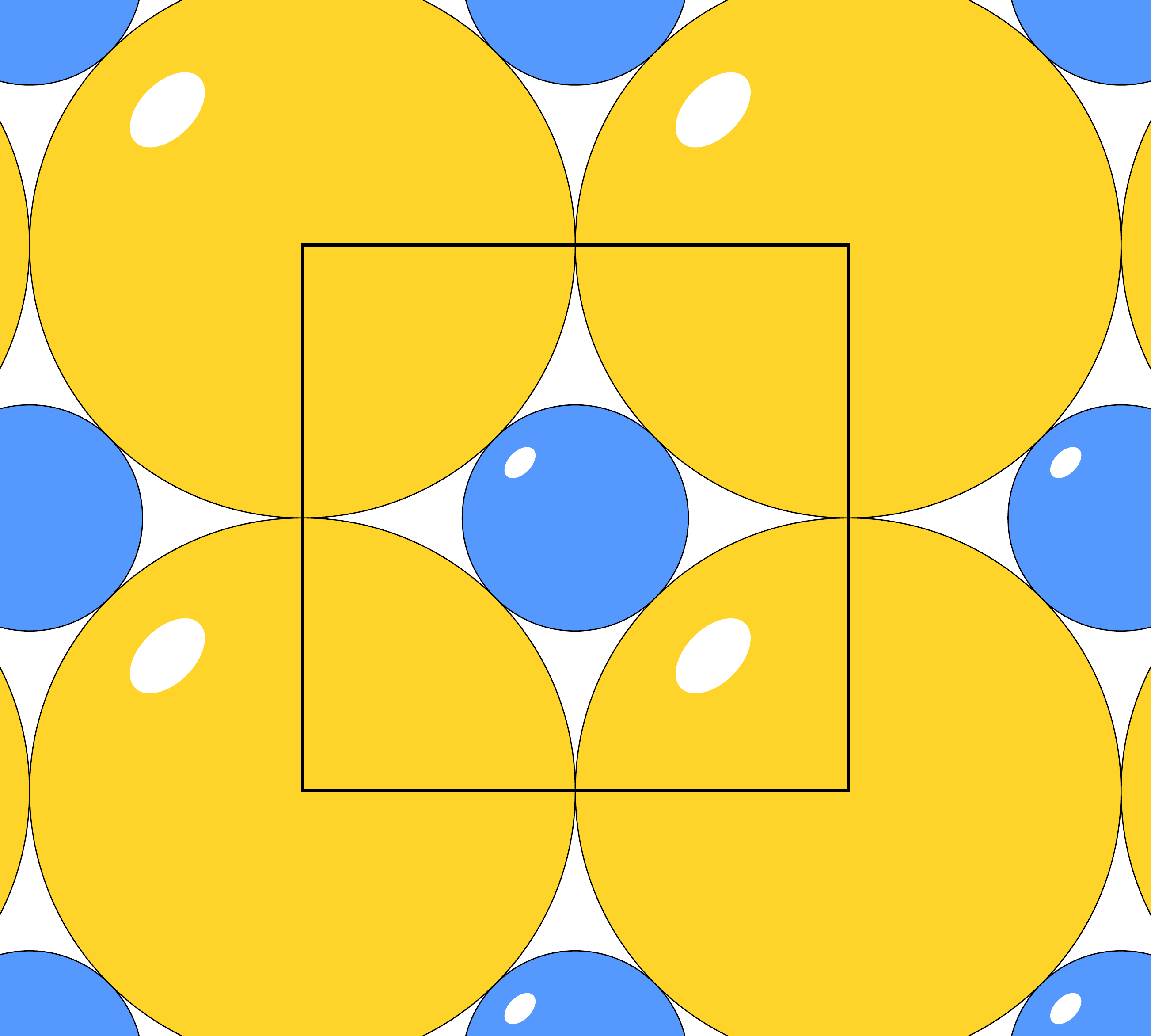} &
  \includegraphics[width=0.3\textwidth]{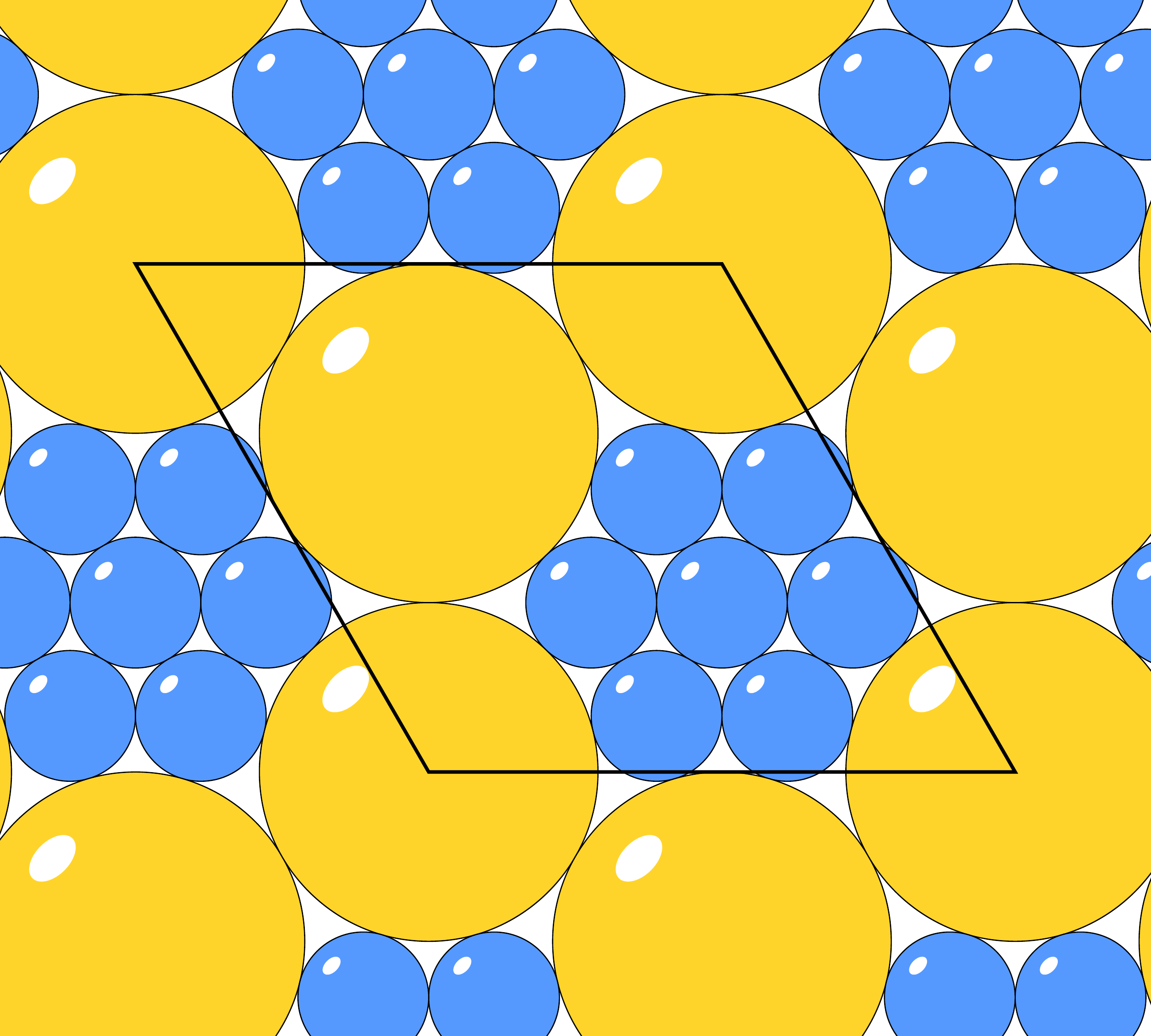} &
  \includegraphics[width=0.3\textwidth]{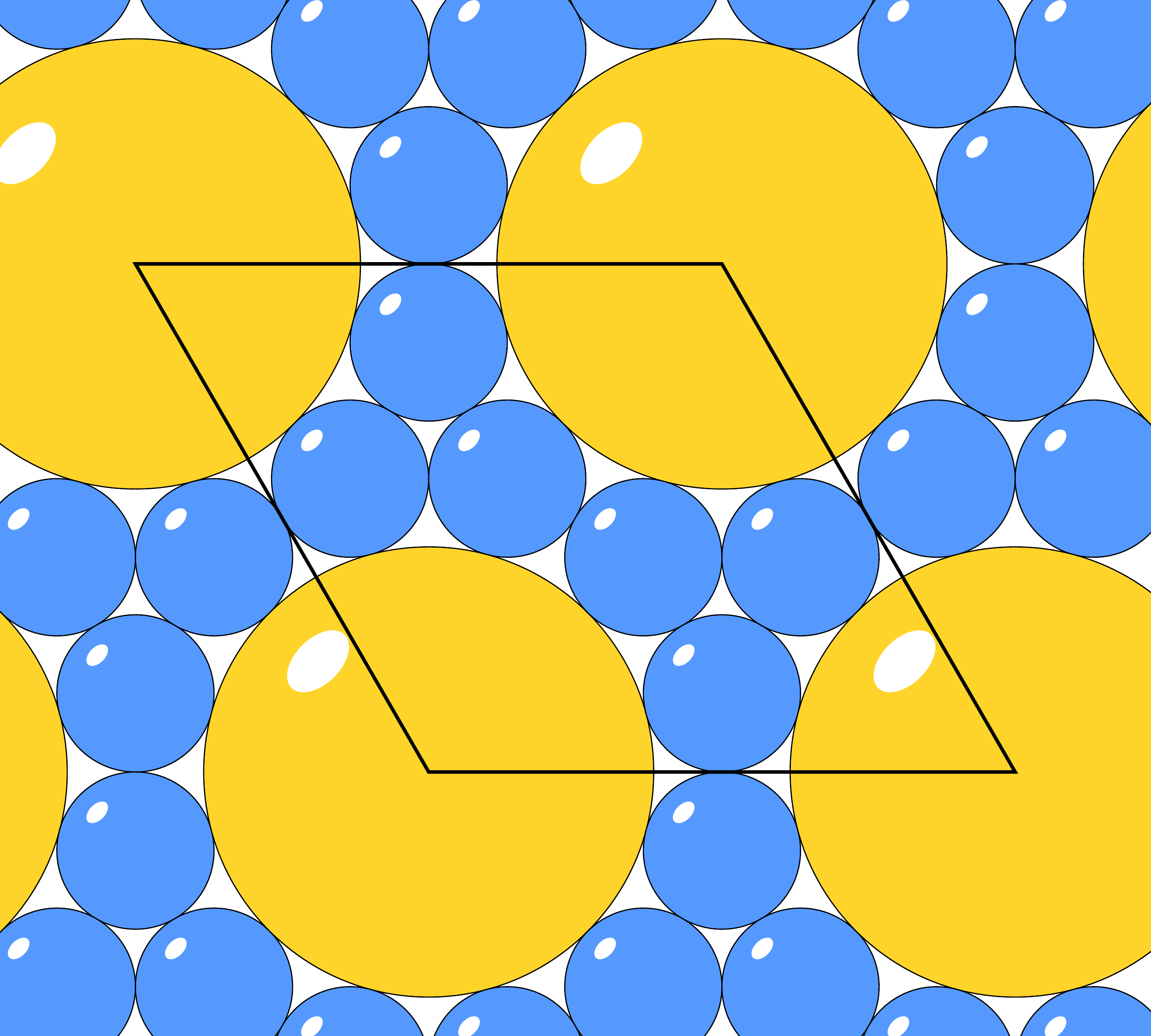}\\
  c7 (E)\hfill 111r & c8 (H)\hfill 111 & c9 (H)\hfill 11rr\\
  \includegraphics[width=0.3\textwidth]{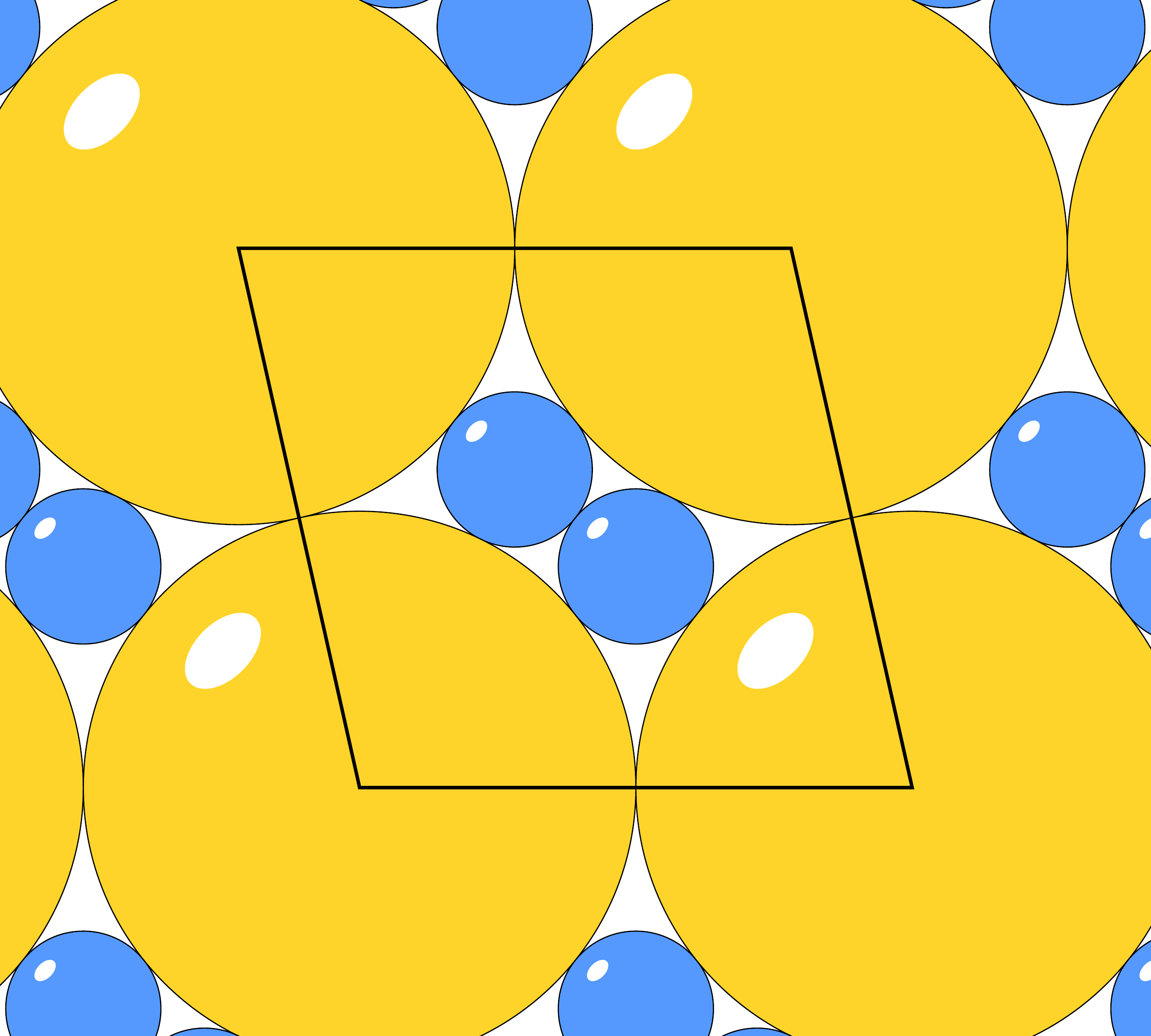} &
  \includegraphics[width=0.3\textwidth]{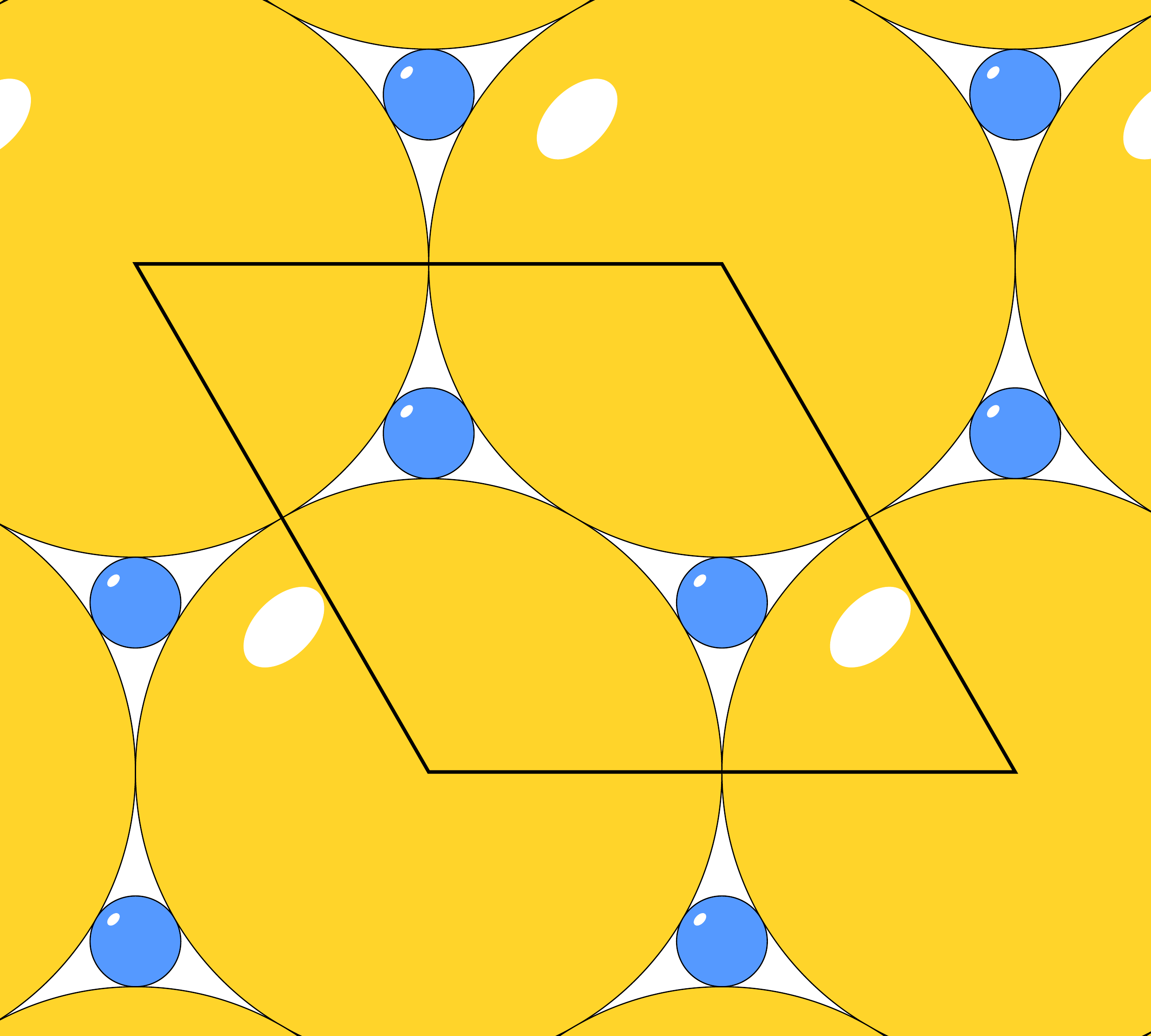} &
  \includegraphics[width=0.3\textwidth]{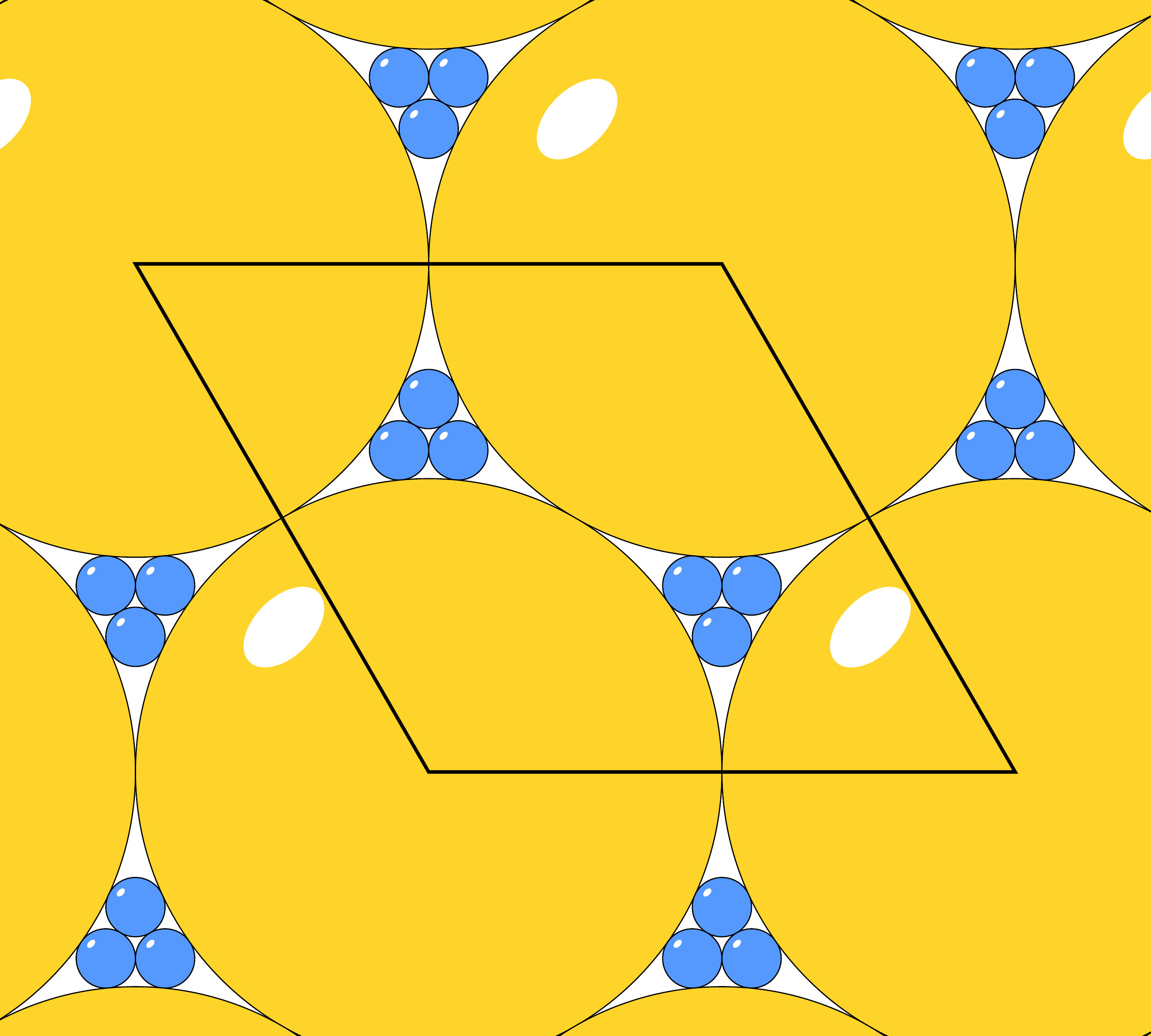}
\end{tabular}
\caption{
An example of compact packing for each of the $9$ possible values of $r<1$ which allow a compact packings by discs of sizes $1$ and $r$.
They are all periodic, with the parallelogram showing a fundamental domain.
The top-right word over $\{\textrm{1},\textrm{r}\}$ codes the {\em corona} of a small disc (see Sec.~\ref{sec:notations}).
A label in c1--c9 is assigned to each case (top-left), followed by a letter in brackets which refers to its {\em type} (see Appendix~\ref{sec:classification}).
}
\label{fig:2packings}
\end{figure}

Recently, it was proven in \cite{Mes20} that there are at most $11462$ pairs $(r,s)$ which allow a compact packing by discs of sizes $1$, $r$ and $s$.
The author provided several examples and suggested that a complete characterization could be beyond the actual capacity of computers.
We here overcome this limitation and we prove:

\begin{theorem}
\label{th:main}
There are exactly $164$ pairs $(r,s)$ which allow a compact packing by discs of sizes $1$, $r$ and $s$, with $0<s<r<1$.
\end{theorem}

All the values $r$ and $s$ are algebraic and their minimal polynomials (which can be quite complicated) are given in the supplementary materials (as well as numerical approximations).
Fig.~\ref{fig:3packings} depicts an example of compact packing for $9$ of these $164$ cases.
The full list is in Appendix~\ref{sec:examples}.
In each case we found a {\em periodic} compact packing, so that it suffices to give its fundamental domain (as done in Fig.~\ref{fig:2packings} and \ref{fig:3packings}).
It is of course difficult to convince oneself with the naked eye that the depicted packings are {\em really} compact.
However, we shall see that their {\em combinatorics} is sufficient to ensure that they are indeed compact.
Note also that, in many cases, more compact packings than the only one depicted are possible (sometimes much more - as for the surprising number $83$ in Fig.~\ref{fig:3packings}).
This is discussed in Appendix~\ref{sec:classification}.

\begin{figure}[hbtp]
\centering
\begin{tabular}{lll}
  3 (E)\hfill 1111 / 11r1r & 27 (E)\hfill 11r / 1s1s1s1s & 33 (L)\hfill 1rr / 11srs1srs\\
  \includegraphics[width=0.3\textwidth]{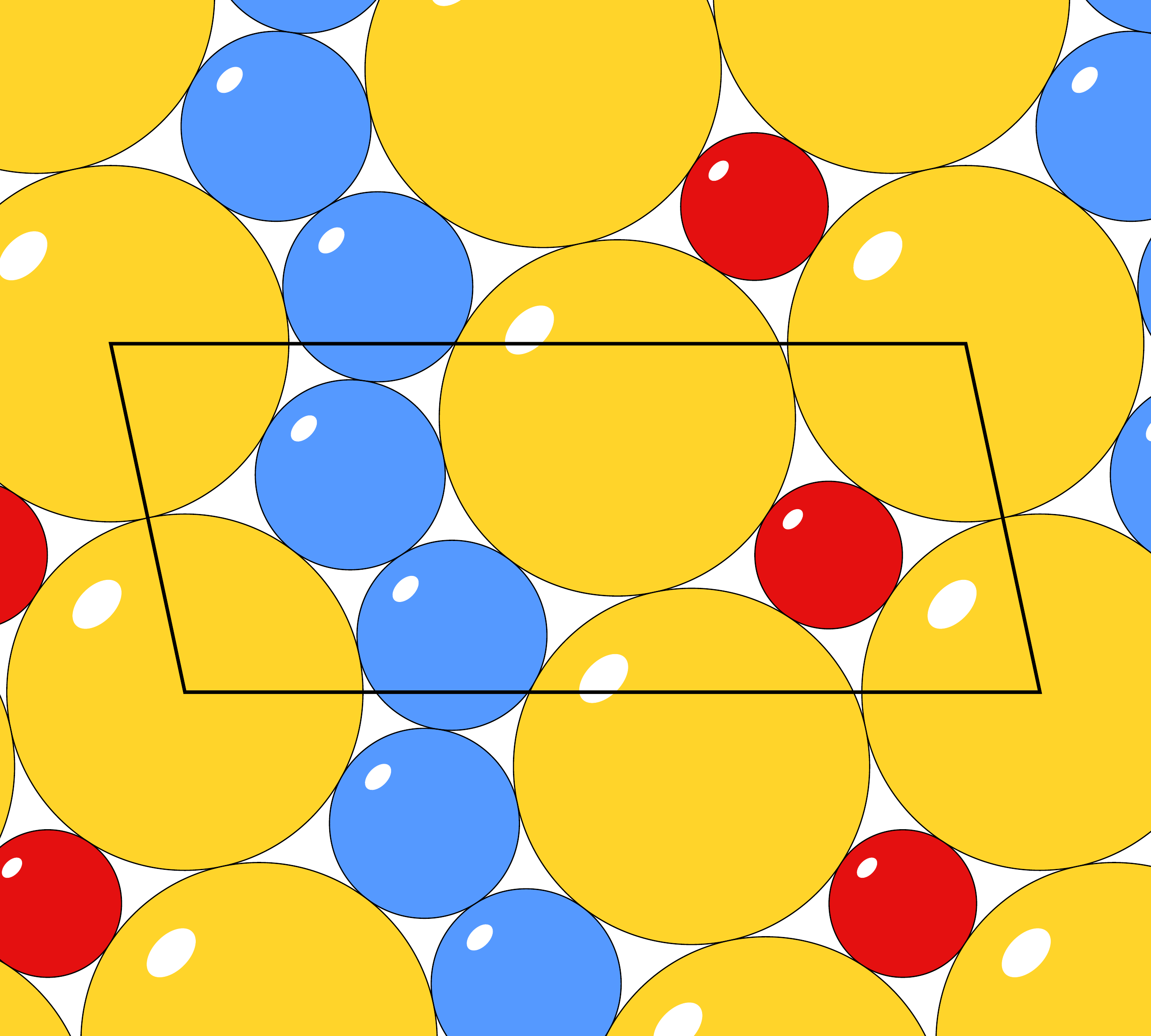} &
  \includegraphics[width=0.3\textwidth]{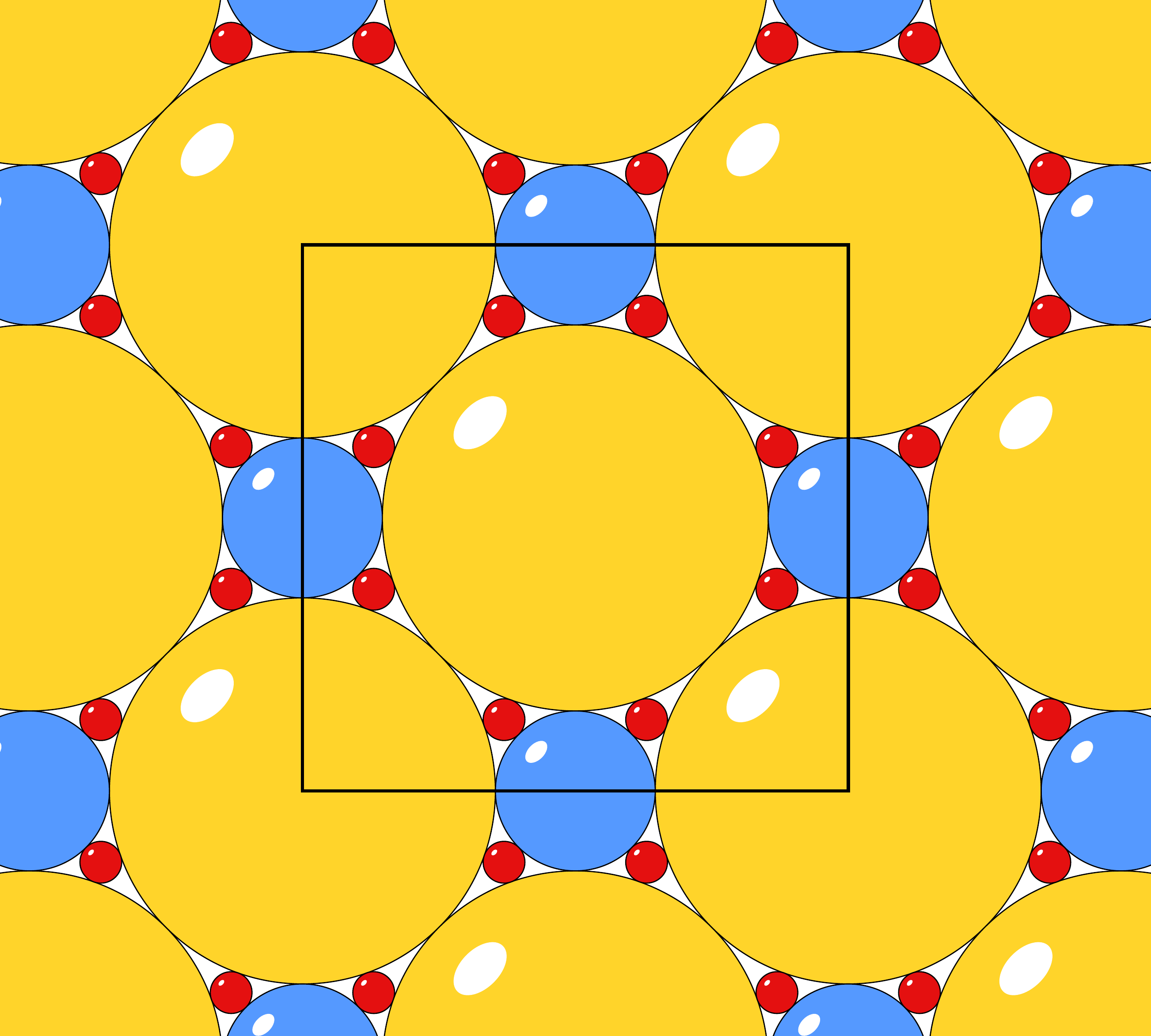} &
  \includegraphics[width=0.3\textwidth]{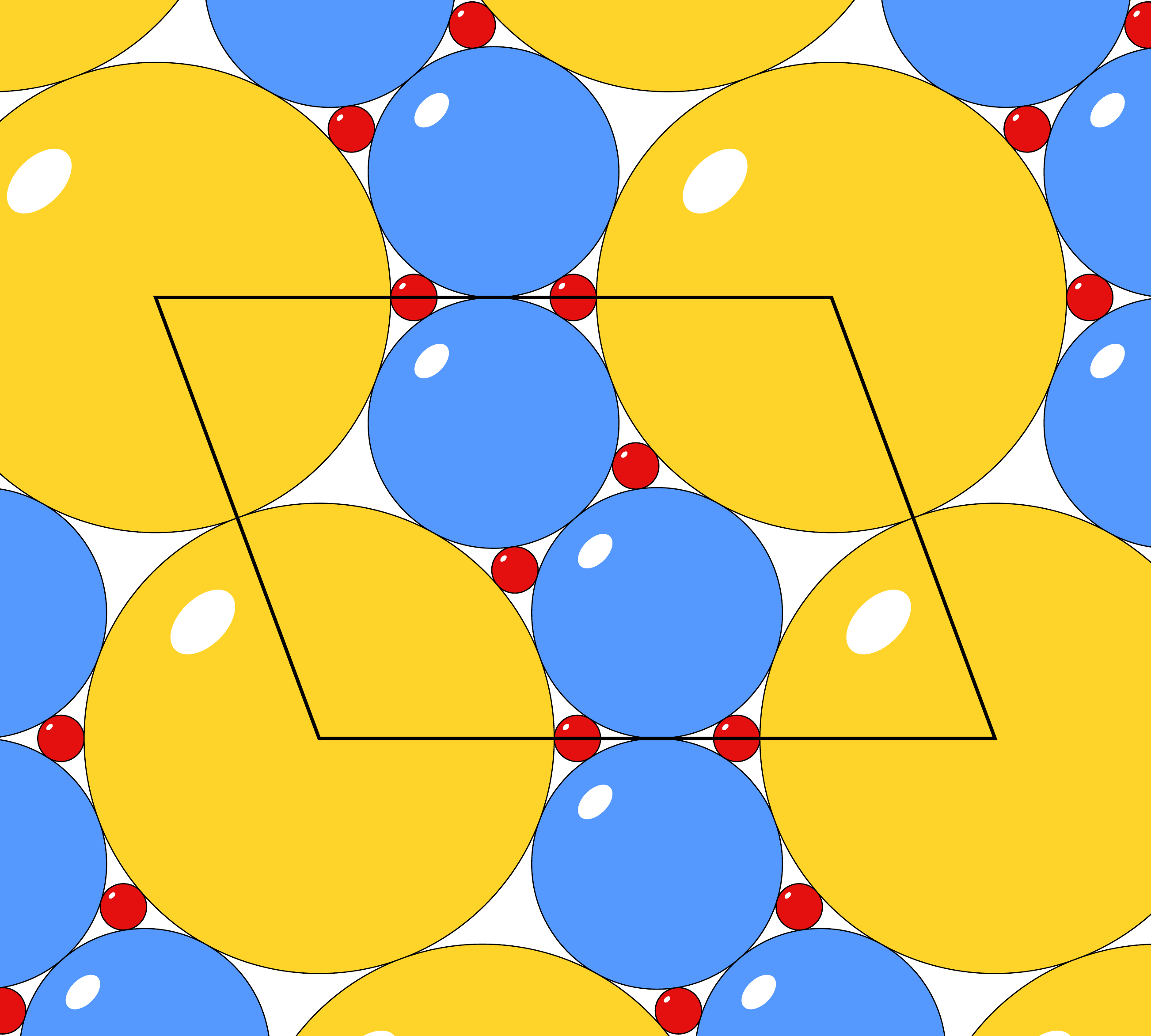}\\
  51 (L)\hfill 111rr / 1rrrrs & 53 (H)\hfill 11r1r / 1r1s1s & 83 (E)\hfill 1r1r / 11r1s\\ 
  \includegraphics[width=0.3\textwidth]{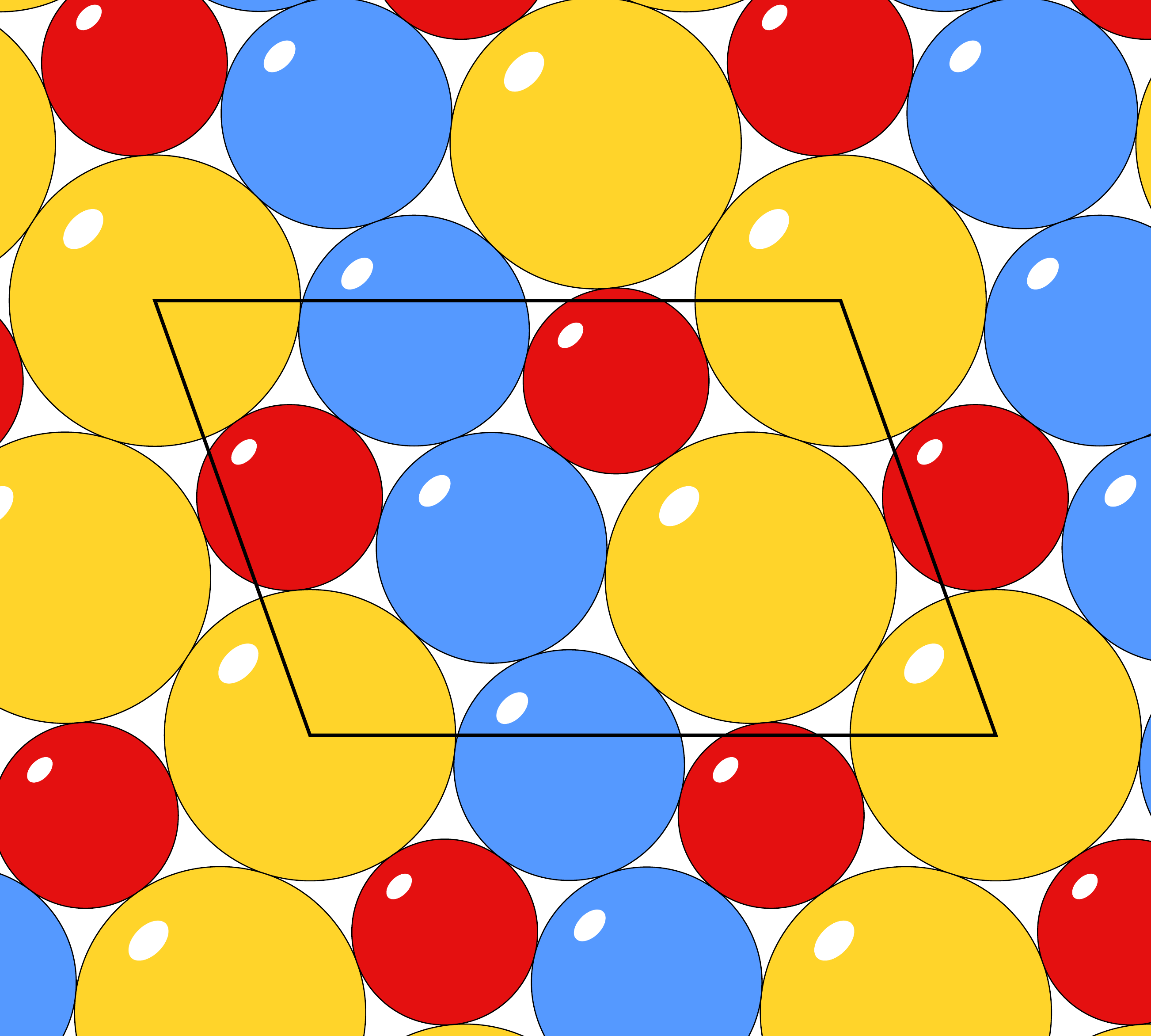} &
  \includegraphics[width=0.3\textwidth]{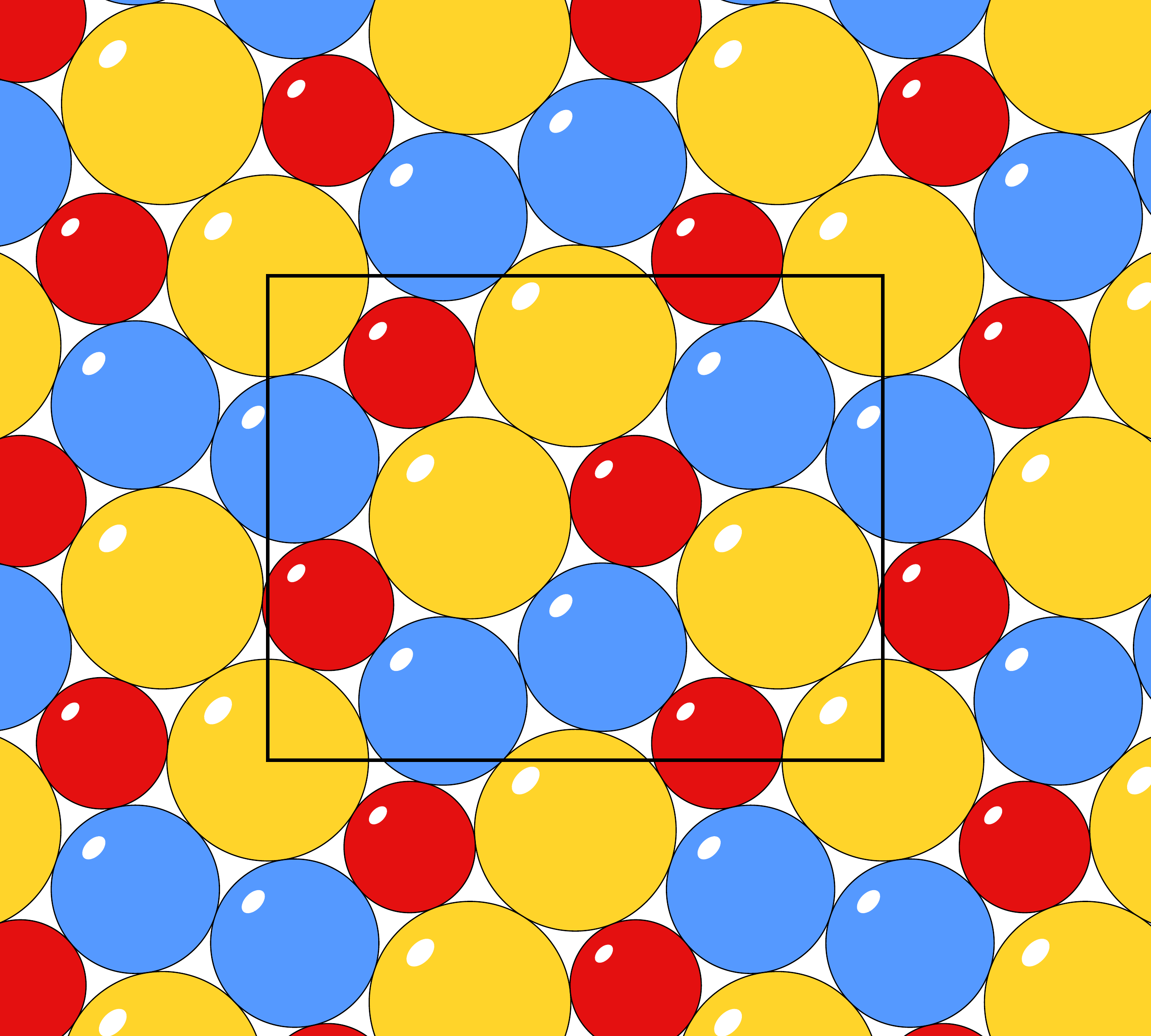} &
  \includegraphics[width=0.3\textwidth]{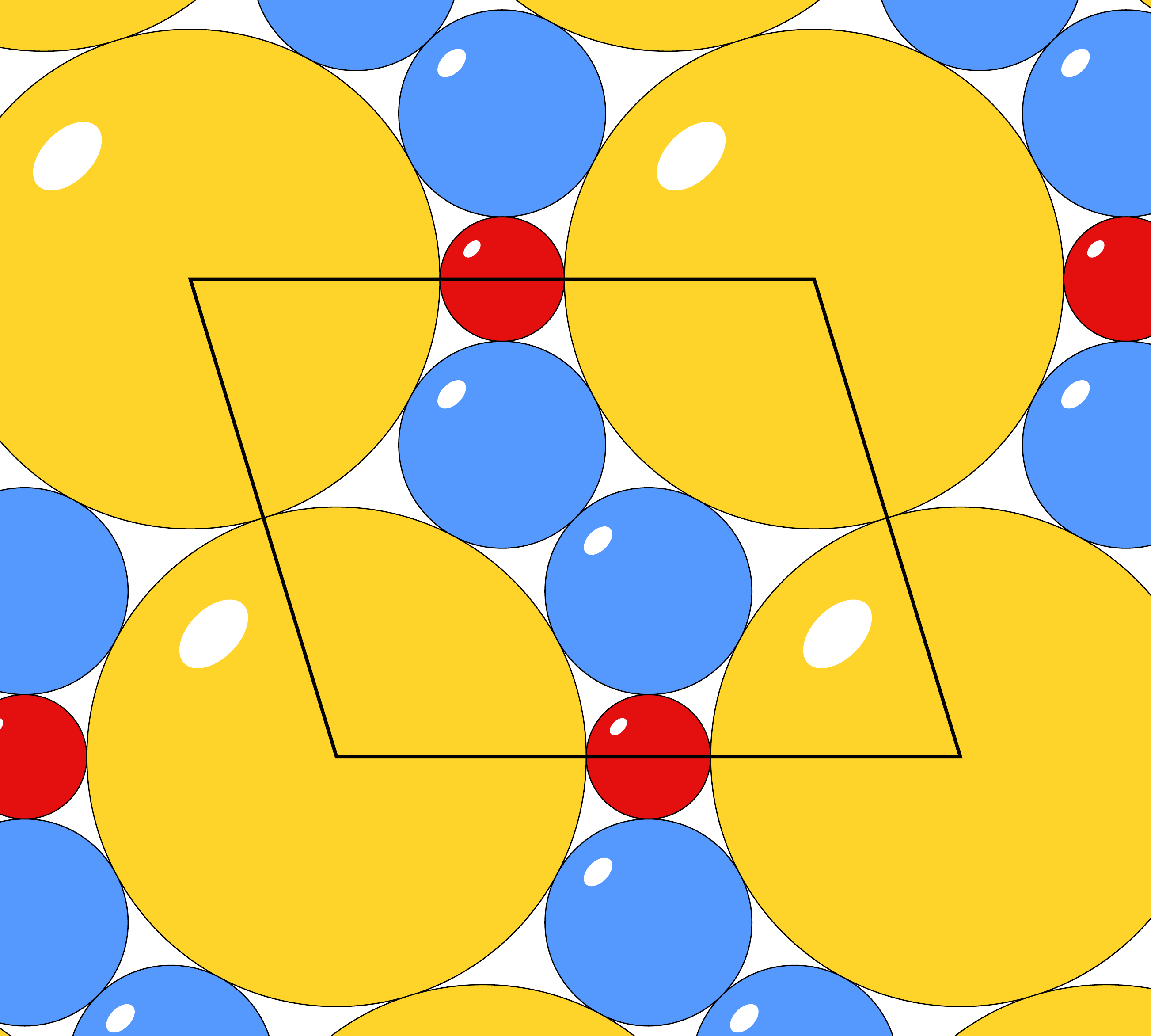}\\
  93 (H)\hfill 1r1rr / 1s1srs & 107 (H)\hfill 1r1s / 1s1s1s & 152 (S)\hfill 1rssr / 1s1sss \\
  \includegraphics[width=0.3\textwidth]{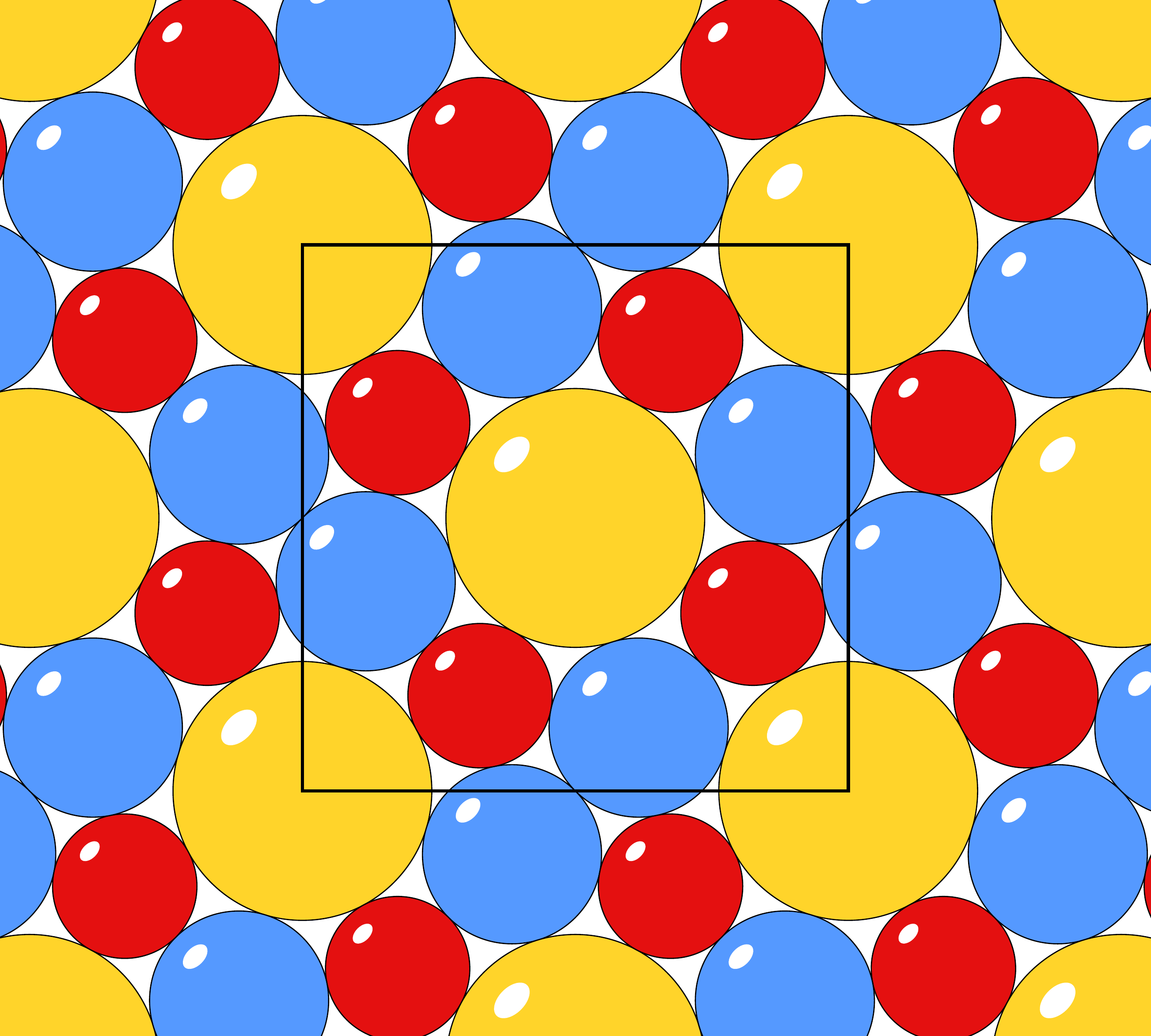} &
  \includegraphics[width=0.3\textwidth]{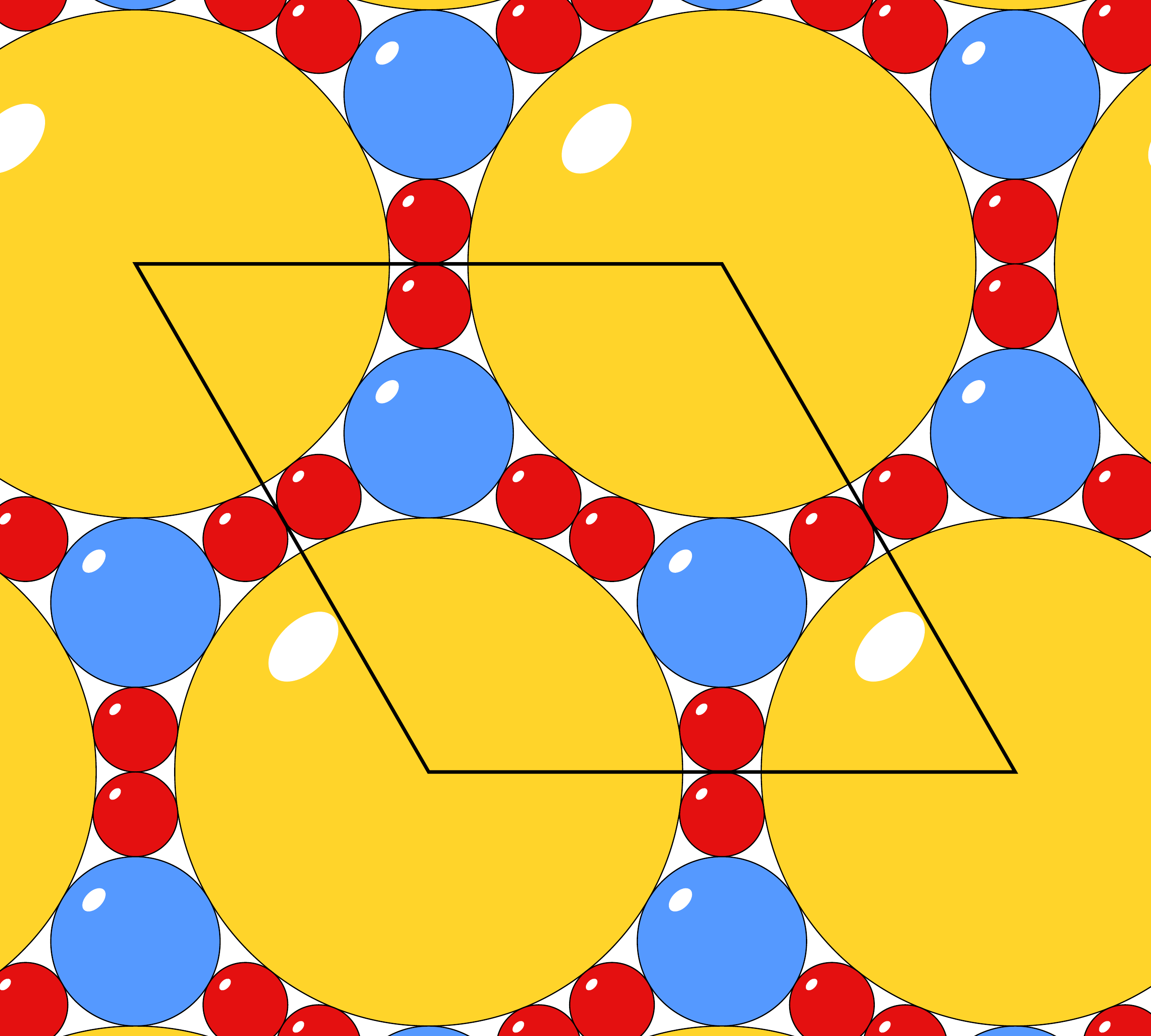} &    
  \includegraphics[width=0.3\textwidth]{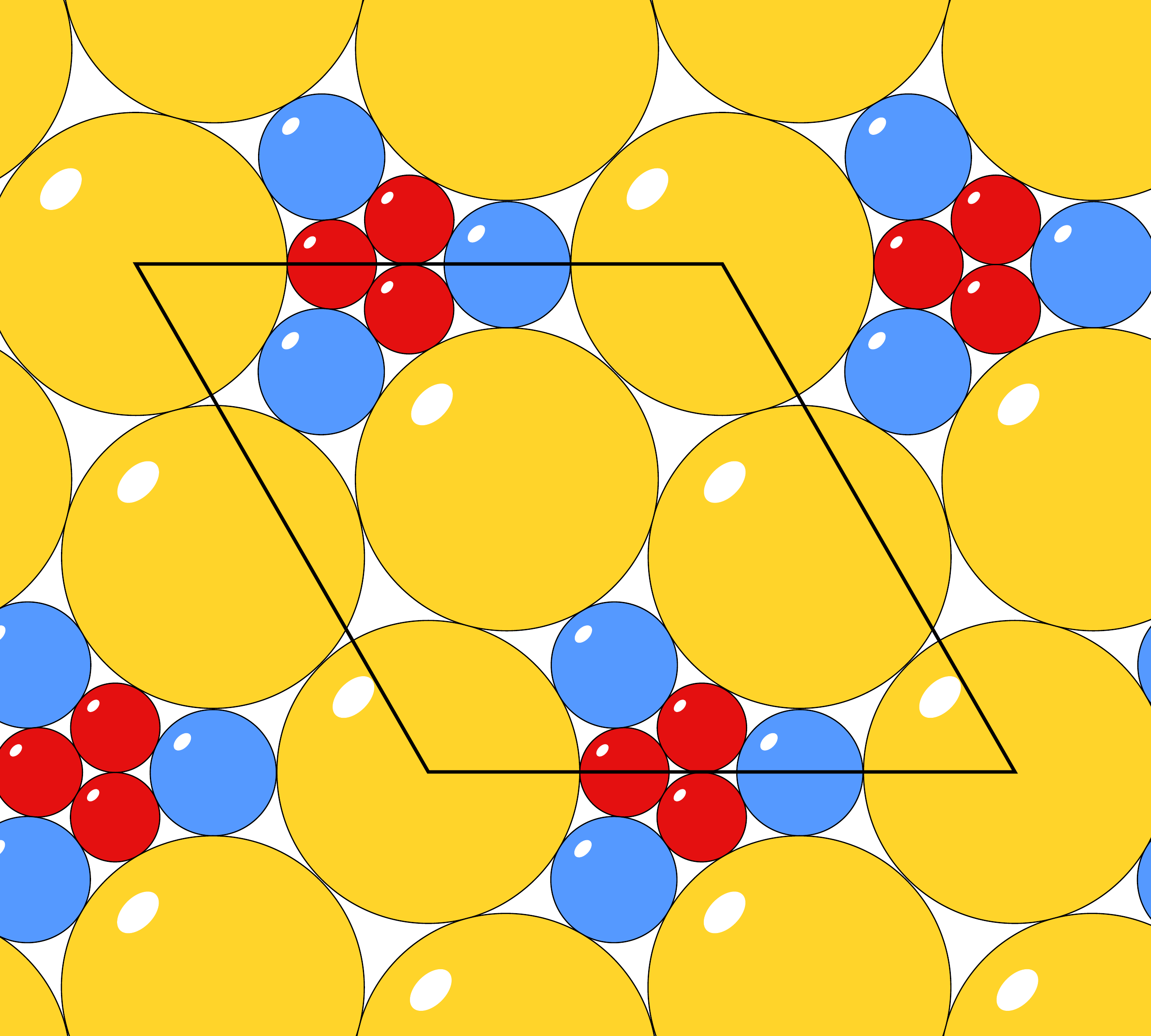}\\
\end{tabular}
\caption{
Some compact packings with three sizes of discs.
They are all periodic, with the parallelogram showing a fundamental domain.
The top-left number refers to numbers in App.~\ref{sec:examples} and the two top-right words over $\{\textrm{1},\textrm{r},\textrm{s}\}$ codes the {\em coronas} of both a small and a medium disc (see Sec.~\ref{sec:notations}).
The letter in brackets refers to the {\em type} of the compact packing (see App.~\ref{sec:classification}).
}
\label{fig:3packings}
\end{figure}

Many compact packings with three sizes of discs can be found by hand.
Those on the first line in Fig.~\ref{fig:3packings}, for example, are easily derived from compact packings with two sizes of discs.
Some others were already known.
Number $107$ appears in \cite{FT64} (p.~187).
Numbers $99$, $104$ and $143$--$146$ appear in \cite{Mes20}.
Number $51$ can be found on a street pavement in Weggis, Switzerland (Rigiblickstrasse~$1$).
The challenge is to find them all.
As we shall see, this is mathematically rather simple, but the profusion of cases and the complexity of the calculations make it a computational challenge.

Before we continue, let us give a motivation to study compact packings.
A central problem in packing theory is to find the maximal density over a set $\mathcal{P}$ of packings, defined by
$$
\delta(\mathcal{P}):=\sup_{P\in\mathcal{P}}\limsup_{k\to\infty}\frac{\textrm{area of the $k\times k$ square covered by $P$}}{k^2}.
$$
The maximal density over packings with only one size of disc has been proven to be $\pi/\sqrt{12}\simeq 0.9069$, attained for the unique possible compact packing (the hexagonal one) \cite{FT43}.
With two sizes of discs, the maximal density turns out to be known only for sizes which allow a compact packing.
Namely, all the compact packings in Fig.~\ref{fig:2packings} have been proven to maximize the density among the packings with the same sizes of discs \cite{Hep00,Hep03,Ken04}, except $c_5$ and $c_9$ which are still in the running.
Compact packings thus seem to be good candidates to provably maximize the density.
So, what about the $164$ new cases with three discs?
Some can be easily ruled out.
Number $33$, for example (depicted top-right in Fig.~\ref{fig:3packings}), has small discs which can be inserted in the holes between two large and a medium discs: this increases the density but yields a non-compact packing.
We shall therefore consider only so-called {\em saturated} packings, {\em i.e.}, such that no further discs can be added.
Does it suffice?

\begin{open}
Assume that a finite set of different discs allow a saturated compact packing.
Consider all the packings by these discs: is the maximal density achieved for a compact packing?
\end{open}

The densest compact packings in cases $1$--$18$ turn out to have only two sizes of discs, so that we do not get anything new.
Cases $24$, $29$--$34$ and $37$--$44$ are ruled out because the densest compact packings are not saturated.
There are still $130$ candidates left\ldots

The compact packings depicted in Appendix~\ref{sec:examples} are actually the densest among the possible compact packings, except in cases $1$--$18$ where the densest compact packings have only two discs (this can be checked with the help of App.~\ref{sec:classification}).
As already mentioned, they are all periodic.
On the one hand, it is a chance because it makes it easy to describe them.
On the other hand, it is a bit disappointing because one of our main goals to extend \cite{Ken06} from two to three sizes of discs was to find a densest aperiodic packing.
This could indeed have been put forward as a rather simple explanation of aperiodic structure of materials known as {\em quasicrystals}.
Maybe a few extra disk sizes would be enough?

\begin{open}
Is there a finite set of different discs which allow compact packings and such that the densest one is aperiodic?
\end{open}

To finish with density of compact packings, let us say a word about higher dimensional packings, namely sphere packings.
The notion of density is easily generalized in $\mathbb{R}^n$.
Then, a sphere packing in $\mathbb{R}^n$ is said to be compact if its contact graph is a homogeneous simplicial complex of dimension $n$.
This coincides with the previous notion for $n=2$.
For $n=3$, it means that it can be seen as a tiling by tetrahedra which can intersect only on a full face, a full edge or a point.
The case of packings with only one size of sphere has been extensively studied (see, {\em e.g.}, \cite{CS99}).
The maximal density is known in dimension $3$ \cite{Hal05}, $8$ \cite{Via17} and $24$ \cite{CKM17}.
In dimension $8$ and $24$, the maximal density turns out to be achieved by a compact packing!
This is not the case in dimension $3$, but there is no compact packing in this case (regular tetrahedra do not tile the space).
Compact packings are thus still good candidates to provably maximize the density in higher dimensions.
In particular, what about the compact packings with two sizes of spheres studied in \cite{Fer18}?

The paper is organized as follows.
Section~\ref{sec:notations} sets out some notations and introduces the important notion of {\em corona}.
In Section~\ref{sec:strategy}, the overall strategy is presented.
As we shall see, the proof heavily relies on computer, and we shall explain it in details how in order to make it reproducible.
All the computations were made with the open-source software SageMath \cite{sage} on our modest laptop, an Intel Core i5-7300U with $4$ cores at $2.60$GHz and $15,6$ Go RAM.
Sections~\ref{sec:small} through \ref{sec:one_small} are dedicated to the proof itself (more details are given in Section~\ref{sec:strategy}).
An example of compact packing for each possible sizes is given in Appendix~\ref{sec:examples}, while the variety of compact packings which is possible for each of these sizes is discussed in Appendix~\ref{sec:classification}.
Last, Appendix~\ref{sec:code} provides indications on the code used on computer to find possible sizes, which can be found in supplementary materials.

\section{Notations}
\label{sec:notations}

The large disc is assumed to have size, $1$.
The sizes of the medium and small discs are denoted by $r$ and $s$, $0<s<r<1$.
We also call {\em x-disc} a disc of size $x$.
In a compact packing, the {\em corona} of a disc is the set of discs it is tangent to.
We shall consider coronas up to isometries.
A corona is said to be {\em small}, {\em medium} or {\em large} depending whether the surrounded disc is small, medium or large.
We shall also call them s-, r- and 1-coronas.
The {\em coding} of a corona is a word over the alphabet $\{\textrm{1},\textrm{r},\textrm{s}\}$: each letter corresponds to a disc and gives its size, with two letter being neighbor if and only if the corresponding discs are tangent.
We may use an exponent when there are many consecutive letters in the coding, for example 11rr$\textrm{s}^{12}$ denotes the corona 11rrssssssssssss ($12$ consecutive s-discs).
Any circular permutation or reversal of this word corresponds to the same corona: we shall usually use the lexicographically smallest coding.
Given pairwise tangent discs of size $x$, $y$ and $z$, $\widehat{xyz}$ denotes the non-oriented angle between the segments which connect the center of the disc of size $y$ to the two other centers.
Figure~\ref{fig:notations} illustrates this.

\begin{figure}
  \centering
  \includegraphics[width=0.4\textwidth]{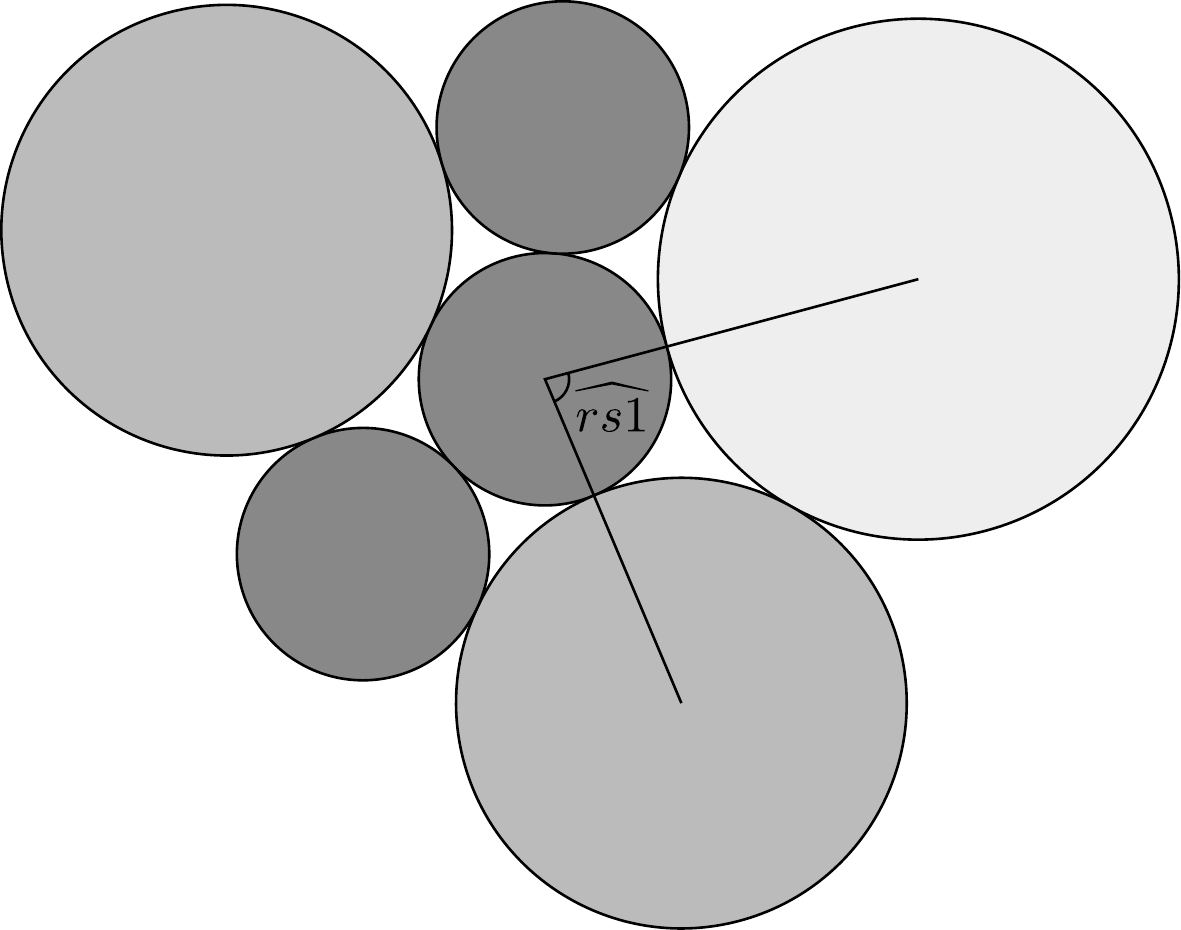}
  \caption{The s-corona 1rsrs and the angle $\widehat{rs1}$.}
  \label{fig:notations}
\end{figure}

\section{Overall strategy}
\label{sec:strategy}

We here sketch the overall strategy to prove Theorem~\ref{th:main}.

\paragraph{Coronas (Sections~\ref{sec:small} and \ref{sec:medium}).}
First, independently of the values of $r$ and $s$, we find a finite set of {\em candidates coronas} containing at least all the s- and r-coronas which can appear in a compact packing by three sizes of discs.
This is easy for s-coronas because an s-disc is surrounded by at most $6$ discs (with equality only if the $6$ discs are small).
But since an r-disc can be surrounded by arbitrarily many s-discs for $s$ small enough, there are infinitely many r-coronas.
However, we are interested not in any r-corona but in those which do appear in a compact packing by three sizes of discs.
Such a packing also has an s-corona which shall yield a lower bound on the ratio $\tfrac{s}{r}$ (Lemma~\ref{lem:ratio_minimal}), hence on the maximal number of s-discs in an r-corona and therefore on the possible r-coronas.

\paragraph{Equations (Sections~\ref{sec:polynomial} and \ref{sec:medium}).}
Then, we associate with each candidate corona an algebraic equation satisfied by any values of $r$ and $s$ which allow a compact packing containing this corona (if any).
The basic idea is that the sum of the angles between the center of consecutive discs in the corona and the center of the surrounded disc is $2\pi$.
For example, the s-corona 1rsrs (Fig.~\ref{fig:notations}) yields the equation $\widehat{1sr}+\widehat{rss}+\widehat{ssr}+\widehat{rss}+\widehat{ss1}=2\pi$.
This equation can then be transformed into an algebraic equation by taking the cosine of both sides and using some trigonometry and algebra (at the cost of some parasitic solutions that must be eliminated afterward).

\paragraph{Computations (Sections~\ref{sec:two_phases}, \ref{sec:two_smalls} and \ref{sec:one_small}).}
We can now solve the systems of algebraic equations associated with each pair of candidate s- and r-coronas.
We need to perform exact solving because we want to be sure that discs are really tangent when they should be.
Exact solving of systems of algebraic equations is not a trivial task, but it is a classic of computer algebra and various ingenious algorithms do exist.
Once $r$ and $s$ are known, we can compute all the coronas they indeed allow, {\em i.e.}, the "local structure" of packings.

\paragraph{Packings (Sections~\ref{sec:two_phases}, \ref{sec:two_smalls}, \ref{sec:one_small} and Appendix~\ref{sec:examples}).}
Last, we have to determine which sets of coronas indeed allow a compact packing of the whole plane.
Determining whether three discs allow a compact packing of the plane amounts to determine whether the $10$ different triangles which connect the centers of three mutually tangent discs can {\em tile} the plane, {\em i.e.}, whether isometric copies of theses triangles can cover the whole plane such that the intersection of any two triangles is either empty, or a vertex, or a complete edge.
This could be hard since the general issue of whether a given finite set of polygons can tile the plane is undecidable \cite{Rob71}.
This has to do with the existence of finite sets of tiles which can tile the plane but only aperiodically (as few as two tiles suffice \cite{Pen78}).
However, all the cases of Theorem~\ref{th:main} actually allow a {\em periodic} compact packing with the three sizes of discs that we managed to find by hand (it is a purely combinatorial issue once the set of possible coronas is known).

\paragraph{Computational issues.}
The polynomial systems associated with pairs  of candidates s- and r-coronas can be quite hard to solve.
For example, the pair 11rrs / 11rr$\textrm{s}^{12}$ yields two bivariate polynomials of degree $28$ and $416$, with the latter being $1.4$Mo when written to a plain text file.
Of course, since it depends on the computer and algorithm used to solve them, it is a subjective issue.
The approach in \cite{Mes20} yields similar polynomial systems, but the attempt to solve them by computing Gröbner basis with the open-source software Singular \cite{singular} succeeded only on a few systems.
In a private communication, Bruno Salvy showed us that Gröbner basis could be computed more efficiently with the algorithm \cite{Fau10}.
Still, some polynomial systems seem out of reach (moreover, the source code of this latter algorithm is non available).
With SageMath \cite{sage}, all the computations we eventually perform take less than one hour on our laptop.
For each pair $P$ and $Q$ of integer polynomials in $r$ and $s$ associated with two coronas, we proceed as follows.
\begin{enumerate}
\item
We use the {\em hidden variable} method (see, {\em e.g.}, Sec. 3.5 of \cite{CLO05} or Sec. 9.2.4 of \cite{sage2}) to compute a set which contains the roots of $P$ and $Q$.
Recall that the {\em resultant} of two univariate polynomials is a scalar which is equal to zero iff the two polynomials have a common root.
If we see $P$ and $Q$ as polynomials in $s$ with coefficients in $\mathbb{Z}[r]$ ($r$ is the "hidden variable"), then their resultant is a polynomial in $r$ which has a root $r_0$ iff $P(r_0,s)$ and $Q(r_0,s)$ have a common root $s$.
Computing the roots of this resultant thus yields the first entries of all the solution $(r,s)$ of $P=Q=0$.
We similarly get the second entries by exchanging the roles of $r$ and $s$.
The cartesian product of these two sets, with the condition $0<s<r<1$, yields the wanted set.
Wrong pairs must then be {\em filtered}.
\item
We first filter using {\em interval arithmetic}, that is, numerical computations with exact error bounds.
This allows to be sure to never reject a true solution.
Namely, we check the equation on the s- and r-coronas (the original non-algebraic equation on angles) and the existence of at least one 1-corona (similar equation on angles).
We also search for all the coronas compatible with the values of $r$ and $s$.
\item
We then filter {\em exactly}.
This is the more time-consuming part.
Actually, with the default $53$-bits precision of SageMath, all the pairs which successfuly passed the interval arithmetic filtering also passed this one.
We check all the coronas found by the previous filtering.
This time, we have to check the equations associated with coronas because no exact computation can be performed on the non-algebraic angles.
The way we proceed shall be detailed after we formalize how an algebraic equation is associated with a corona (Section~\ref{sec:polynomial}).
\end{enumerate}

\paragraph{Combinatorial issues.}
The computational tactic discussed above does not suffice to solve all the polynomial systems.
For example, computing the resultant for the polynomials associated with the pair 11rrs / 11rr$\textrm{s}^{12}$ exhausted the memory of our laptop.
Moreover, the number of candidates pairs of coronas (hence of polynomial systems) to check is huge.
Indeed, in Sections~\ref{sec:small} and \ref{sec:medium}, trying to bound this number as sharply as possible (with no more algebra than in \cite{Ken06}), we get $16805$ pairs.
Actually, a key factor to prove Theorem~\ref{th:main} is to split the compact packings in three combinatorial classes:
\begin{enumerate}
\item
The packings where no small and medium disc are adjacent (Section~\ref{sec:two_phases}).
These packings are said to be {\em large separated} (the large discs form the phase interface).
The problem primarily reduces to the already solved problem of finding compact packings with two sizes of discs.
This class yields $18$ cases of Theorem~\ref{th:main} (numbers 1--18 in Appendix~\ref{sec:examples}).
\item
The packings with two different s-coronas other than ssssss (Section~\ref{sec:two_smalls}).
This allows to use two s-coronas instead of an s- and an r-coronas with minor changes in the strategy.
The point is that there are much less s-coronas than r-coronas, and the associated equations are much simpler (mainly because there are fewer discs in an s-corona).
This class yields only one case of Theorem~\ref{th:main} (number 19 in Appendix~\ref{sec:examples}).
\item
The packings with only one s-corona other than ssssss (Section~\ref{sec:one_small}).
Here we have to consider the $16805$ pairs of s- and r-coronas and their complicated associated equations.
However, the hypothesis that there is only one s-corona (besides ssssss, which is always possible) greatly helps.
Indeed, it yields a strong (but simple) constraint on the two neighboring discs of each s-disc in a r-corona (namely, if the sequence xsy appears in the r-corona, then xry must appear in the s-corona).
Moreover, as a rule of thumb, r-coronas associated with complicated equations contain many s-discs, thus many such constraints to satisfy: this often (more than nine time out of ten) rules out a pair without any computation!
This class yields $145$ cases of Theorem~\ref{th:main} (numbers 20--164 in Appendix~\ref{sec:examples}).
\end{enumerate}

\paragraph{Other optimizations.}
Many other optimizations could be designed.
We actually introduced many ones that we later removed, because there is a trade-off between the time required to check the computations and the details required to explain the optimizations.
The above strategy, with its computational and combinatorial optimizations, allows to check all the cases in less than one hour of computation on our laptop.
The code is in Python, using powerful SageMath functions: it has only slightly more than 700 lines, comments included (see Appendix~\ref{sec:code}).

\section{Small coronas}
\label{sec:small}

An s-corona contains at most $6$ discs, with equality only if there are only s-discs.
There are thus finitely many different s-coronas.
We want here to be as precise as possible in order to reduce the number of cases to be considered afterward.
Let us define for $\vec{k}=(k_1,\ldots,k_6)\in\mathbb{N}^6$ the function
$$
S_{\vec{k}}(r,s):=k_1\widehat{1s1}+k_2\widehat{1sr}+k_3\widehat{1ss}+k_4\widehat{rsr}+k_5\widehat{rss}+k_6\widehat{sss}.
$$
This function counts the angles to pass from disc to disc in an $s$-corona.
To each s-corona corresponds a vector $\vec{k}$, called its {\em angle vector}, which satisfies
\begin{equation}\label{eq:s_corona}
S_{\vec{k}}(r,s)=2\pi.
\end{equation}
We have to find the possible values $\vec{k}$ for which the equation $S_{\vec{k}}(r,s)=2\pi$ admits a solution $0<s<r<1$.
The angles which occur in $S_{\vec{k}}(r,s)$ decrease with $s$, except $\widehat{sss}$, and increase with $r$.
This yields the following inequalities, strict except for the s-corona ssssss:
$$
S_{\vec{k}}(r,s)\leq \lim_{r\to 1\atop s\to 0}S_{\vec{k}}(r,s)=k_1\pi+k_2\pi+k_3\frac{\pi}{2}+k_4\pi+k_5\frac{\pi}{2}+k_6\frac{\pi}{3},
$$
$$
S_{\vec{k}}(r,s)\geq \inf_r\lim_{s\to r}S_{\vec{k}}(r,s)=\lim_{r\to 1\atop s\to 1}S_{\vec{k}}(r,s)=k_1\frac{\pi}{3}+k_2\frac{\pi}{3}+k_3\frac{\pi}{3}+k_4\frac{\pi}{3}+k_5\frac{\pi}{3}+k_6\frac{\pi}{3}.
$$
The existence of $(r,s)$ such that $S_{\vec{k}}(r,s)=2\pi$ thus yields the inequalities
$$
k_1+k_2+k_3+k_4+k_5+k_6< 6<3k_1+3k_2+\frac{3}{2}k_3+3k_4+\frac{3}{2}k_5+k_6,
$$
except for the s-corona ssssss ($k_1=\ldots=k_5=0$ and $k_6=6$).
An exhaustive search on computer yields $383$ possible values for $\vec{k}$.
\begin{figure}
  \centering
  \includegraphics[width=0.3\textwidth]{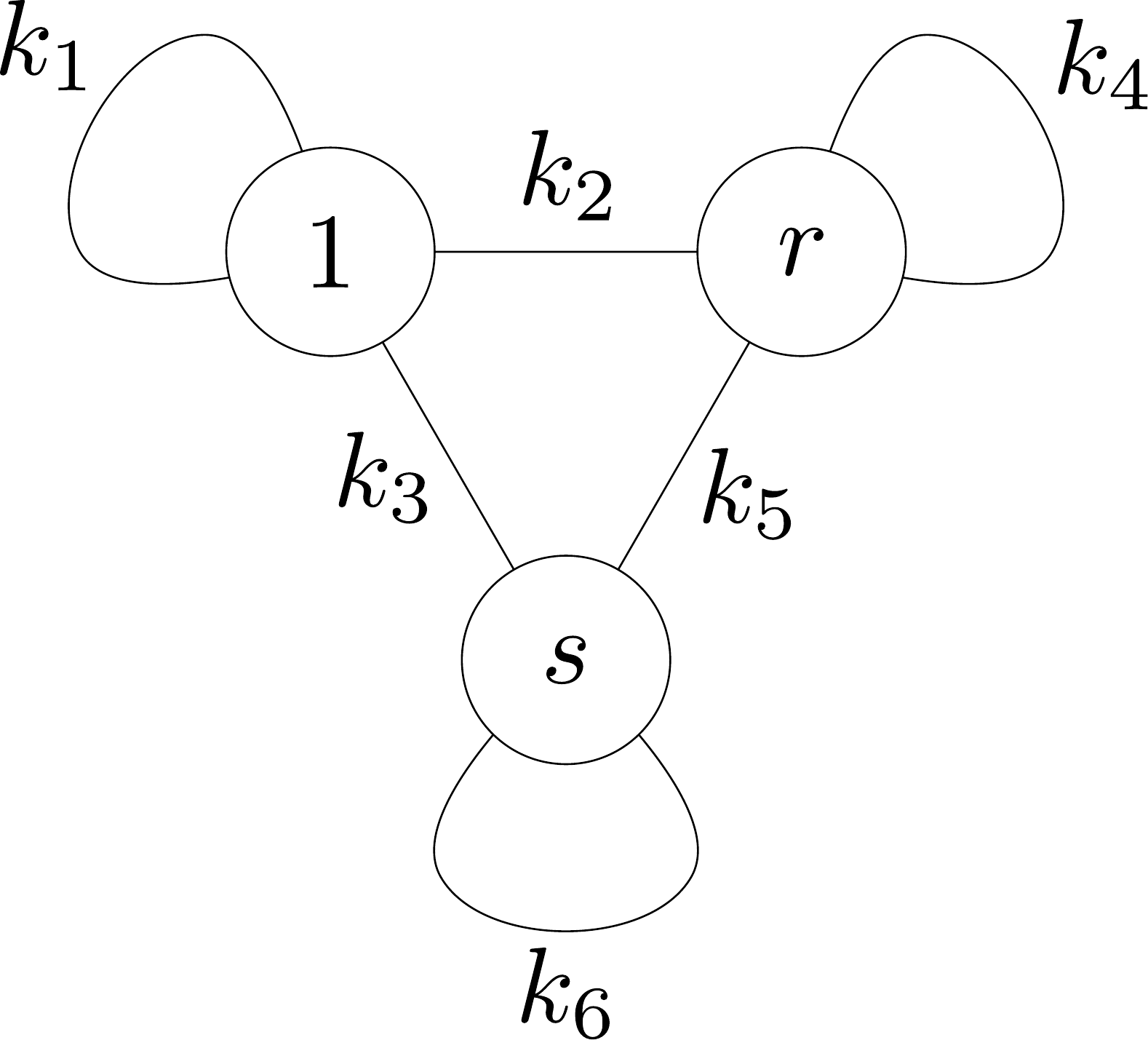}
  \caption{The cycles of this graph encode coronas associated with $\vec{k}$.}
  \label{fig:couronnable}
\end{figure}
For each of them, one shall check that there indeed exists a coding over $\{\textrm{1},\textrm{r},\textrm{s}\}$ with this angle vector\footnote{For example, $\vec{k}=(0,3,0,0,0,0)$ corresponds to three angles $\widehat{1sr}$ around an $s$-disc: this is combinatorially impossible.}.
This is the case if the graph depicted in Fig.~\ref{fig:couronnable} contains a cycle such that $\vec{k}$ counts the number of times each edge appear in this cycle.
If this cycle contains a loop, then either it contains only this loop or this loop must be accessible from another vertex.
This yields the conditions:
\begin{eqnarray}
(k_1=0) \lor (k_2\neq 0 \lor k_3\neq 0) \lor (k_2=k_3=k_4=k_5=k_6=0)\\
(k_4=0) \lor (k_2\neq 0 \lor k_5\neq 0) \lor (k_1=k_2=k_3=k_5=k_6=0)\\
(k_6=0) \lor (k_3\neq 0 \lor k_5\neq 0) \lor (k_1=k_2=k_3=k_4=k_5=0)
\end{eqnarray}
The cycle must do $k_0:=\min(k_2,k_3,k_5)$ round trips around the three vertices, and possibly some round trips between two vertices which use an even number of times the edge which connects these vertices.
This yields the conditions:
\begin{equation}
  k_2-k_0\in2\mathbb{N},\qquad
  k_3-k_0\in2\mathbb{N},\qquad
  k_5-k_0\in2\mathbb{N}.
\end{equation}
All the previous conditions eventually lead to $56$ possible angle vectors.
Tab.~\ref{tab:s_coronas} gives, for each, a coding of the corresponding corona\footnote{An angle vector could correspond to several coronas. For example, $(0,2,2,0,2,0)$ corresponds to 1rsr1s, 1r1srs, 1rs1sr and 1rs1rs. But this does not happen here.}.

\begin{table}[hbtp]
  \centering
  \begin{tabular}{cccccccccc}
    rrrrr & rrrrs & rrrss & rrsrs & rrrr & rrsss & rsrss & rrrs & rrr & rrss\\
    11111 & 1111s & 111ss & 11s1s & 1111 & 11sss & 1s1ss & 111s & 111 & 11ss\\
    1111r & 111rs & 11rss & 11srs & 111r & 1rsss & 1srss & 11rs & 11r & 1rss\\
    111rr & 11r1s & 1r1ss & 1rs1s & 11rr &  &  & 1r1s & 1rr & \\
    11r1r & 11rrs & 1rrss & 1rsrs & 1r1r &  &  & 1rrs &  & \\
    11rrr & 1r1rs & r1rss & rrs1s & 1rrr &  &  & r1rs &  & \\
    1r1rr & 1rr1s &  &  &  &  &  &  &  & \\
    1rrrr & 1rrrs &  &  &  &  &  &  &  & \\
    & r11rs &  &  &  &  &  &  &  & \\
    & r1rrs &  &  &  &  &  &  &  & \\
  \end{tabular}
  \caption{
    The $55$ possible s-coronas, besides ssssss.
    Those on the first line have no $1$-disc and those on the second line no r-disc.
    In each column, the codings are all equal if we replace each 1 by an r (this is used in Lemma~\ref{lem:ratio_minimal}).}
  \label{tab:s_coronas}
\end{table}

\section{Polynomial associated with a small corona}
\label{sec:polynomial}

We here associate with an s-corona a polynomial equation in $r$ and $s$ such that any pair $(r,s)$ for which the corona can be geometrically represented corresponds to a solution of this equation.
First take the cosines of both sides of equation \ref{eq:s_corona} and fully expand the left-hand side.
Substract $\cos(2\pi)=1$ to both sides.
This yields a polynomial equation in cosines and sines of the angles occuring in $S_{\vec{k}}$.
The cosine law then allows to replace each cosine by a rational fraction in $r$ and $s$.
Namely, the cosines of $\widehat{1s1}$, $\widehat{1sr}$, $\widehat{1ss}$, $\widehat{rsr}$, $\widehat{rss}$ and $\widehat{sss}$ are respectively replaced by
$$
1-\frac{2}{(1+s)^2},\quad
1-\frac{2r}{(r+s)(1+s)},\quad
\frac{s}{1+s},\quad
1-\frac{2r^2}{(r+s)^2},\quad
\frac{s}{r+s},\quad
\frac{1}{2}.
$$
The squares of the sines follow (same order):
$$
\frac{4s(s+2)}{(s + 1)^4},\quad
\frac{4rs(r + s + 1)}{(s + 1)^2(r + s)^2},\quad
\frac{2s + 1}{(s + 1)^2},\quad
\frac{4s(2r + s)r^2}{(r + s)^4},\quad
\frac{r(r + 2s)}{(r + s)^2},\quad
\frac{3}{4}.
$$
We express sines (same order) using auxiliary variables $X_1$ through $X_6$:
$$
\frac{2X_1}{(s + 1)^2},\quad
\frac{2X_2}{(s + 1)(r + s)},\quad
\frac{X_3}{s + 1},\quad
\frac{2rX_4}{(r + s)^2},\quad
\frac{X_5}{r + s},\quad
\frac{X_6}{2},
$$
where the squares of the auxiliary variables $X_1,\ldots,X_6$ are (same order):
$$
s(s+2),\quad
rs(r+s+1),\quad
2s+1,\quad
s(2r+s),\quad
r(r + 2s),\quad
3.
$$
This yields a polynomial system in $r$, $s$ and the $X_i$'s.

One can use the above expression of $X_i^2$ as a polynomial in $r$ and $s$ to remove any power $k\geq 2$ of $X_i$.
We then successively eliminate each $X_i$ as follows.
We write the equation $AX_i+B=0$, where $A$ and $B$ do not contain any occurence of $X_i$.
We then multiply both sides by $AX_i-B$, so that no solution is lost.
This yields the new equation $A^2X_i^2-B^2=0$, where $X_i^2$ can be replaced by its expression as a polynomial in $r$ and $s$.
In short:
$$
AX_i+B=0~\rightarrow~ A^2X_i^2-B^2=0.
$$
Doing this successively for each $X_i$ eventually yields the wanted polynomial equation in $r$ and $s$.
We simplify it by removing multiplicities of factors as well as factors which have clearly no roots $0<s<r<1$.

\begin{table}[hbtp]
  \centering
  \begin{tabular}{lll}
  11111 & $5s^4+20s^3+10s^2-20s+1$ & $0.701$ \\
  1111s & $s^4 - 10s^2 - 8s + 9$ & $0.637$ \\
  111ss & $s^8 - 8s^7 - 44s^6 - 232s^5 - 482s^4 - 24s^3 + 388s^2 - 120s + 9$ & $0.545$ \\
  11s1s & $8s^3 + 3s^2 - 2s - 1$ & $0.533$ \\
  1111 & $s^2 + 2s - 1$ & $0.414$ \\
  11sss & $9s^4 - 12s^3 - 26s^2 - 12s + 9$ & $0.386$ \\
  1s1ss & $s^4 - 28s^3 - 10s^2 + 4s + 1$ & $0.349$ \\
  111s & $2s^2 + 3s - 1$ & $0.280$ \\
  111 & $3s^2 + 6s - 1$ & $0.154$ \\
  11ss & $s^2 - 10s + 1$ & $0.101$ 
  \end{tabular}
  \caption{s-coronas without r-discs, associated polynomials and root $s\in(0,1)$.}
  \label{tab:petites_couronnes_deux_disques}
\end{table}

For the $10$ s-coronas without r-disc, we get a polynomial in $s$ (Tab.~\ref{tab:petites_couronnes_deux_disques}).
Each of the $10$ s-coronas without large disc yields the same polynomial as the corona where r has been replaced by 1, with the variable $\tfrac{s}{r}$ instead of $s$.
The $35$ reamining s-coronas yield an explicit polynomial which can be rather complex (Tab.~\ref{tab:degre_petites_couronnes} gives the degrees).
For example, the s-corona 11rs yields:
$$
r^2s^4 - 2r^2s^3 - 2rs^4 - 23r^2s^2 - 28rs^3 + s^4 - 24r^2s - 58rs^2 - 2s^3 + 16r^2 - 8rs + s^2.
$$

\begin{table}[hbtp]
  \centering
  \begin{tabular}{lr|lr|lr|lr|lr}
1r1r	&	$2$	&	11rr	&	$4$	&	1rrsr	&	$6$	&	1rsss	&	$8$	&	111rr	&	$12$	\\
1r1s	&	$2$	&	1rss	&	$4$	&	1111r	&	$7$	&	1srrs	&	$8$	&	11rrr	&	$12$	\\
1rsr	&	$2$	&	11r1s	&	$6$	&	11r1r	&	$7$	&	1srss	&	$8$	&	111rs	&	$18$	\\
111r	&	$3$	&	11rs	&	$6$	&	1rs1s	&	$7$	&	1r1rr	&	$10$	&	1rrrs	&	$18$	\\
11r	&	$3$	&	11rsr	&	$6$	&	11srs	&	$8$	&	1rrrr	&	$10$	&	11rss	&	$24$	\\
1rr	&	$3$	&	1rr1s	&	$6$	&	1r1ss	&	$8$	&	1rsrs	&	$10$	&	1rrss	&	$24$	\\
1rrr	&	$3$	&	1rrs	&	$6$	&	1rssr	&	$8$	&	1r1rs	&	$11$	&	11rrs	&	$28$	\\
  \end{tabular}
  \caption{Degree of the polynomial in $r$ and $s$ associated with each s-corona.}
  \label{tab:degre_petites_couronnes}
\end{table}

If discs of sizes $s$, $r$ and $1$ are compatible with a corona, then $(r,s)$ is a solution of the equation associated with this corona.
But some solutions of this equation may not correspond to a corona: they are parasitic solutions.
The first reason is that taking the cosines of the equations on angles can lead to solutions which are true only modulo $2\pi$.
Such cases are however easily detected just by a numerical computation.
The second reason is that, to eliminate the $X_i$'s, we multiply by factors $AX_i-B$ which introduce new solutions.
To detect whether a given pair $(r,s)$ comes from these multiplying factors, we check by interval arithmetic whether $AX_i-B$ could be equal to zero, and only if it does (it is very rare), then we substitute the exact values of $r$ and $s$ in the initial equation in $r$, $s$ and the $X_i$'s (this is more time-consuming).

\section{Medium coronas and associated polynomials}
\label{sec:medium}

Since $s$ can be arbitrarily smaller than $r$, there can be infinitely many s-discs in an r-corona.
Let us however see that this cannot happen in a compact packing.

\begin{lemma}
\label{lem:ratio_minimal}
The ratio $\tfrac{s}{r}$ is uniformly bounded from below in compact packings with three sizes of discs.
\end{lemma}

\begin{proof}
Consider a compact packings with three sizes of discs.
It contains an s-disc and not only s-discs, thus an s-corona other than ssssss.
By replacing each 1 by an r in the coding of this s-corona, Table~\ref{tab:s_coronas} shows that we get the coding of a new s-corona.
In this new corona, the ratio $\tfrac{s}{r}$ is smaller.
Indeed, the 1-discs have been ''deflated'' into r-discs, so that the perimeter of the corona decreased, whence the size of the surrounded small disc too.
But there is at most $10$ possible ratios $\tfrac{s}{r}$ for an s-corona without large discs: they correspond to the values computed in Tab.~\ref{tab:petites_couronnes_deux_disques} for compact packing with two sizes of discs (the smallest is $5-2\sqrt{6}\simeq 0.101$).
\end{proof}

This lemma ensures that the number of s-discs in an r-corona is uniformly bounded in compact packings.
There is thus only finitely many different r-coronas in compact packings.
To find them all, we proceed similarly as for s-coronas.
We define for $0<s<r<1$ and $\vec{l}=(l_1,\ldots,l_6)\in\mathbb{N}^6$ the function
$$
M_{\vec{l}}(r,s):=l_1\widehat{1r1}+l_2\widehat{1rr}+l_3\widehat{1rs}+l_4\widehat{rrr}+l_5\widehat{rrs}+l_6\widehat{srs}.
$$
We have to find the possible values $\vec{l}$ for which the equation $M_{\vec{l}}(r,s)=2\pi$ admits a solution $0<s<r<1$.
Actually, since the solution should correspond to an r-corona which occur in a packing, we can assume that it satisfies $\tfrac{s}{r}\geq \alpha$, where $\alpha$ is the lower bound on $\tfrac{s}{r}$ given by the s-coronas which occur in this packing.
The angles which occur in $M_{\vec{l}}(r,s)$ decrease with $r$ (except $\widehat{rrr}$) and increase with $s$.
This yields the following inequalities, strict except for the r-corona rrrrrr:
$$
M_{\vec{l}}(r,s)\leq  \sup_r\lim_{s\to r} M_{\vec{l}}(r,s)=l_1\pi+l_2\frac{\pi}{2}+l_3\frac{\pi}{2}+l_4\frac{\pi}{3}+l_5\frac{\pi}{3}+l_6\frac{\pi}{3},
$$
$$
M_{\vec{l}}(r,s)\geq \inf_r\lim_{\tfrac{s}{r}\to\alpha}M_{\vec{l}}(r,s)=l_1\frac{\pi}{3}+l_2\frac{\pi}{3}+l_3u_\alpha+l_4\frac{\pi}{3}+l_5u_\alpha+l_6v_\alpha,
$$
where $\widehat{1rs}$ has been bounded from below by $\widehat{rrs}$ for any $r$, and the limits $u_\alpha$ and $v_\alpha$ of $\widehat{rrs}$ and $\widehat{srs}$ when $\tfrac{s}{r}\to\alpha$ are obtained via the cosine law:
$$
u_\alpha:=\arccos\left(\frac{1}{1+\alpha}\right)
\quad\textrm{and}\quad
v_\alpha:=\arccos\left(1-\frac{2\alpha^2}{(1+\alpha)^2}\right).
$$
The existence of $(r,s)$ such that $M_{\vec{l}}(r,s)=2\pi$ thus yields the inequalities
$$
l_1+l_2+l_4+\frac{3}{\pi}(l_3u_\alpha+l_5u_\alpha+l_6v_\alpha)< 6< 3l_1+\frac{3}{2}l_2+\frac{3}{2}l_3+l_4+l_5+l_6,
$$
except for the r-corona rrrrrr ($l_1=\ldots=l_5=0$ and $l_6=6$).
We also impose
$$
l_1+l_2+l_4+\frac{1}{2}(l_3+l_5)<6,
$$
which says that an r-corona (other than rrrrrr) contains at most $5$ r- or 1-discs.
An exhaustive search on computer, which also checks whether there indeed exists a coding for each possible $\vec{l}$, eventualy yields the possible r-coronas for each value of $\alpha$.
Tab.~\ref{tab:r_coronas} gives the numbers of such coronas.

\begin{table}[hbtp]
  \centering
  \begin{tabular}{c|c|c|c|c|c|c|c|c|c}
    rrrrr & rrrrs & rrrss & rrsrs & rrrr & rrsss & rsrss & rrrs & rrr & rrss\\
    $0.701$ & $0.637$ & $0.545$ & $0.533$ & $0.414$ & $0.386$ & $0.349$ & $0.280$ & $0.154$ & $0.101$\\
    $84$ & $94$ & $130$ & $143$ & $197$ & $241$ & $272$ & $386$ & $889$ & $1654$
  \end{tabular}
  \caption{
  Each s-corona whose 1-discs have been deflated in r-discs (first line) yields a lower bound $\alpha$ on $\tfrac{s}{r}$ in any compact packing which contains it (second line), and thus an upper bound on the number of possible r-coronas in this packing (third line).
  }
  \label{tab:r_coronas}
\end{table}

Since any r-corona for some lower bound $\alpha$ also appears for a smaller $\alpha$, there are at most $1654$ different r-coronas.
Now, each s-corona in the $k$-th column of Table~\ref{tab:s_coronas} and each r-corona in the same $k$-th column of Table~\ref{tab:r_coronas} form a pair which could appear in the same compact packing.
The total number of pairs, namely $16805$, is thus the element-wise product of the vectors
$$
(8,10,6,6,6,3,3,6,4,3)
\quad\textrm{and}\quad
(84,94,130,143,197,241,272,386,889,1654).
$$
A polynomial in $r$ and $s$ is associated with each medium corona as done for the small coronas in Sec.~\ref{sec:polynomial}.
The cosine law yields the cosines of $\widehat{1r1}$, $\widehat{1rr}$, $\widehat{1rs}$, $\widehat{rrr}$, $\widehat{rrs}$ and $\widehat{srs}$ (in this order):
$$
1-\frac{2}{(1+r)^2},\quad
\frac{r}{1+r},\quad
1-\frac{2s}{(r+s)(1+r)},\quad
\frac{1}{2},\quad
\frac{r}{r+s},\quad
1-\frac{2s^2}{(r+s)^2}.
$$
The squares of the sines follow (same order):
$$
\frac{4r(r + 2)}{(r + 1)^4},\quad
\frac{2r + 1}{(r + 1)^2},\quad
\frac{4rs(r + s + 1)}{(r + 1)^2(r + s)^2},\quad
\frac{3}{4},\quad
\frac{s(2r + s)}{(r + s)^2},\quad
\frac{4r(r + 2s)s^2}{(r + s)^4}.
$$
We express sines (same order) using auxiliary variables:
$$
\frac{2X_7}{(r + 1)^2},\quad
\frac{X_8}{r + 1},\quad
\frac{2X_2}{(r + 1)(r + s)},\quad
\frac{X_6}{2},\quad
\frac{X_4}{r + s},\quad
\frac{2sX_5}{(r + s)^2},
$$
where $X_1,\ldots,X_6$ are defined in Sec.~\ref{sec:polynomial} and $X_7$ et $X_8$ have respective squares
$$
r(r+2),\quad
2r+1.
$$
The $X_i$'s elimination eventually yields a bivariate polynomial in $r$ and $s$.
These polynomials are usually larger than those associated with s-coronas (because there are generally more discs - up to $33$ - in a medium corona than in a small one).
The degree can be as high as $416$ (for the s-corona 11rr$\textrm{s}^{12}$) and the mean degree is $57.88$ (standard deviation $50.16$).
Computing them all for the $1654$ r-coronas takes 2h21min on our laptop and yields a $35$Mo file (this is actually not necessary, as we shall see).

\section{Large separated}
\label{sec:two_phases}

A compact packing by discs of sizes $1$, $r$ and $s$ is said to be {\em large separated} if it does not contain any adjacent s- and r-discs, {\em i.e.}, there is always a large disc between a small disc and a medium one.
Any s-corona has only s- and 1-discs: this yields one of the $10$ equations in $s$ given by Tab.~\ref{tab:petites_couronnes_deux_disques}.
Any r-corona has only r- and 1-discs: this also yields one of these $10$ equations, up to the replacement of each $s$ by an $r$.
With $s<r$, this yields $\binom{10}{2}=45$ possible pairs $(r,s)$.

For each of these $45$ pairs, we compute all the possible coronas, first by interval arithmetic and then exactly, as explained Section~\ref{sec:strategy}.
This takes around 4min on our laptop.
We find a periodic packing with all the sizes of discs for $18$ of these $45$ cases (numbered $1$ to $18$ in Appendix~\ref{sec:examples}).
They turn out to be exactly those with a 1-corona which contains both an s- and an r-disc and can act as a pivotal point to connect s- and r-discs of the packing.

What about the $27$ other cases?
One could imagine packings in which s- and r-discs are separated by more 1-discs.
The point of the following lemma is to show that this would need a very special 1-corona, namely one with an r-disc and three consecutive 1-discs:

\begin{lemma}
  \label{lem:biphase}
  Assume that disc sizes are such that none of the following is possible:
  \begin{enumerate}
  \item an s-corona with an r-disc or conversely;
  \item a 1-corona with both an s-disc and an r-disc;
  \item a 1-corona with an r-disc and three consecutive 1-discs.
  \end{enumerate}
  Then, no compact packing with all the sizes of discs is possible
\end{lemma}

\begin{proof}
Let us call {\em critical} a 1-disc which has a neighboring r-disc.
We claim that any neighboring 1-disc $D'$ of a critical disc $D$ is critical itself.
Indeed, consider the two discs which are neighbor of both $D$ and $D'$.
They cannot be both 1-discs, otherwise together with $D'$ they would form three consecutive 1-disc in the corona of $D$, which the third hypothesis rules out.
None of them can be an s-disc because $D$ would have an s-disc and an r-disc in its corona, which the second hypothesis rules out.
Thus, at least one is an r-disc: $D'$ is critical as claimed.

Now, assume that such a packing exists and get a contradiction.
Consider a pair of closest r- and s-discs (according the graph distance in the contact graph of the packing).
The first hypothesis ensures that they are not neighbor.
Consider a shortest path between them.
It has only 1-discs.
Walking from the r-disc to the s-disc along this path, the first disc is critical by definition, and all the following ones also according to the above claim.
But the last one has a neighboring s-disc: this contradicts the second hypothesis.
\end{proof}

\begin{figure}[hbtp]
  \includegraphics[width=\textwidth]{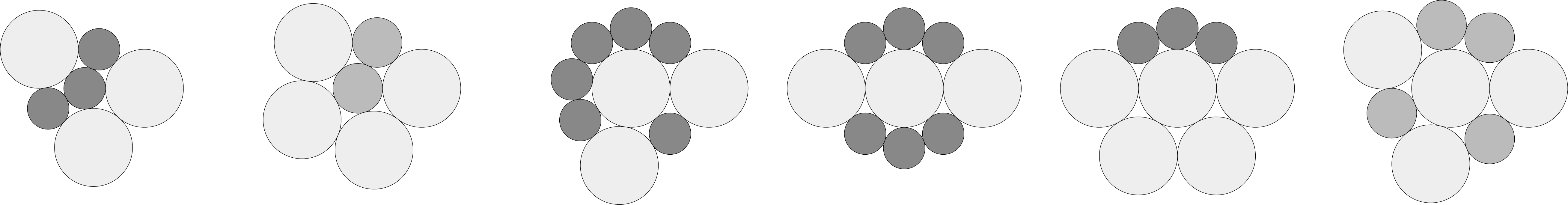}
  \caption{A case (given by all the possible coronas) ruled out by Lemma~\ref{lem:biphase}.}
  \label{fig:biphase_impossible_1}
\end{figure}

Lemma~\ref{lem:biphase} rules out $20$ of the $27$ remaining cases.
Actually, Lemma~\ref{lem:biphase} is purely combinatorial and still holds under any permutation of the disc types.
In particular, exchanging types s and r rules out $13$ of these $27$ remaining cases: $9$ are also ruled out by the original lemma, while $4$ are newly ruled out.
On the whole it rules out $20+4$ of the $27$ cases, so that there are still three remaining cases, depicted in Fig.~\ref{fig:biphase_impossible_2}.
They satisfy the two first hypothesis of Lemma~\ref{lem:biphase} but not the third one.
The following lemma rules them out:

\begin{figure}[hbtp]
  \includegraphics[width=\textwidth]{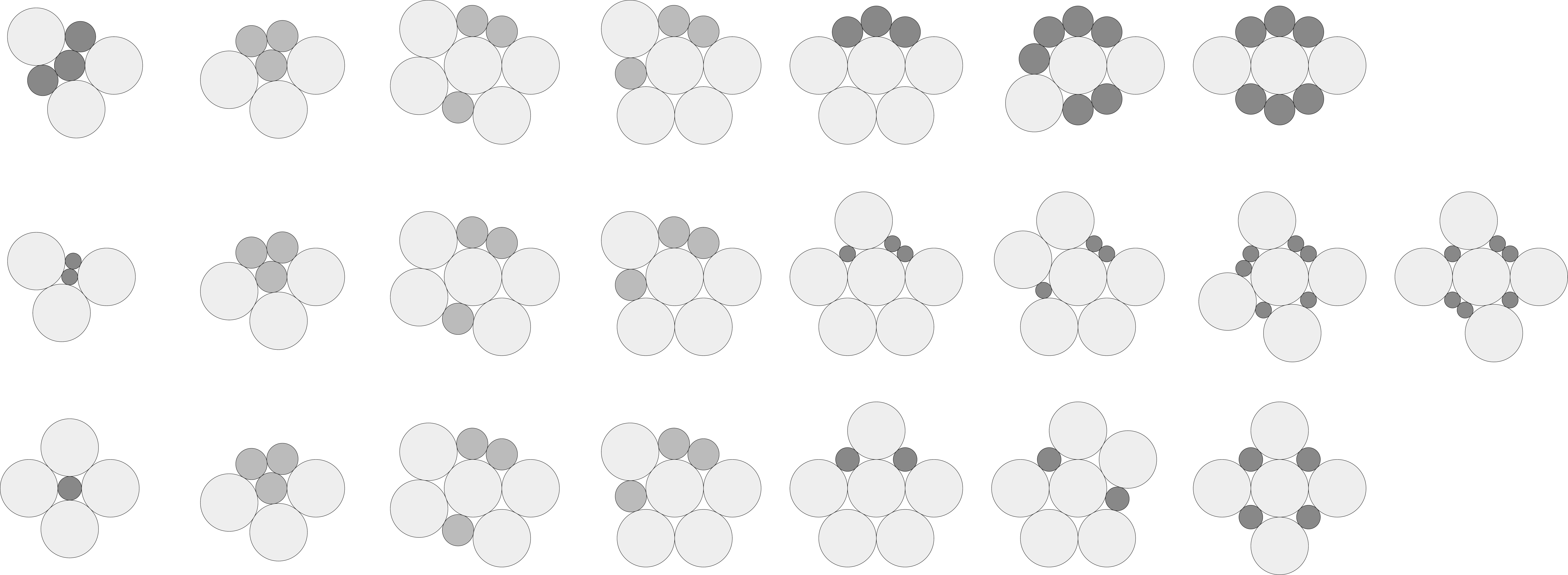}
  \caption{Three cases (one per line, given by all the possible coronas) which do not allow any compact packing of the plane with all the three sizes of discs.}
  \label{fig:biphase_impossible_2}
\end{figure}

\begin{lemma}
  \label{lem:biphase2}
  None of the case depicted in Figure~\ref{fig:biphase_impossible_2} allow a compact packing of the plane with all the three sizes of discs.
\end{lemma}

\begin{proof}
We keep the notion of critical disc introduced in the proof of Lemma~\ref{lem:biphase}.
Here, only the first and second hypothesis of Lemma~\ref{lem:biphase} is fulfilled because each case has the corona 111rr1r (fourth position on each line in Fig.~\ref{fig:biphase_impossible_2}).
A 1-disc  to a critical disc is no more necessarily critical itself.
However, we claim that if a 1-disc has a critical neighbor, then all its neighbors are critical discs (see below).
The proof of Lemma~\ref{lem:biphase} then goes the same way.
On a path of 1-discs going from an r-disc to an s-disc, the first 1-disc is critical (the initial r-disc is a neighbor) and any non-critical one (if any) appears between two critical ones.
In particular, the last 1-disc is critical (the final s-disc is a neighbor), which contradicts the second hypothesis of Lemma~\ref{lem:biphase} and proves Lemma~\ref{lem:biphase2}.
  
Let us prove the above claim (disc labels refer to Fig.~\ref{fig:biphase_impossible_3}).
Consider a 1-disc {\bf a} which is not critical but has a critical neighbor {\bf b}.
The neighbors {\bf c} and {\bf d} of both {\bf a} and {\bf b} are neither s-disc because of the second hypothesis of Lemma~\ref{lem:biphase}, nor r-disc because {\bf a} is not critical.
They are thus 1-discs.
Only one of the coronas depicted in Fig.~\ref{fig:biphase_impossible_2} is possible for {\bf b}: 111rr1r.
This corona ensures that the neighbor of {\bf b} and {\bf c} is an r-disc, as well as the neighbor of {\bf b} and {\bf d}.
In particular, both {\bf c} and {\bf d} are critical, and the second hypothesis of Lemma~\ref{lem:biphase} ensures that they do not have a neighboring s-disc.
The neighbor {\bf e} of {\bf c} and {\bf a} is thus a 1-disc, as well as the neighbor {\bf f} of {\bf d} and {\bf a}.
The corona of {\bf a} now contains $5$ 1-discs: according to Fig.~\ref{fig:biphase_impossible_2}, a sixth 1-disc {\bf g} completes this corona.
The three consecutive neighbors {\bf e}, {\bf a} and {\bf b} of {\bf c} ensure that its corona is 111rr1r, so that the neighbor of {\bf c} and {\bf e} is an r-disc.
The same holds for each disc of the corona of {\bf a}: the neighbor (other than {\bf a}) of two consecutive discs of the corona is an r-disc.
This proves that all the discs of the corona of {\bf a} are critical.
\end{proof}

\begin{figure}[hbtp]
  \centering
  \includegraphics[width=\textwidth]{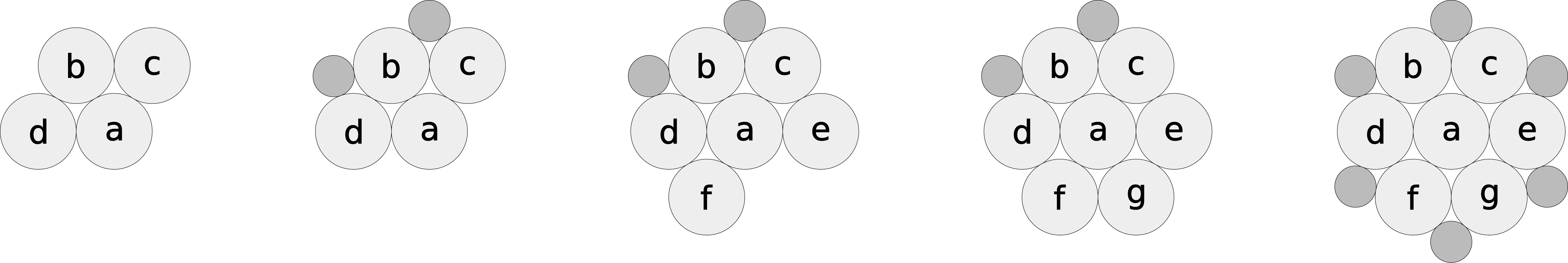}
  \caption{In the three cases whose coronas are depicted in Fig.~\ref{fig:biphase_impossible_2}, any $1$-disc of a (hypothetical) packing is either critical or surrounded only by critical discs.}
  \label{fig:biphase_impossible_3}
\end{figure}

\section{Two small coronas}
\label{sec:two_smalls}

We consider compact packings for pairs $(r,s)$ which allow at least two different s-coronas other than ssssss.
We want to find $(r,s)$ by solving the polynomial system associated with these two s-coronas.
Since there are only $55$ s-coronas, there are at most $\binom{55}{2}=1485$ such systems.
Moreover, the equations associated with s-coronas are much simpler.

If two different s-coronas contain only s- and 1-discs, then they characterize two different values of $s$.
Similarly, if they contain only s- and r-discs, then they characterize two different values of $\tfrac{s}{r}$.
We thus rule out these pairs.
We get $1395$ systems of two polynomial equations in $r$ and $s$.
We then follow the computational tactic explained Section~\ref{sec:strategy}.

The hidden variable method fails on the three pairs rrs / 11rss, rrss / 1111r and 1rrss / 1111r.
The computation of the algebraic roots of the resultants indeed raises (various) exceptions.
This could be overcome, but we just compute roots with interval arithmetic instead.
Indeed, this yields $179$ pairs $(r,s)$ of intervals which are all ruled out by the interval arithmetic filtering.

The hidden variable method works on the $1392$ other pairs and yields $13239$ pairs $(r,s)$, with $0<s<r<1$.
Arithmetic filtering rules out all but $37$ of them.
The exact filtering validate all these pairs and yields all the possible coronas in each case.
We find a periodic packing for only one of these $37$ cases, namely the pair of s-coronas 1srrs / 1s1ss (number $19$ in Appendix~\ref{sec:examples}).
Its coronas are depicted in Fig.~\ref{fig:2pc_possible_1}.

\begin{figure}[hbtp]
  \centering
  \includegraphics[width=\textwidth]{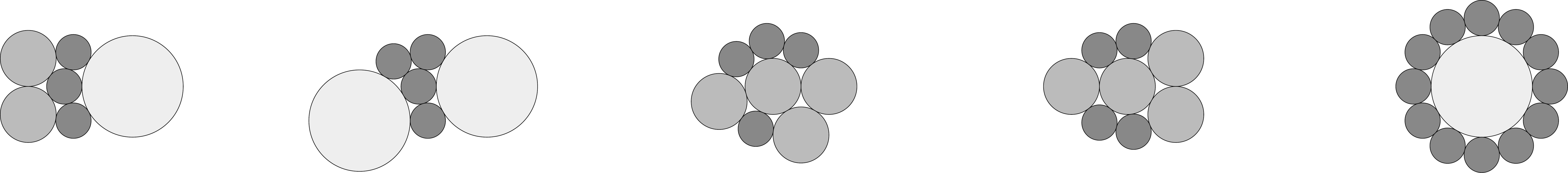}
  \caption{The only case which allows two small coronas and a valid packing.}
  \label{fig:2pc_possible_1}
\end{figure}

Actually, the case 1srrs / 1s1ss is the only one with an s-corona which contains both an r-disc and a 1-disc and can act as a pivotal point to connect r- and 1-discs of the packing.
In the $36$ other cases, r- and 1-discs are never in contact: only s-separated packings are possible.
The situation is thus very similar to Section~\ref{sec:two_phases} and Lemma~\ref{lem:biphase} again helps.
Precisely, the permutation (s,r,1) on the disc types in Lemma~\ref{lem:biphase} rules out $32$ of these $36$ cases and the permutation (s,1,r) rules out the four last ones (as well as $28$ already ruled out).


Actually, even if the only remaining case allows the two small coronas 1srrs and 1s1ss, only the first one can appear in a valid packing:

\begin{proposition}
  There is no compact packing with three sizes of discs which contains both the small coronas 1srrs and 1s1ss.
\end{proposition}

\begin{proof}
  There is no other small corona (besides ssssss) compatible with the values of $r$ and $s$ characterized by these two coronas, and the only compatible large corona is $\textrm{s}^{12}$.
  Assume that a small corona 1s1ss appears (Fig.~\ref{fig:2pc_possible_2}, left).
  Each of the two large discs of this corona must be surrounded by small discs (Fig.~\ref{fig:2pc_possible_2}, center).
  A small disc with a factor s1ss in its corona (one of them is pointed in Fig.~\ref{fig:2pc_possible_2}, left) must have a corona 1s1ss.
  This yields a third large disc, also surrounded by small discs (Fig.~\ref{fig:2pc_possible_2}, right).
  This argument can be repeated on each small disc with factor s1ss in its corona (pointed in Fig.~\ref{fig:2pc_possible_2}, right), forbidding any medium disc to ever appear.
\end{proof}

\begin{figure}[hbtp]
  \centering
  \includegraphics[width=\textwidth]{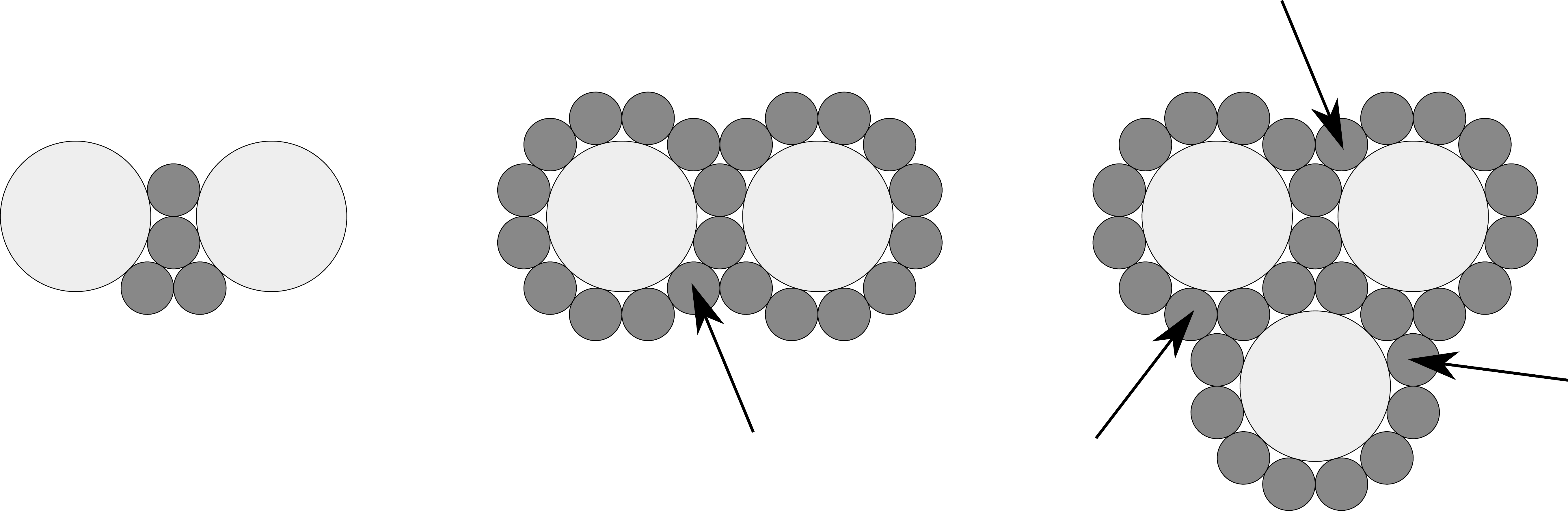}
  \caption{
  The small coronas 1srrs and 1s1ss are possible, but if the latter appears then the packing is made of large discs surrounded by small ones and centered on a triangular grid: it does not contain any medium disc.}
  \label{fig:2pc_possible_2}
\end{figure}

\section{One small corona}
\label{sec:one_small}

The remaining class is the main one.
The packings contain adjacent s- and r-discs (otherwise they are large separated) and allow only one s-corona (besides ssssss).
In particular, any such packing has an r-corona which contains an s-disc.
We shall here always assume that in the pair of s- and r-corona used to get the equations in $r$ and $s$, the r-corona contains an s-disc.
This indeed yields a strong combinatorial constraint.
Consider, for example, the pair 11rrs / 11rr$\textrm{s}^{12}$ (computing the resultants for the associated equations exhausted the memory of our laptop, as mentioned in Section~\ref{sec:strategy}).
Whenever, in the r-corona, there is an s-disc between an x-disc and an y-disc, the corona of this s-disc must contain an r-disc between an x-disc and an y-disc.
Here, this shows that the coding of the s-corona must contain the factors rrs, srs and sr1.
However, neither srs nor sr1 appear in 11rrs.
This pair can thus be ruled out without any further computation.

Recall that an angle vector can generally have different codings (not for s-coronas, however).
In the above example,  11rr$\textrm{s}^{12}$ is the unique coding of $(1,1,1,1,1,11)$.
But the angle vector $(0, 0, 4, 0, 6, 10)$, for example, admits $1022$ different codings.
To rule a case out, each of these codings must be checked.
Formally, a small angle vector $\vec{k}$ is said to {\em cover} a medium angle vector $\vec{l}$ if there exists a coding of $\vec{l}$ such that, for any factor xsy of this coding, the (unique) coding of $\vec{k}$ contains xry.
One also says that $\vec{k}$ {\em pre-covers} $\vec{l}$ if there exists a coding of $\vec{l}$ such that, for any factor xs of this coding, the coding of $\vec{k}$ contains xr.
This latter condition is weaker but it can be directly checked on the angle vector (which precisely counts the factors of length $2$), hence faster.

Consider the $16805$ candidates pairs small/medium angle vectors.
Keeping only those where the s-corona contains an r-disc and the r-corona an s-disc reduces to $12265$ pairs.
Checking the pre-covering condition reduces to $2889$ pairs.
Checking then the covering condition reduces to $803$ pairs.
These $803$ pairs correspond to $192$ different medium coronas out of $1654$ initially, and the associated equations are generally much simpler\footnote{
Computing the $192$ polynomials takes $2$min on our laptop and yields a $256$Ko file.
The mean degree is $14$, the maximum one $80$ for 11rrsrss.
This has to be compared with the statistics provided for all the $1654$ polynomials at the end of Section~\ref{sec:medium}.}.

The hidden variable method fails on $23$ of these $803$ pairs.
For $8$ of them, the computation of the algebraic roots of the resultants raises (various) exceptions.
Again, we compute the roots with interval arithmetic and filter them with interval arithmetic.
Only one pair remains: 1rr1s / 11rrs.
Lemma~\ref{lem:1pc_impossible_2} (page~\pageref{lem:1pc_impossible_2}) rules out this pair with a combinatorial argument.
The $15$ other pairs fail because the two resultants are zero, that is, there is a continuum of possibles pairs $(r,s)$.
We rule them out as follows.
In $13$ out of these $15$ cases, neither the s- nor the r-corona does contain a 1-disc.
In order to allow a packing with three sizes of discs, there must be an s-corona or an r-corona with a 1-disc.
If there is such a corona with an s-disc and an r-disc adjacent, then this case appears in the initial list of $16805$ candidates pairs and is handled elsewhere.
Otherwise, there is either an s-corona with only 1- and s-discs, or an r-corona with only 1- and r-discs.
In the former case, the corona appears in Tab.~\ref{tab:petites_couronnes_deux_disques} and characterizes $s$.
The same holds in the latter case, with $r$ instead of $s$.
In both cases, and for each of these $10$ new possible coronas, we consider the polynomial system formed by the equation associated to the new corona and the two ones associated with the initial pair of corona.
We find (using Gröbner basis) that none of these systems does have a solution.
These 13 cases are thus ruled out.
The two remaining cases are 1rr / 1r1srs and 11r / 111s1s.
The first case is not possible: the s-corona 1rr tells us that the s-discs must be in the interstices between one 1-disc and two r-discs, but then the r-corona 1r1srs should be an r-corona 1r1r with s-discs in the interstices, but the r-corona 1r1r is forbidden ($r$ should be arbitrarily small).
In the second case, the s-corona 11r tells us that the s-discs must be in the interstices between two 1-discs and one r-disc, but then the r-corona 111s1s is not completed because one could add s-discs between the central r-discs and two consecutives 1-discs, yielding the r-corona 1s1s1s1s.
This is thus a subcase of 11r/1s1s1s1s, which is handled elsewhere.
All the cases where the hidden variable method fails are thus ruled out.

The hidden variable method works on the $780$ other pairs and yields $56968$ algebraic pairs $(r,s)$, with $0<s<r<1$.
Arithmetic filtering rules out all but $202$ of them.
The exact filtering rules out $27$ of these pairs.
One is the pair 1srrs / rrsrsss which turns out to also allow the s-corona 1s1ss: this is actually the only one case which allows a packing Section~\ref{sec:two_smalls} (number $19$ in Appendix~\ref{sec:examples}).
The $26$ other ones are detected as ''duplicates'' once all the possible coronas are computed, that is, there is another pair which yields the same values of $(r,s)$ and the same set of possible coronas.
The $175$ other cases are validated and all the possible coronas computed.
Computing exactly all the large coronas is the most time-consuming part of all the paper ($30$min on our laptop, {\em i.e.}, half of the total computation time).

We find a periodic packing for $145$ of these $175$ cases (numbers $20-164$ in Appendix~\ref{sec:examples}).
We shall rule out the other cases by two combinatorial lemmas.
The first one rules out $24$ cases:

\begin{lemma}
  \label{lem:1pc_impossible_1}
  If a compact packing contains an s-corona 1rss, 11rss, 1rrss or 1srss, then it must contain another s-corona (other than ssssss).
\end{lemma}

\begin{proof}
  The proof does not rely on the value of $r$ and $s$. 
  The four cases are similar and depicted in Fig.~\ref{fig:1pc_impossible_1}.
  Polygons around letters link the text and the figure.
  Consider an \fond{4}{s}-disc and an \fond{5}{s}-disc in its corona.
  The \fond{5}{s}-corona is determined and yields in the corona of the \fond{6}{s}-disc a factor (1ss1 in the first case, rssr in the other ones) which appears neither in these four \fond{4}{s}-coronas nor in ssssss.
\end{proof}

\begin{figure}[hbtp]
  \centering
  \includegraphics[width=\textwidth]{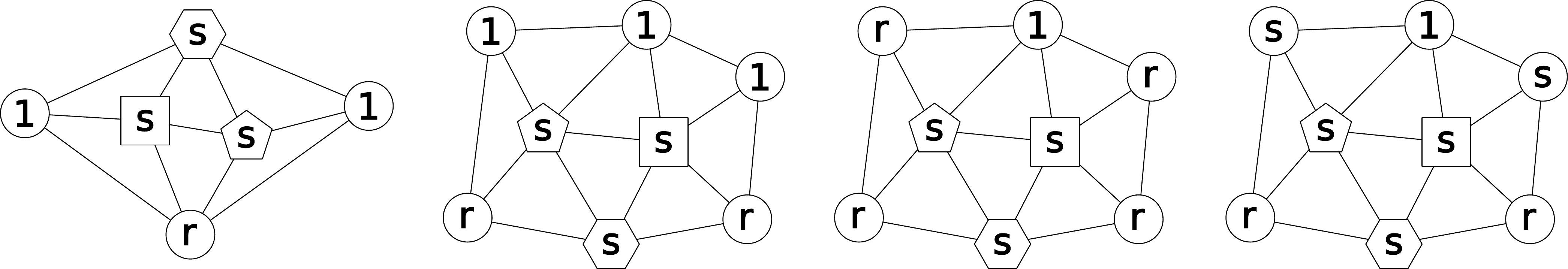}
  \caption{Four s-coronas which ensure that there must be another s-corona.}
  \label{fig:1pc_impossible_1}
\end{figure}

The second lemma rules out the $6$ last pairs, as well as the pair 1rr1s/11rrs, for which no exact filtering was performed (because the hidden variable method yields only interval for $r$ and $s$):

\begin{lemma}
  \label{lem:1pc_impossible_2}
  The small/medium coronas 1rsrs/1rr1ss, 11rr/11rrs, 1rr1s/11rrs, rrrrr/1rsrsr, rrrrs/11rssr, rrrss/11rssr and rrrs/11rssr do not allow a packing with all the three sizes of discs (without any other small corona, except ssssss).
\end{lemma}

\begin{proof}
  We check the $7$ cases one by one, using the values of $r$ and $s$ to determine (with a computer) all the possible coronas and then relying on a short combinatorial argument (illustrated next to it).
  
\noindent
\begin{minipage}{0.73\textwidth}
  {\bf 1rsrs/1rr1ss.}
  The values of $r$ and $s$ allow no other r-corona.
  In the \fond{4}{s}-corona, the \fond{4}{r}-disc has three neighbor s-discs.
  This is incompatible with the r-corona 1rr1ss.
\end{minipage}
\hfill
\begin{minipage}{0.17\textwidth}
\includegraphics[width=\textwidth]{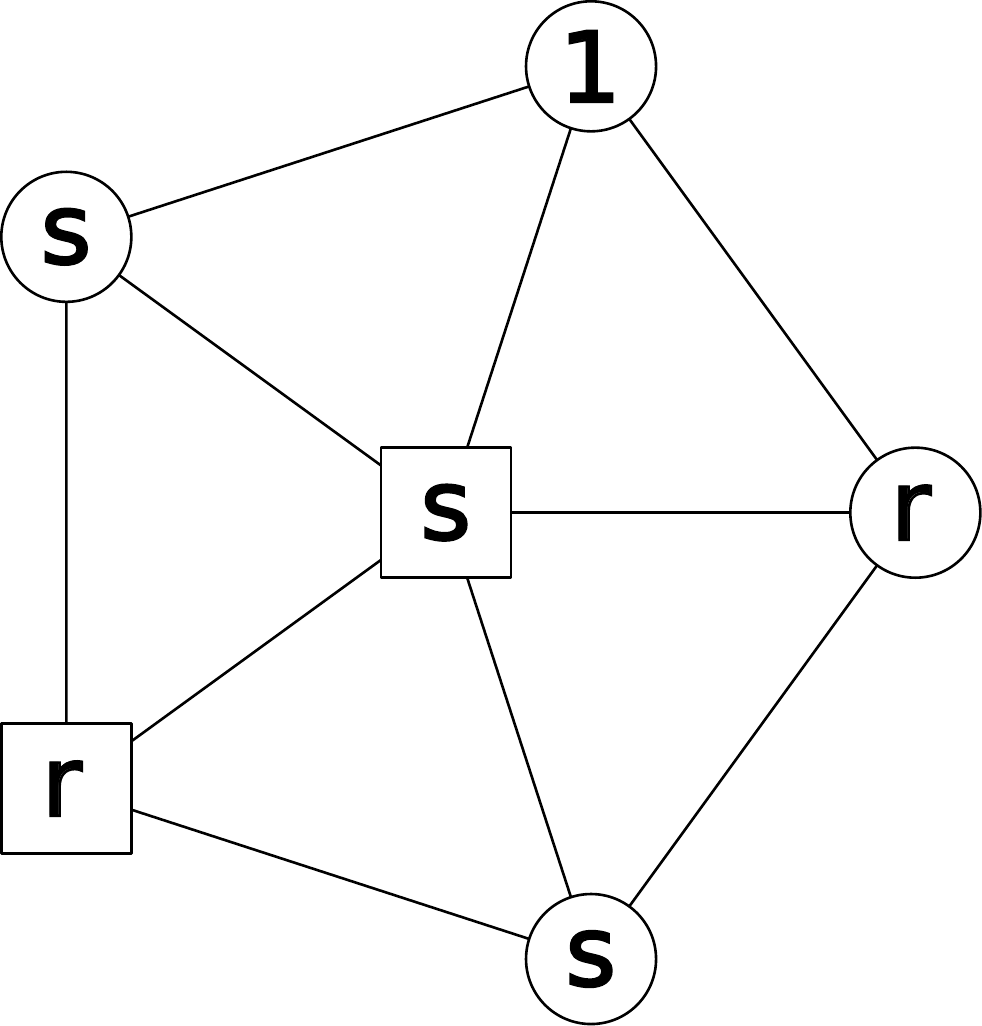}
\end{minipage}
\bigskip
\noindent
\begin{minipage}{0.24\textwidth}
\includegraphics[width=\textwidth]{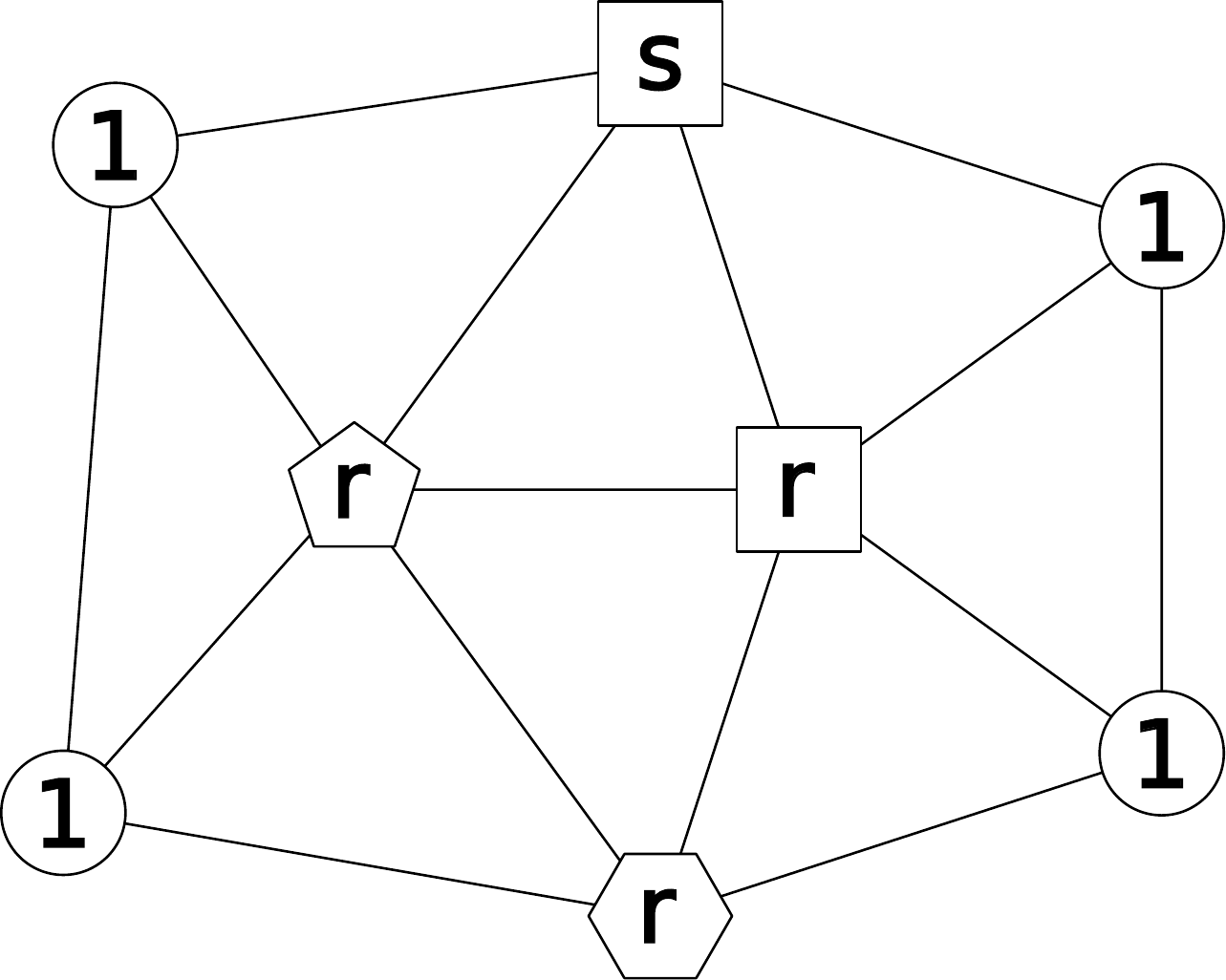}
\end{minipage}
\hfill
\begin{minipage}{0.66\textwidth}
  {\bf 11rr/11rrs.}
  The values of $r$ and $s$ allow no other r-corona.
  In the corona of the \fond{4}{r}-disc, the corona of the \fond{5}{r}-disc enforces the factor 1rr1 in the corona of the \fond{6}{r}-disc.
  This is incompatible with the r-coronna 11rrs.
\end{minipage}
\bigskip
\noindent
\begin{minipage}{0.68\textwidth}
  {\bf 1rr1s/11rrs.}
  The values of $r$ and $s$ allow no other r-corona.
  In the corona of the \fond{4}{r}-disc, the corona of the \fond{5}{r}-disc enforces the factor srrs in the corona of the \fond{6}{r}-disc.
  This is incompatible with the r-corona 11rrs.
\end{minipage}
\hfill
\begin{minipage}{0.22\textwidth}
\includegraphics[width=\textwidth]{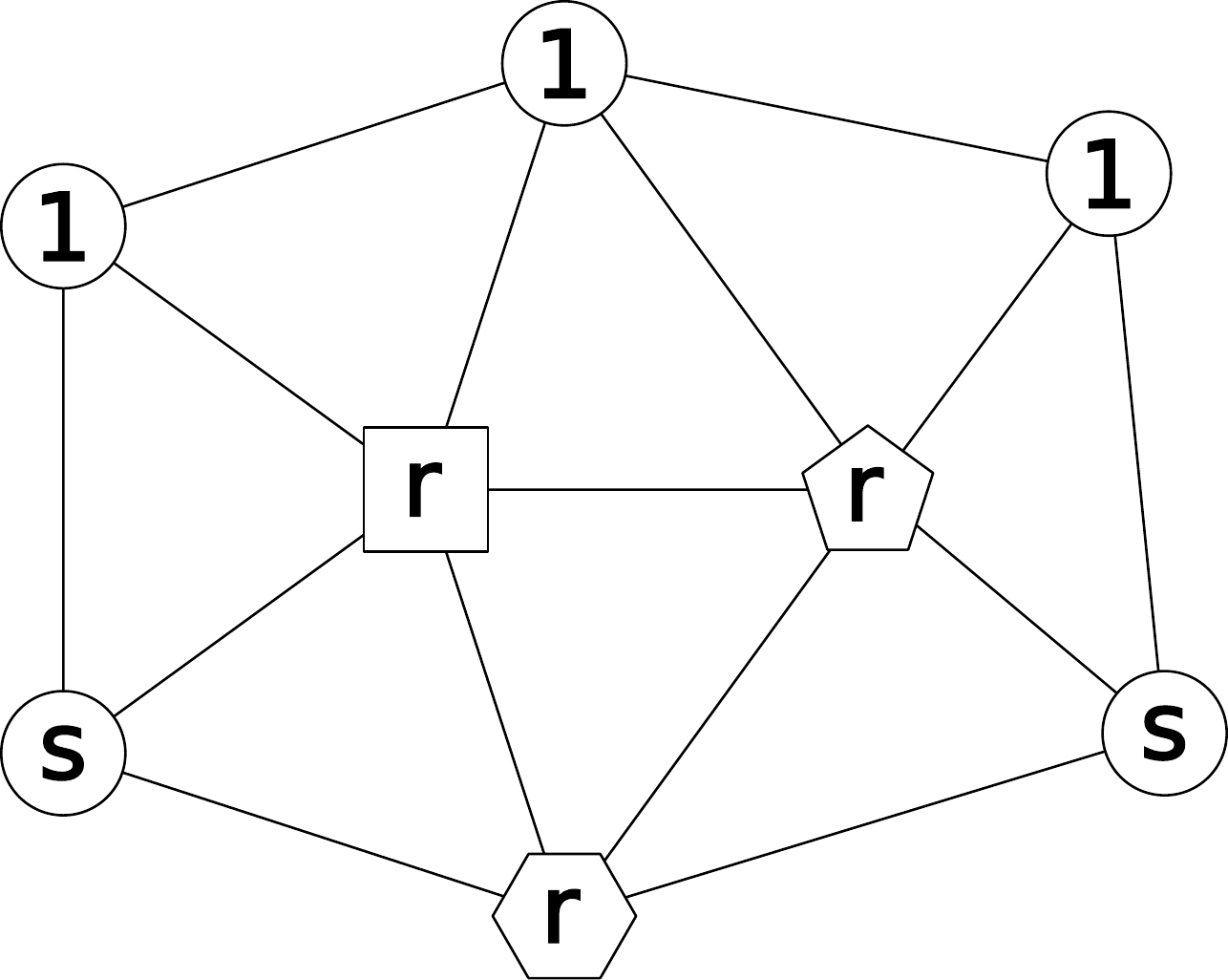}
\end{minipage}
\bigskip
\noindent
\begin{minipage}{0.45\textwidth}
\includegraphics[width=\textwidth]{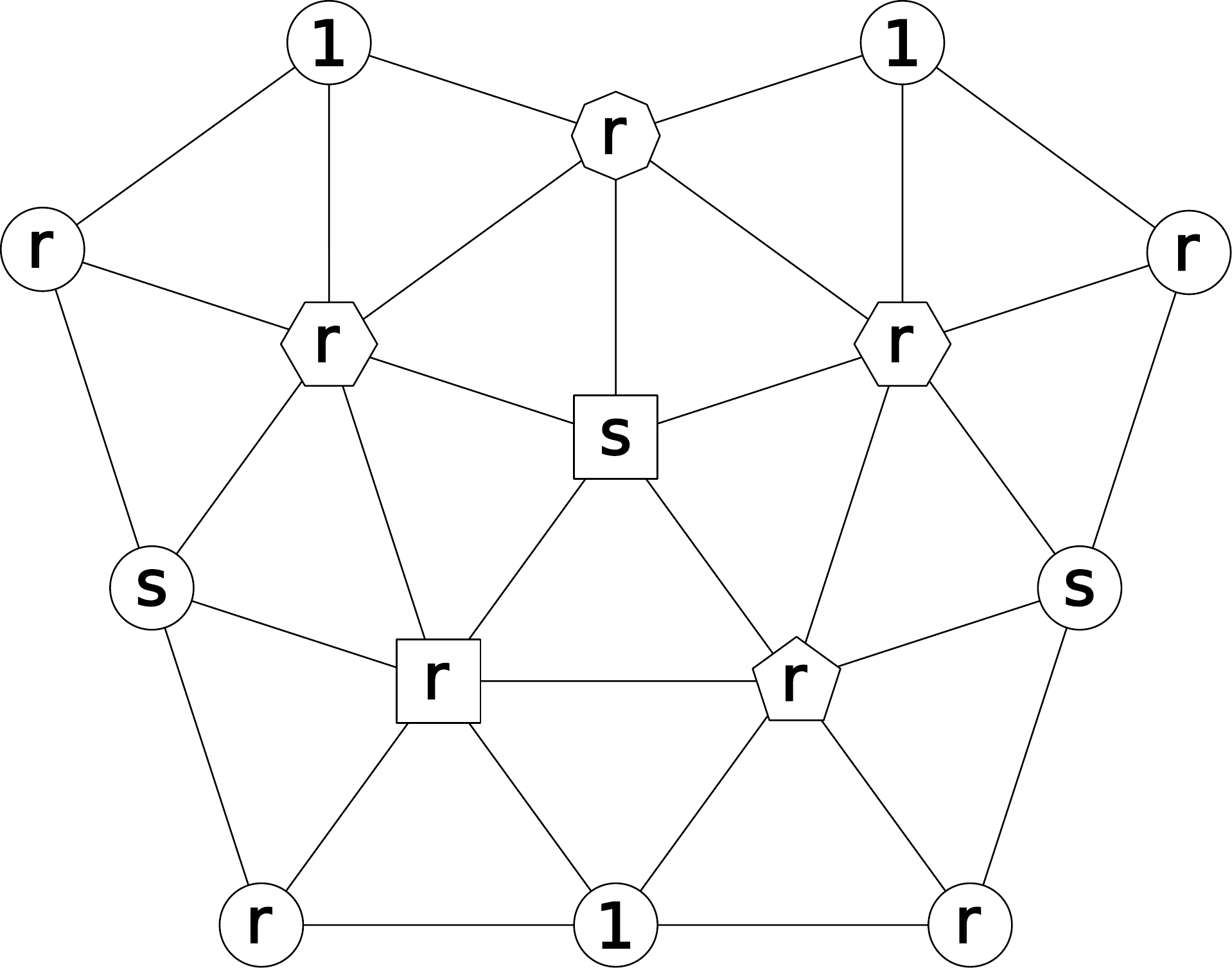}
\end{minipage}
\hfill
\begin{minipage}{0.45\textwidth}
  {\bf rrrrr/1rsrsr.}
  The values of $r$ and $s$ allow no other r-corona.
  In the \fond{4}{s}-corona, there two symmetric ways to draw the corona of the \fond{4}{r}-disc.
  Once it is done, the corona of the neighboring \fond{5}{r}-disc is determined.
  The coronas of the two \fond{6}{r}-discs are also determined.
  This enforces two 1-discs in the \fond{8}{r}-corona.
  This is incompatible with the r-corona 1rsrsr.
\end{minipage}
\bigskip
\noindent
\begin{minipage}{0.7\textwidth}
  {\bf rrrrs/11rssr.}
  The values of $r$ and $s$ also allow the r-coronas rsrsrss and 1111r, as well as the 1-corona 1r1r1rr.
  Exchanging types r and 1 in Lemma~\ref{lem:biphase} ensures that any packing with three sizes of discs must contain the r-corona 11rssr.
  Consider an \fond{4}{r}-corona 11rssr.
  The corona of the \fond{5}{r}-disc must be 11rssr and its position is determined.
  This enforces a factor srrs in the \fond{4}{s}-corona.
  This is incompatible with the s-corona rrrrs.
\end{minipage}
\hfill
\begin{minipage}{0.2\textwidth}
\includegraphics[width=\textwidth]{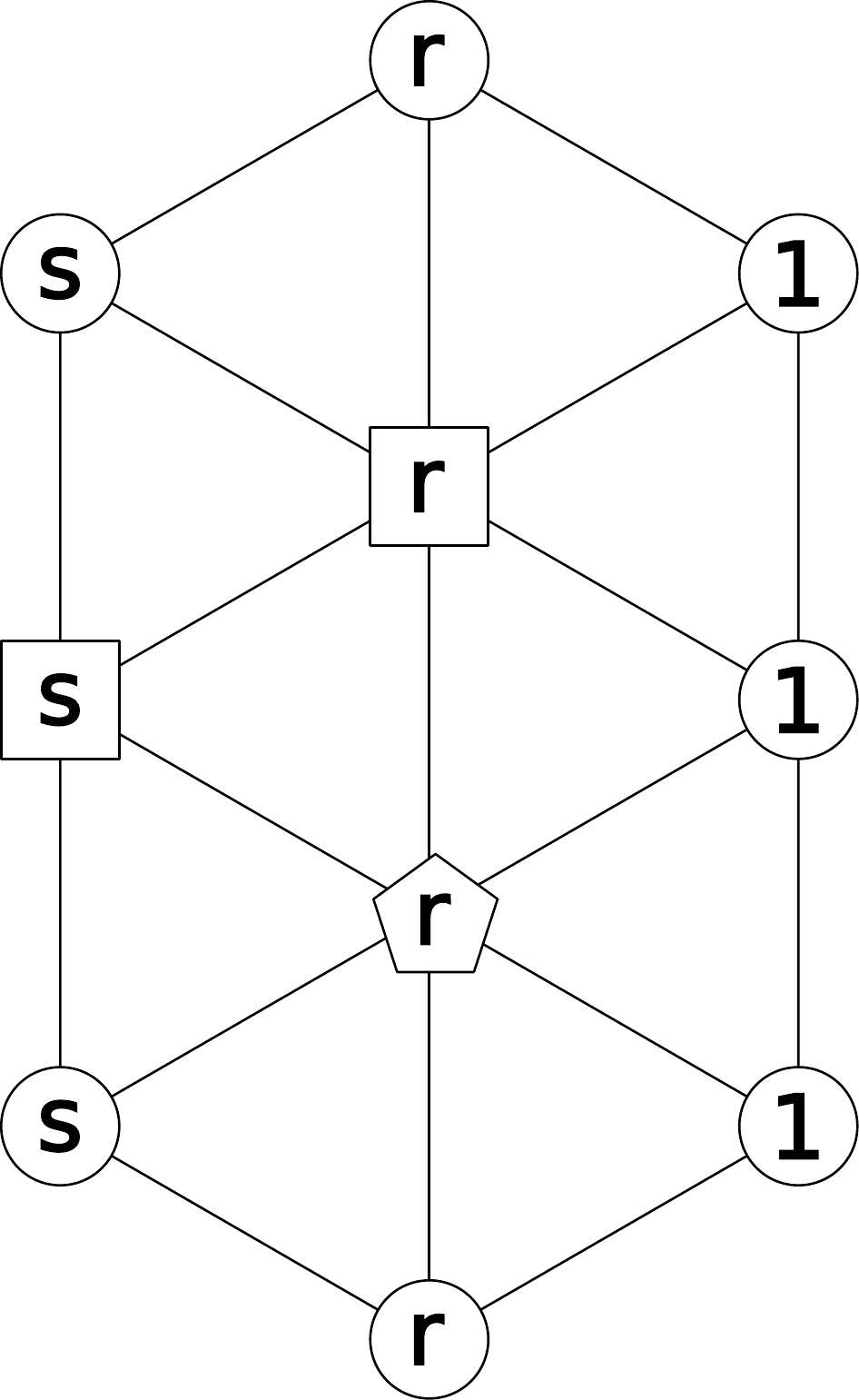}
\end{minipage}
\bigskip
\noindent
\begin{minipage}{\textwidth}
  {\bf rrrss/11rssr.}
  The values of $r$ and $s$ also allow the r-coronas rrsrrss, rrrsrss and 111rr, as well as the 1-coronas 11r11rr and 111r1rr.
  The argument works exactly as for the previous case and forces an impossible factor srrs in the s-corona.
\end{minipage}
\bigskip
\noindent
\begin{minipage}{0.3\textwidth}
\includegraphics[width=\textwidth]{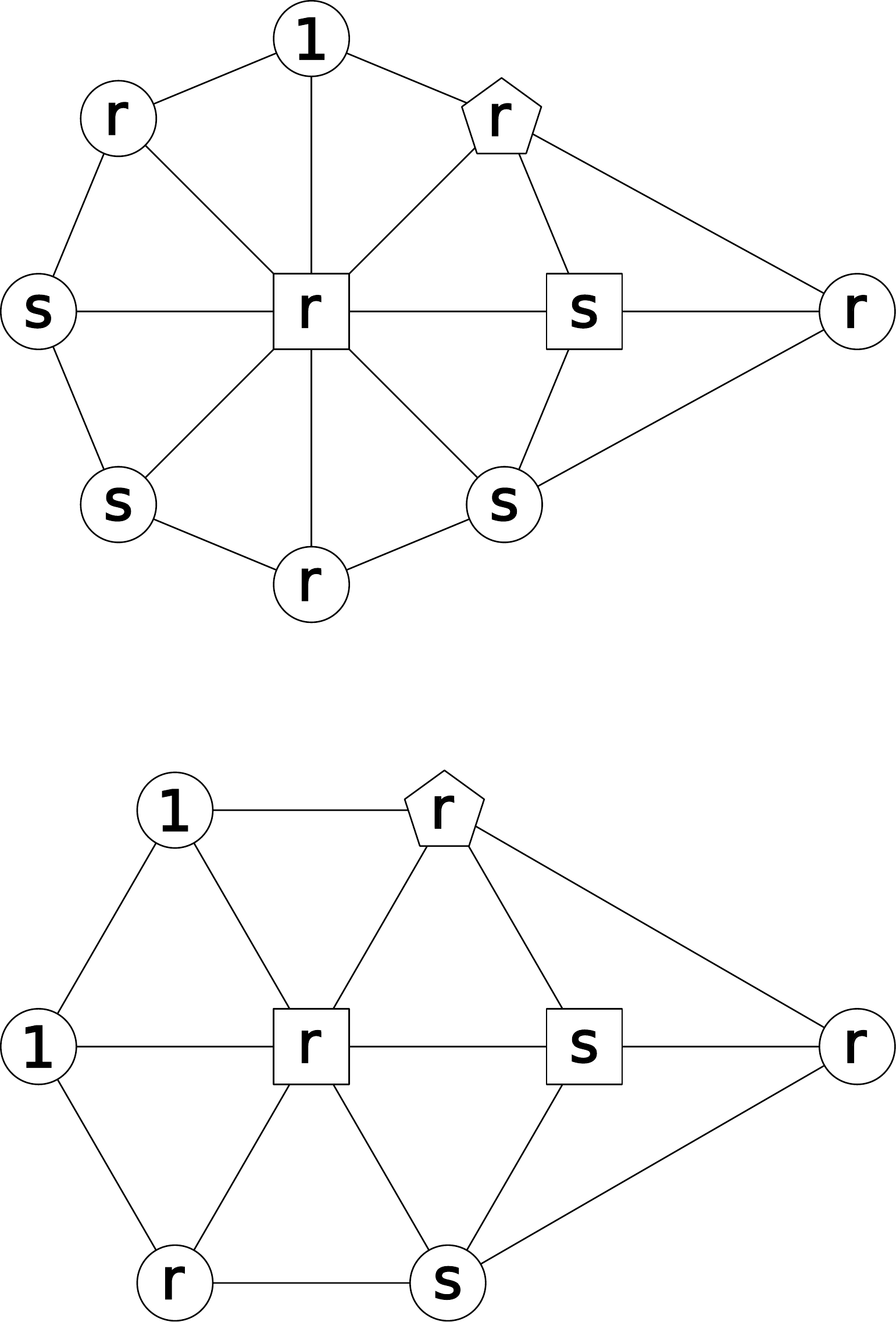}
\end{minipage}
\hfill
\begin{minipage}{0.6\textwidth}
  {\bf rrrs/11rssr.}
  The values of $r$ and $s$ also allow the r-coronas rsrsrsrsss, rsrsrssrss, rsrssrsrss, rrrrsrss, rrrsrrss, rrrsrss, 1rsrsssr, 1rssrssr and 111r.
  The r-coronas 1rsrsssr and rsrsrsrsss are actually impossible because the s-corona rrrs forbid three consecutive s-discs in any r-corona.
  Exchanging types r and 1 in Lemma~\ref{lem:biphase} ensures that any packing with three sizes of discs must contain an r-corona 1rssrssr or 11rssr.
  In both \fond{4}{r}-coronas, the \fond{4}{s}-corona rrrs yields a factor 1rsr in the \fond{5}{r}-corona.
  This is incompatible with the allowed r-coronas (1rsrsssr has been discarded).
\end{minipage}

\end{proof}

\appendix
\section{Compact packings}
\label{sec:examples}

Figure~\ref{fig:repartition} acts as a map, showing the distribution of the $164$ cases of Theorem~\ref{th:main}.
A periodic compact packing is then depicted for each case.
The letter in brackets refers to the {\em type} of the compact packing, see Appendix~\ref{sec:classification}.
The codings of a small and a medium corona from which the values of $r$ and $s$ can be computed is given top-right of each picture.
Numbers $1$--$18$ are large separated packings, number  $19$ is the unique which admits two small coronas and numbers $20$--$164$ are those which admit a unique small corona (besides ssssss, as usual).
In some cases, the small discs are very small and barely visible (numbers $32$, $37$--$40$, $41$--$44$): the small corona however indicates where they are.

\begin{figure}[hbtp]
\centering
\includegraphics[width=\textwidth]{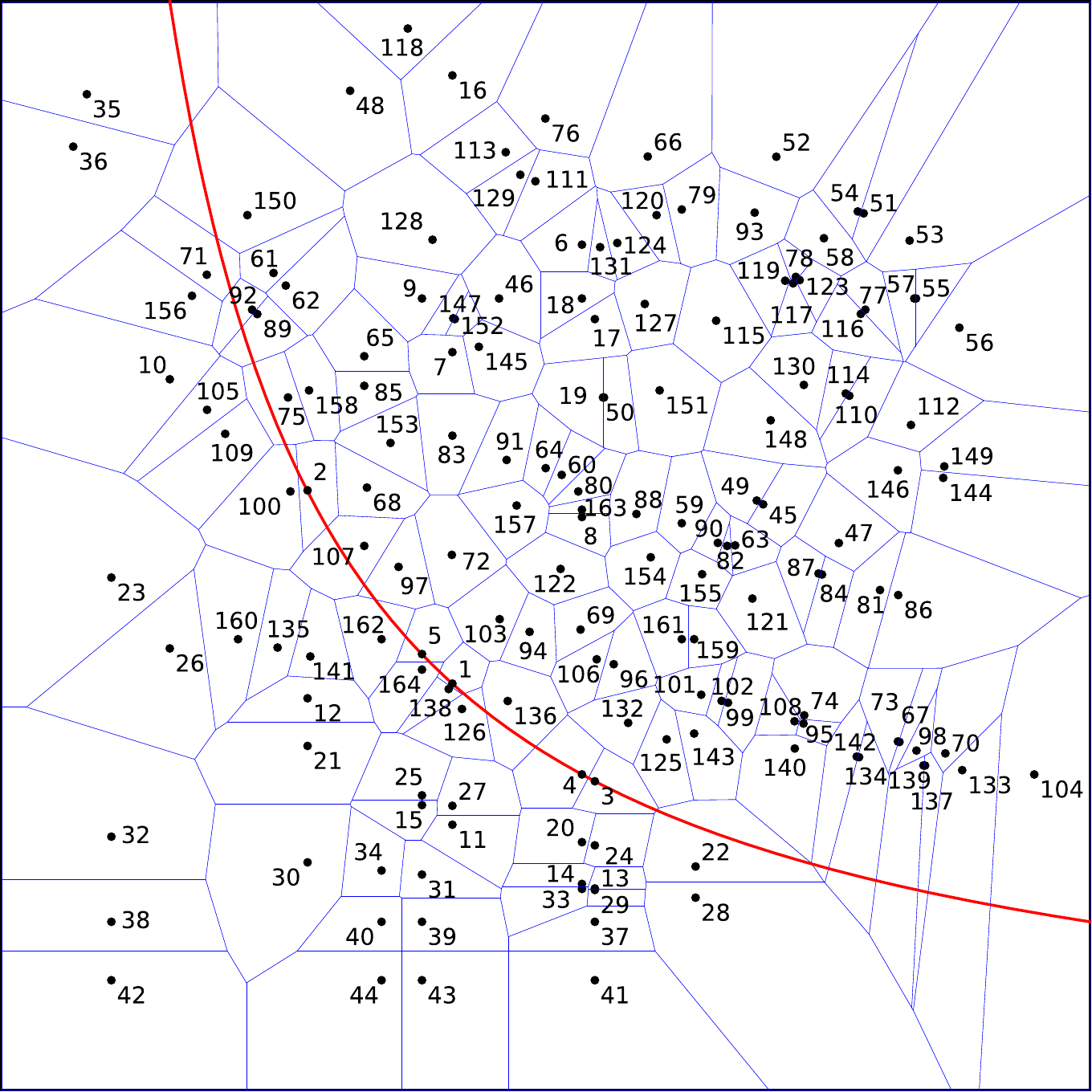}
\caption{
The $164$ pairs $(r,s)$, with abscissa $r$ and ordinate $\tfrac{s}{r}$.
Those below the hyperbola are such that an s-disc fits in the hole between three 1-discs (there are often derived from two disc packings).
Voronoï cells just aim to give an idea of how close are two pairs.
}
\label{fig:repartition}
\end{figure}

\noindent
\begin{tabular}{lll}
  1 (E)\hfill 111 / 1111 & 2 (E)\hfill 111 / 111r & 3 (S)\hfill 111 / 111rr\\
  \includegraphics[width=0.3\textwidth]{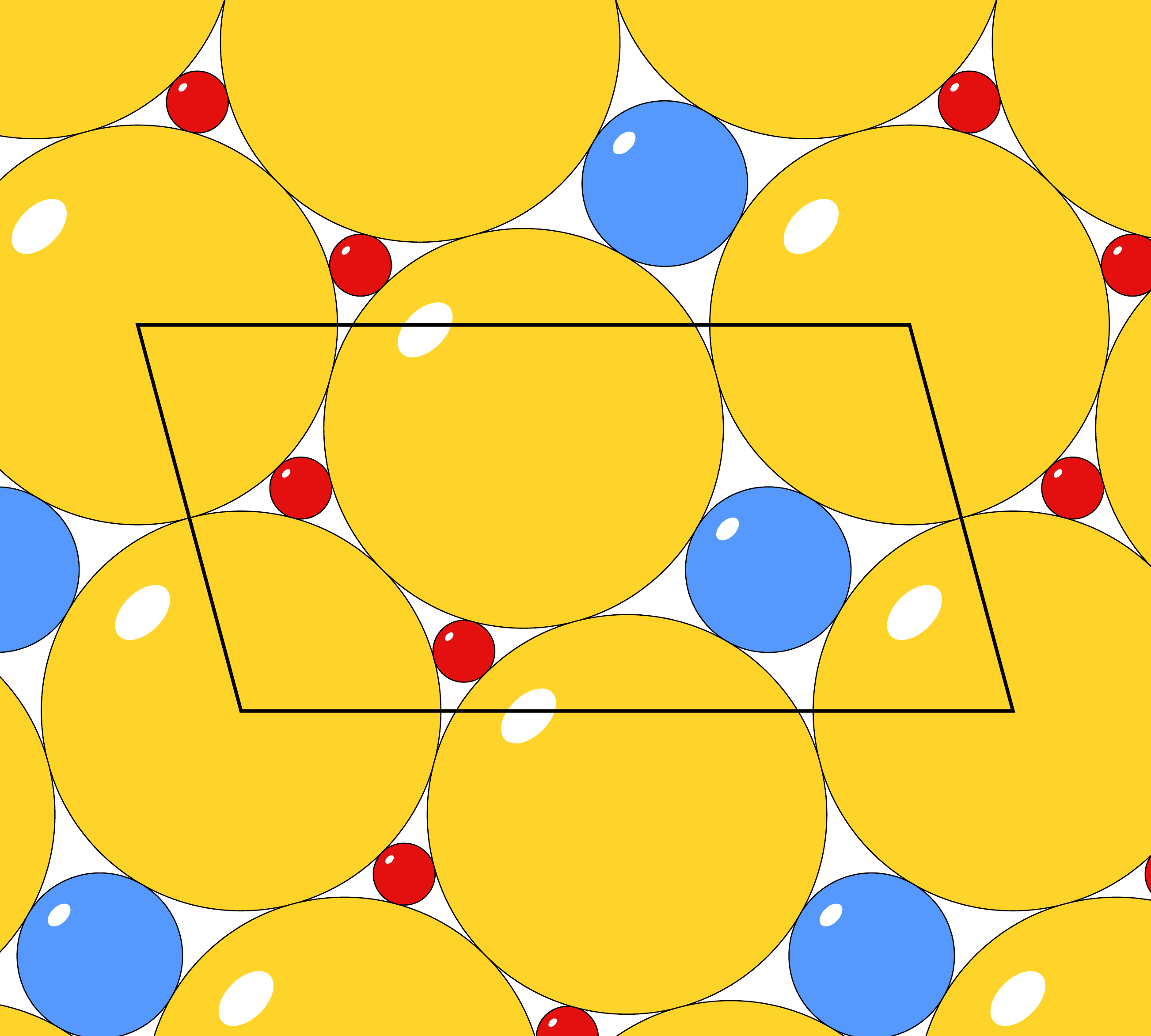} &
  \includegraphics[width=0.3\textwidth]{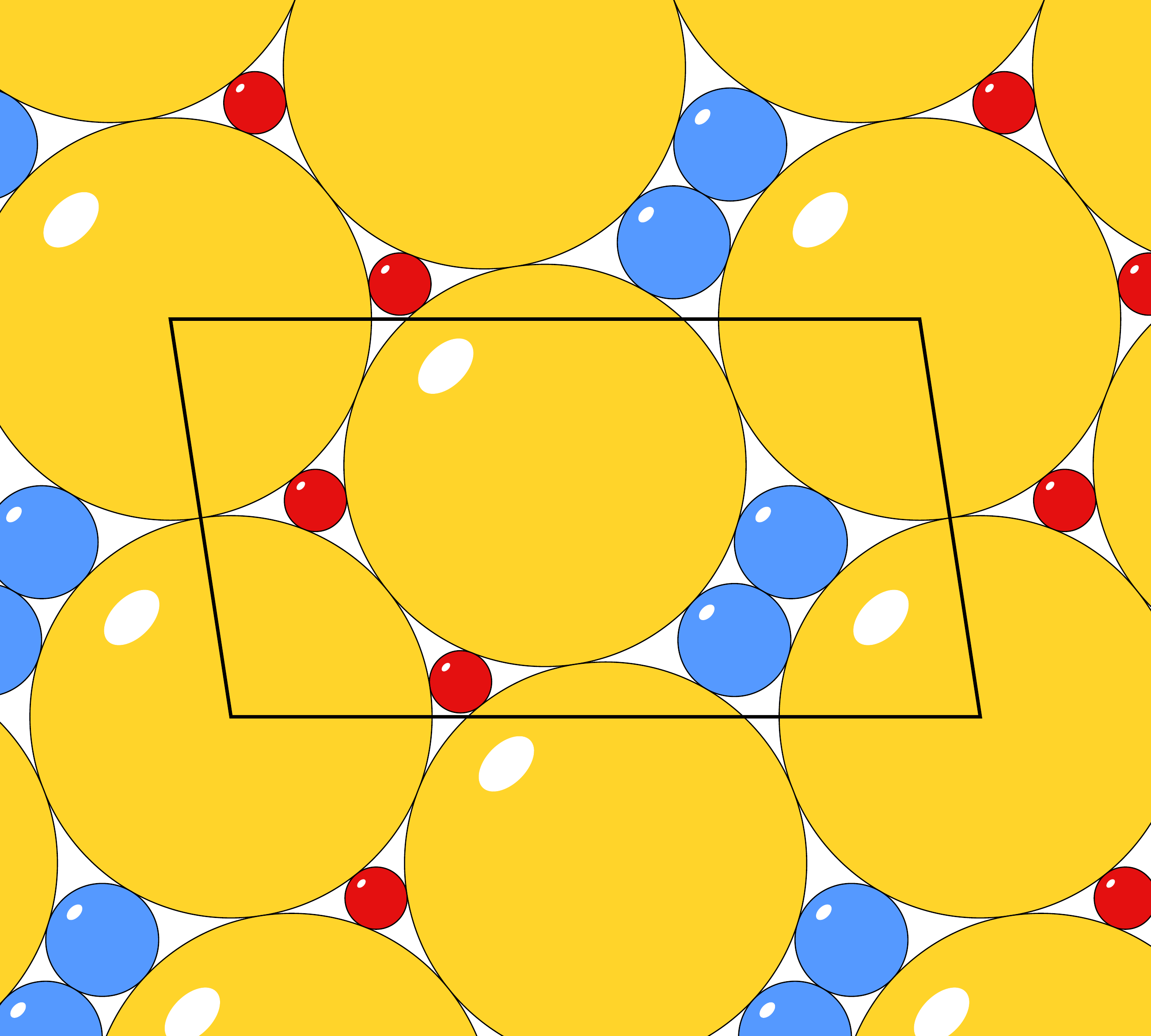} &
  \includegraphics[width=0.3\textwidth]{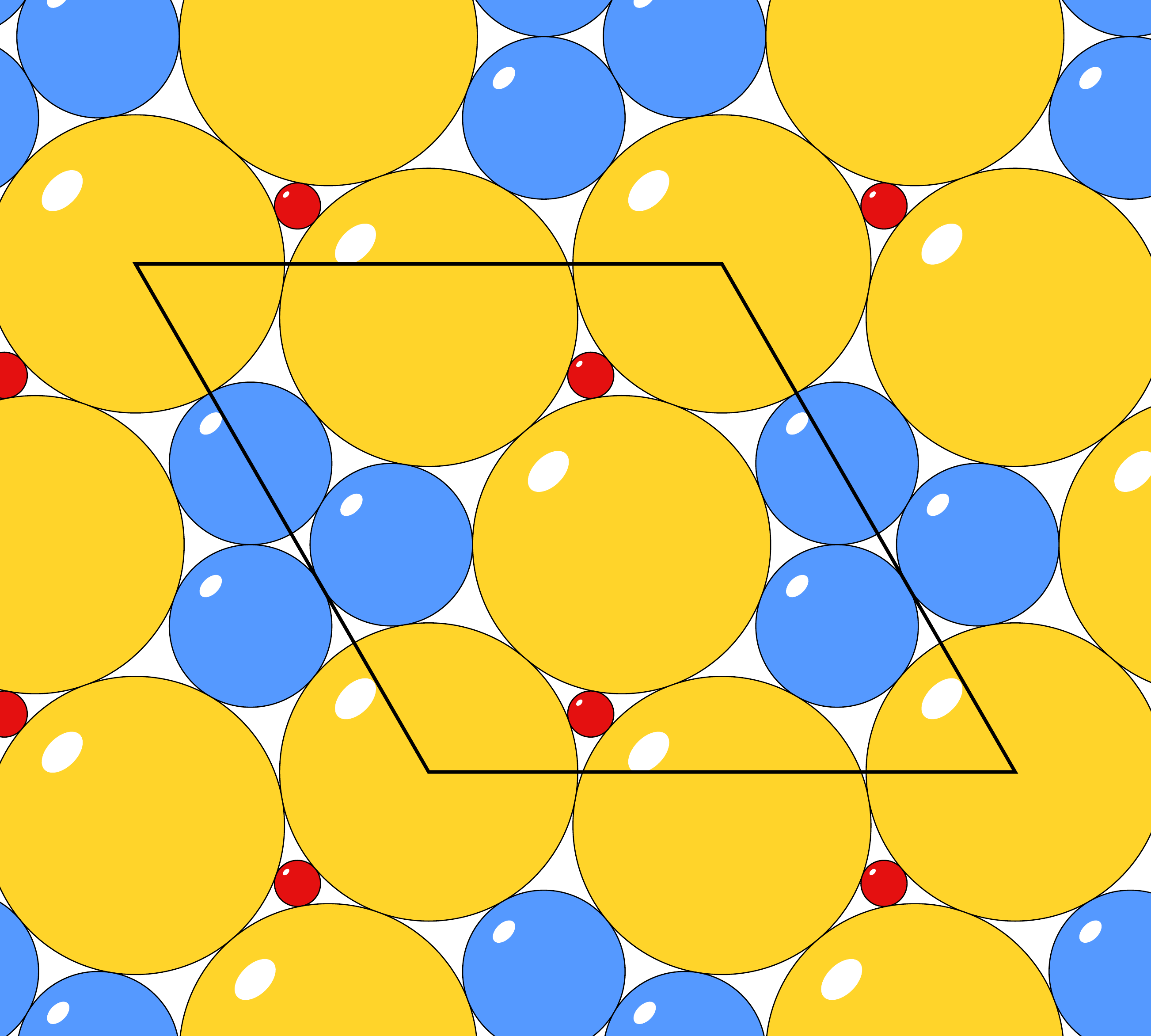}
\end{tabular}
\noindent
\begin{tabular}{lll}
  4 (L)\hfill 111 / 11r1r & 5 (H)\hfill 111 / 11rrr & 6 (E)\hfill 1111 / 11r1r\\
  \includegraphics[width=0.3\textwidth]{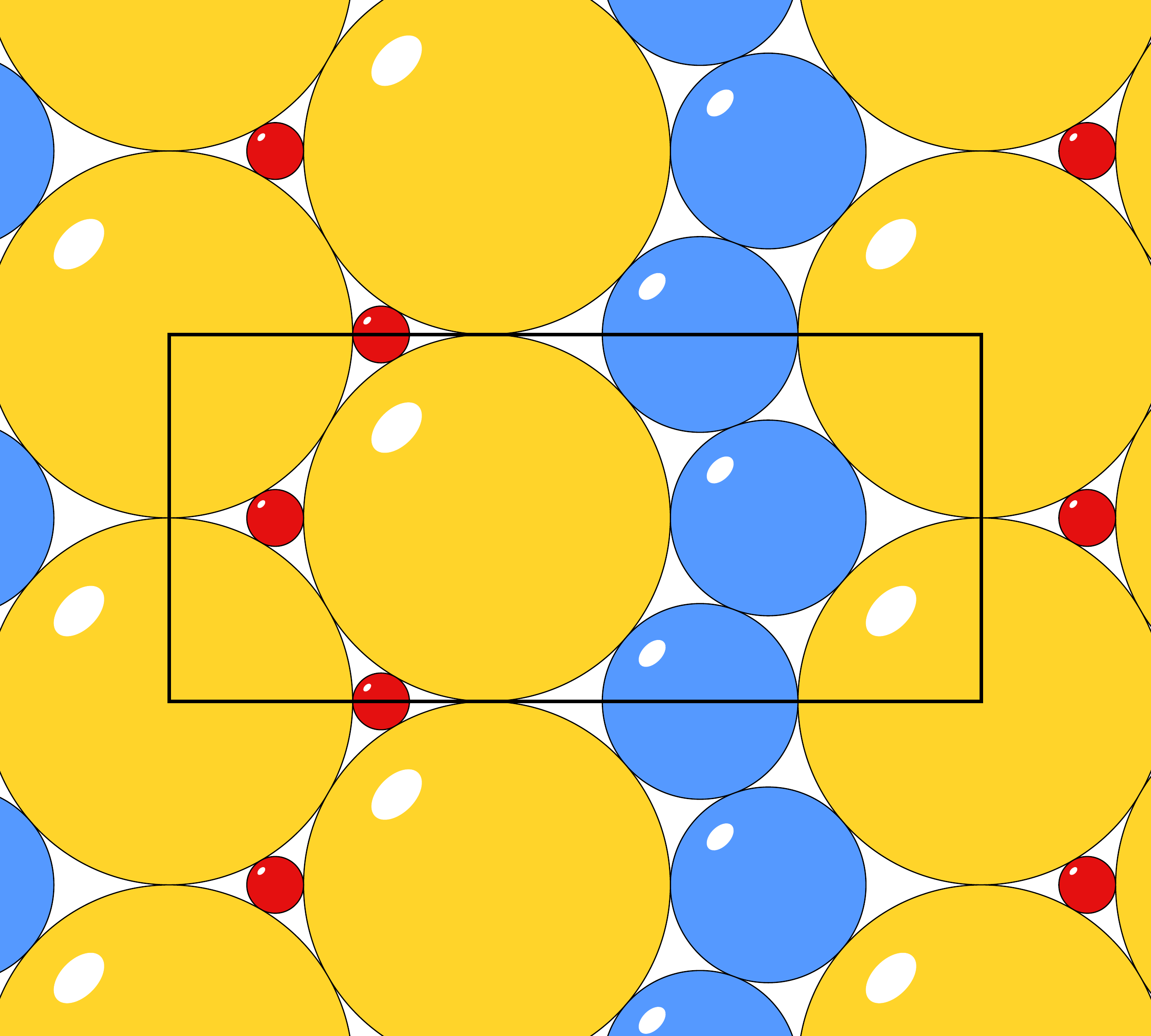} &
  \includegraphics[width=0.3\textwidth]{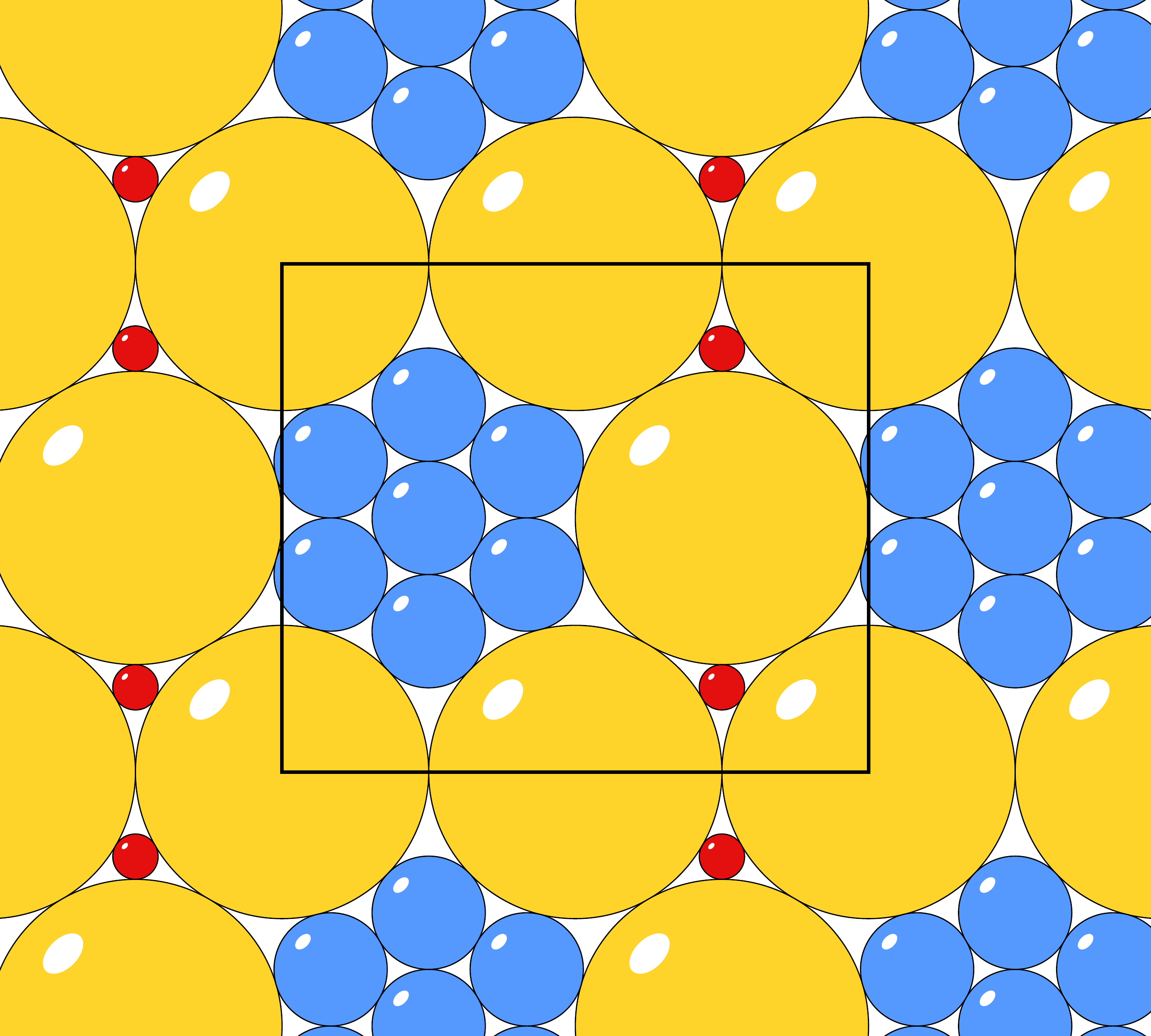} &
  \includegraphics[width=0.3\textwidth]{packing_1111_11r1r.pdf}
\end{tabular}
\noindent
\begin{tabular}{lll}
  7 (E)\hfill 111s / 1111 & 8 (E)\hfill 111s / 11r1r & 9 (E)\hfill 111s / 11rrr\\
  \includegraphics[width=0.3\textwidth]{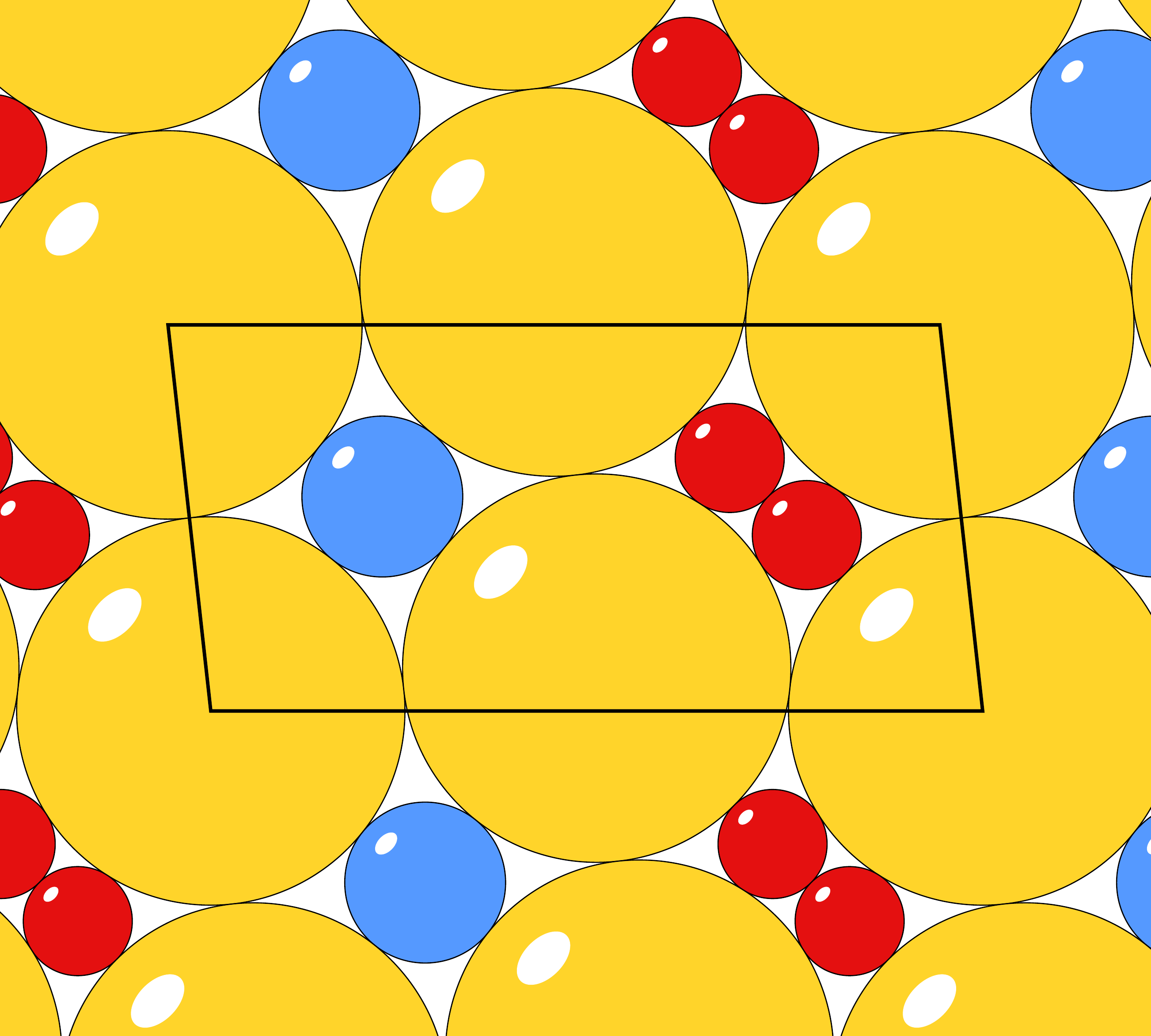} &
  \includegraphics[width=0.3\textwidth]{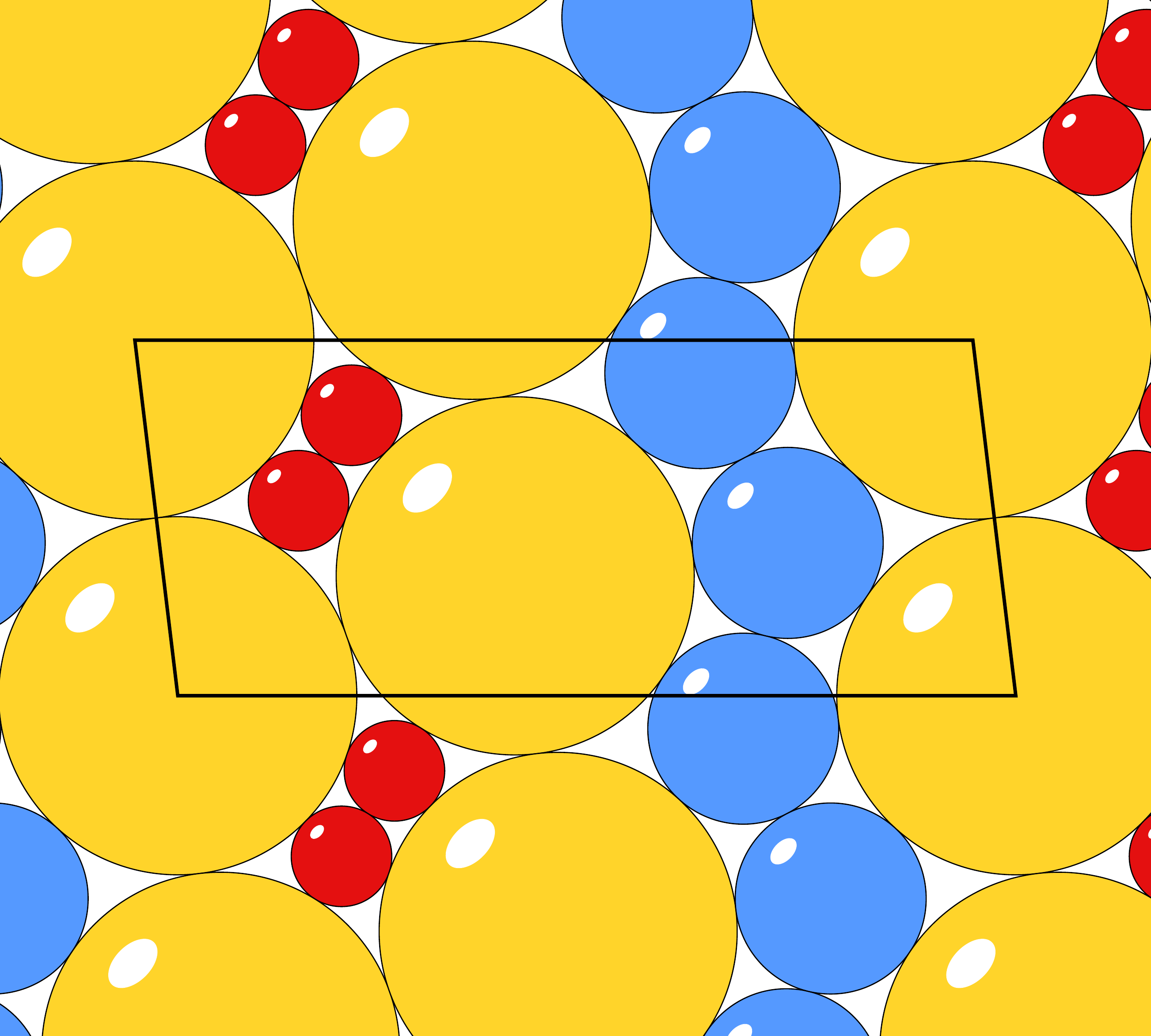} &
  \includegraphics[width=0.3\textwidth]{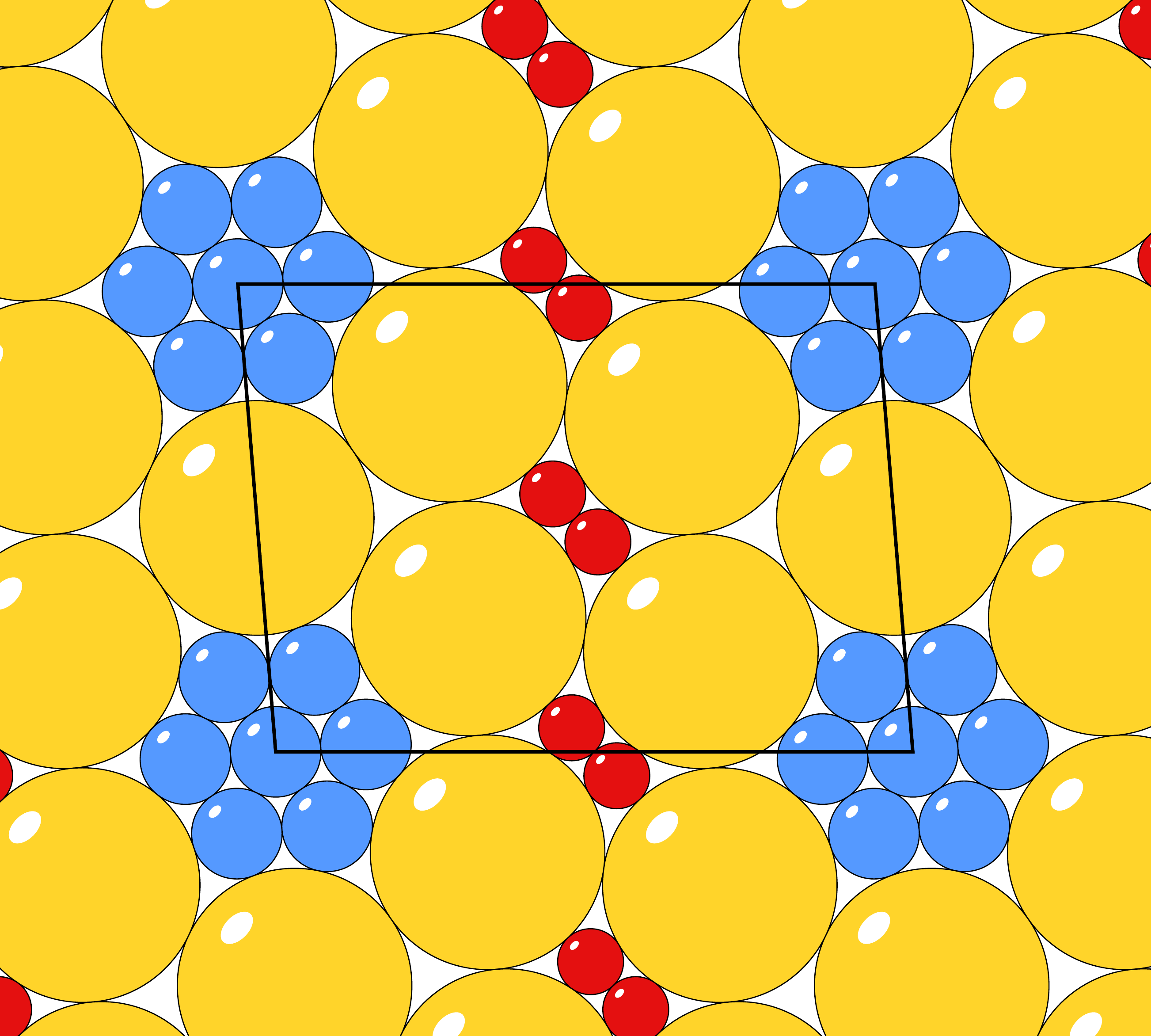}
\end{tabular}
\noindent
\begin{tabular}{lll}
  10 (H)\hfill 11ss / 111 & 11 (E)\hfill 11ss / 1111 & 12 (E)\hfill 11ss / 111r\\
  \includegraphics[width=0.3\textwidth]{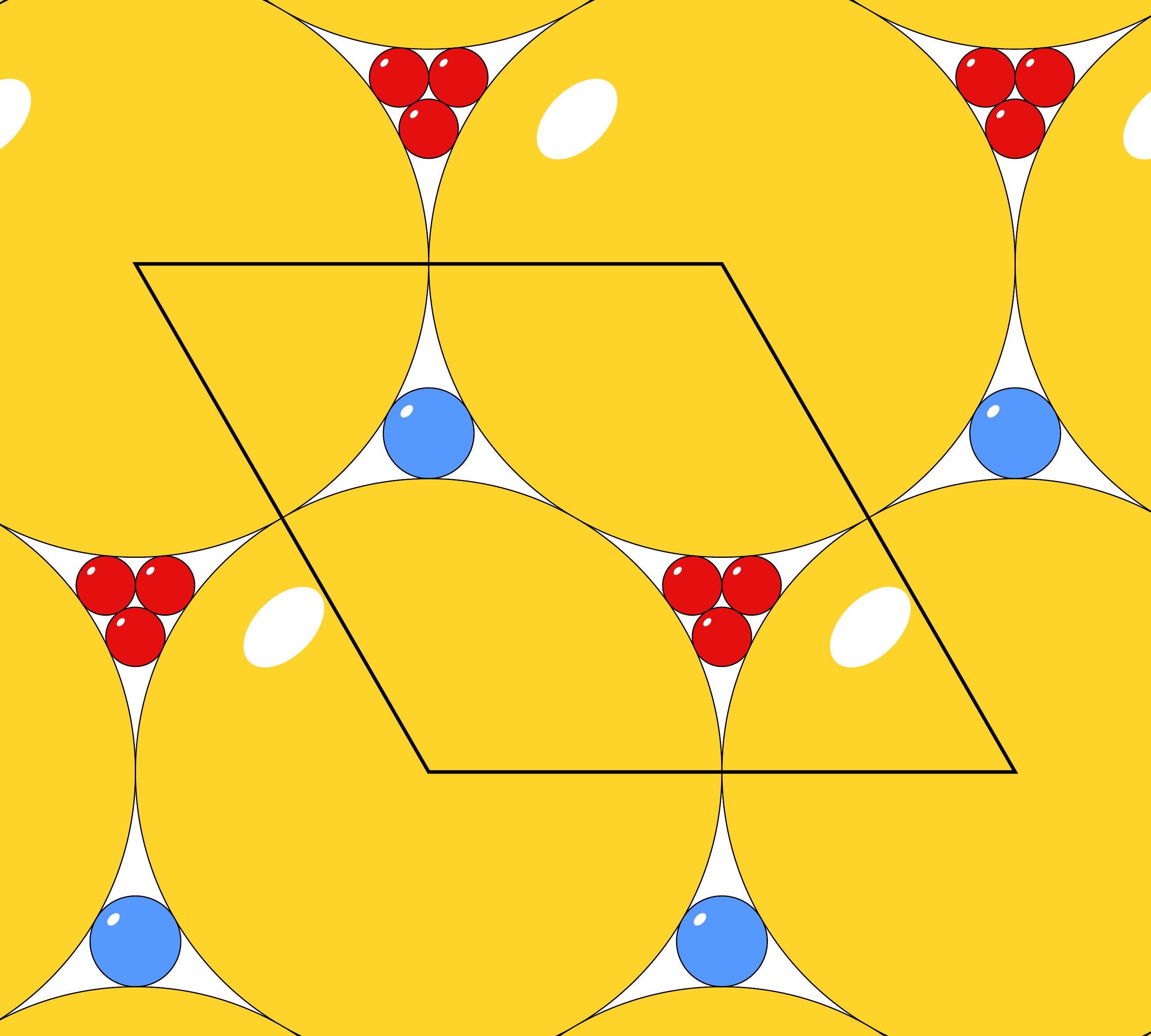} &
  \includegraphics[width=0.3\textwidth]{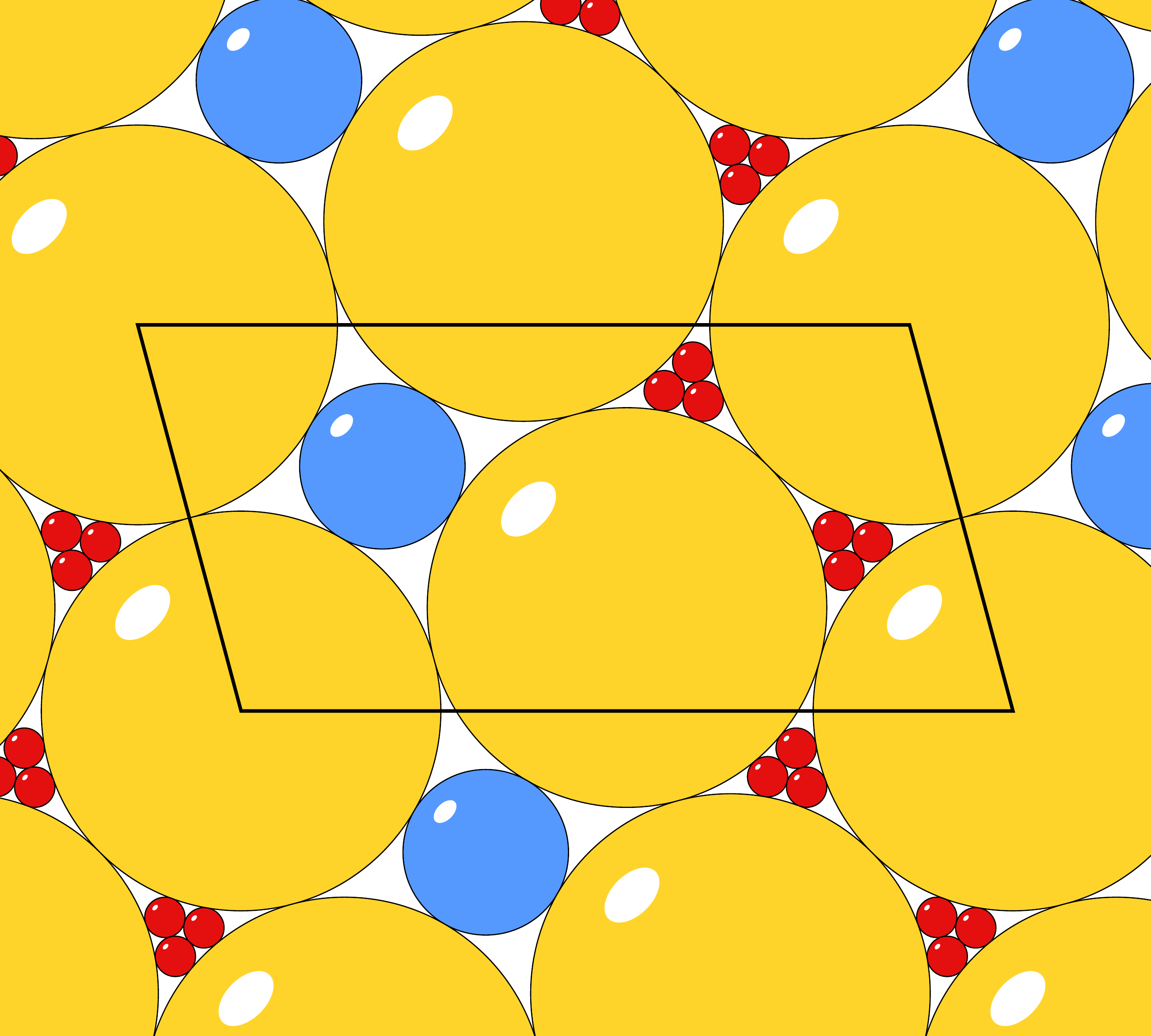} &
  \includegraphics[width=0.3\textwidth]{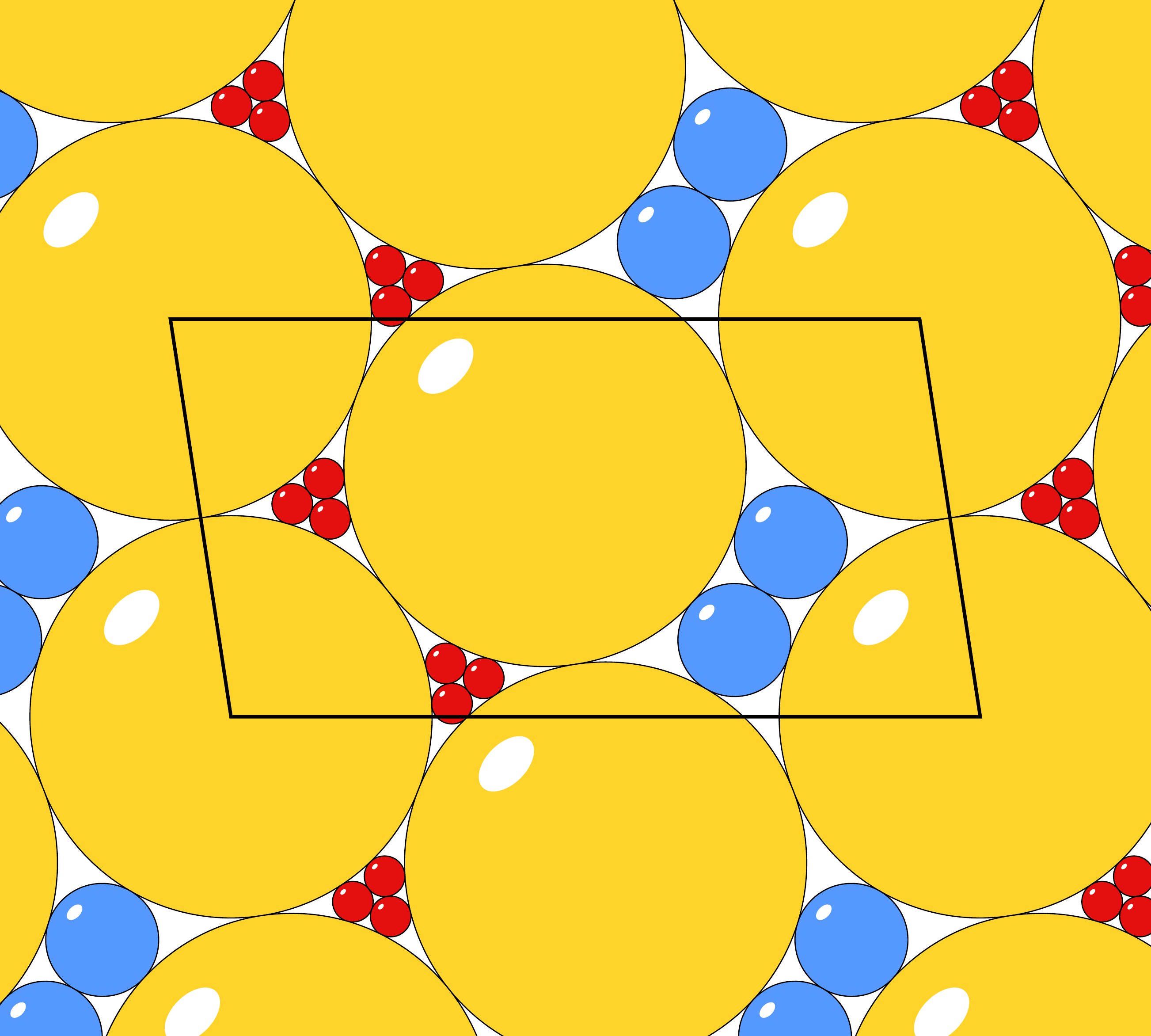}
\end{tabular}
\noindent
\begin{tabular}{lll}
  13 (S)\hfill 11ss / 111rr & 14 (L)\hfill 11ss / 11r1r & 15 (H)\hfill 11ss / 11rrr\\
  \includegraphics[width=0.3\textwidth]{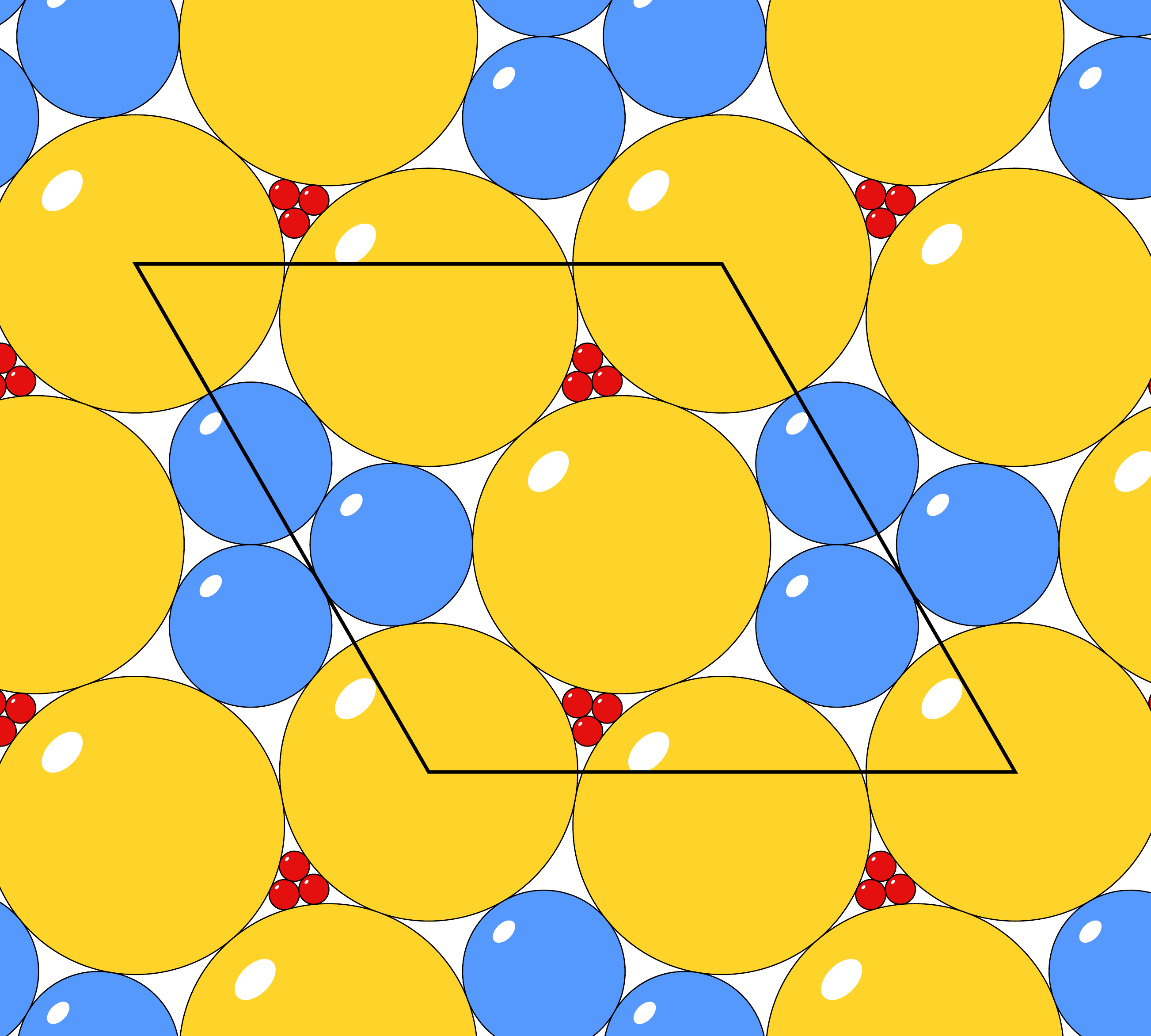} &
  \includegraphics[width=0.3\textwidth]{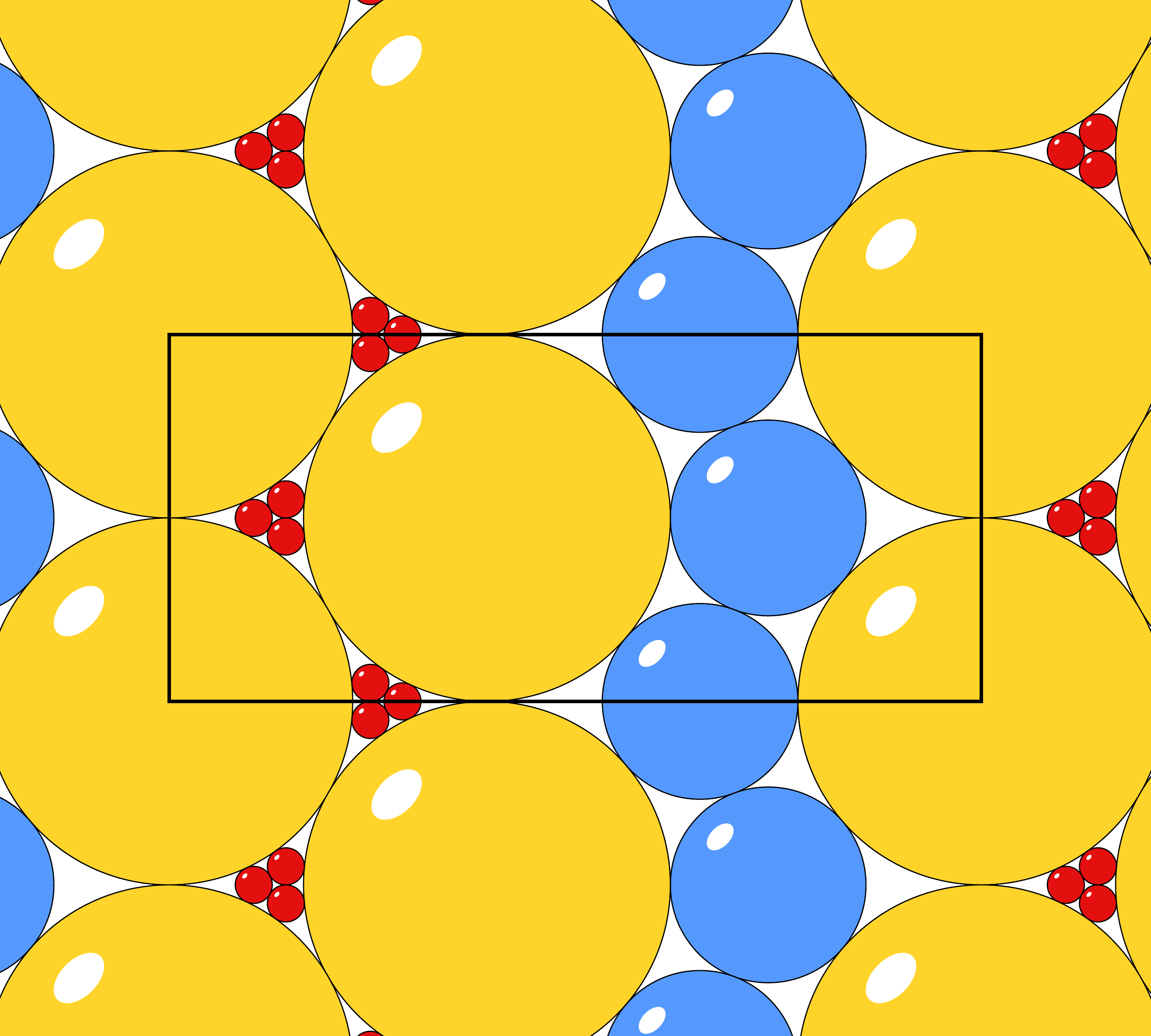} &
  \includegraphics[width=0.3\textwidth]{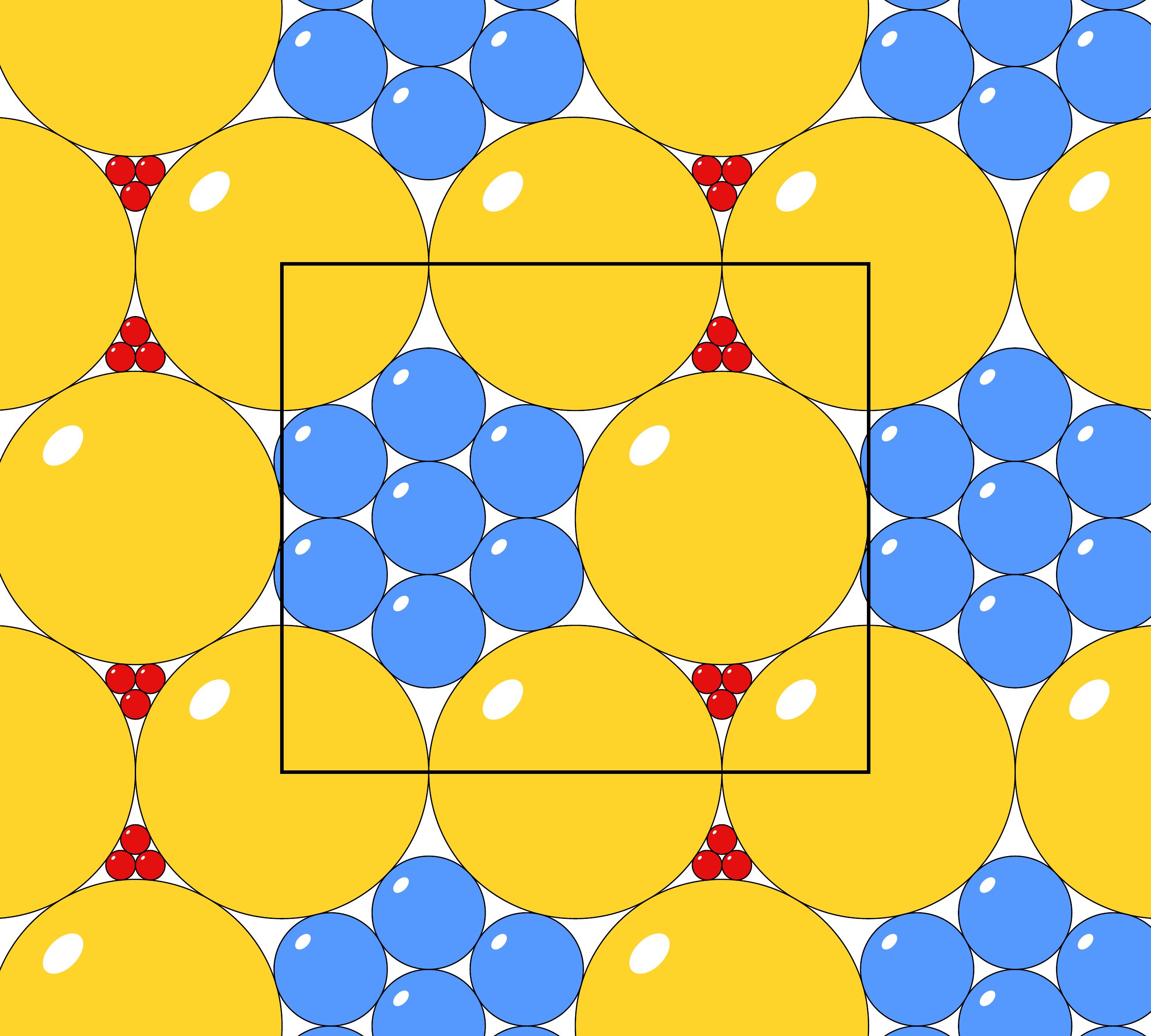}
\end{tabular}
\noindent
\begin{tabular}{lll}
  16 (E)\hfill 11sss / 1111 & 17 (S)\hfill 11sss / 111rr & 18 (L)\hfill 11sss / 11r1r\\
  \includegraphics[width=0.3\textwidth]{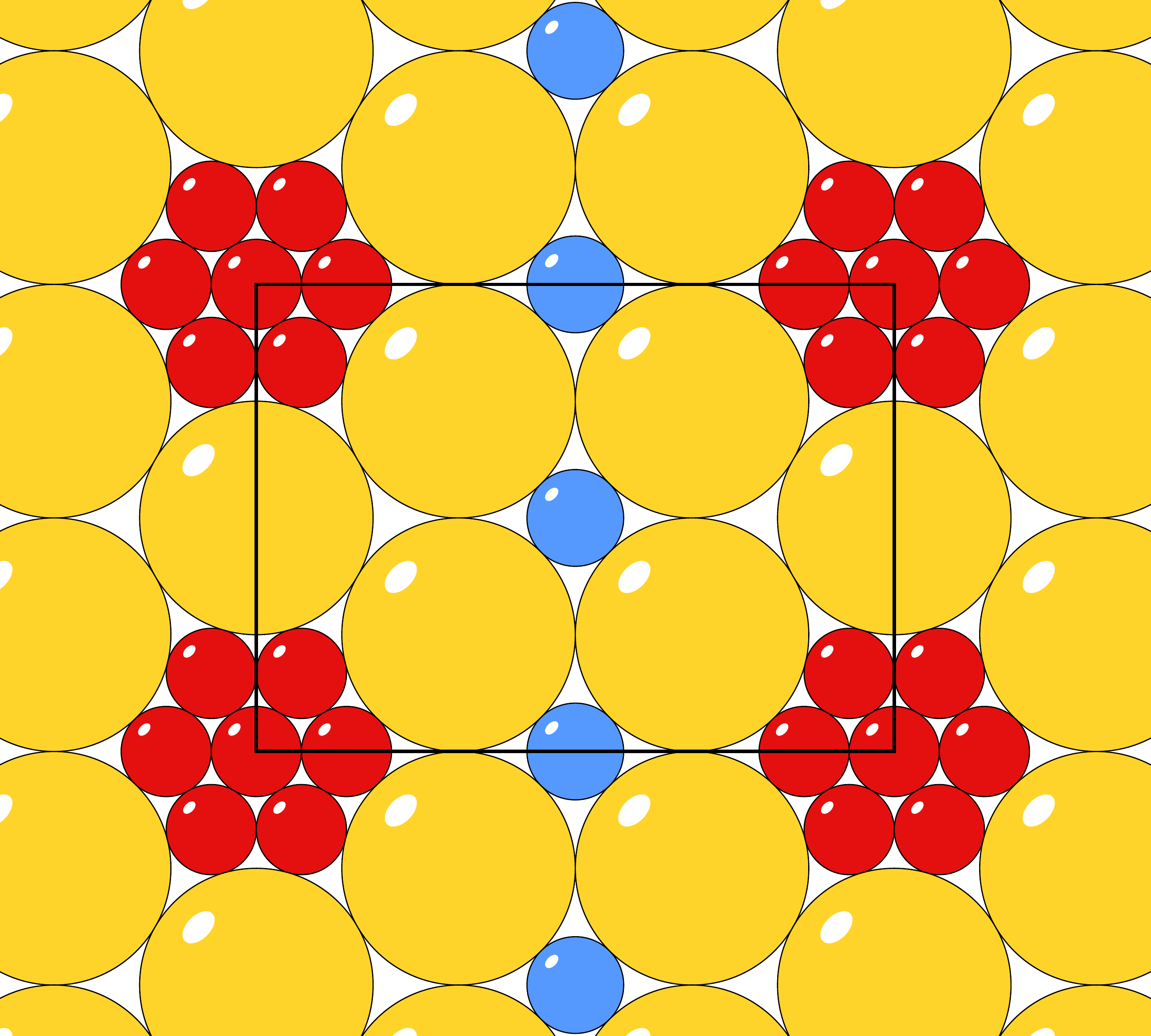} &
  \includegraphics[width=0.3\textwidth]{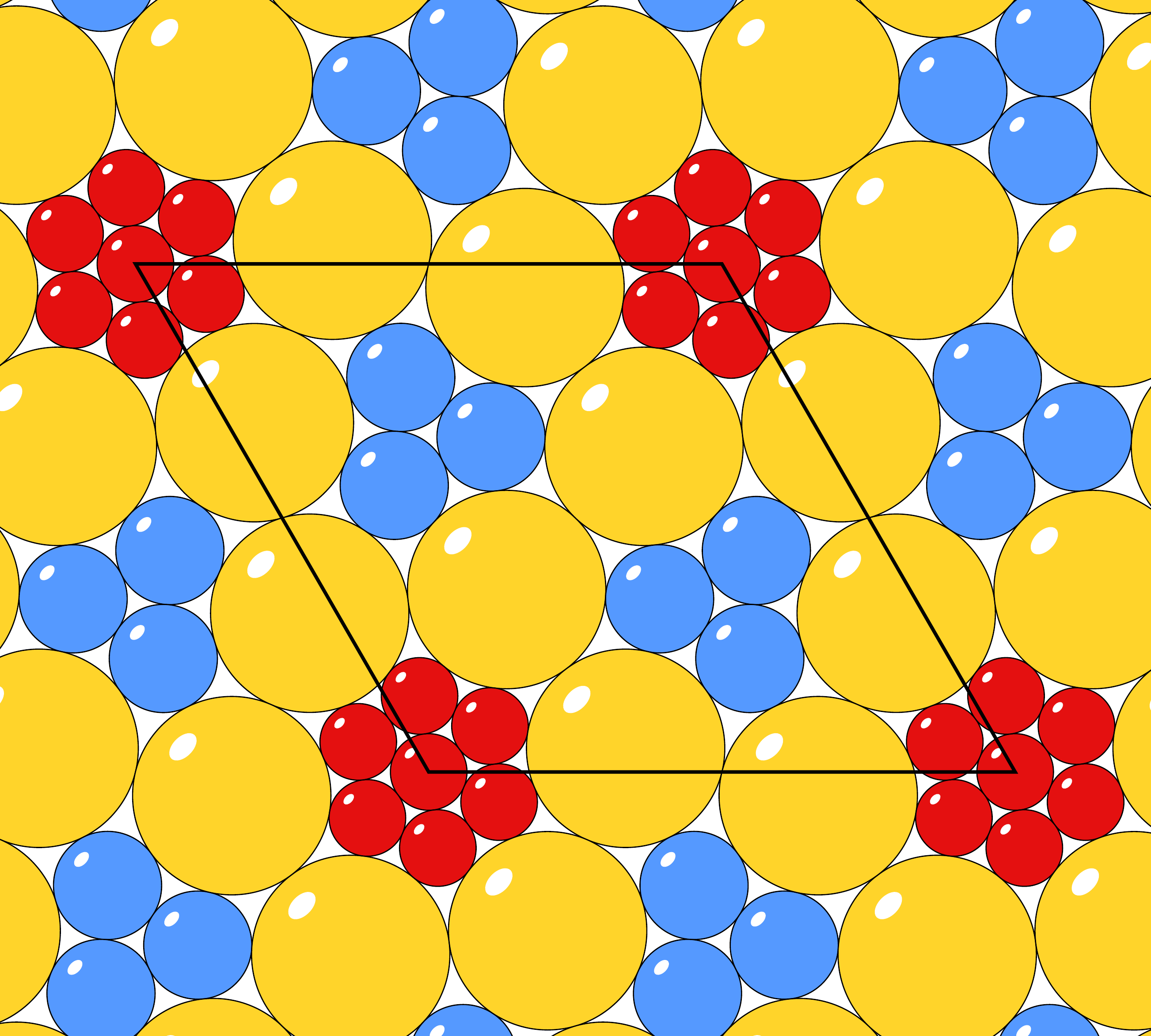} &
  \includegraphics[width=0.3\textwidth]{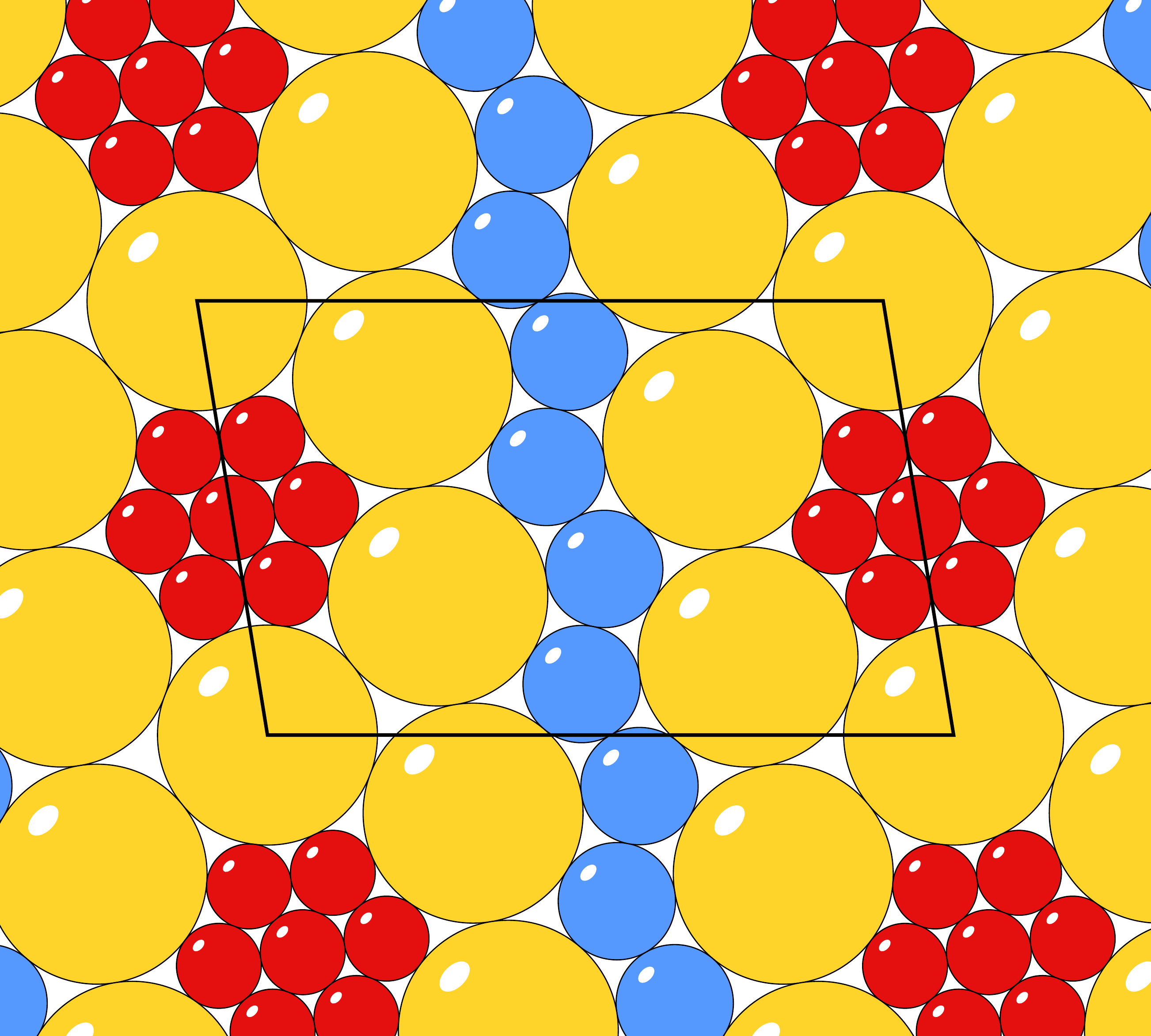}
\end{tabular}
\noindent
\begin{tabular}{lll}
  19 (H)\hfill 1srrs / rrssrss & 20 (L)\hfill 11r / 1r1s1r & 21 (E)\hfill 11r / 1r1s1s\\
  \includegraphics[width=0.3\textwidth]{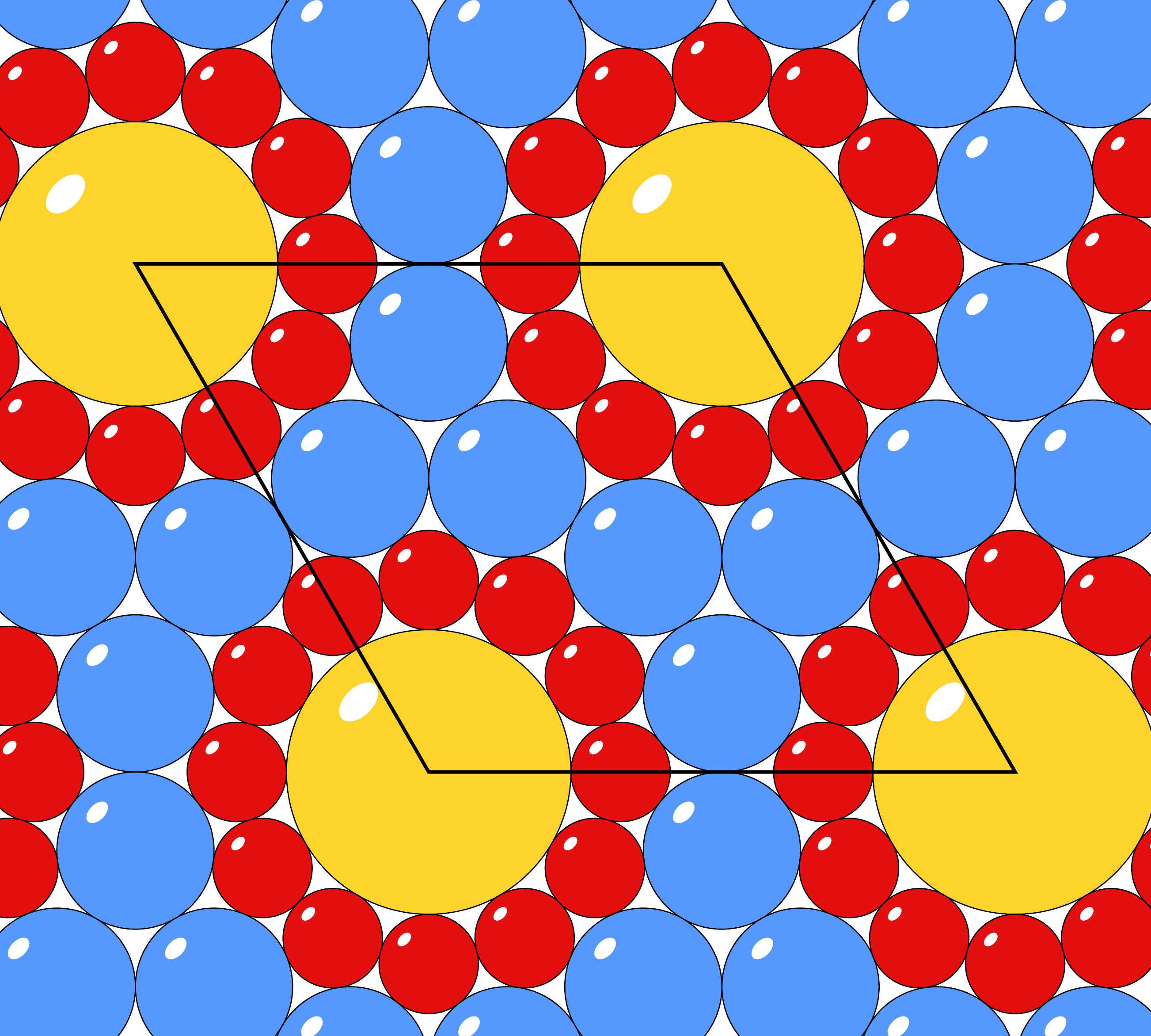} &
  \includegraphics[width=0.3\textwidth]{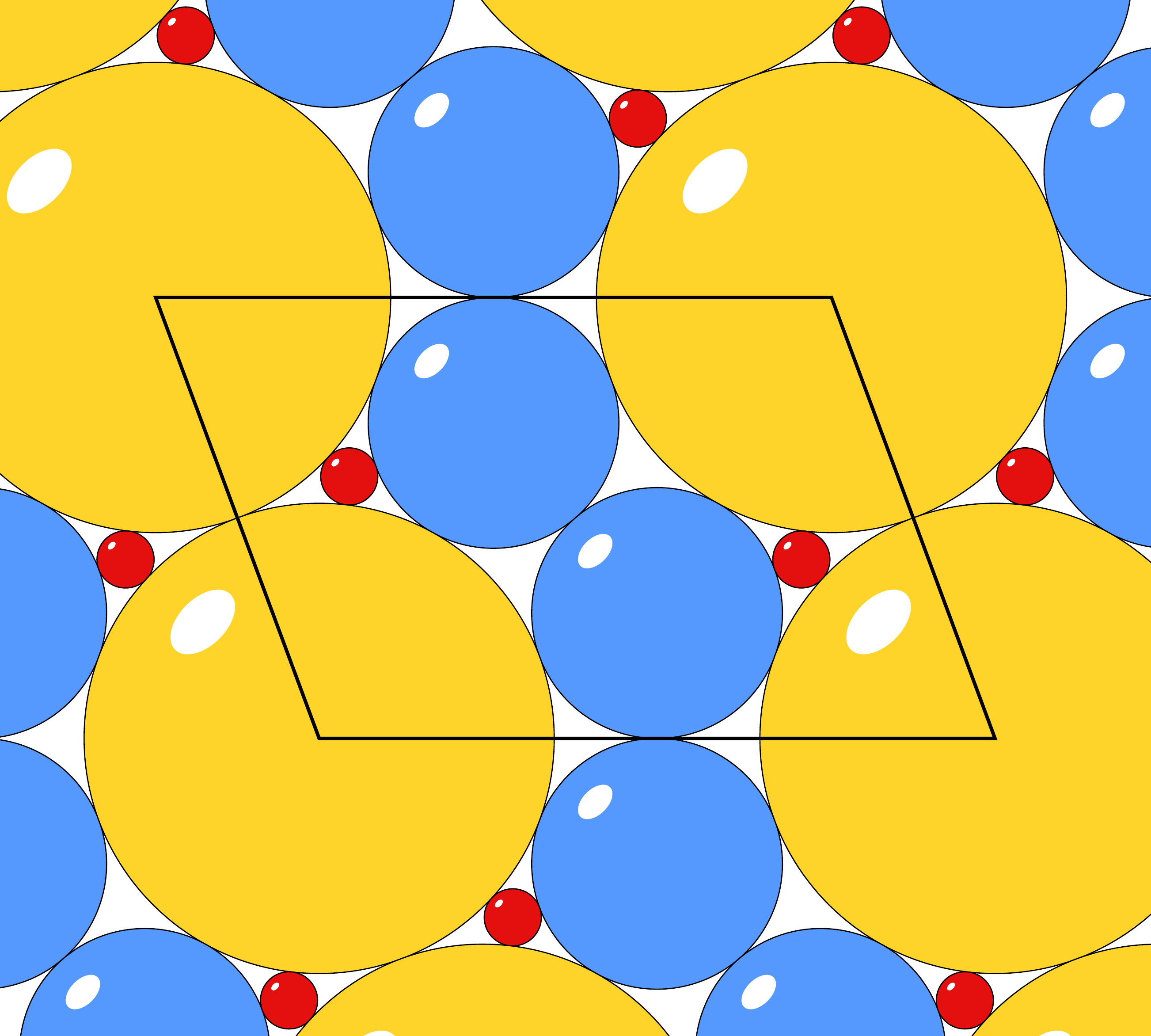} &
  \includegraphics[width=0.3\textwidth]{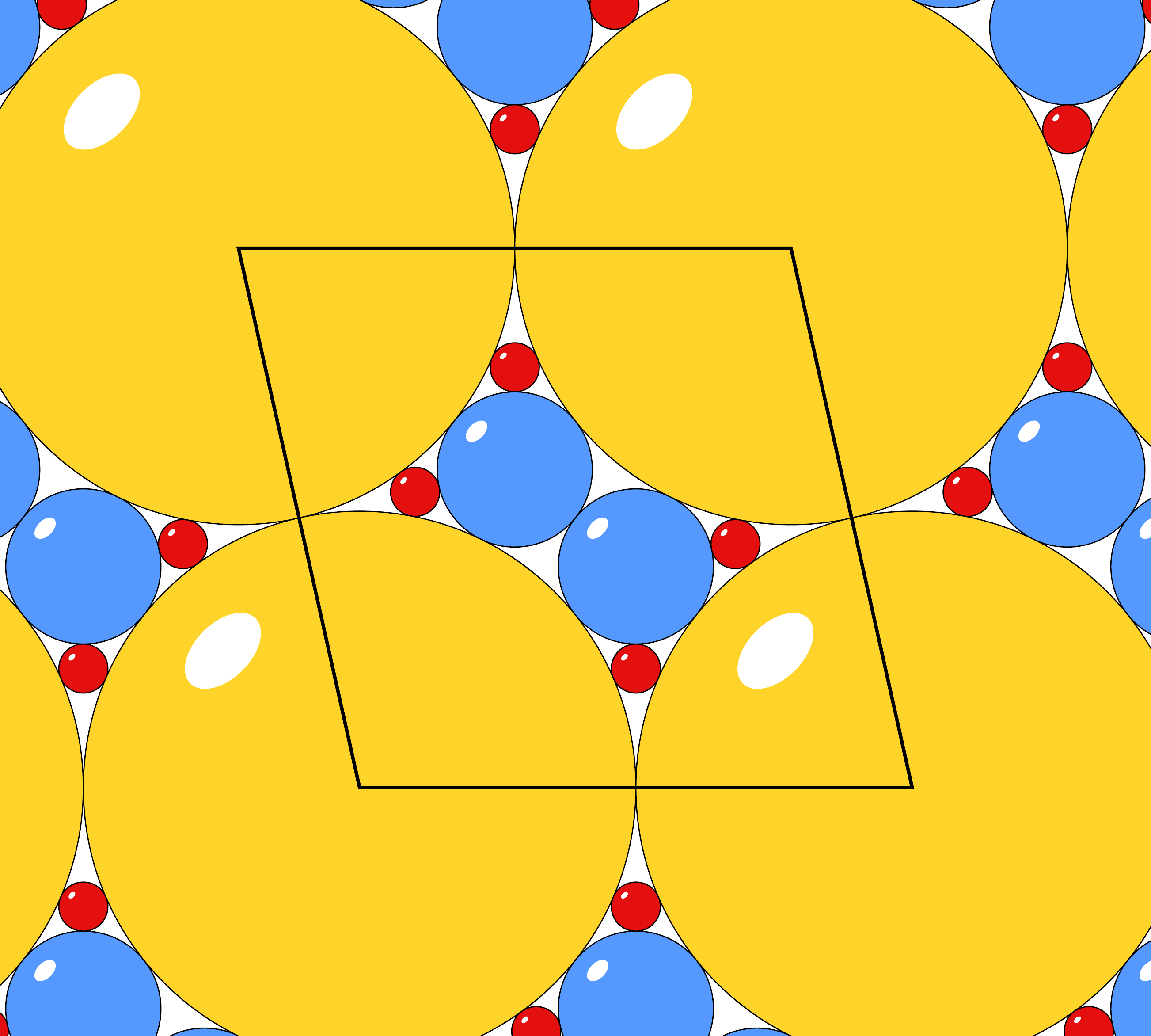}
\end{tabular}
\noindent
\begin{tabular}{lll}
  22 (L)\hfill 11r / 1r1s1s1s & 23 (H)\hfill 11r / 1rr1s & 24 (S)\hfill 11r / 1rr1s1s\\
  \includegraphics[width=0.3\textwidth]{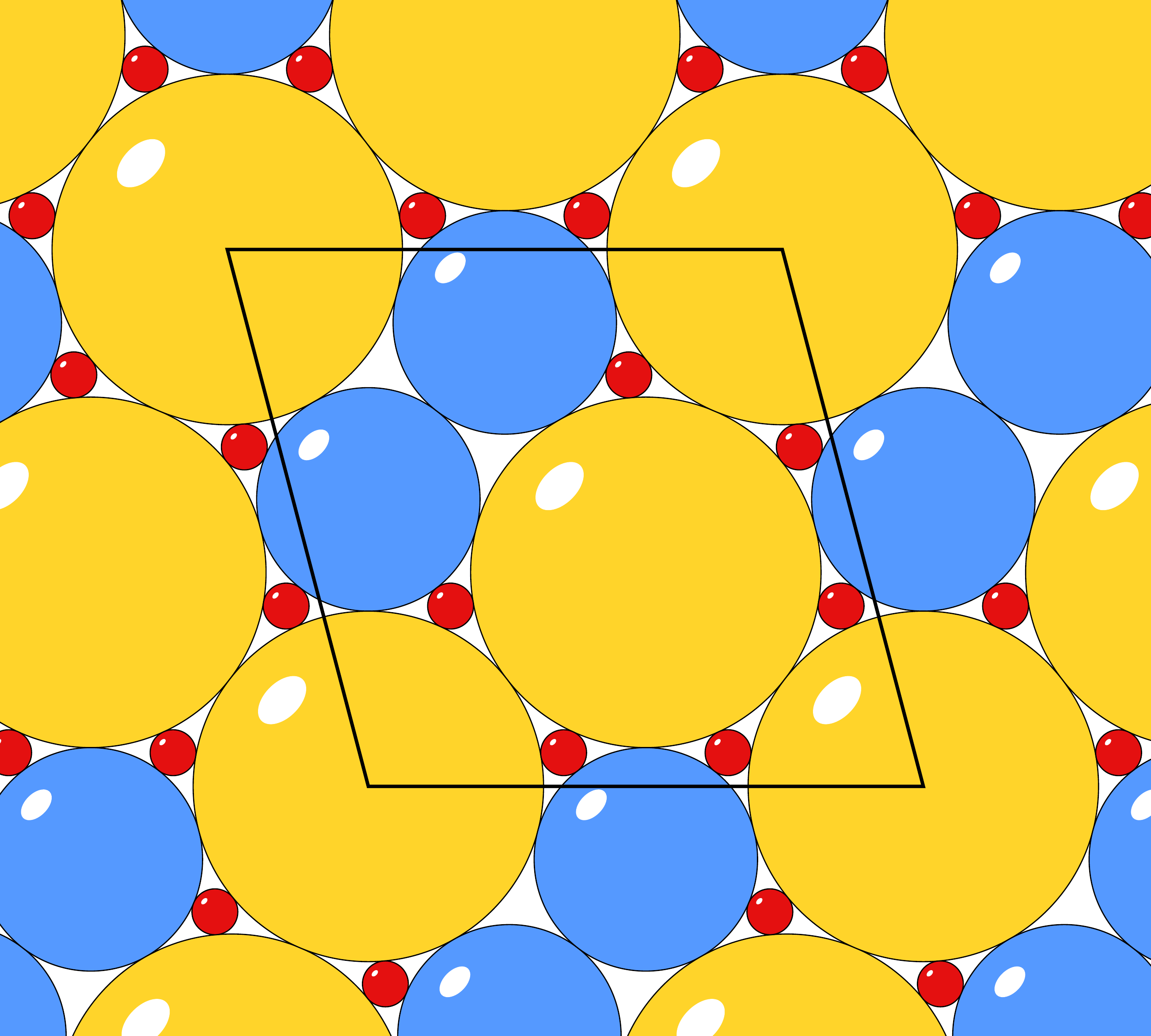} &
  \includegraphics[width=0.3\textwidth]{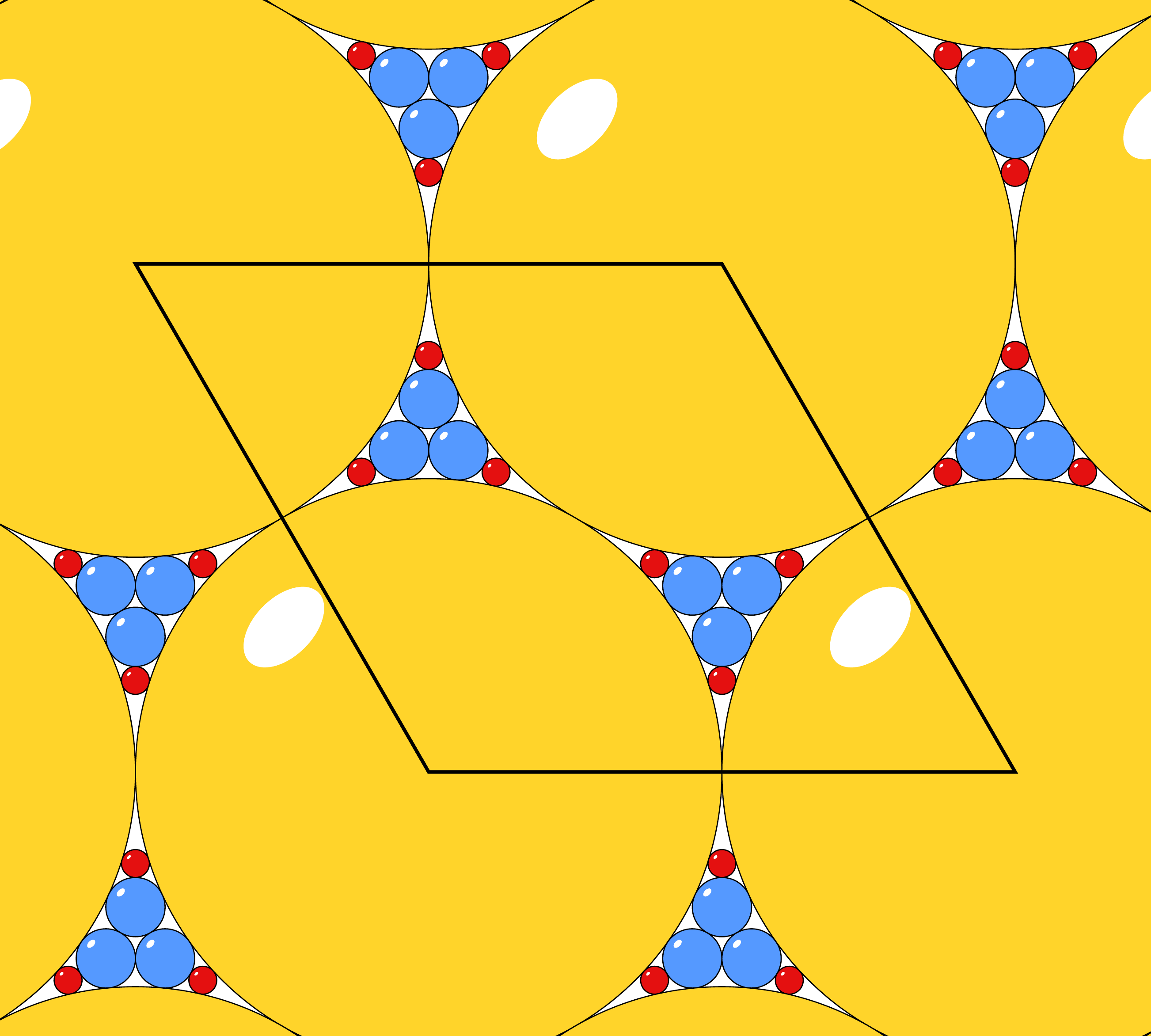} &
  \includegraphics[width=0.3\textwidth]{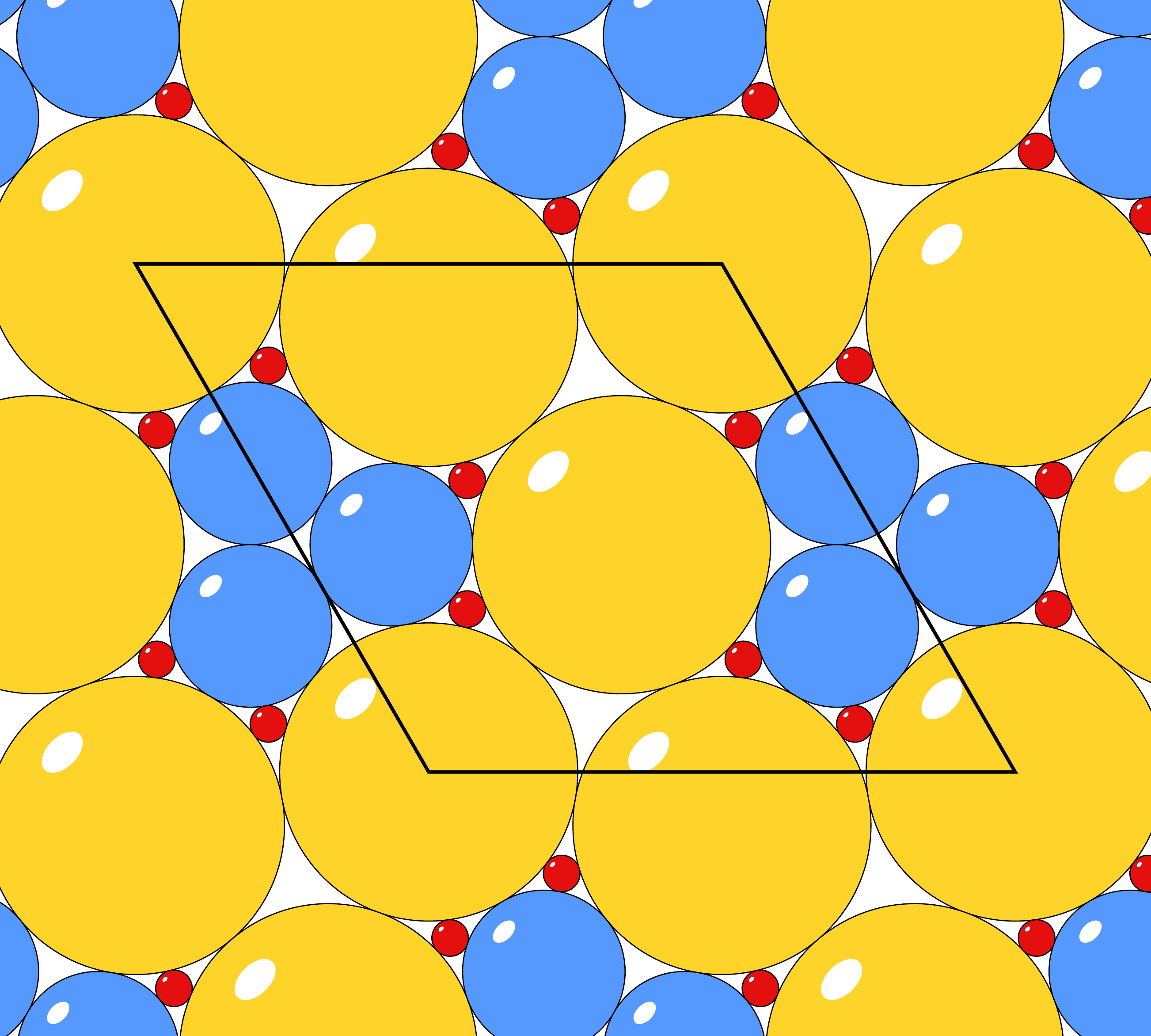}
\end{tabular}
\noindent
\begin{tabular}{lll}
  25 (H)\hfill 11r / 1rrr1s & 26 (H)\hfill 11r / 1s1s1s & 27 (E)\hfill 11r / 1s1s1s1s\\
  \includegraphics[width=0.3\textwidth]{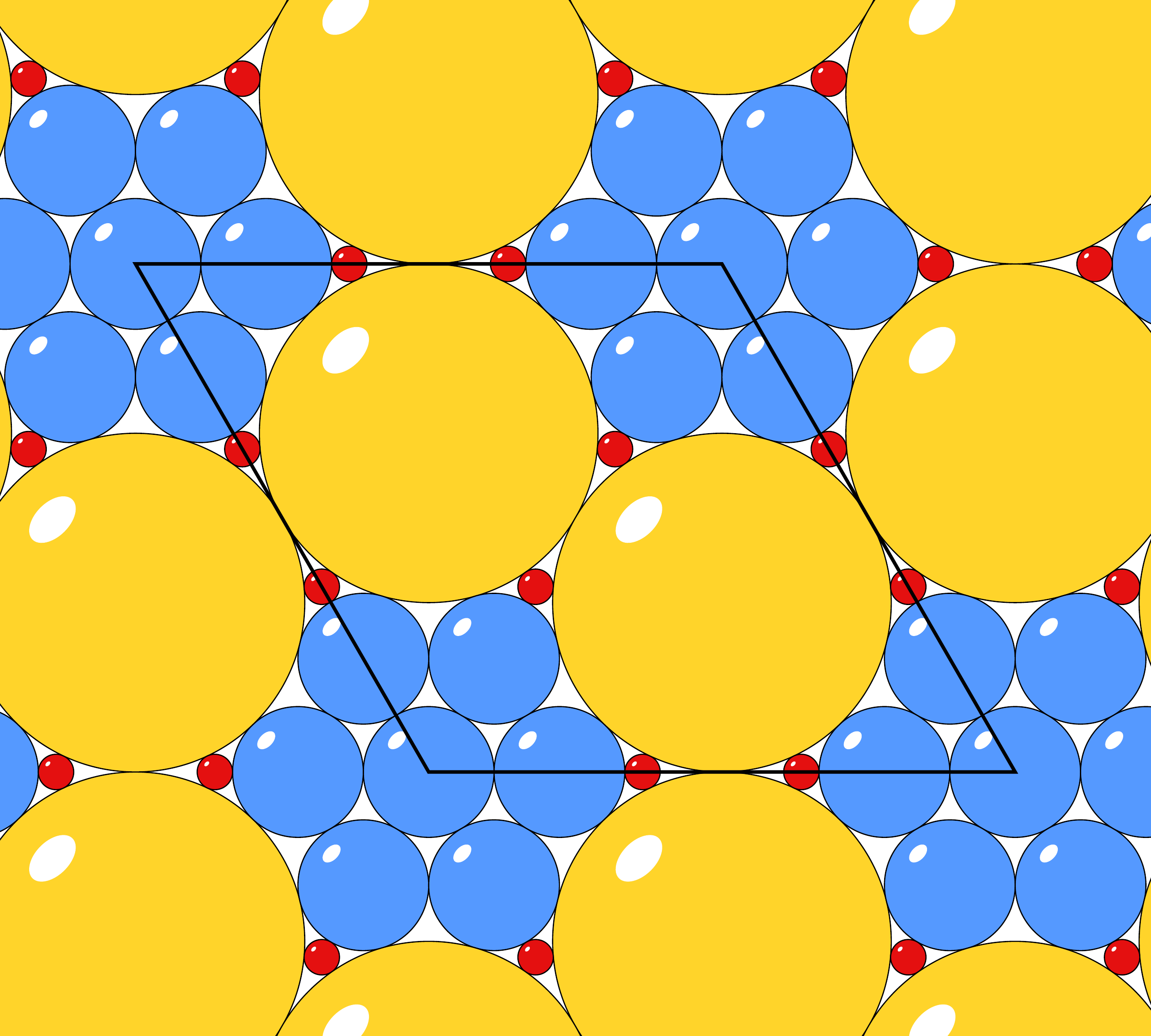} &
  \includegraphics[width=0.3\textwidth]{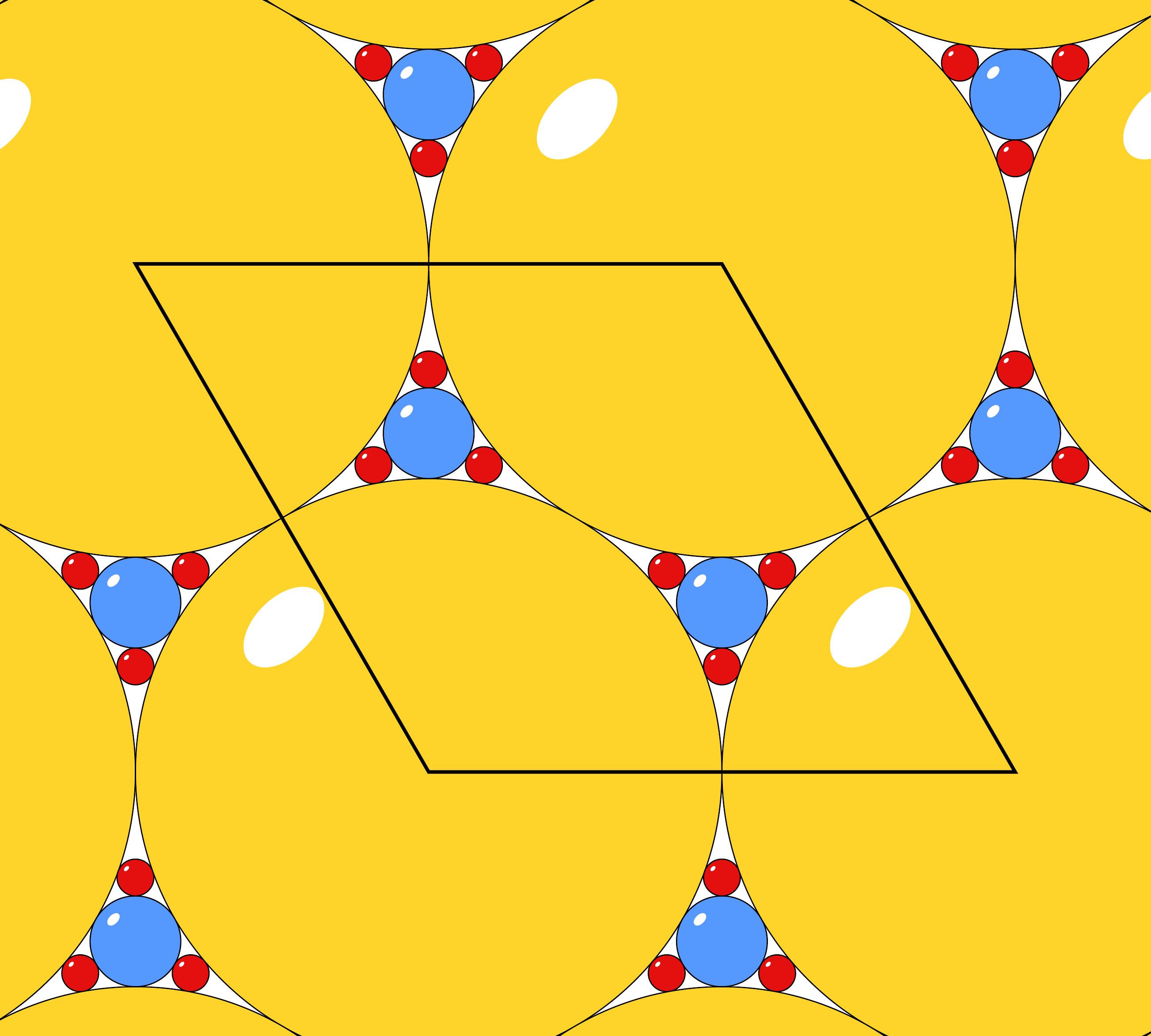} &
  \includegraphics[width=0.3\textwidth]{packing_11r_1s1s1s1s.pdf}
\end{tabular}
\noindent
\begin{tabular}{lll}
  28 (L)\hfill 1rr / 1111srs & 29 (S)\hfill 1rr / 111srrs & 30 (E)\hfill 1rr / 111srs\\
  \includegraphics[width=0.3\textwidth]{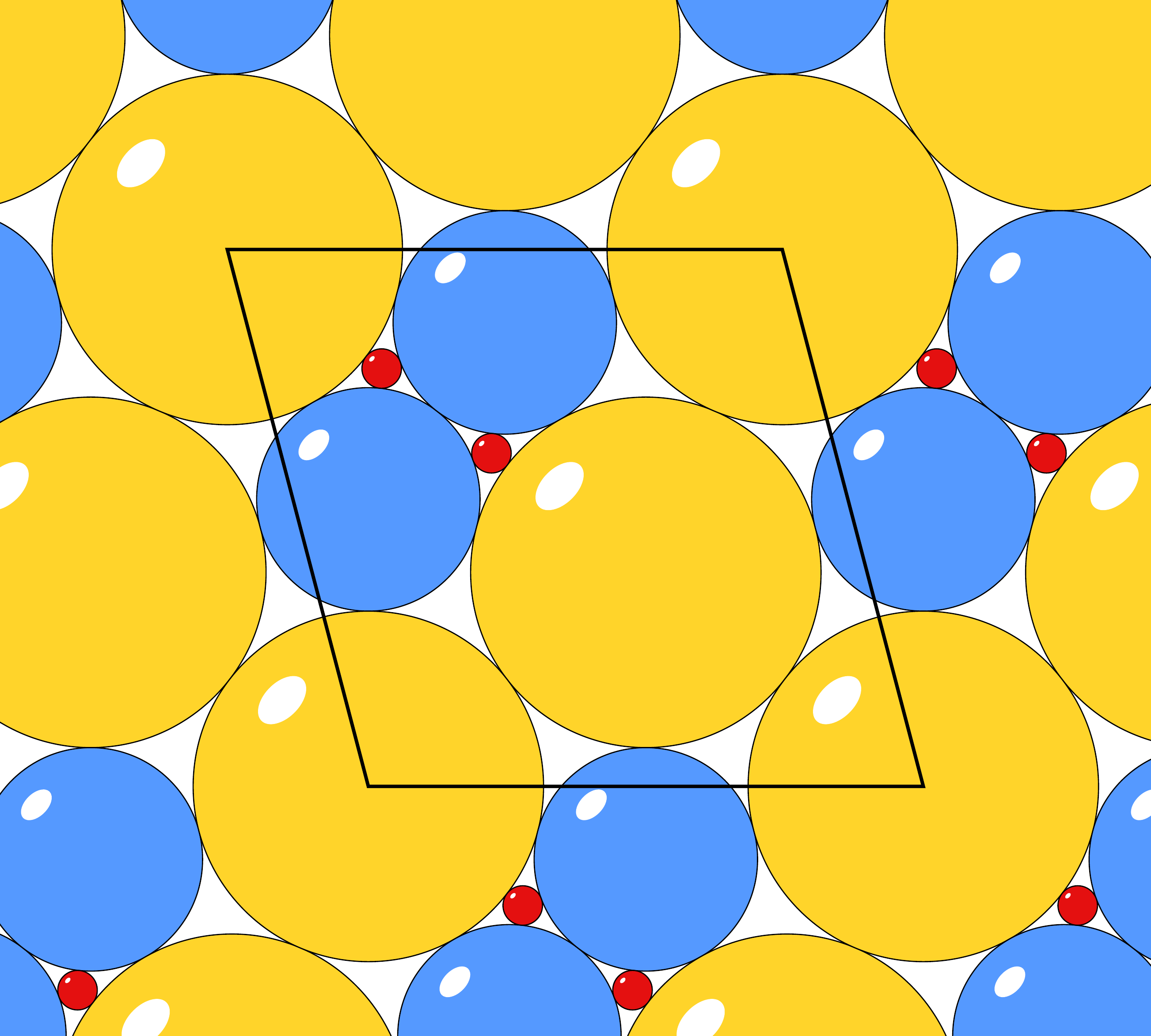} &
  \includegraphics[width=0.3\textwidth]{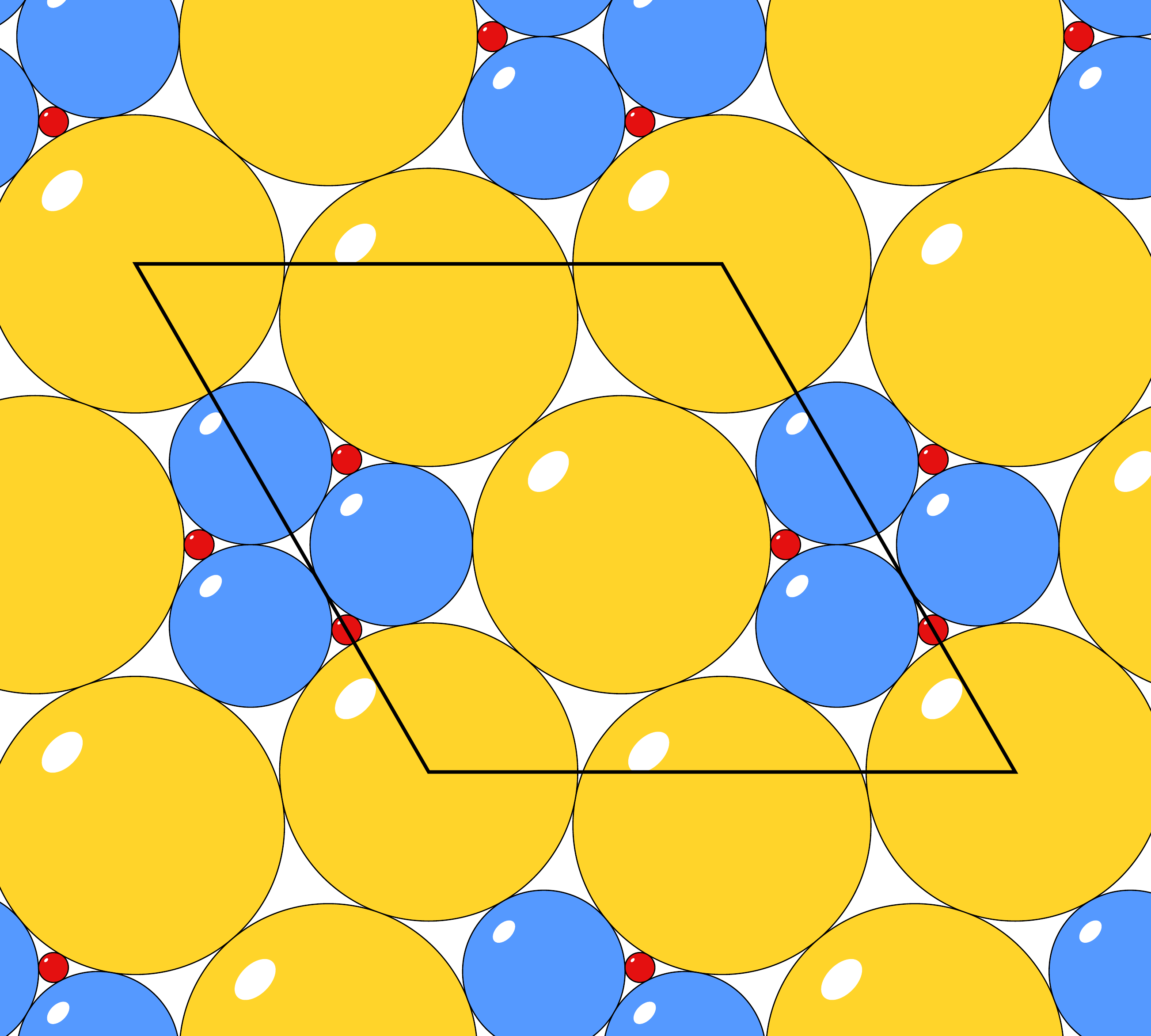} &
  \includegraphics[width=0.3\textwidth]{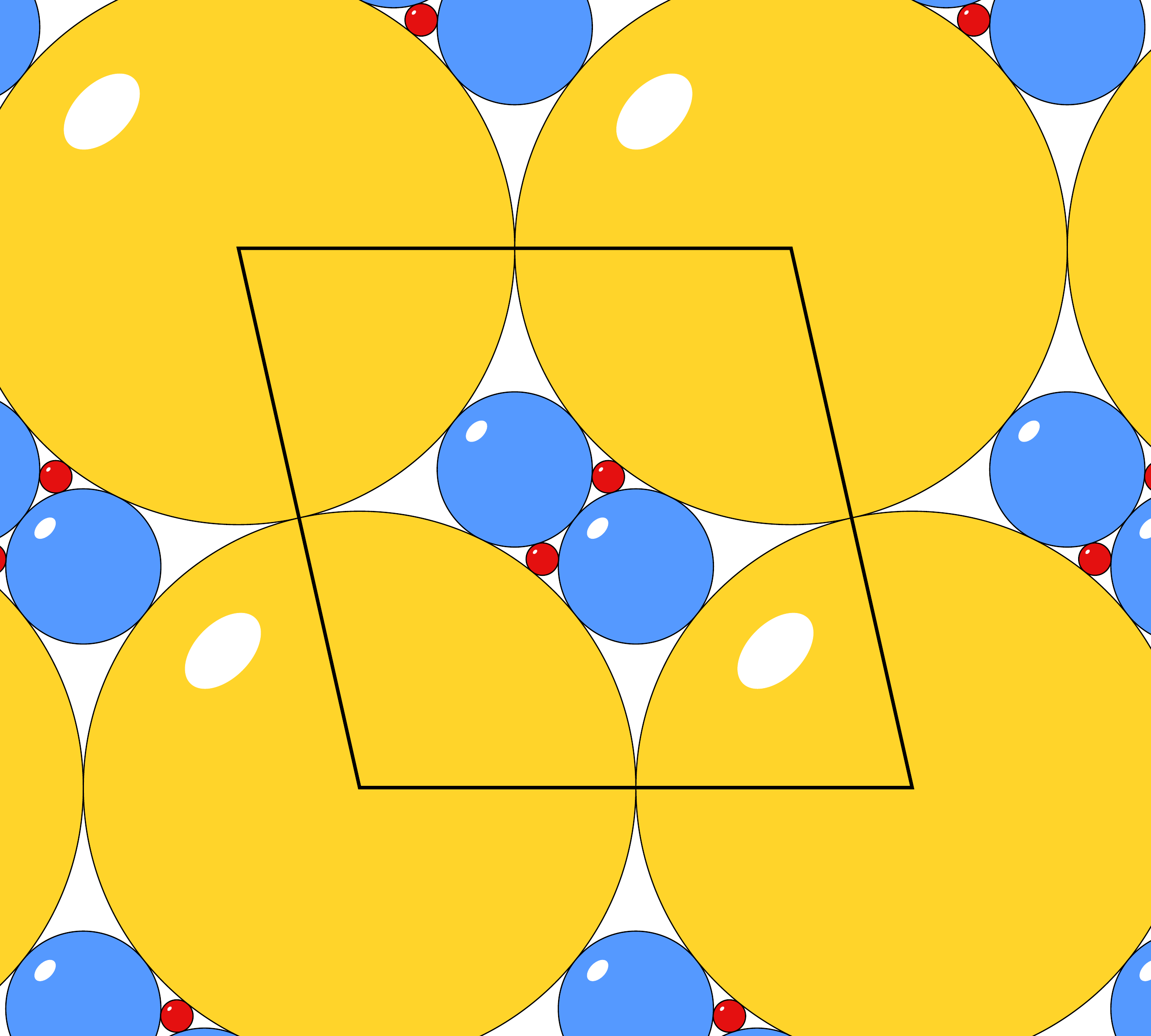}
\end{tabular}
\noindent
\begin{tabular}{lll}
  31 (H)\hfill 1rr / 11srrrs & 32 (H)\hfill 1rr / 11srrs & 33 (L)\hfill 1rr / 11srs1srs\\
  \includegraphics[width=0.3\textwidth]{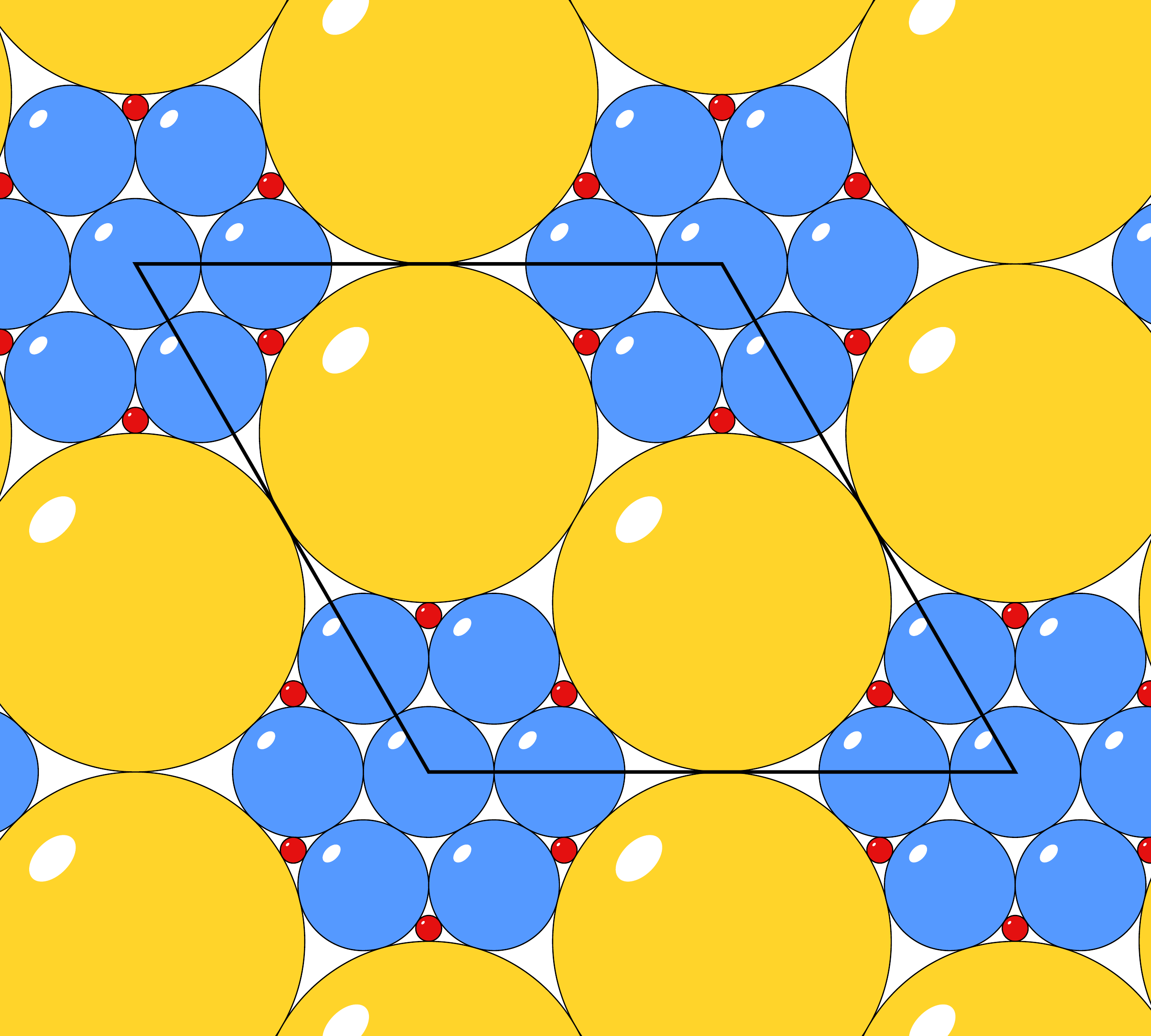} &
  \includegraphics[width=0.3\textwidth]{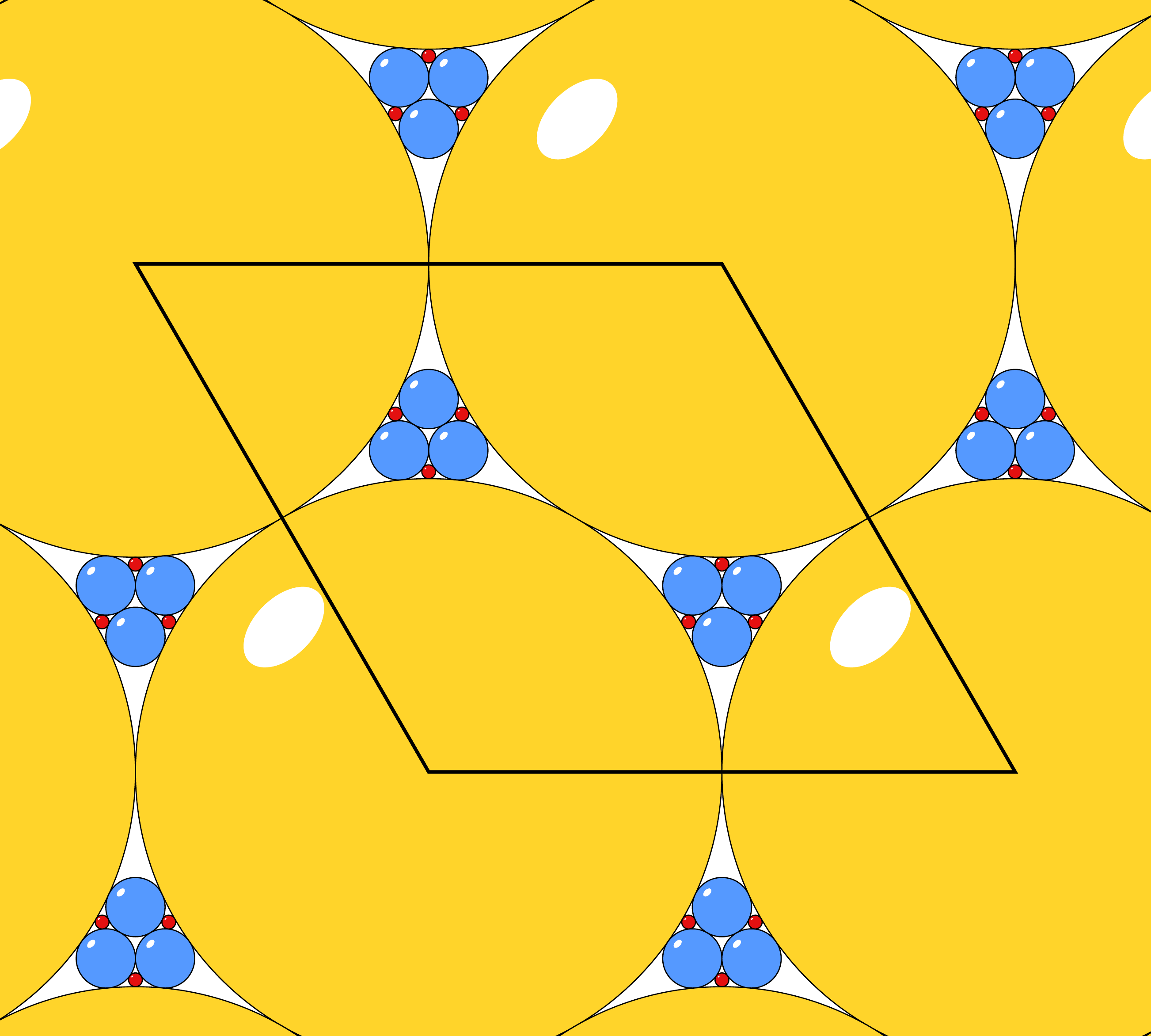} &
  \includegraphics[width=0.3\textwidth]{packing_1rr_11srs1srs.pdf}
\end{tabular}
\noindent
\begin{tabular}{lll}
  34 (H)\hfill 1rr / 1srrs1srs & 35 (H)\hfill 1rssr / 11ss & 36 (H)\hfill 1rsss / 11ss\\
  \includegraphics[width=0.3\textwidth]{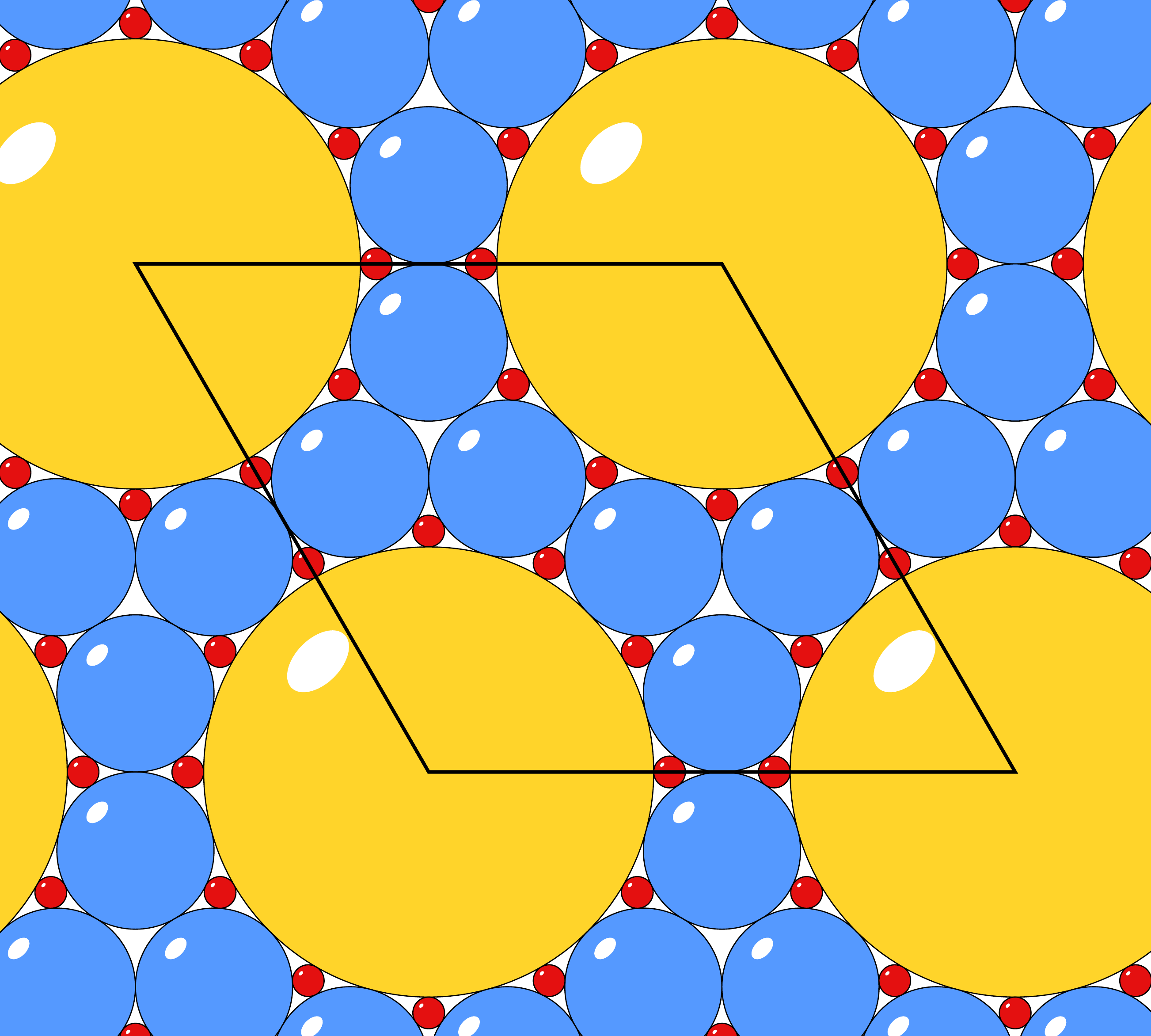} &
  \includegraphics[width=0.3\textwidth]{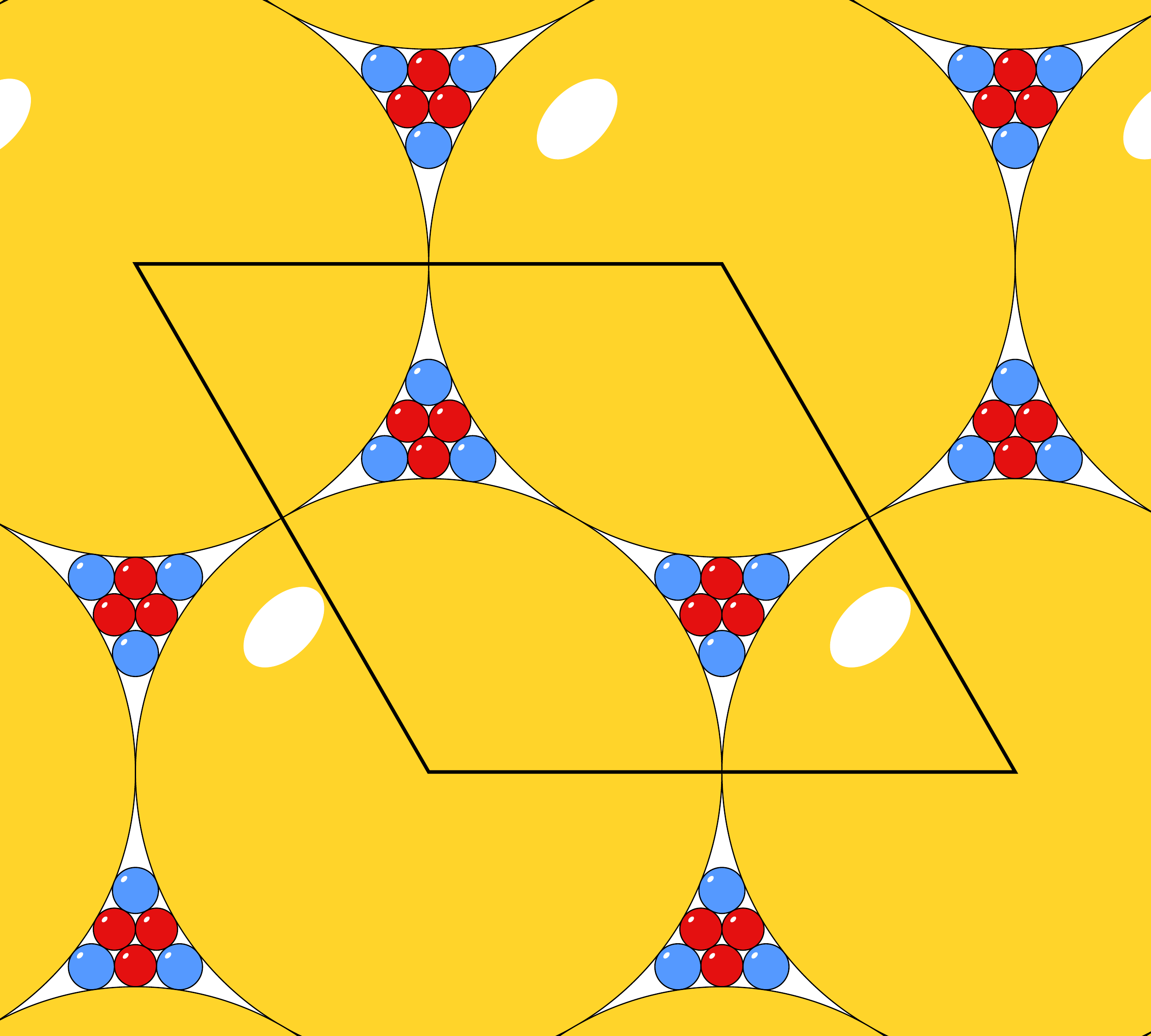} &
  \includegraphics[width=0.3\textwidth]{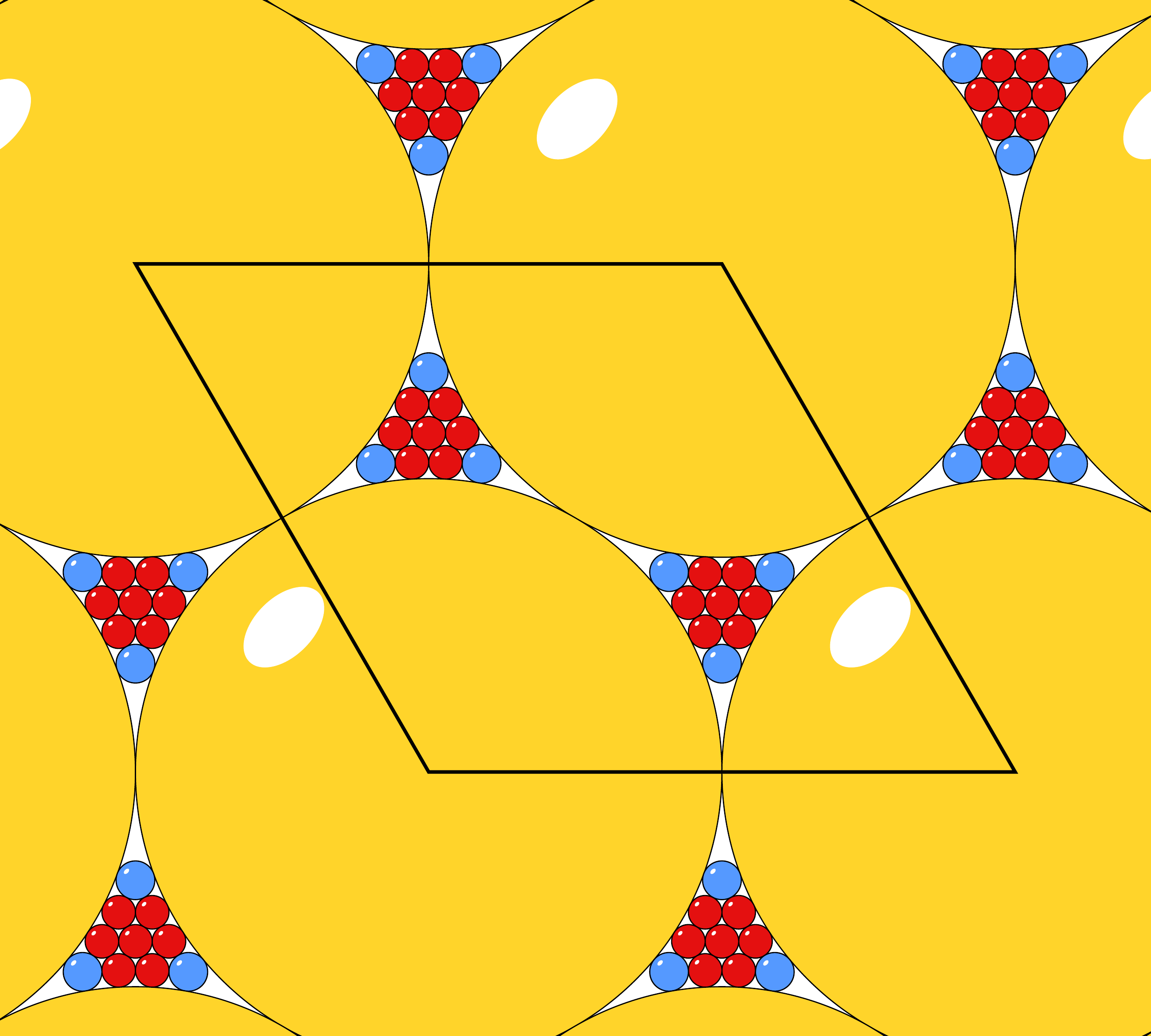}
\end{tabular}
\noindent
\begin{tabular}{lll}
  37 (S)\hfill rrr / 111rsr & 38 (H)\hfill rrr / 11rsr & 39 (H)\hfill rrr / 11rsrsr\\
  \includegraphics[width=0.3\textwidth]{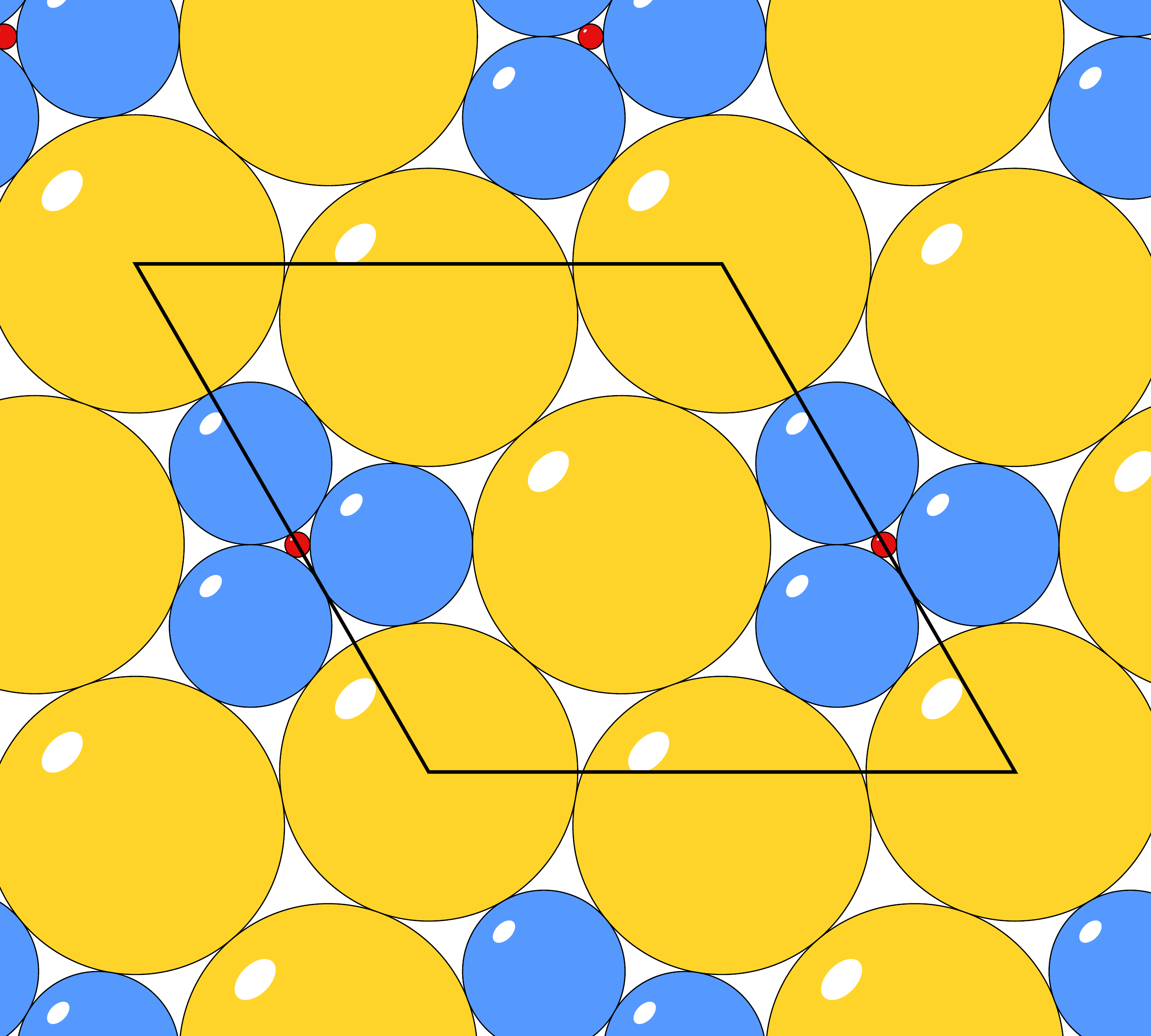} &
  \includegraphics[width=0.3\textwidth]{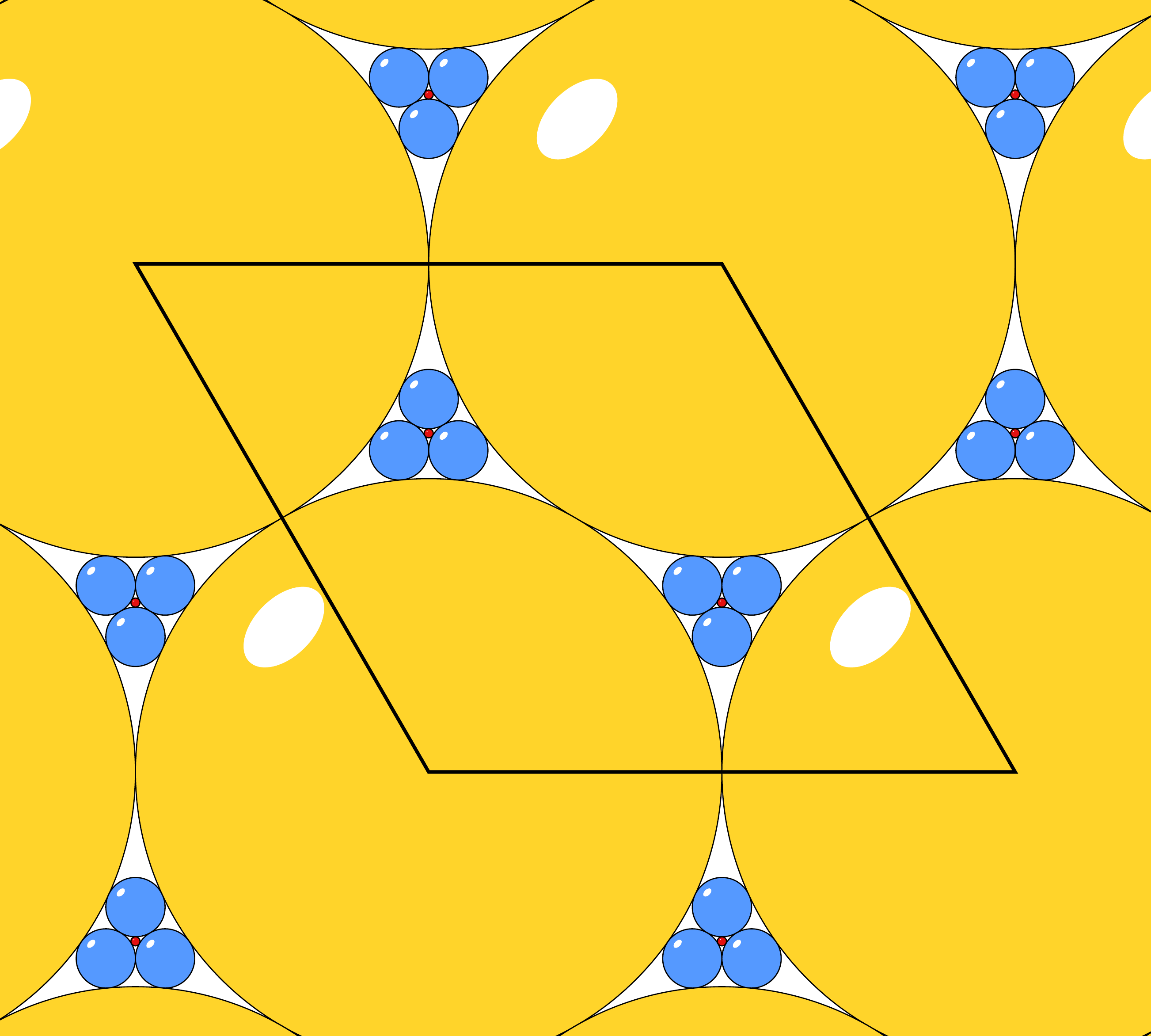} &
  \includegraphics[width=0.3\textwidth]{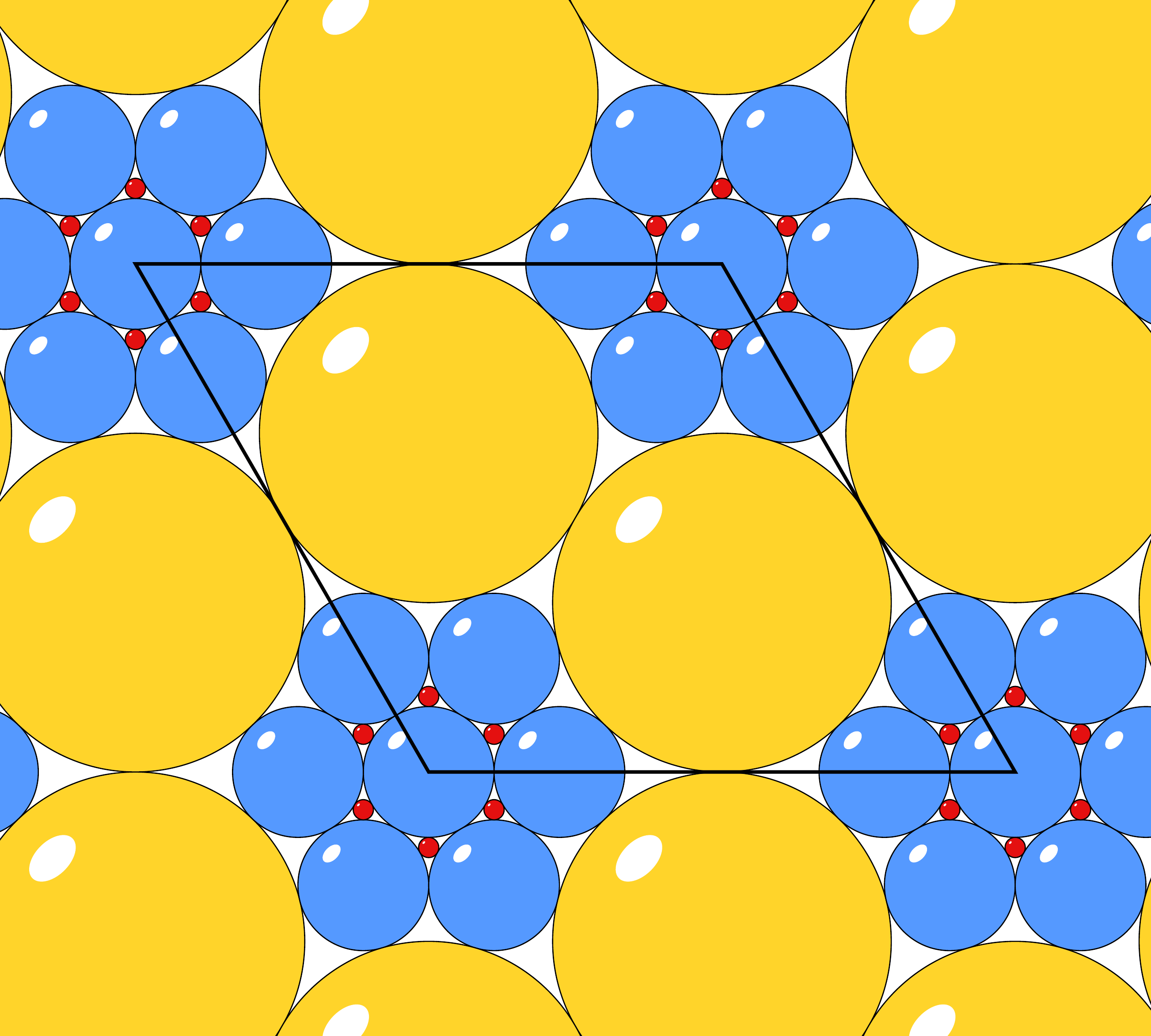}
\end{tabular}
\noindent
\begin{tabular}{lll}
  40 (H)\hfill rrr / 1r1rsr & 41 (S)\hfill rrss / 111rssr & 42 (H)\hfill rrss / 11rssr\\
  \includegraphics[width=0.3\textwidth]{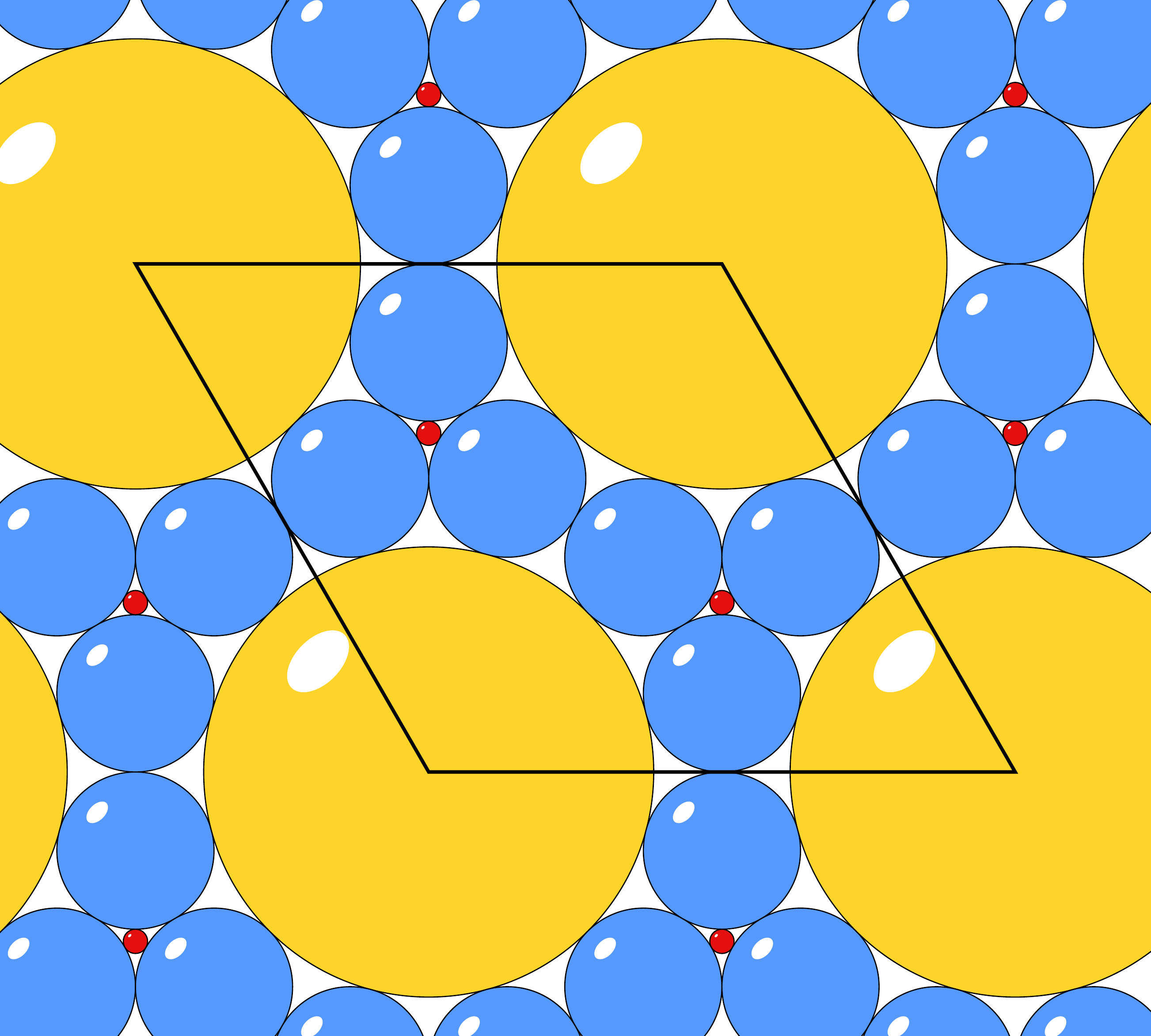} &
  \includegraphics[width=0.3\textwidth]{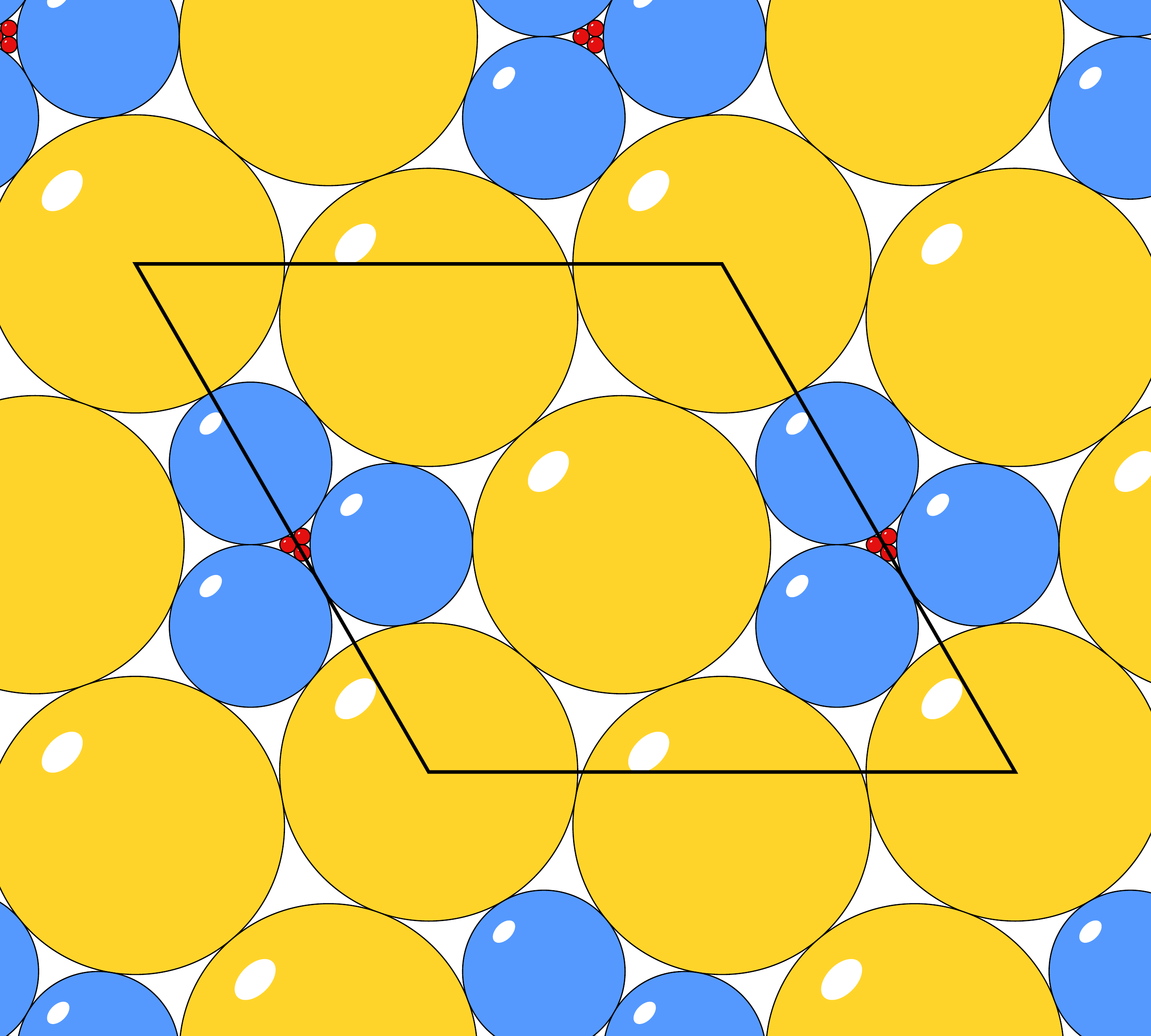} &
  \includegraphics[width=0.3\textwidth]{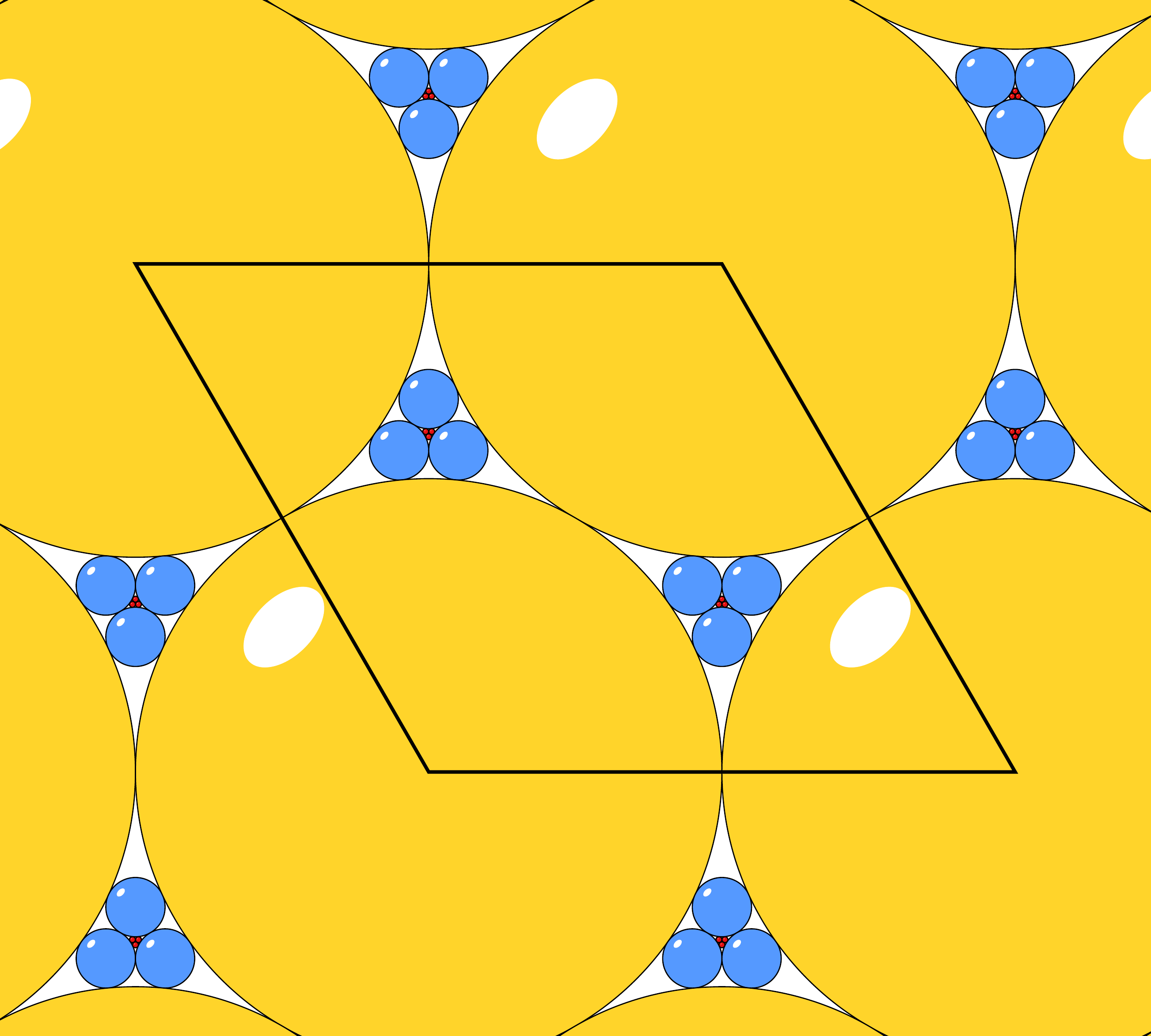}
\end{tabular}
\noindent
\begin{tabular}{lll}
  43 (H)\hfill rrss / 11rssrssr & 44 (H)\hfill rrss / 1r1rssr & 45 (L)\hfill 111r / 111s1s\\
  \includegraphics[width=0.3\textwidth]{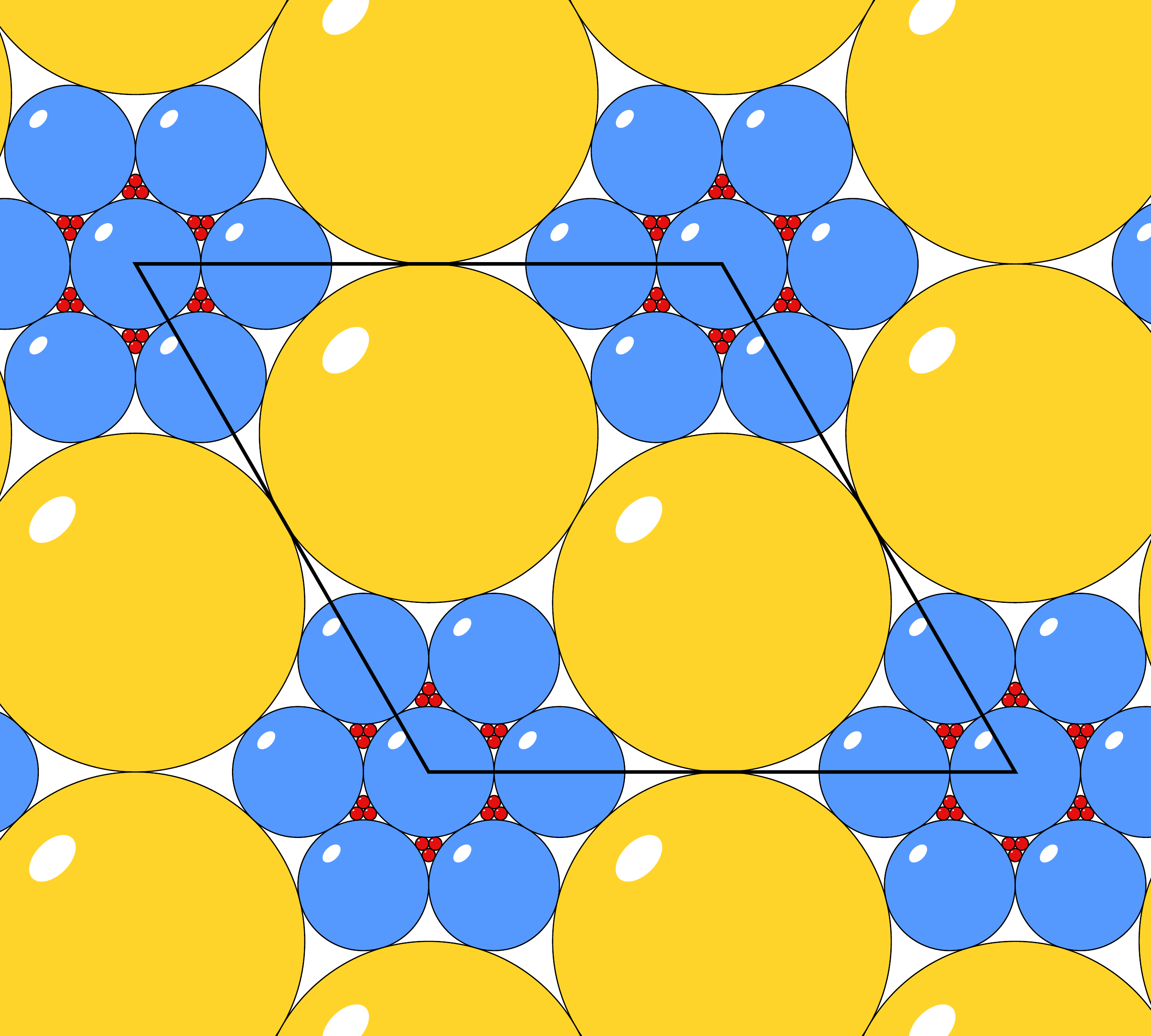} &
  \includegraphics[width=0.3\textwidth]{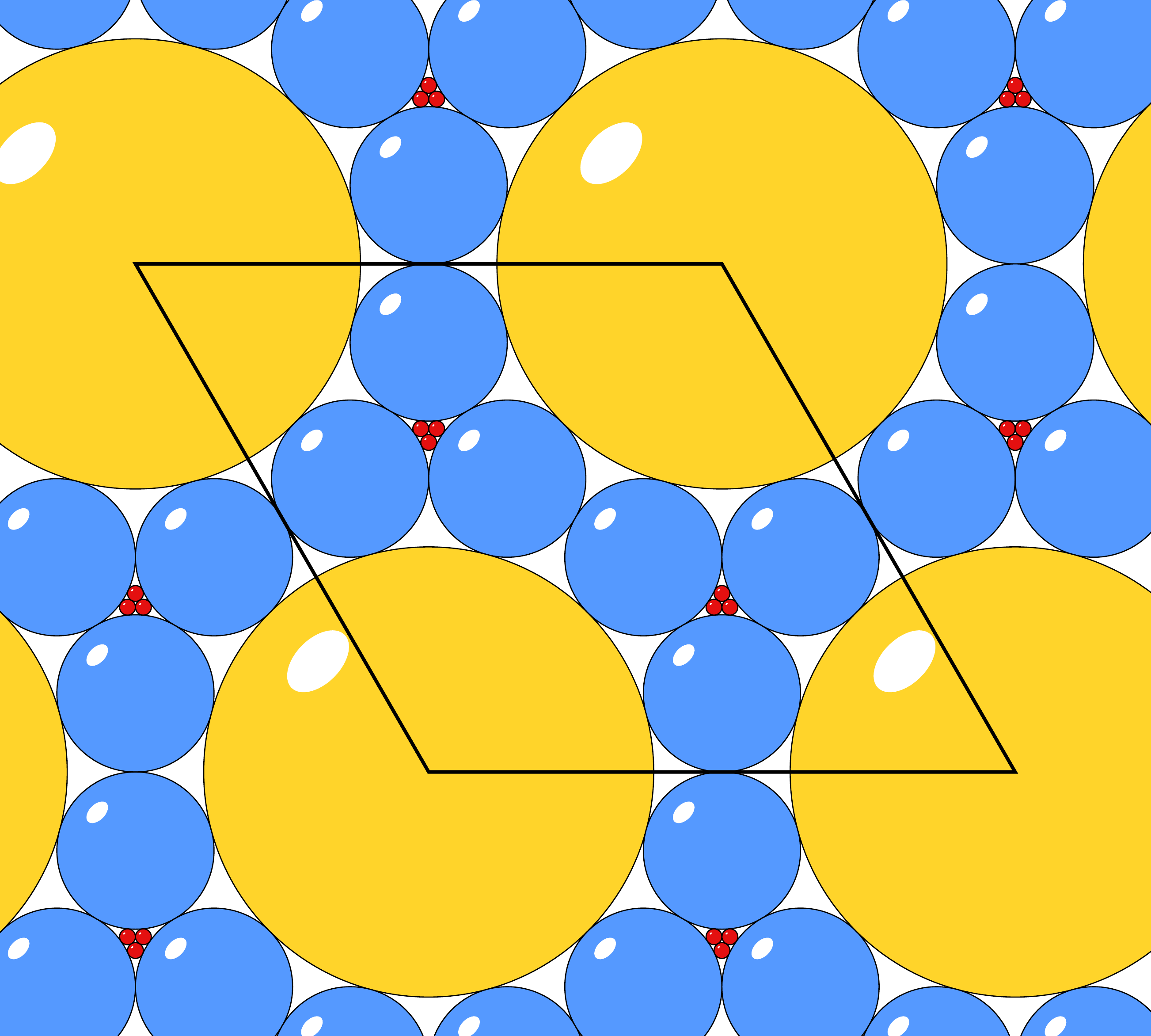} &
  \includegraphics[width=0.3\textwidth]{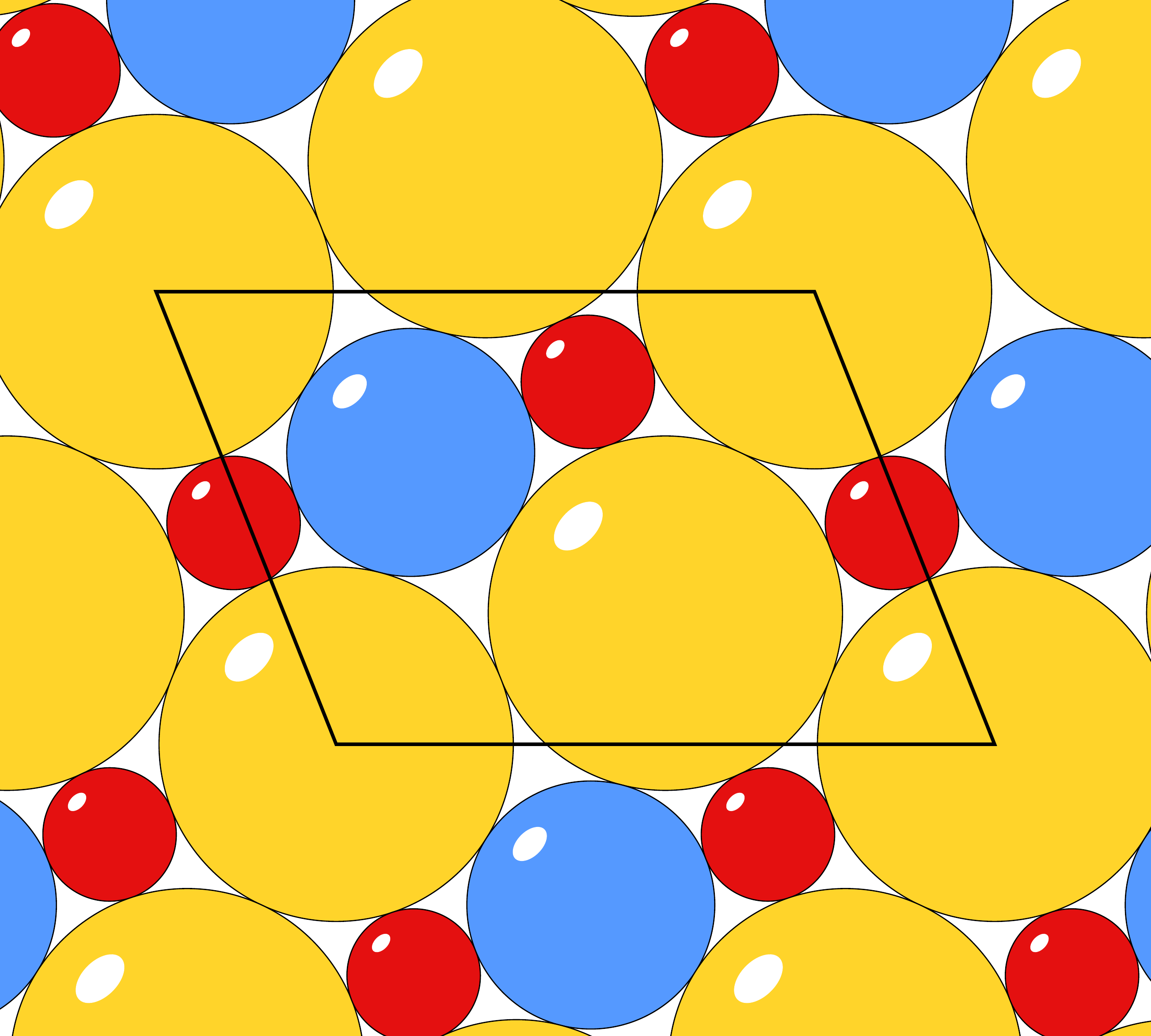}
\end{tabular}
\noindent
\begin{tabular}{lll}
  46 (L)\hfill 111r / 11r1s & 47 (H)\hfill 111r / 1r1r1s & 48 (S)\hfill 111r / 1rr1s\\
  \includegraphics[width=0.3\textwidth]{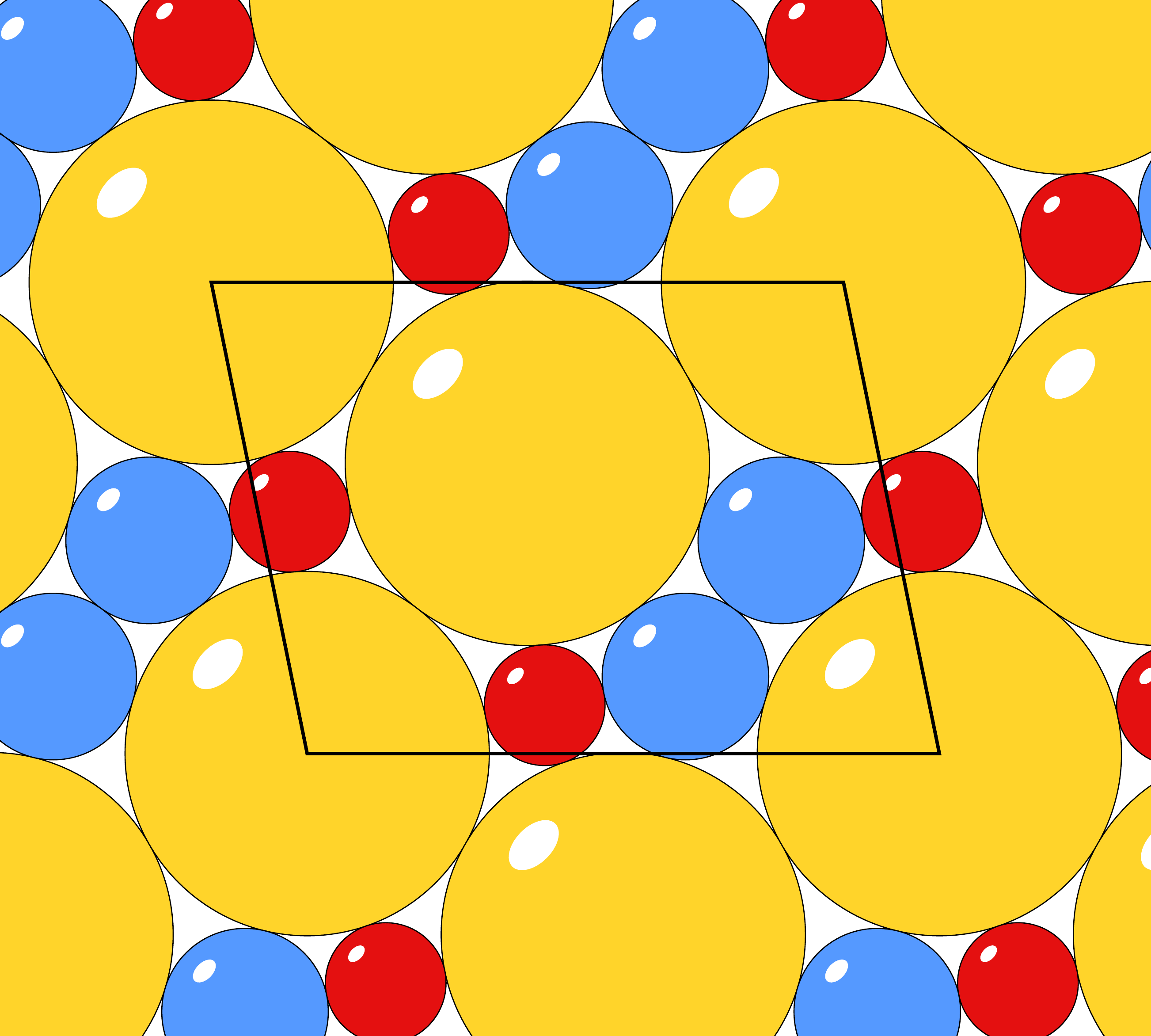} &
  \includegraphics[width=0.3\textwidth]{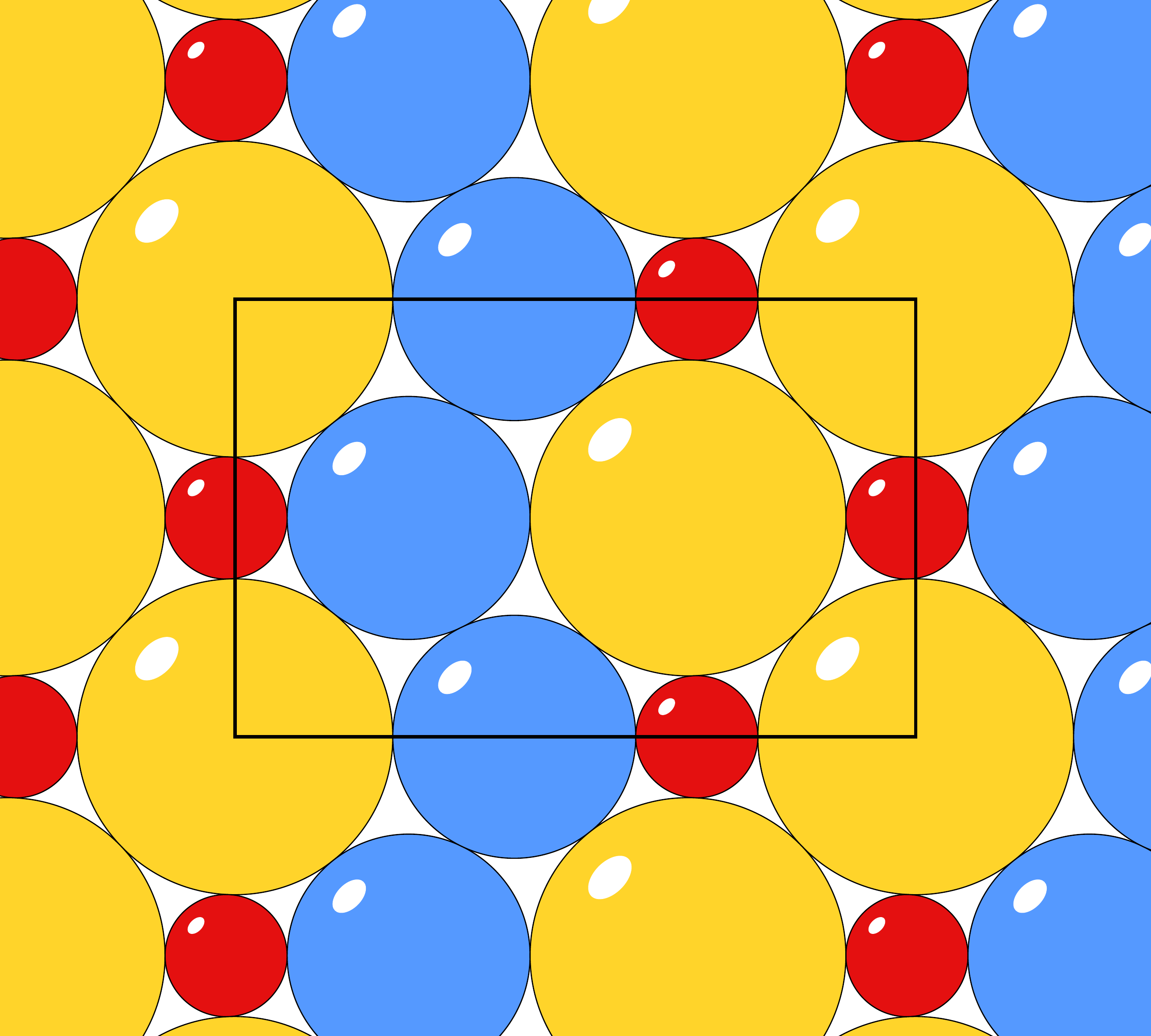} &
  \includegraphics[width=0.3\textwidth]{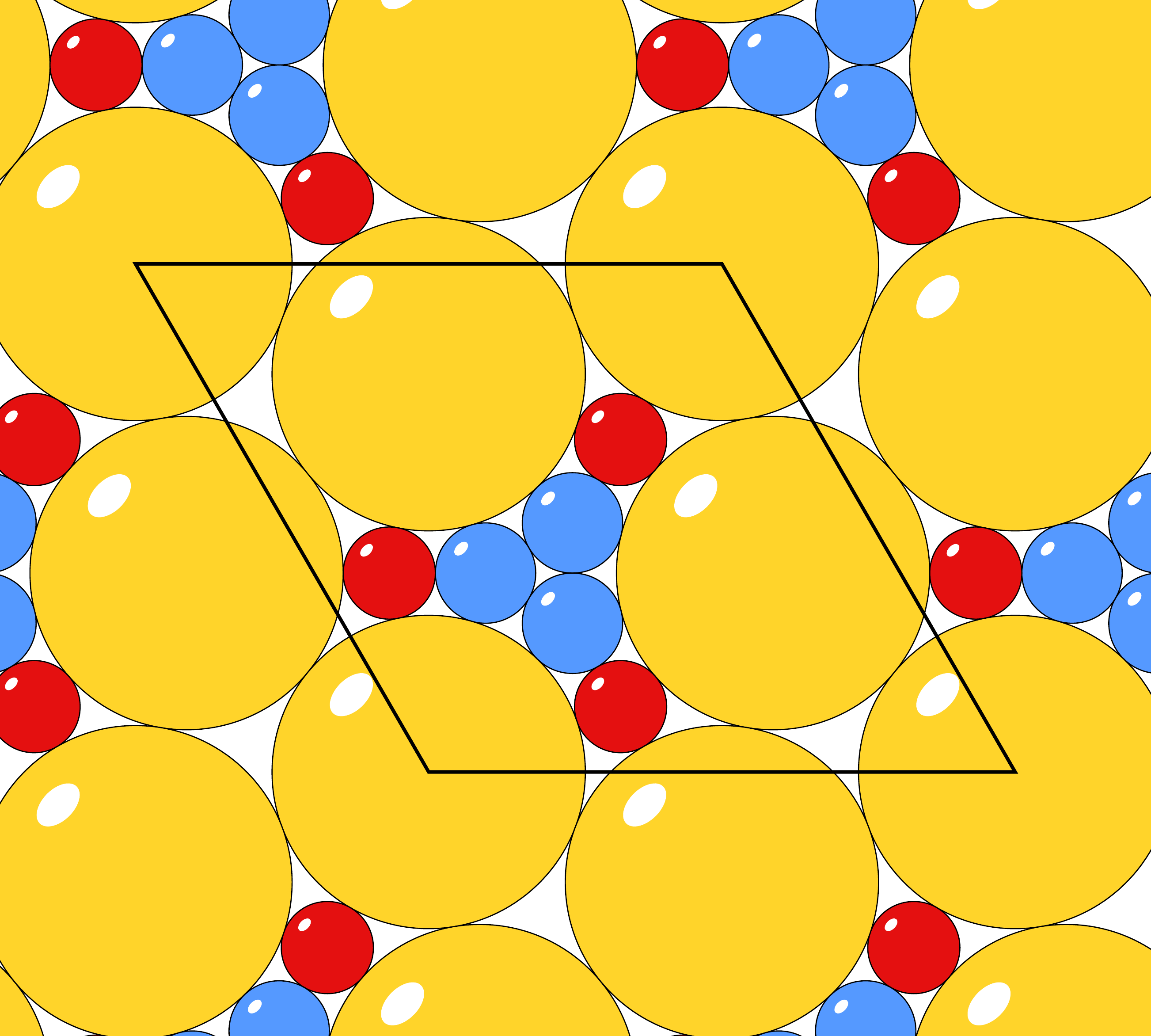}
\end{tabular}
\noindent
\begin{tabular}{lll}
  49 (H)\hfill 111r / 1rrr1s & 50 (S)\hfill 111r / 1s1s1s & 51 (L)\hfill 111rr / 1rrrrs\\
  \includegraphics[width=0.3\textwidth]{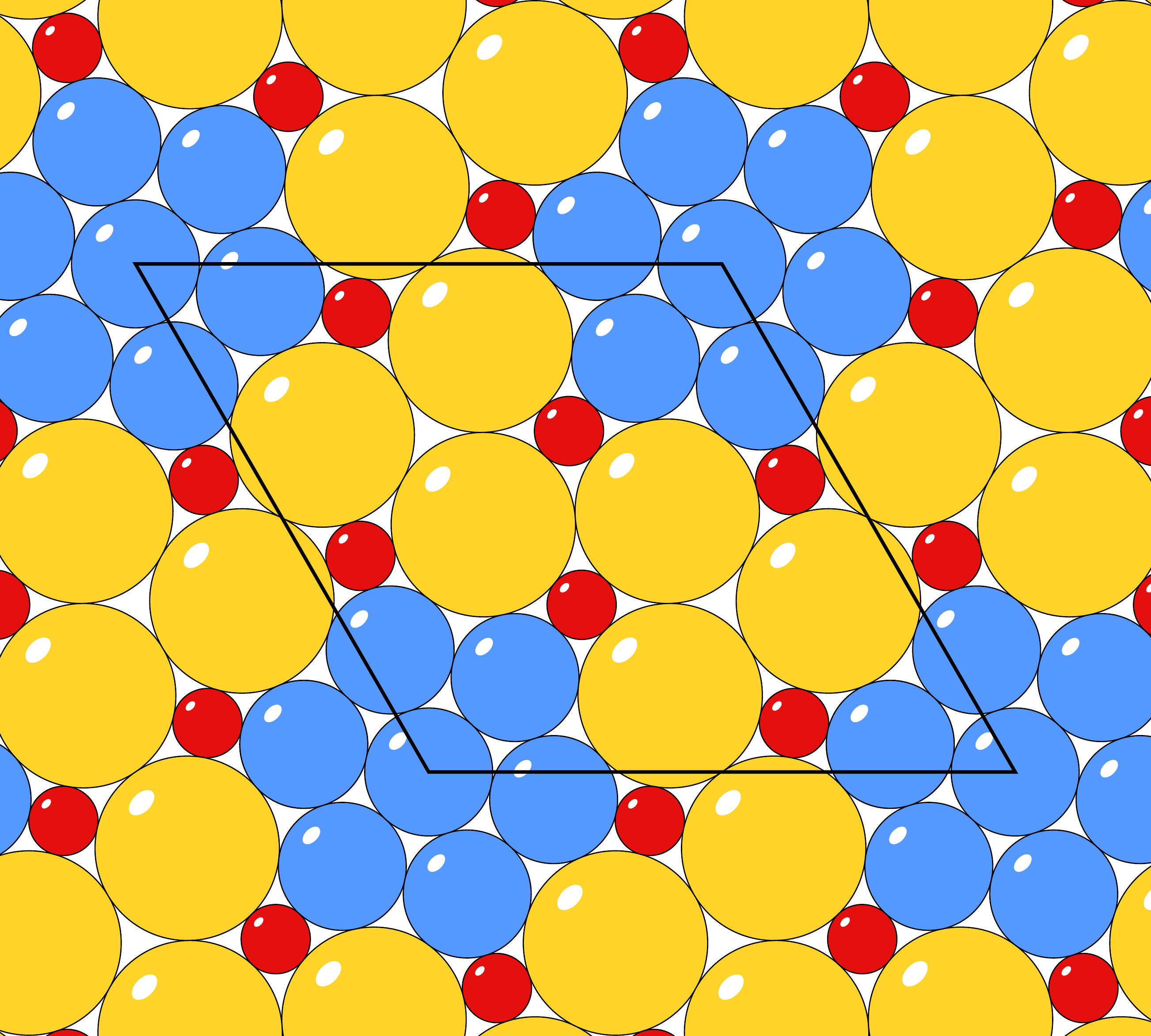} &
  \includegraphics[width=0.3\textwidth]{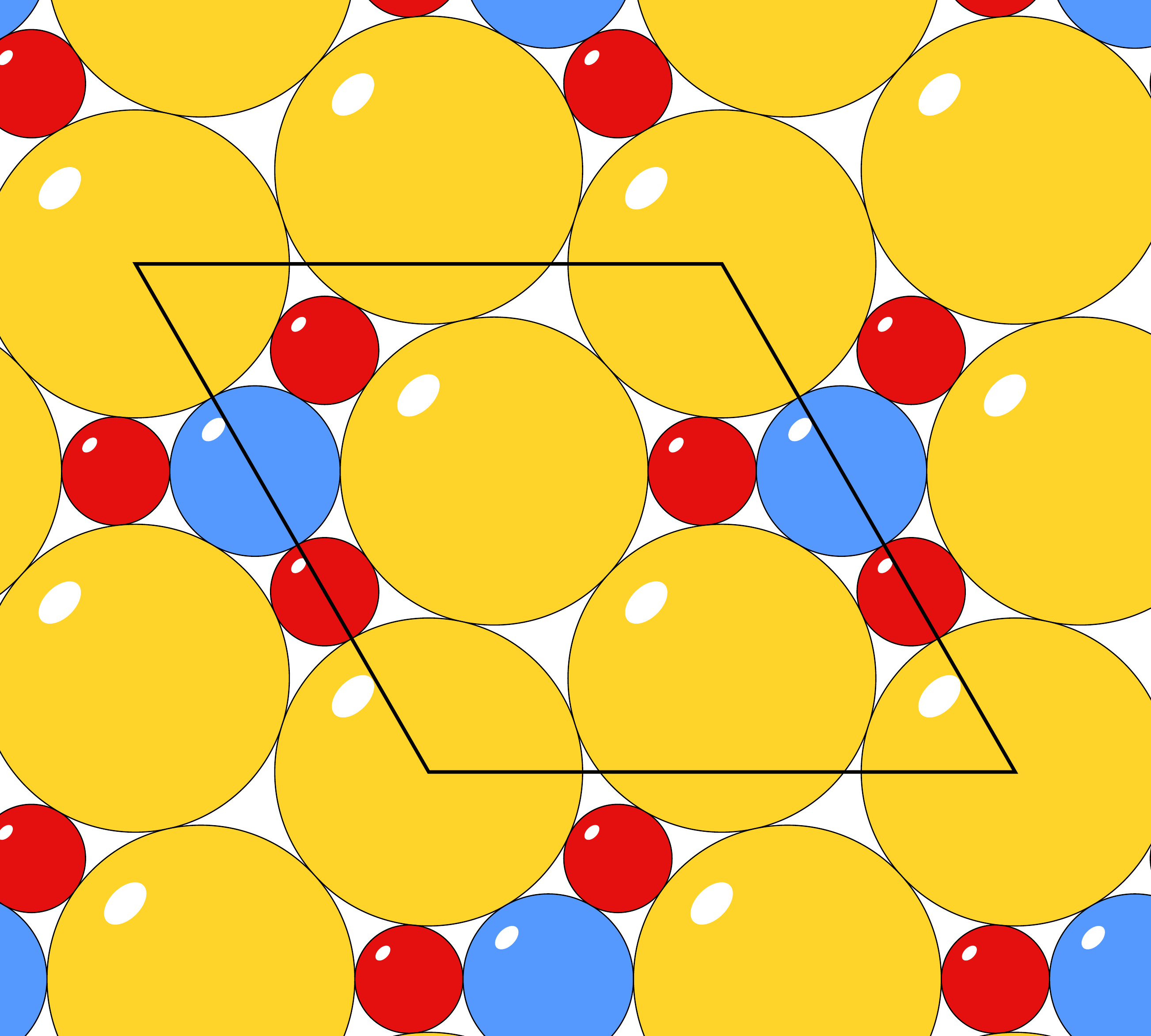} &
  \includegraphics[width=0.3\textwidth]{packing_111rr_1rrrrs.pdf}
\end{tabular}
\noindent
\begin{tabular}{lll}
  52 (H)\hfill 111rr / 1srrrs & 53 (H)\hfill 11r1r / 1r1s1s & 54 (H)\hfill 11r1r / 1s1s1s\\
  \includegraphics[width=0.3\textwidth]{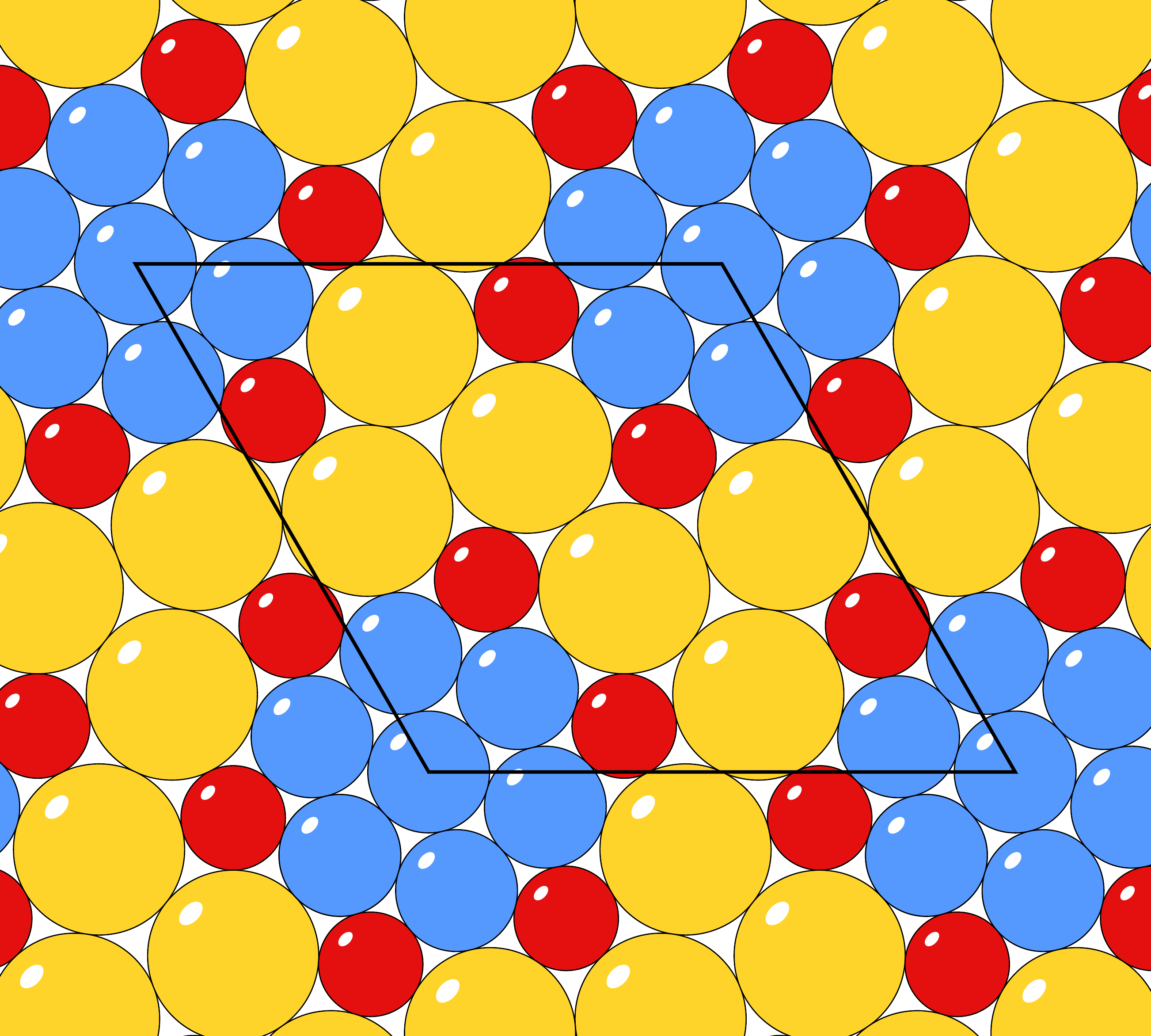} &
  \includegraphics[width=0.3\textwidth]{packing_11r1r_1r1s1s.pdf} &
  \includegraphics[width=0.3\textwidth]{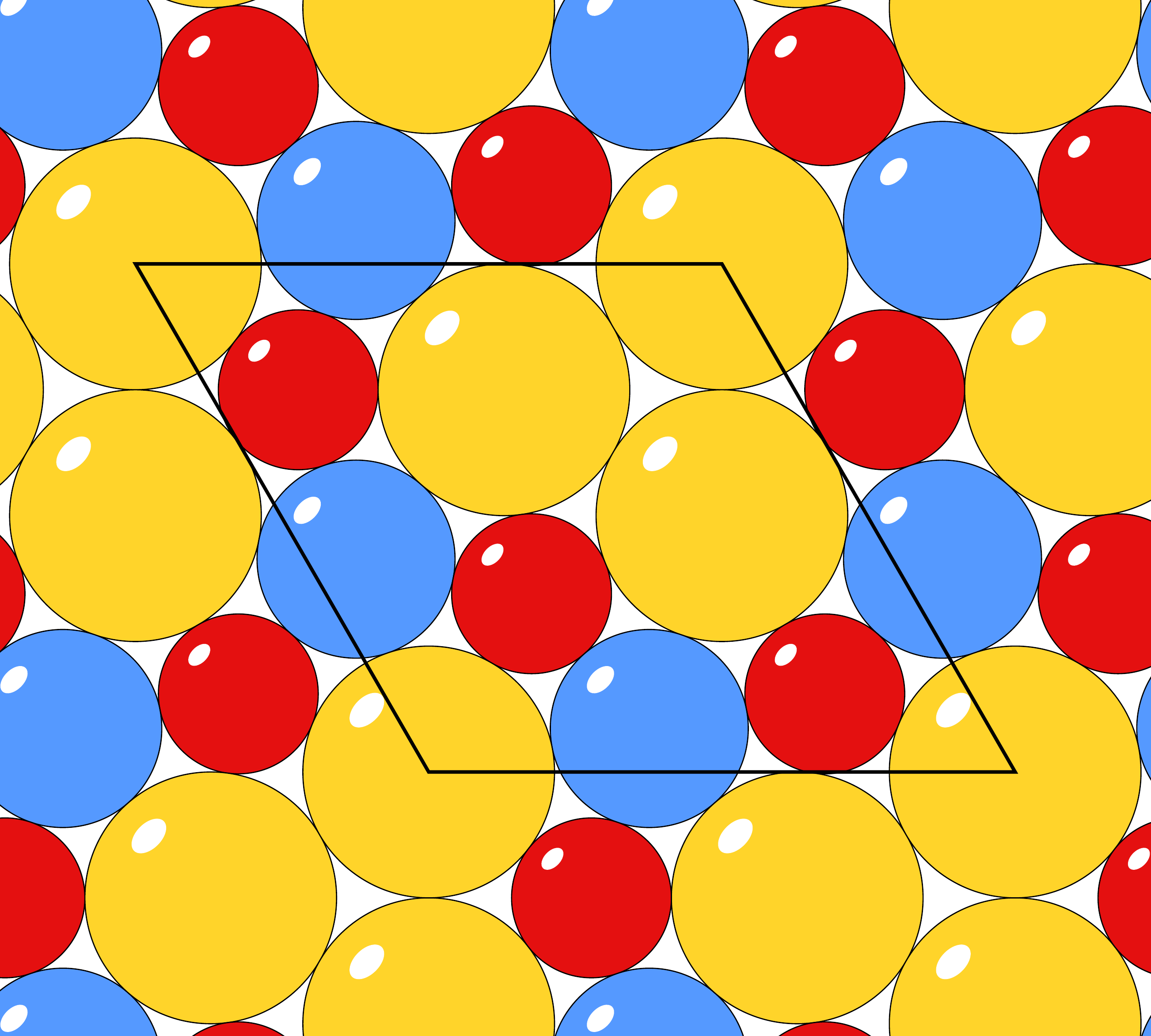}
\end{tabular}
\noindent
\begin{tabular}{lll}
  55 (L)\hfill 11r1s / 111s1s & 56 (L)\hfill 11r1s / 1r1r1s & 57 (H)\hfill 11r1s / 1rrr1s\\
  \includegraphics[width=0.3\textwidth]{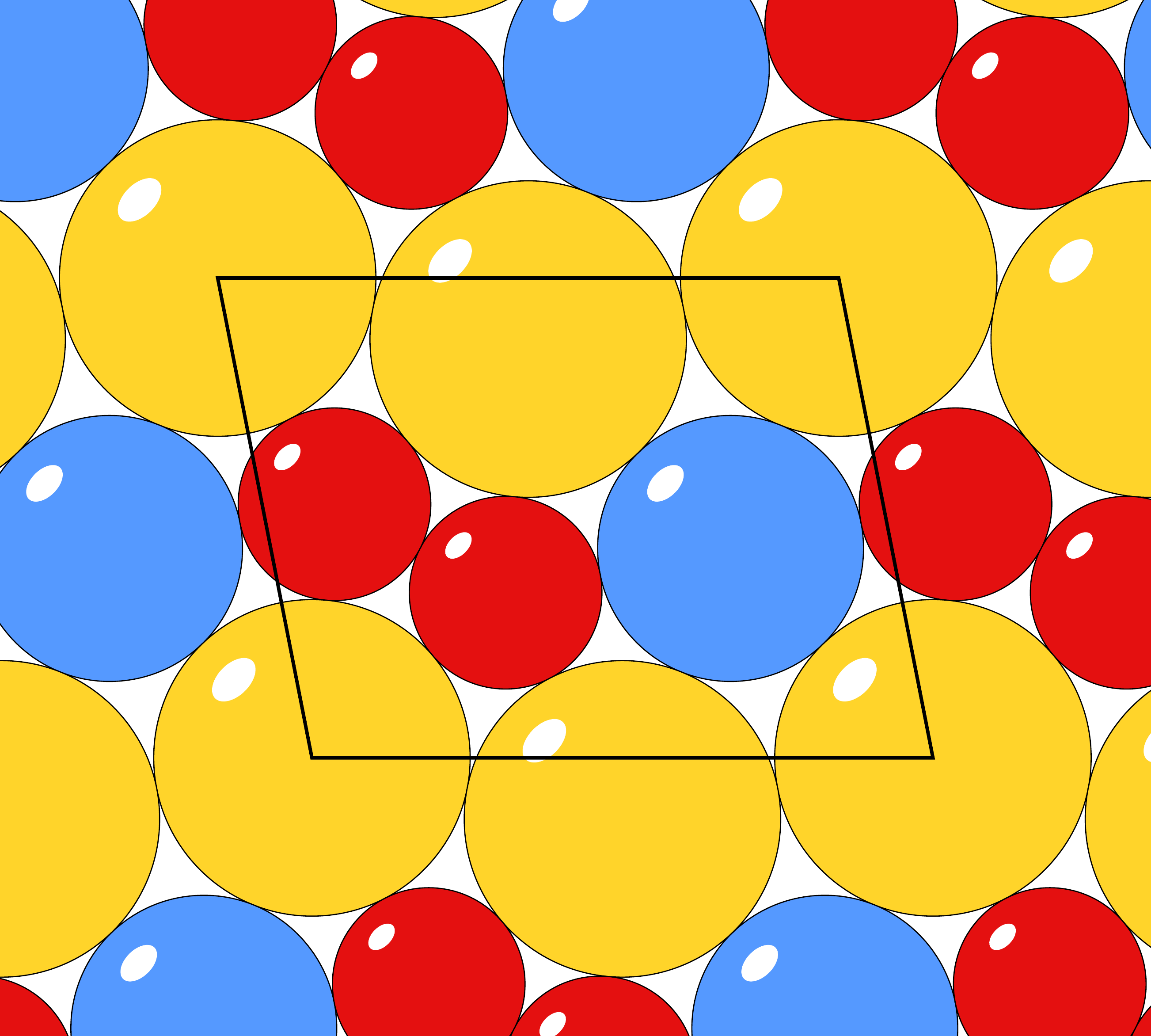} &
  \includegraphics[width=0.3\textwidth]{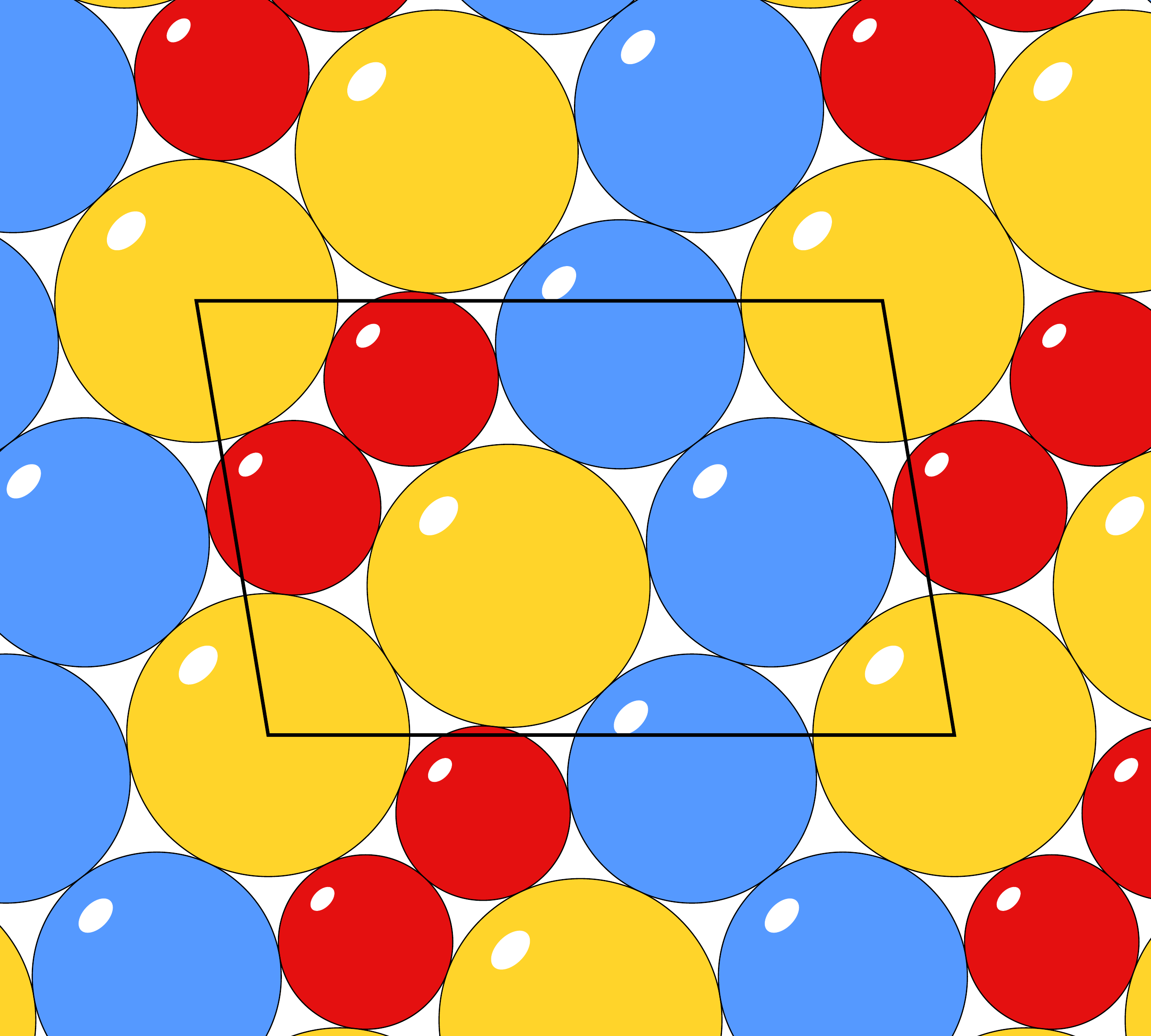} &
  \includegraphics[width=0.3\textwidth]{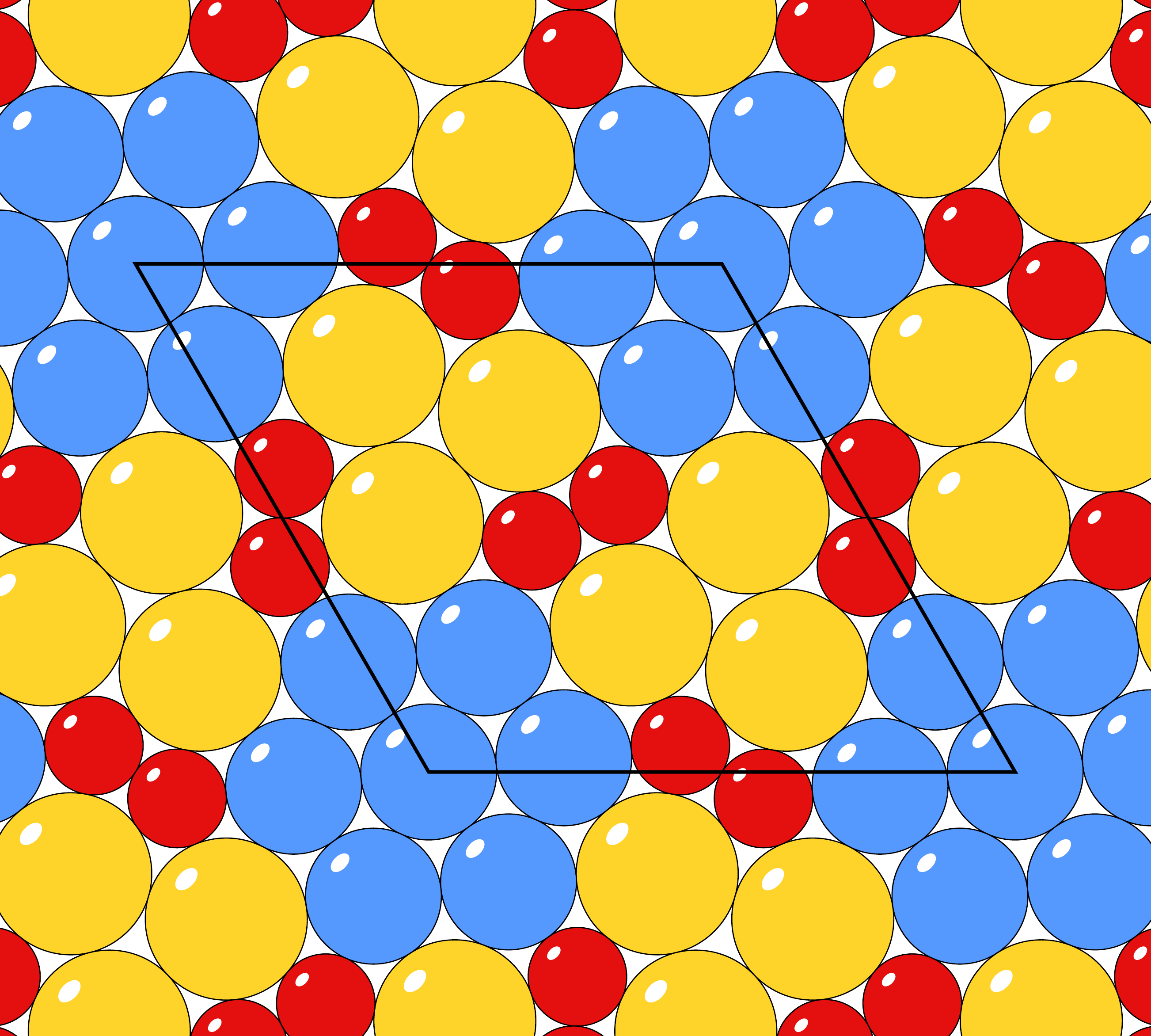}
\end{tabular}
\noindent
\begin{tabular}{lll}
  58 (H)\hfill 11r1s / 1s1s1s & 59 (L)\hfill 11rr / 111srs & 60 (H)\hfill 11rr / 11srrs\\
  \includegraphics[width=0.3\textwidth]{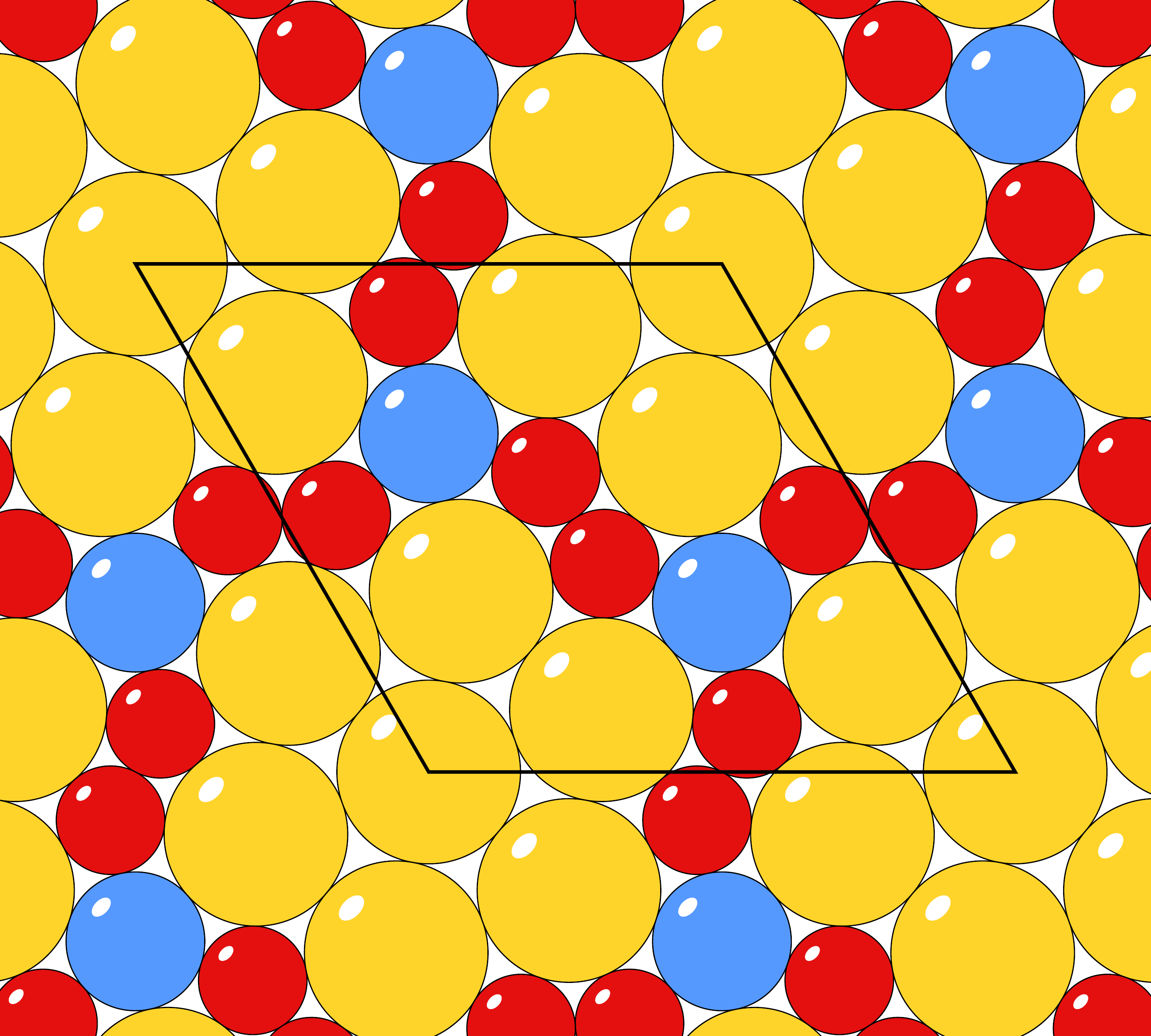} &
  \includegraphics[width=0.3\textwidth]{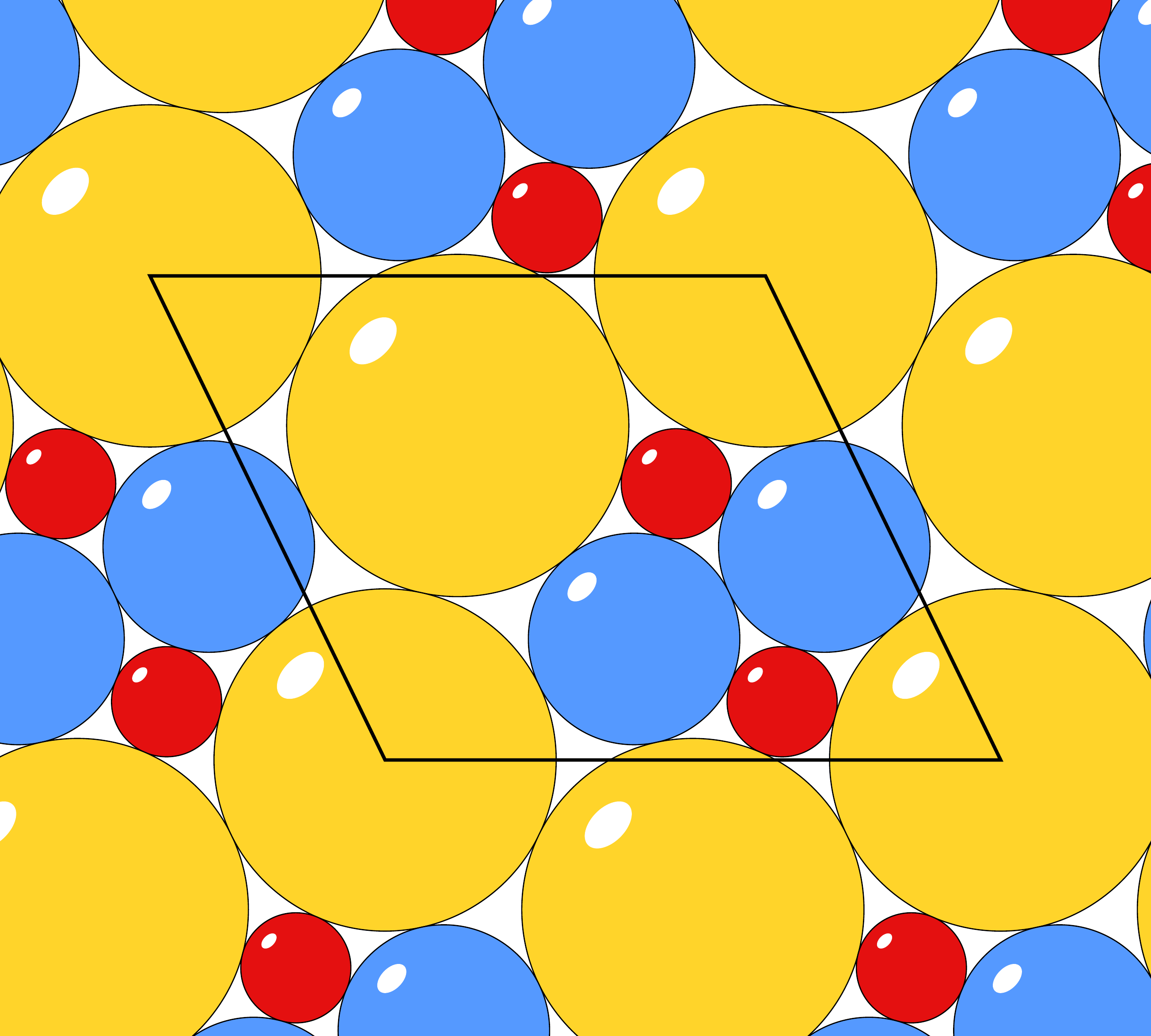} &
  \includegraphics[width=0.3\textwidth]{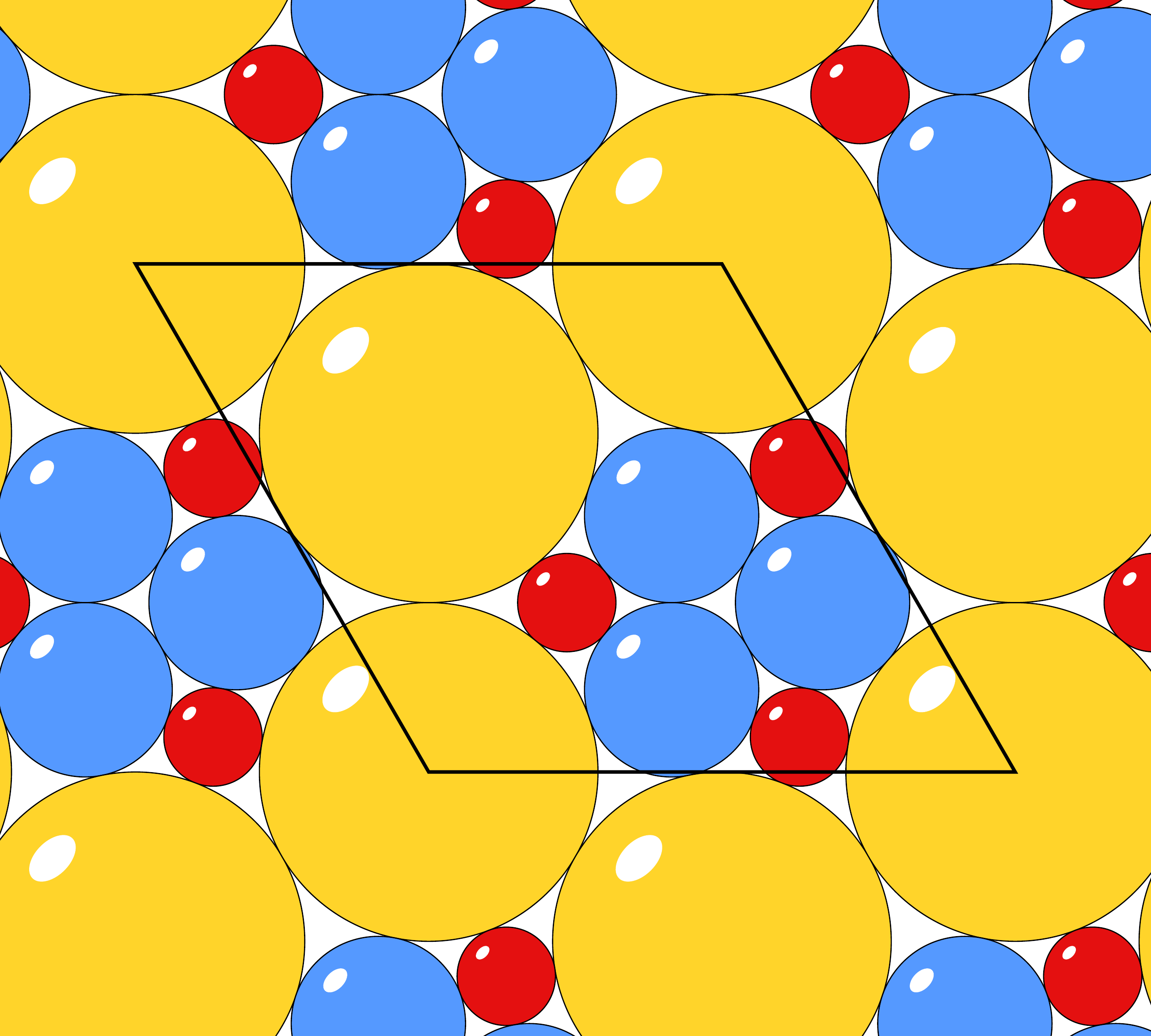}
\end{tabular}
\noindent
\begin{tabular}{lll}
  61 (E)\hfill 11rr / 11srs & 62 (L)\hfill 11rr / 1r1rs & 63 (H)\hfill 11rr / 1rr1rs\\
  \includegraphics[width=0.3\textwidth]{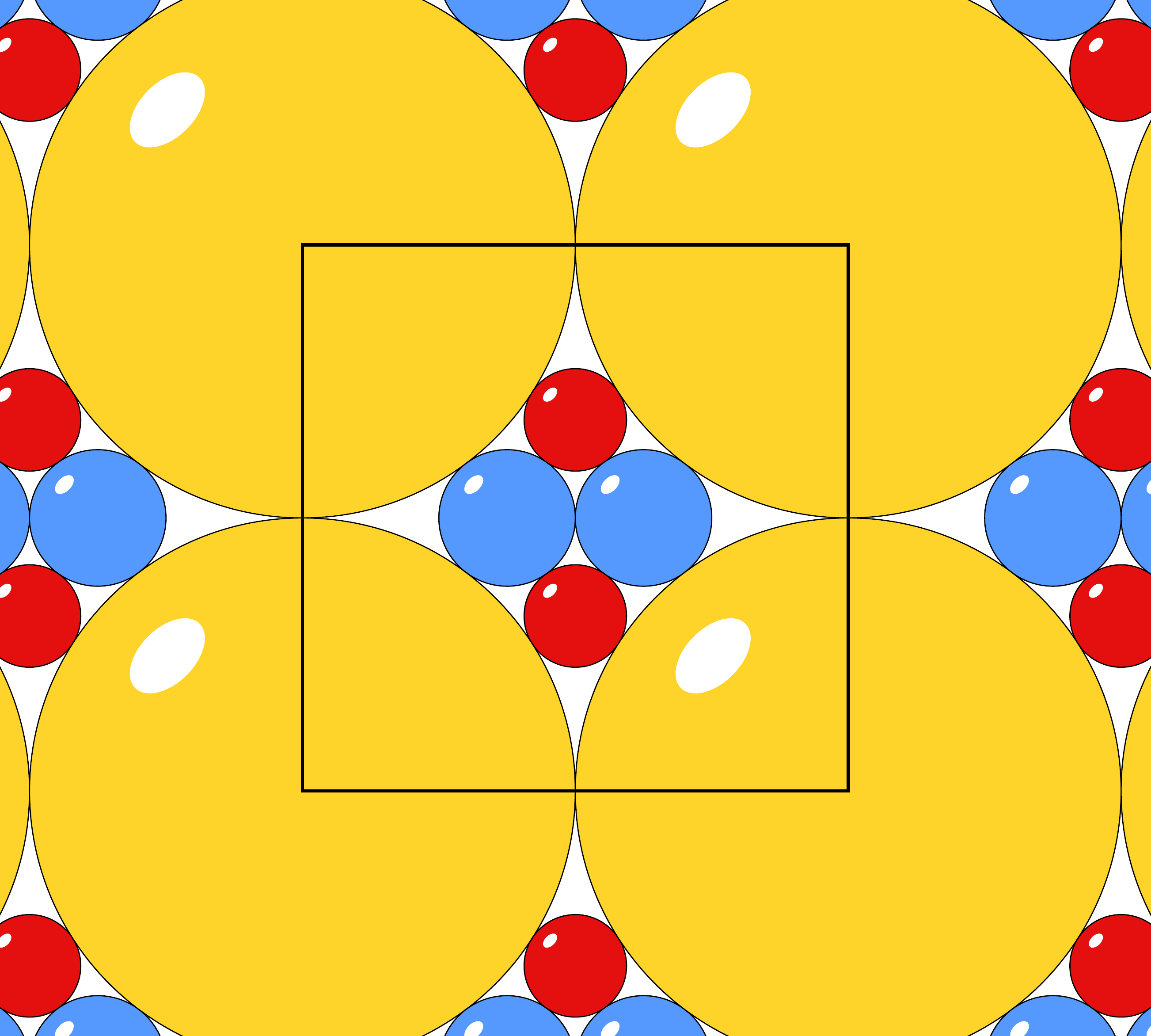} &
  \includegraphics[width=0.3\textwidth]{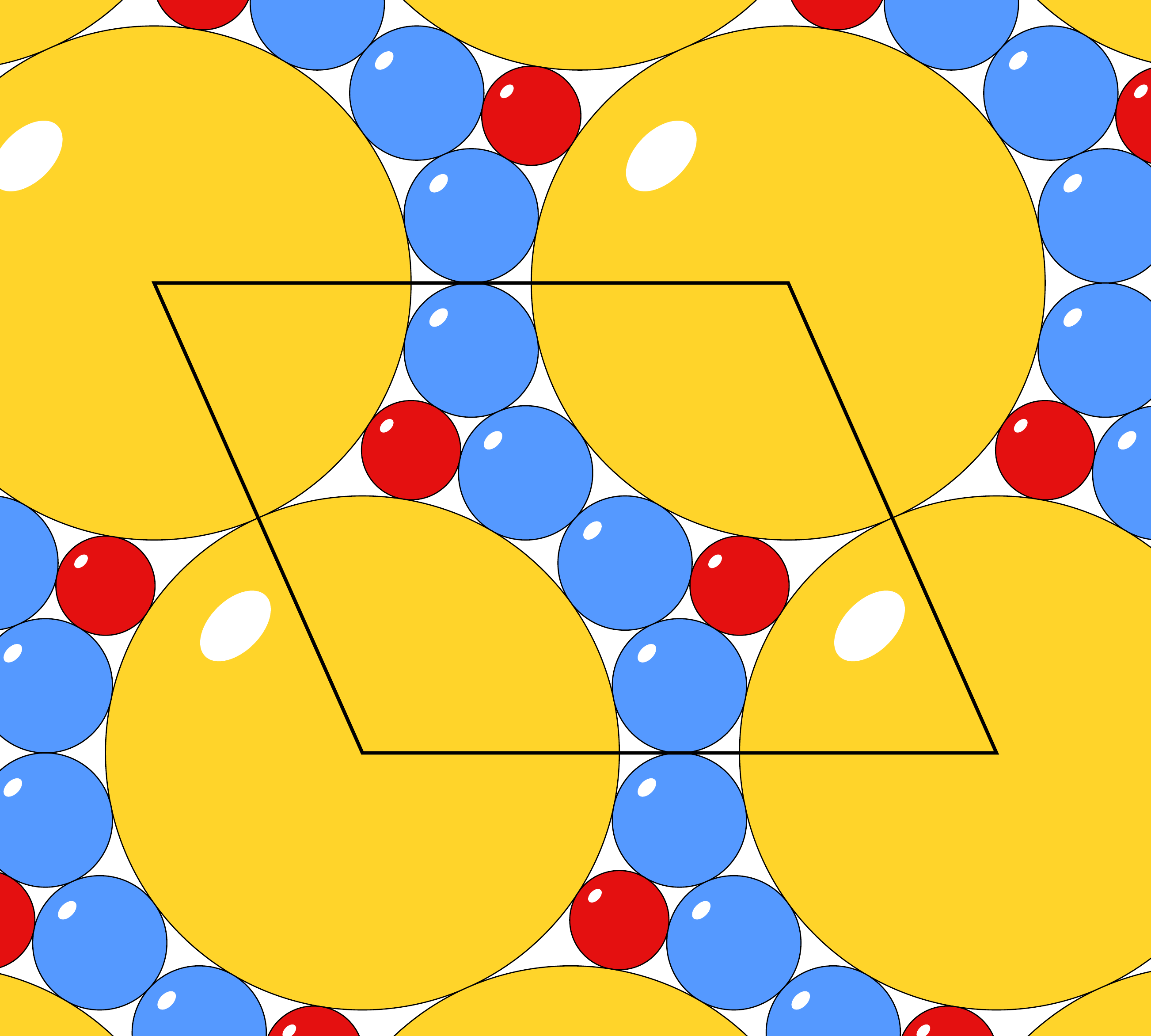} &
  \includegraphics[width=0.3\textwidth]{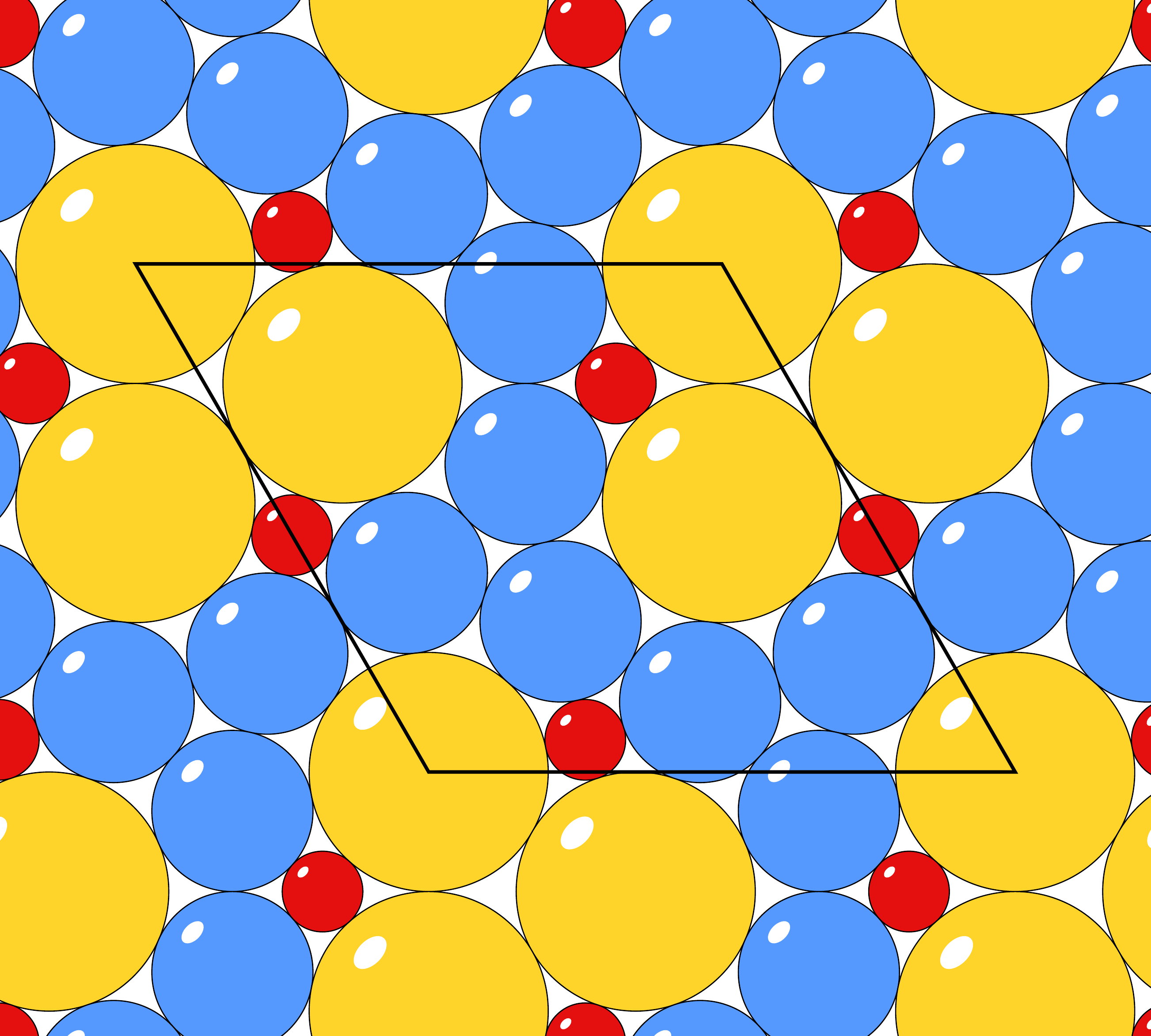}
\end{tabular}
\noindent
\begin{tabular}{lll}
  64 (L)\hfill 11rr / 1rrrrs & 65 (H)\hfill 11rr / 1srrrs & 66 (L)\hfill 11rrr / 1srsrs\\
  \includegraphics[width=0.3\textwidth]{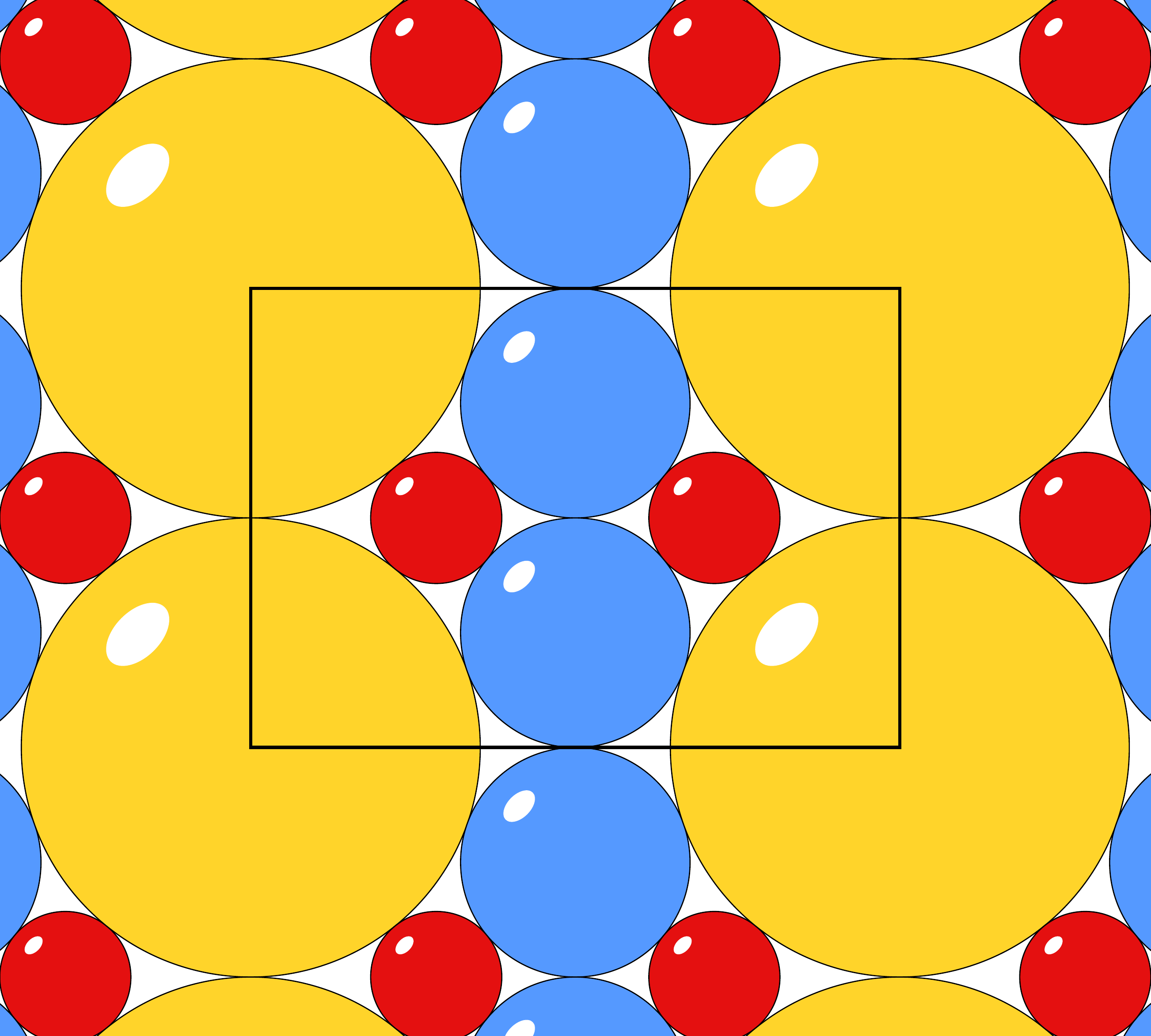} &
  \includegraphics[width=0.3\textwidth]{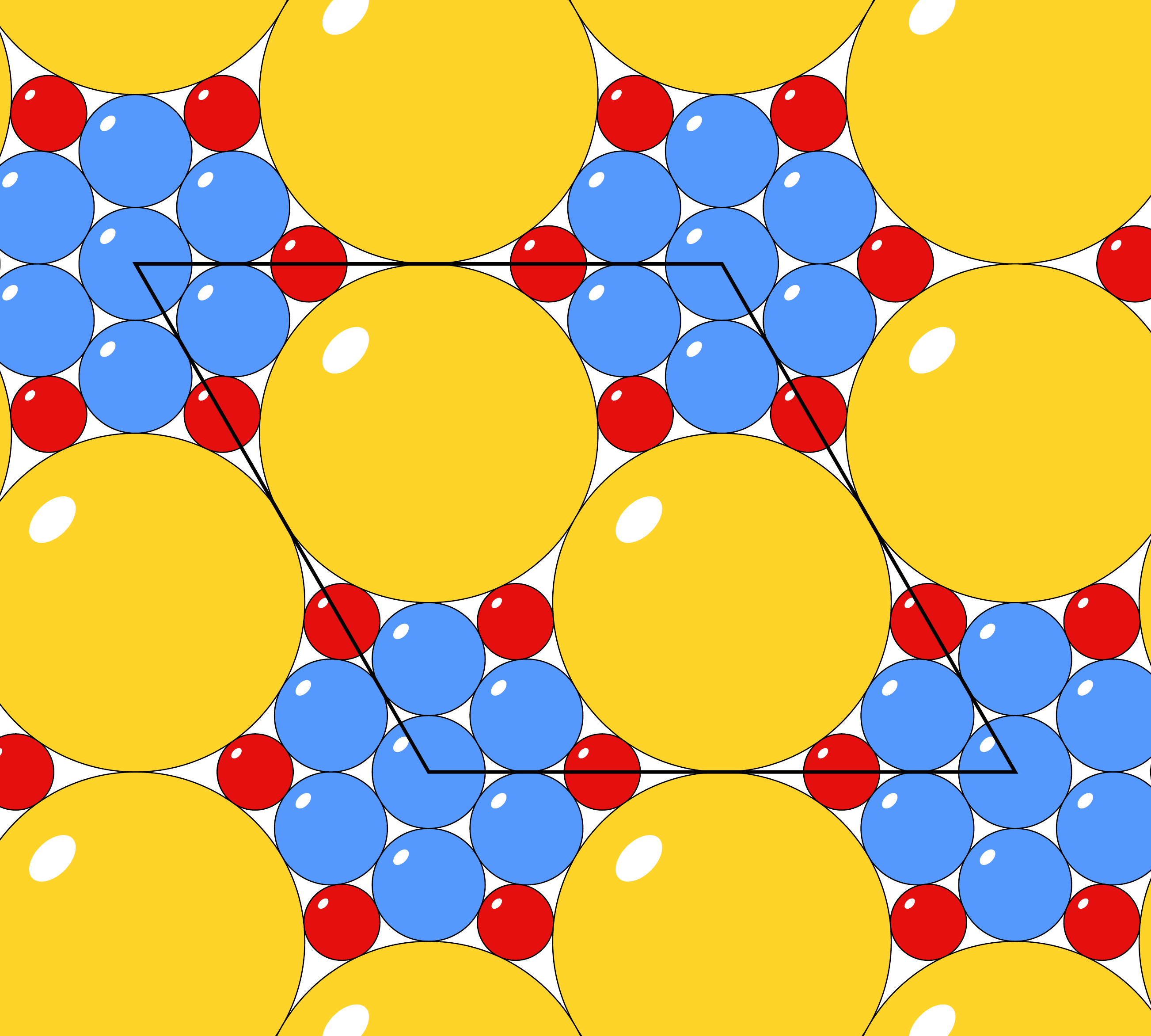} &
  \includegraphics[width=0.3\textwidth]{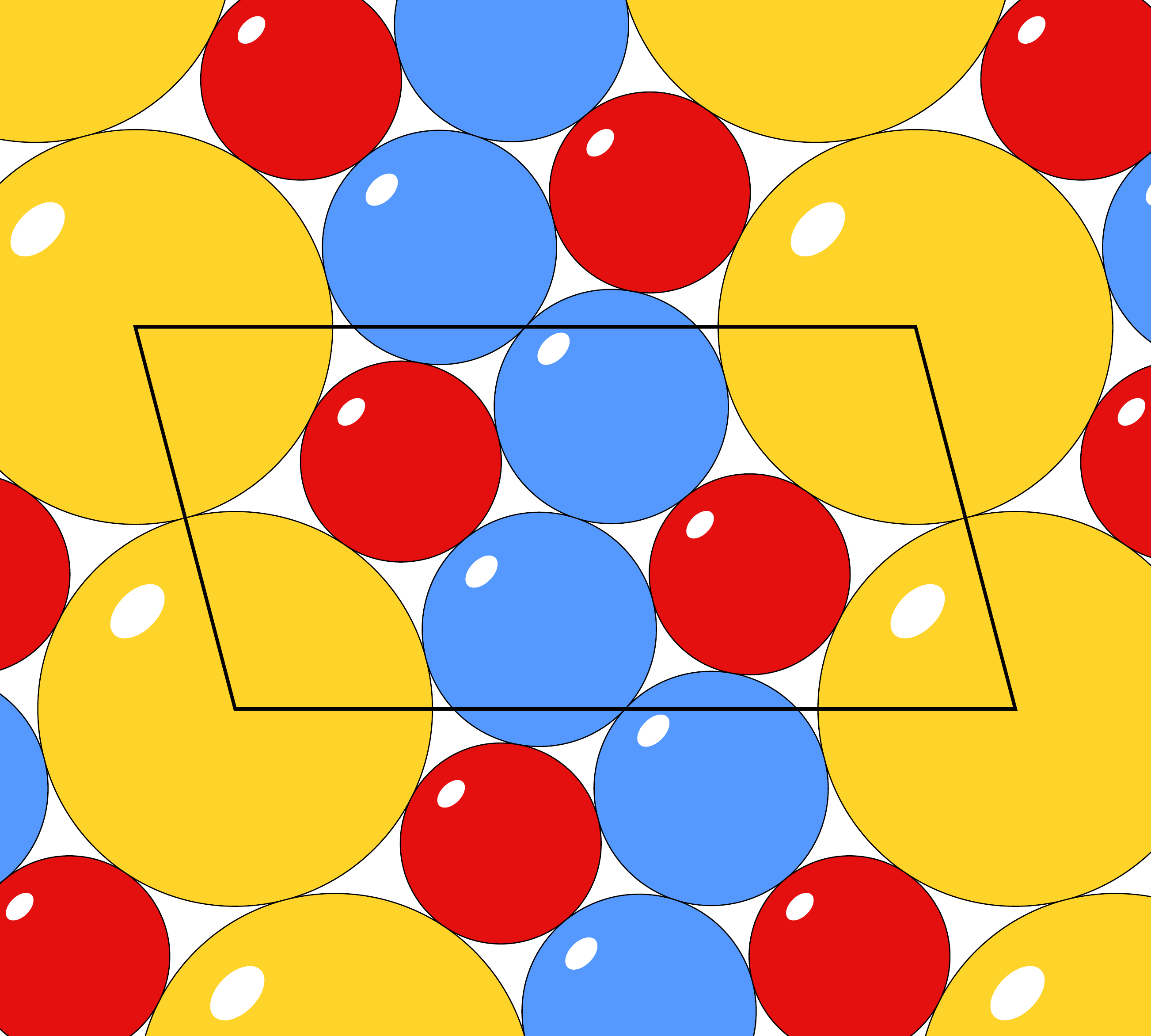}
\end{tabular}
\noindent
\begin{tabular}{lll}
  67 (L)\hfill 11rs / 111s1sss & 68 (E)\hfill 11rs / 111ss & 69 (L)\hfill 11rs / 11r1ss\\
  \includegraphics[width=0.3\textwidth]{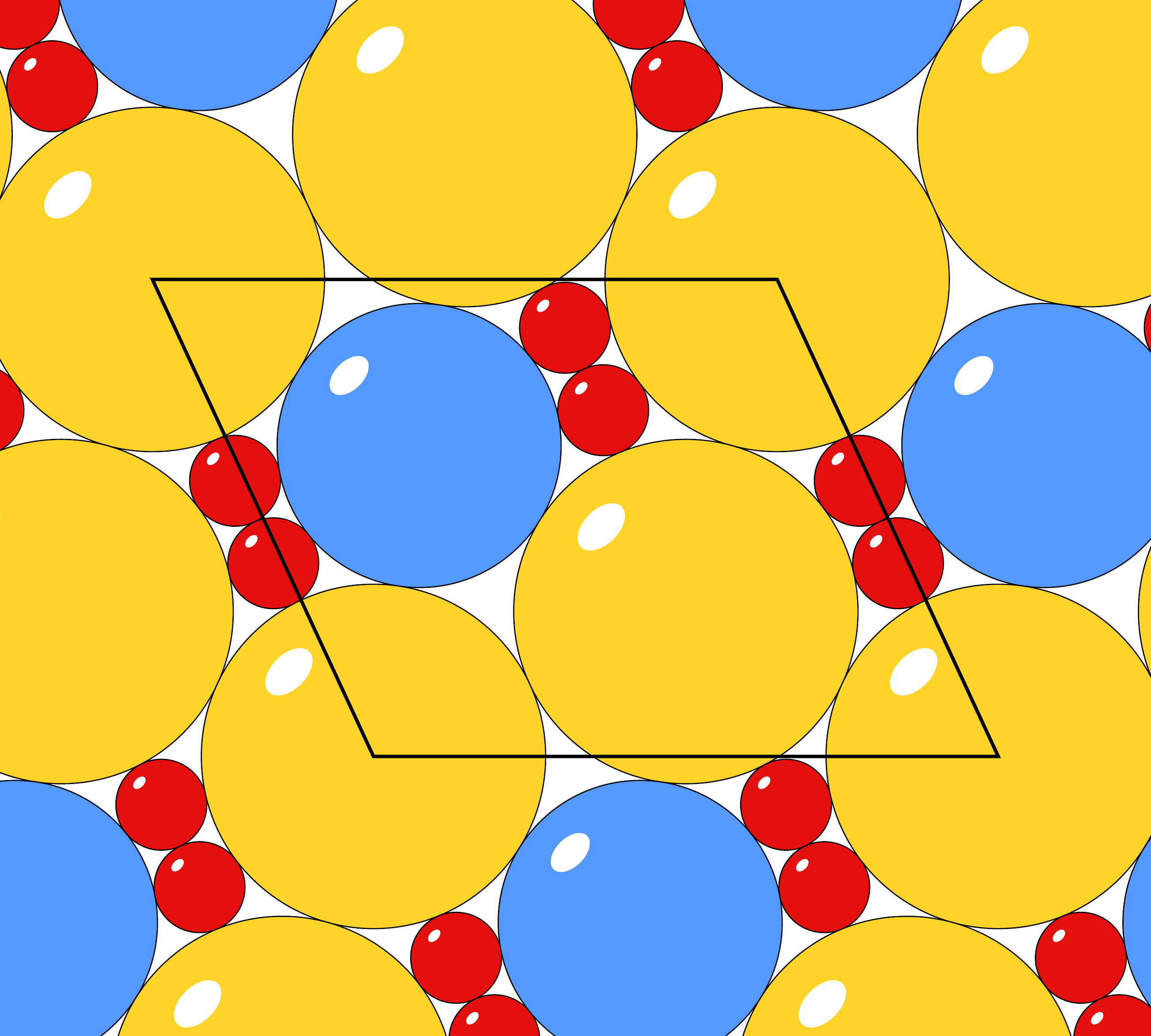} &
  \includegraphics[width=0.3\textwidth]{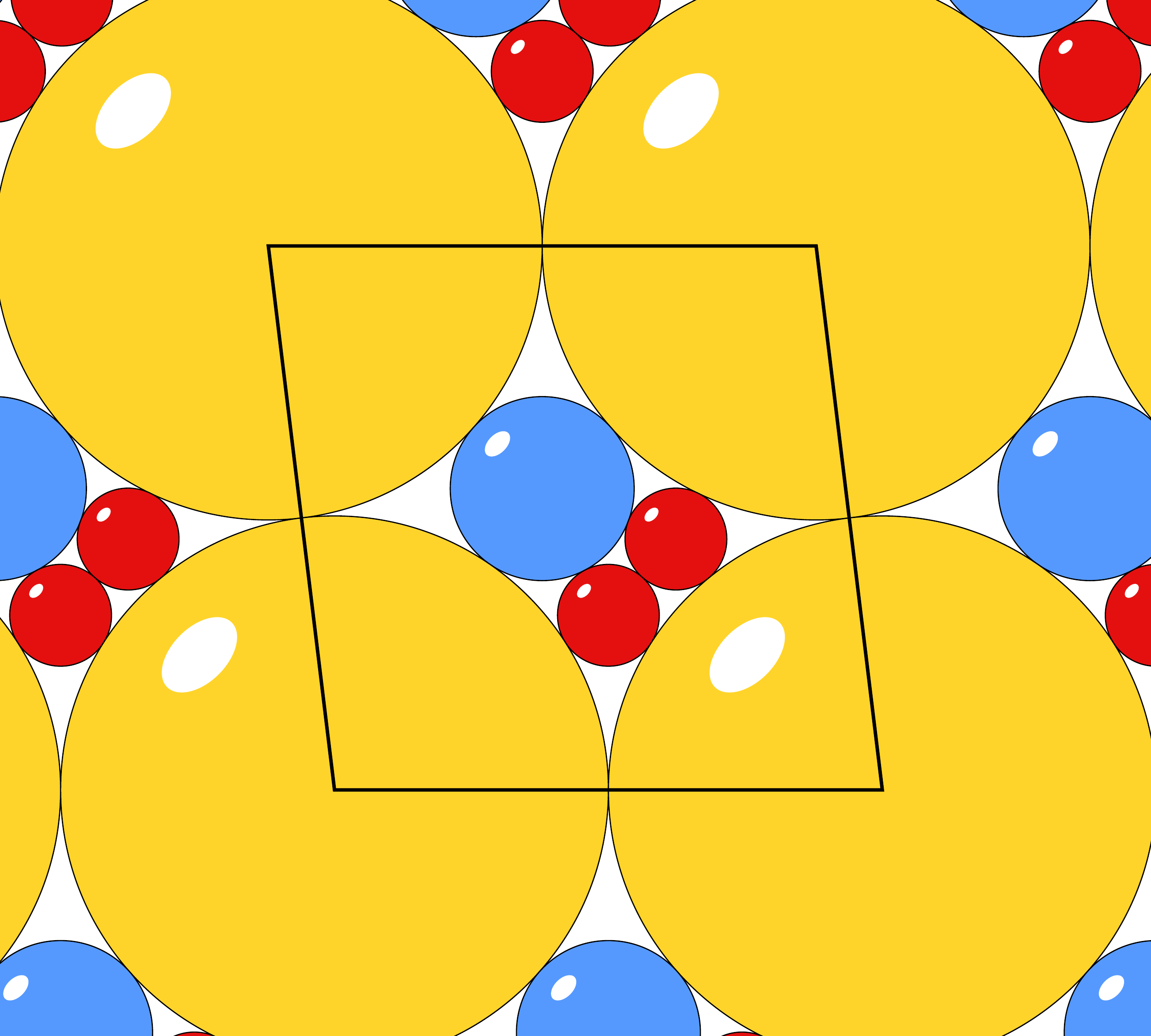} &
  \includegraphics[width=0.3\textwidth]{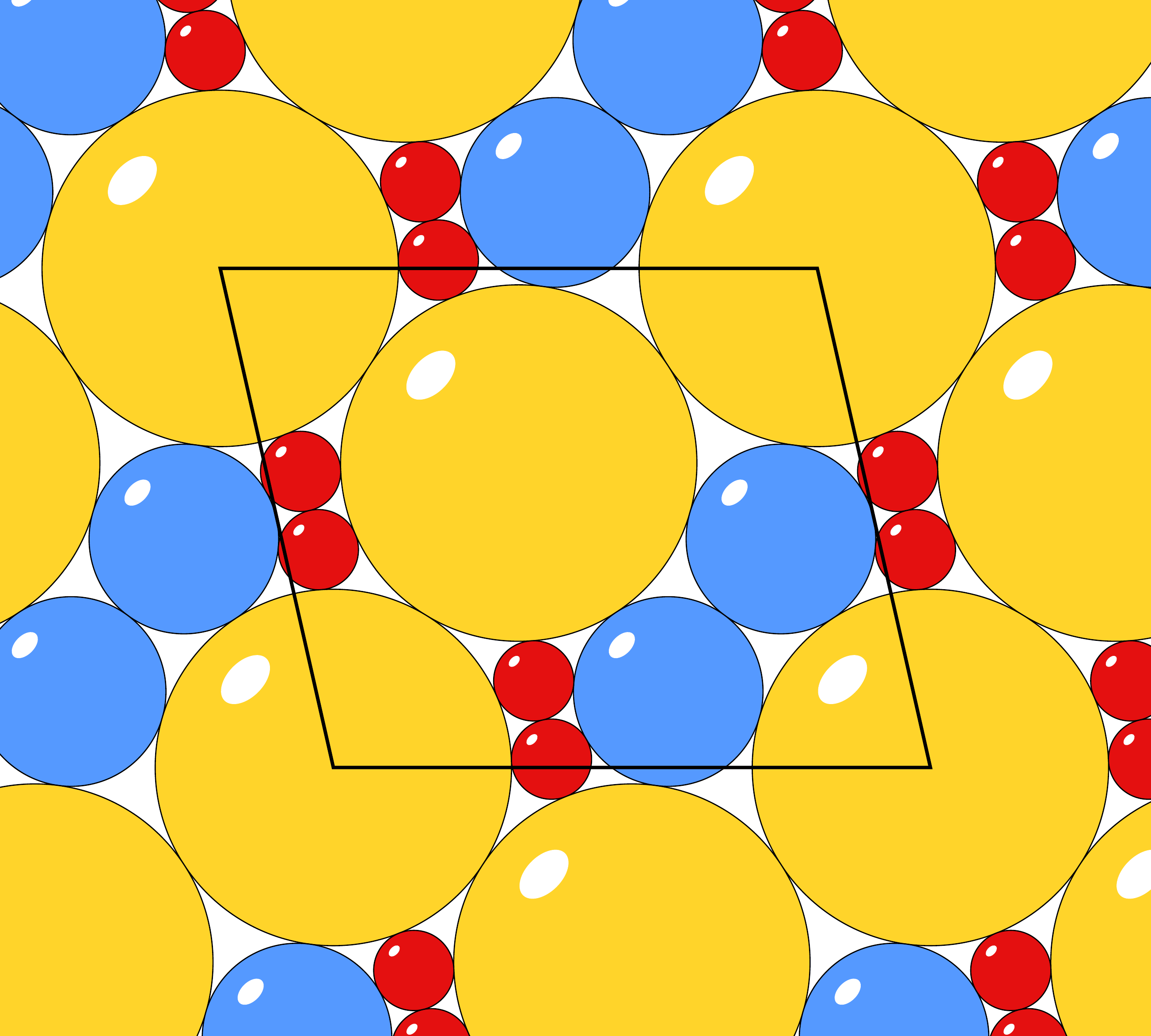}
\end{tabular}
\noindent
\begin{tabular}{lll}
  70 (H)\hfill 11rs / 1r1r1ss & 71 (E)\hfill 11rs / 1r1ss & 72 (S)\hfill 11rs / 1rr1ss\\
  \includegraphics[width=0.3\textwidth]{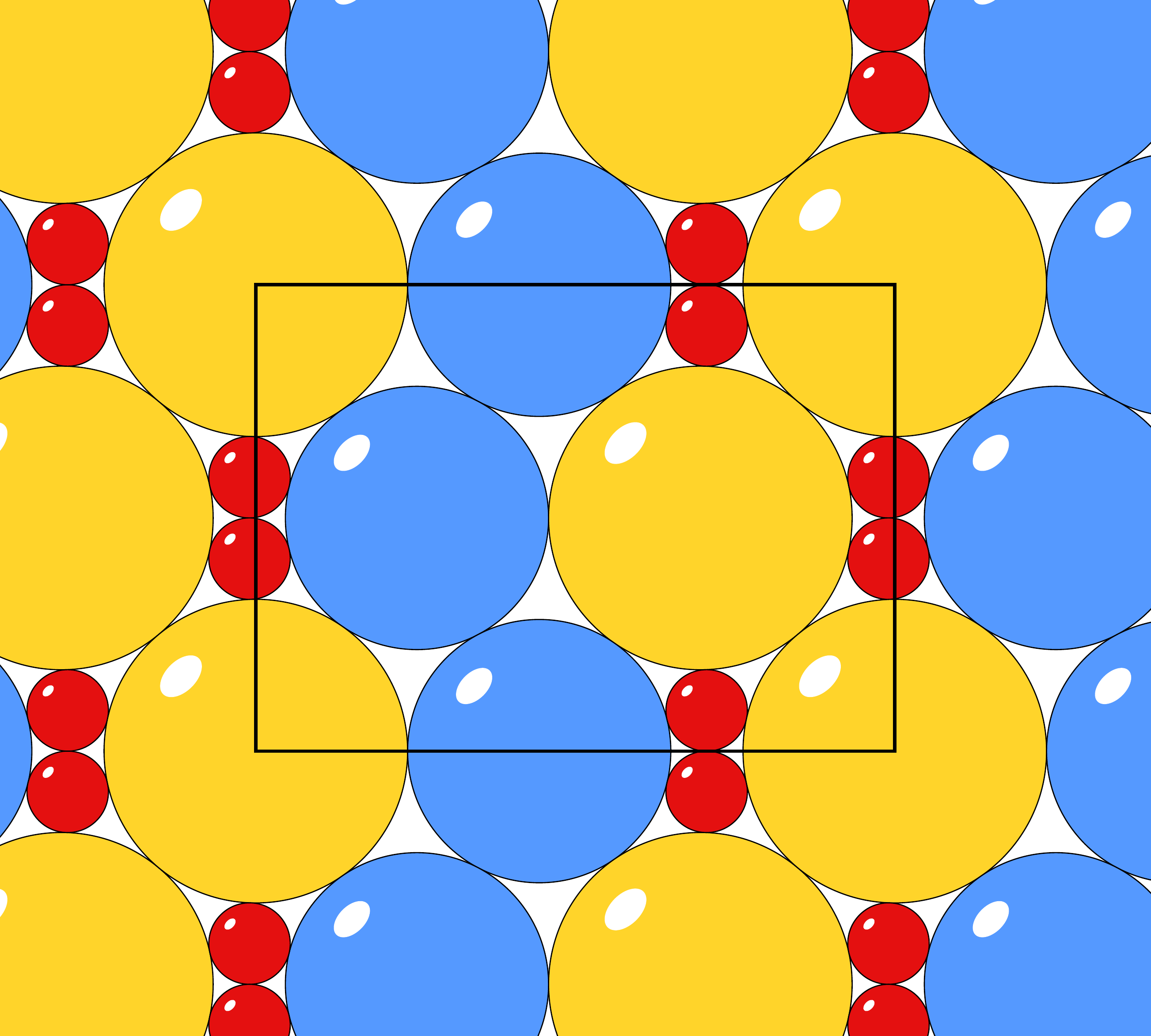} &
  \includegraphics[width=0.3\textwidth]{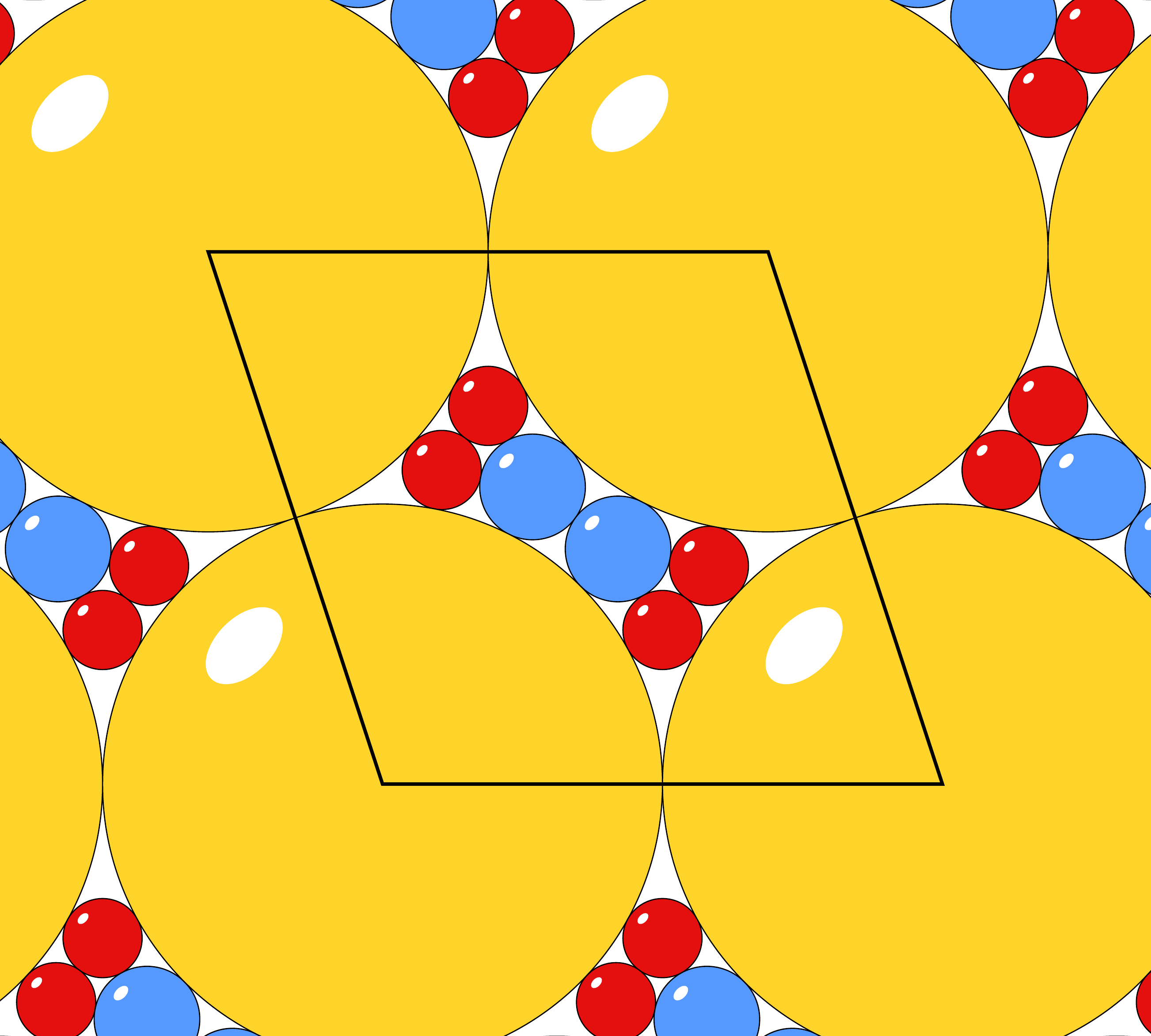} &
  \includegraphics[width=0.3\textwidth]{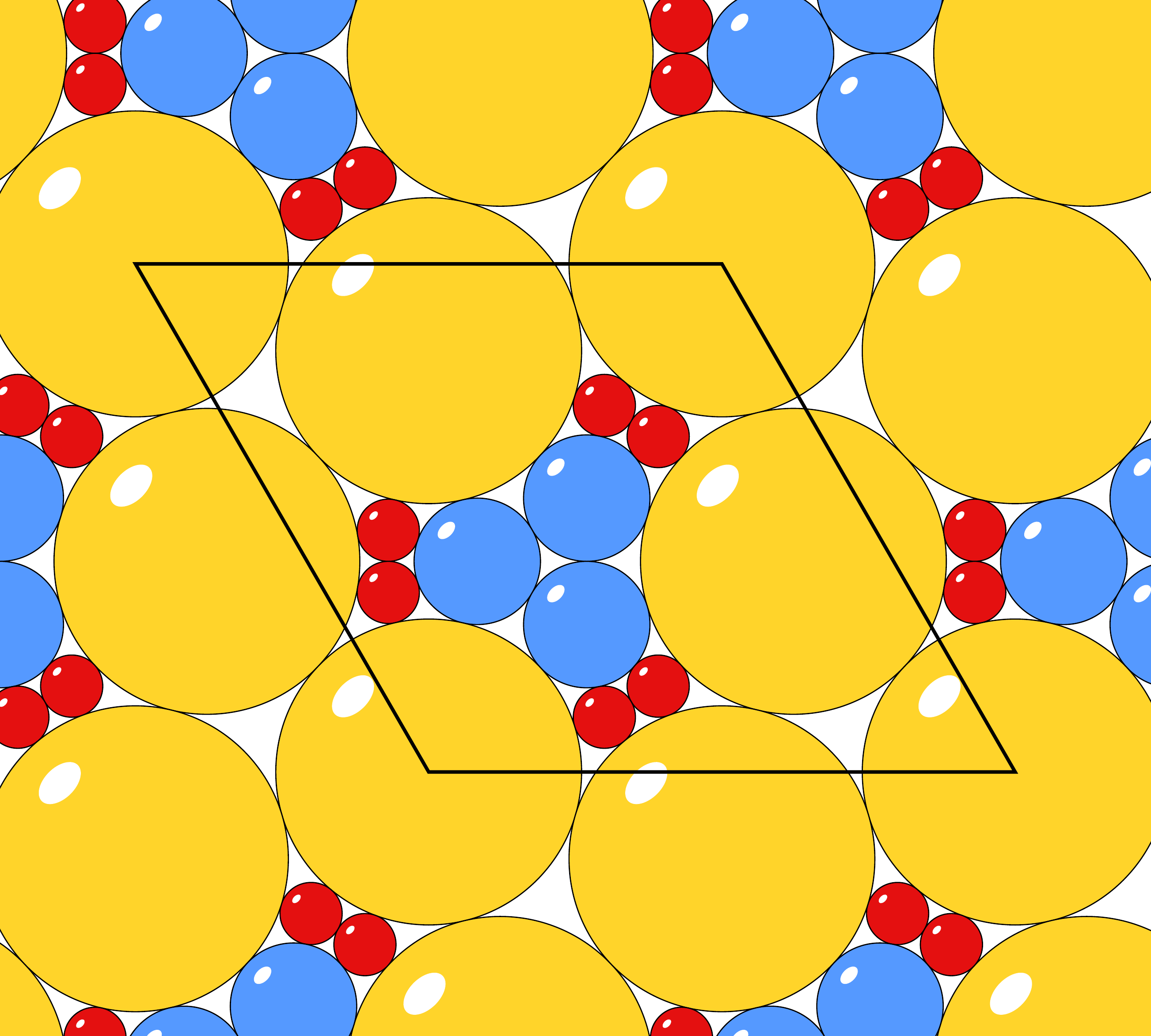}
\end{tabular}
\noindent
\begin{tabular}{lll}
  73 (H)\hfill 11rs / 1rrr1ss & 74 (S)\hfill 11rs / 1s1s1ssss & 75 (E)\hfill 11rs / 1s1sss\\
  \includegraphics[width=0.3\textwidth]{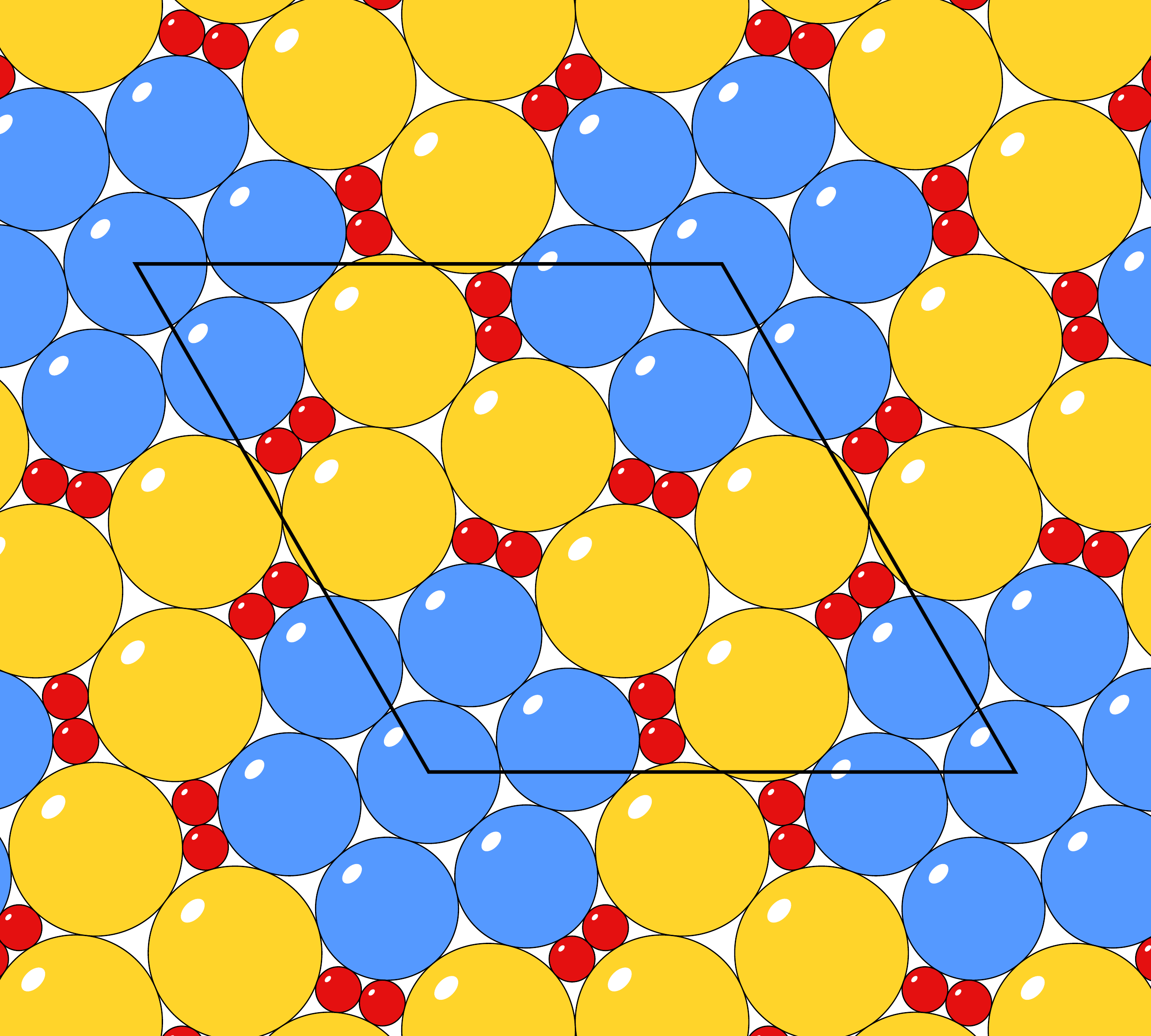} &
  \includegraphics[width=0.3\textwidth]{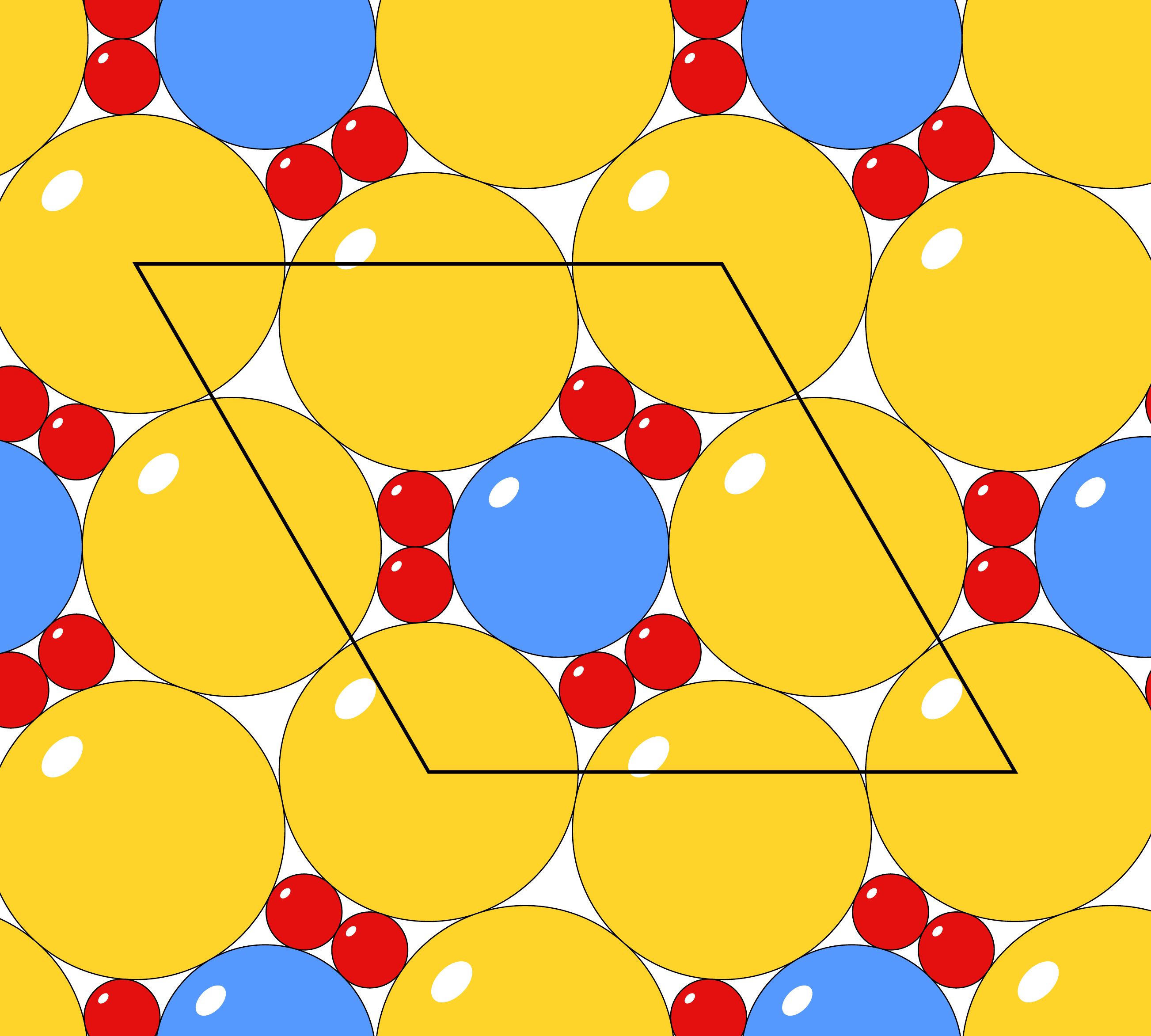} &
  \includegraphics[width=0.3\textwidth]{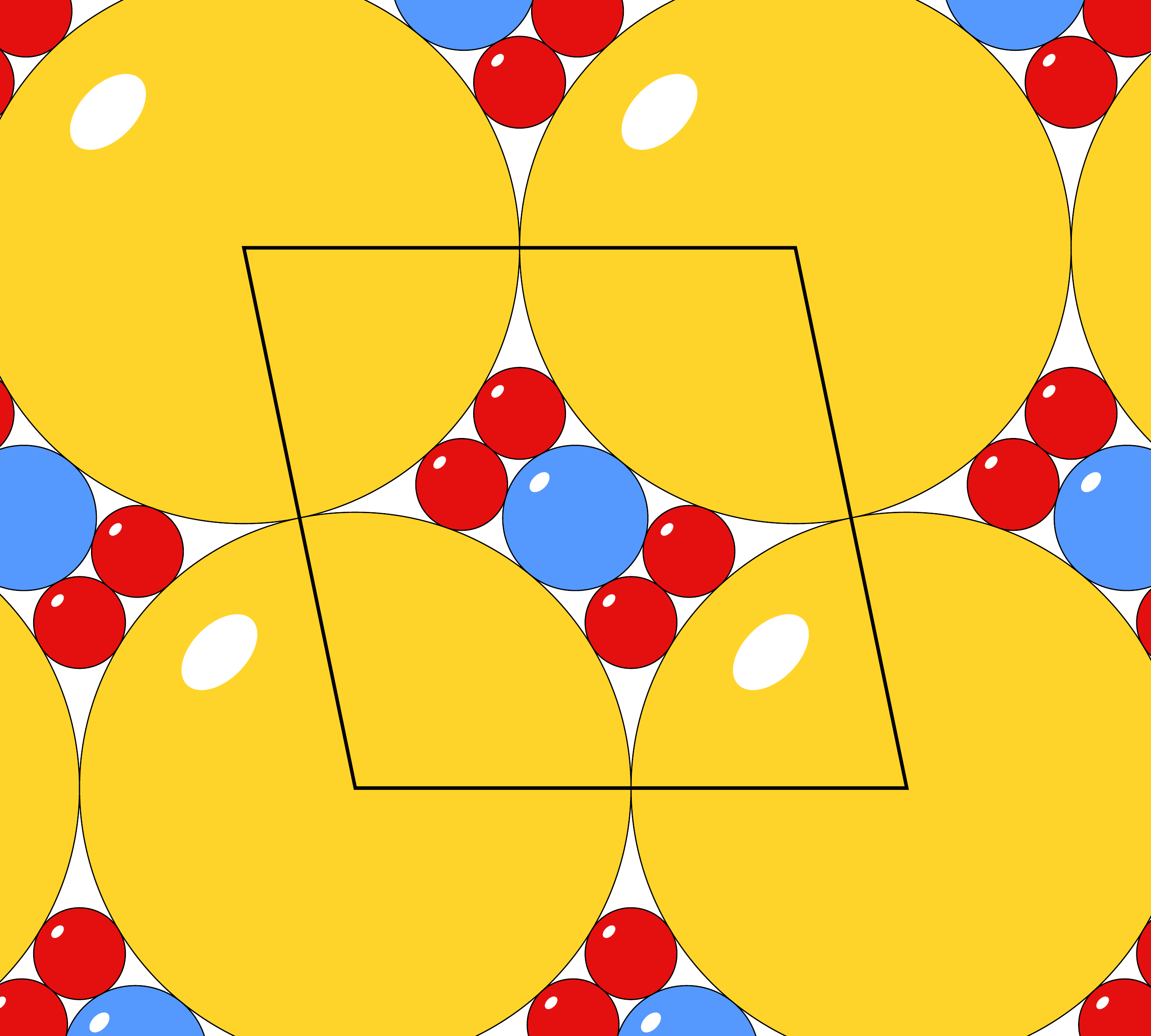}
\end{tabular}
\noindent
\begin{tabular}{lll}
  76 (L)\hfill 11rsr / 111ss & 77 (L)\hfill 11rsr / 11r1ss & 78 (H)\hfill 11rsr / 1rr1ss\\
  \includegraphics[width=0.3\textwidth]{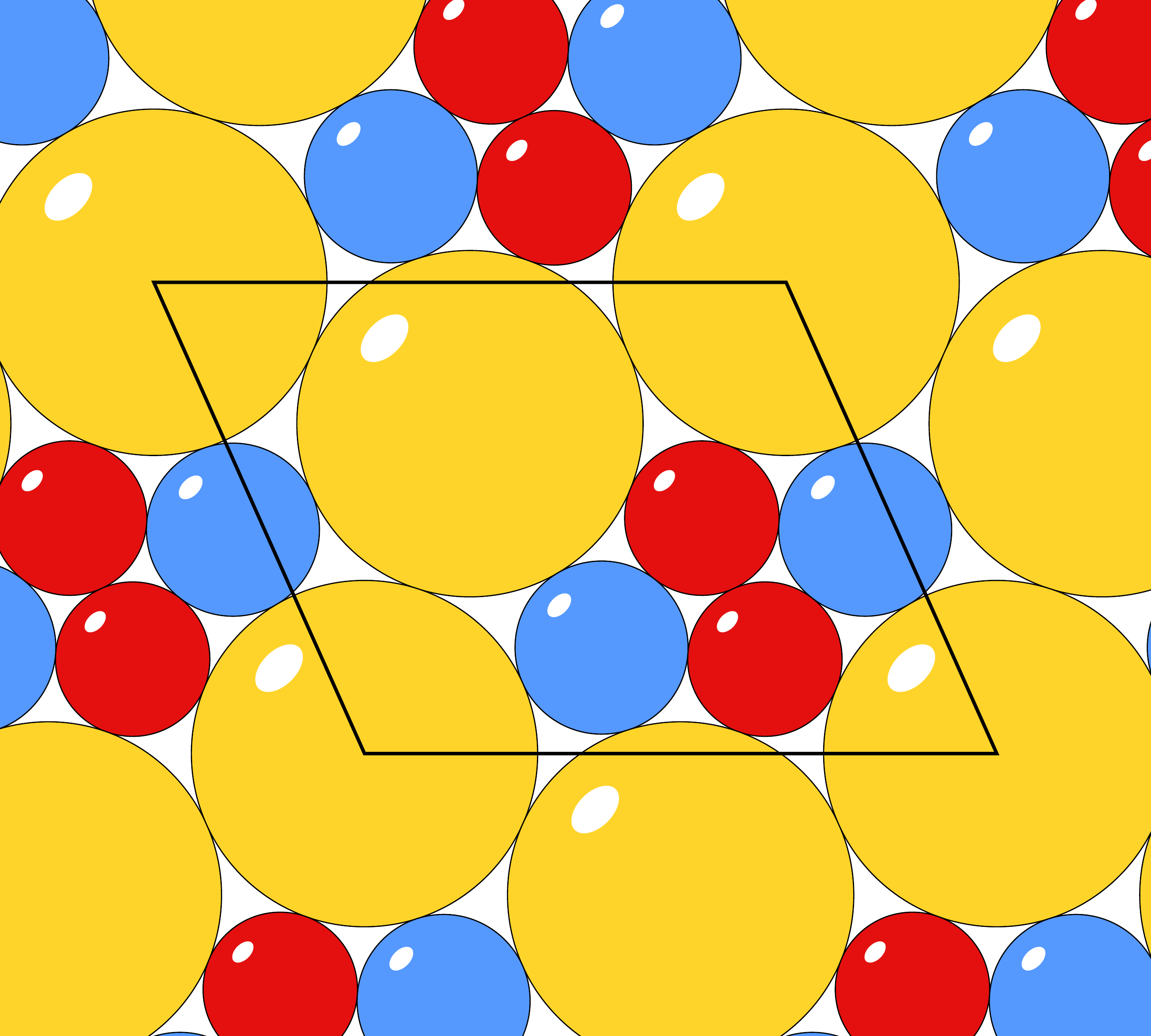} &
  \includegraphics[width=0.3\textwidth]{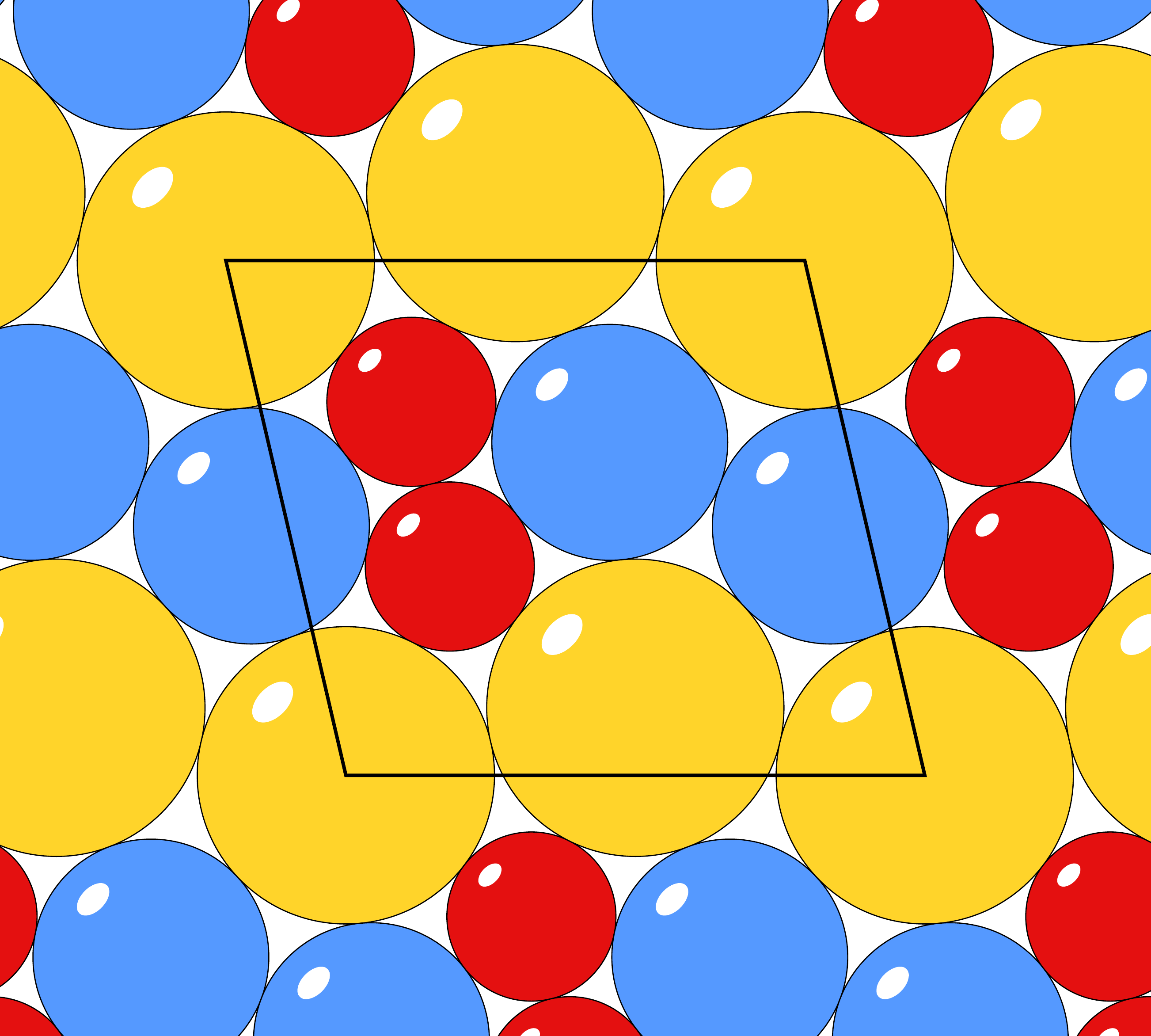} &
  \includegraphics[width=0.3\textwidth]{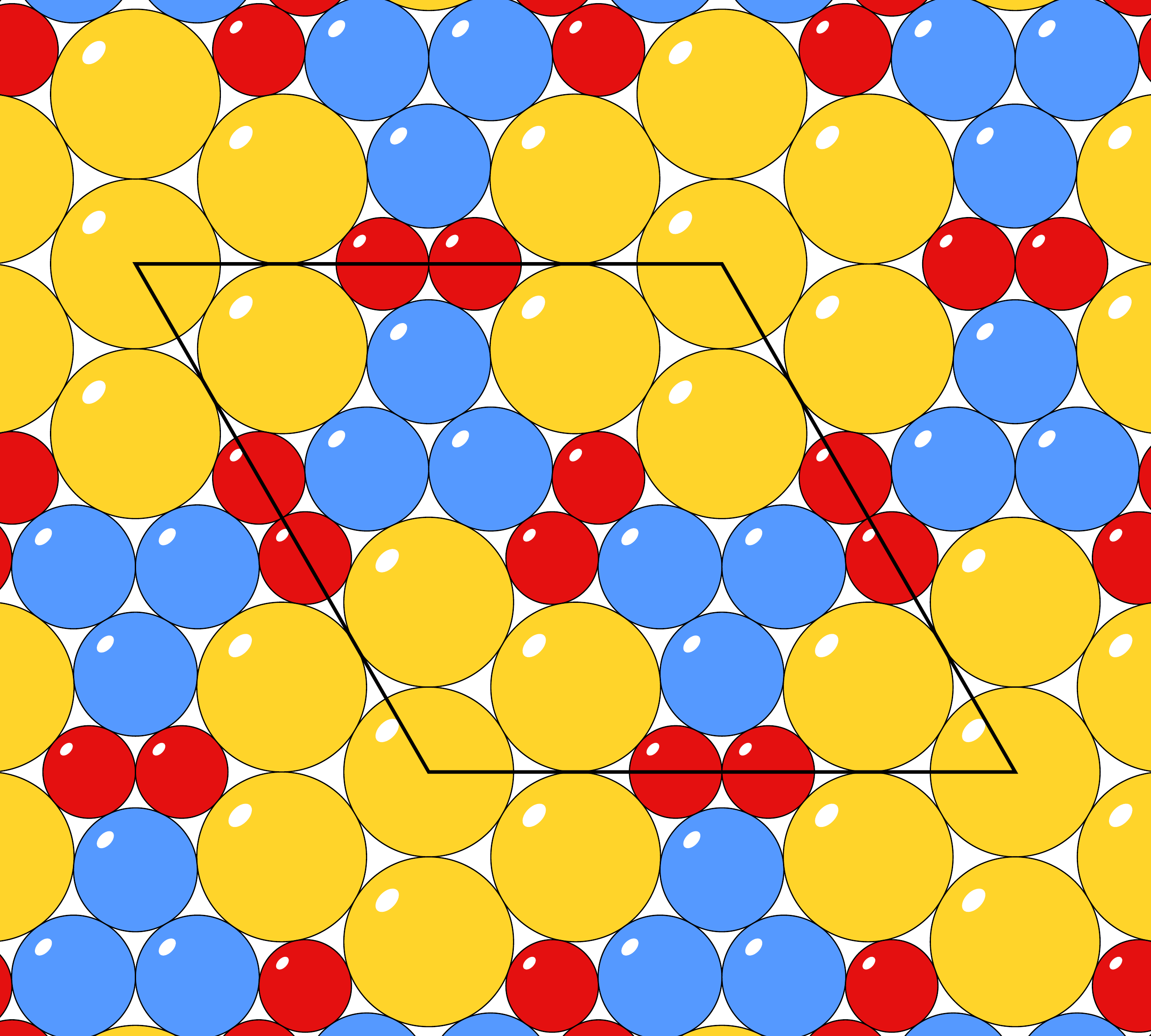}
\end{tabular}
\noindent
\begin{tabular}{lll}
  79 (L)\hfill 11rsr / 1s1sss & 80 (L)\hfill 1r1r / 1111s & 81 (L)\hfill 1r1r / 111r1s\\
  \includegraphics[width=0.3\textwidth]{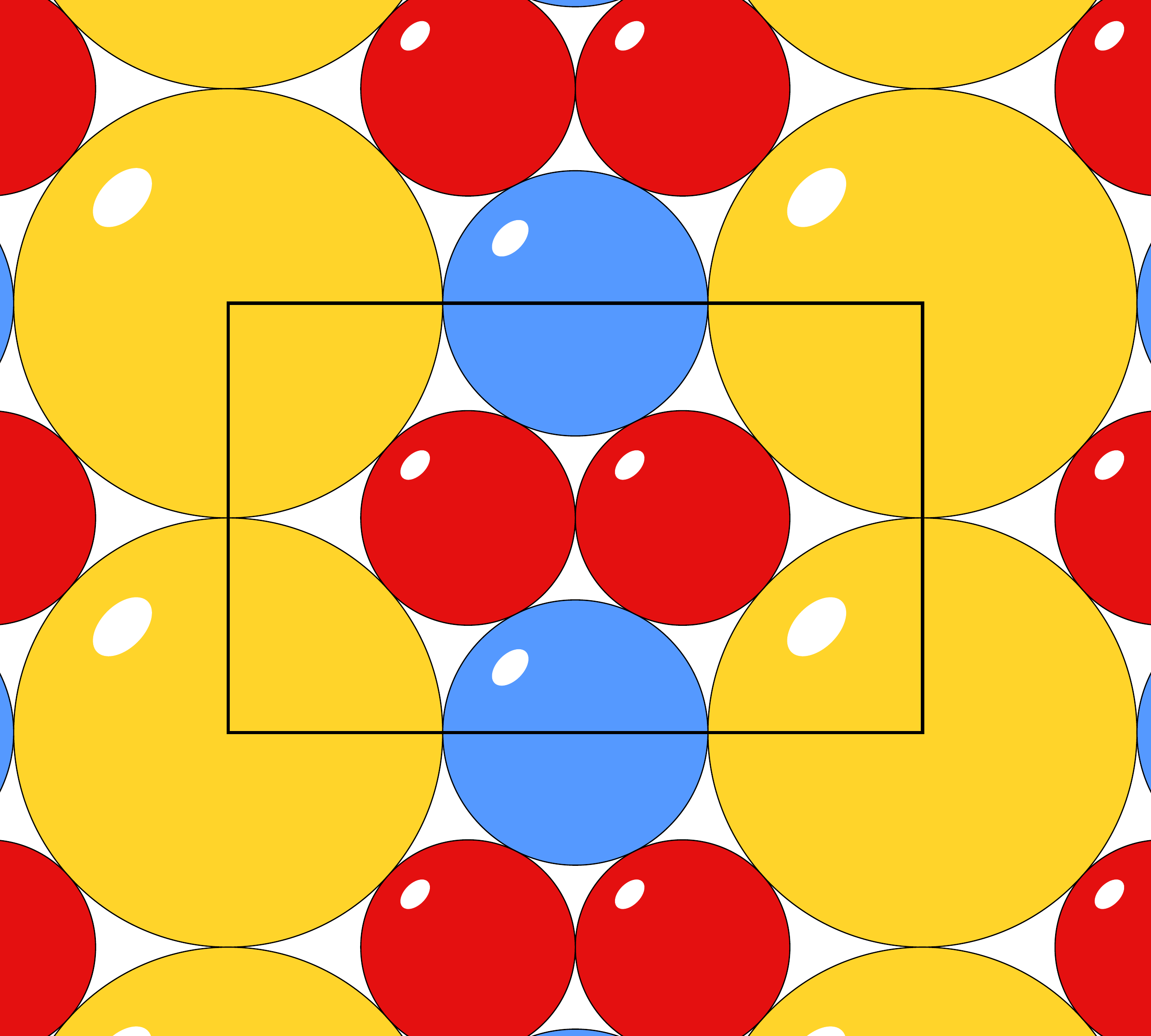} &
  \includegraphics[width=0.3\textwidth]{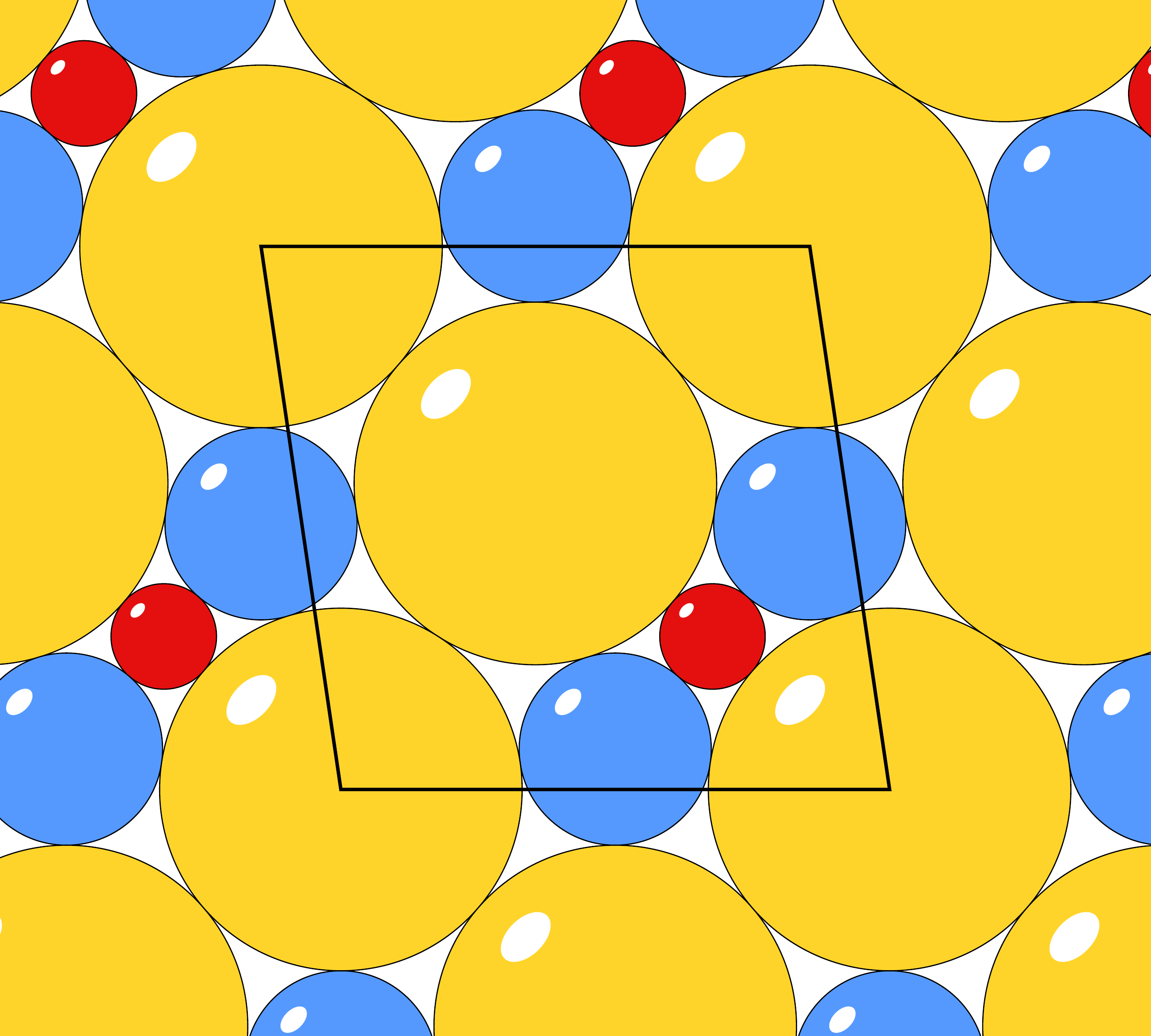} &
  \includegraphics[width=0.3\textwidth]{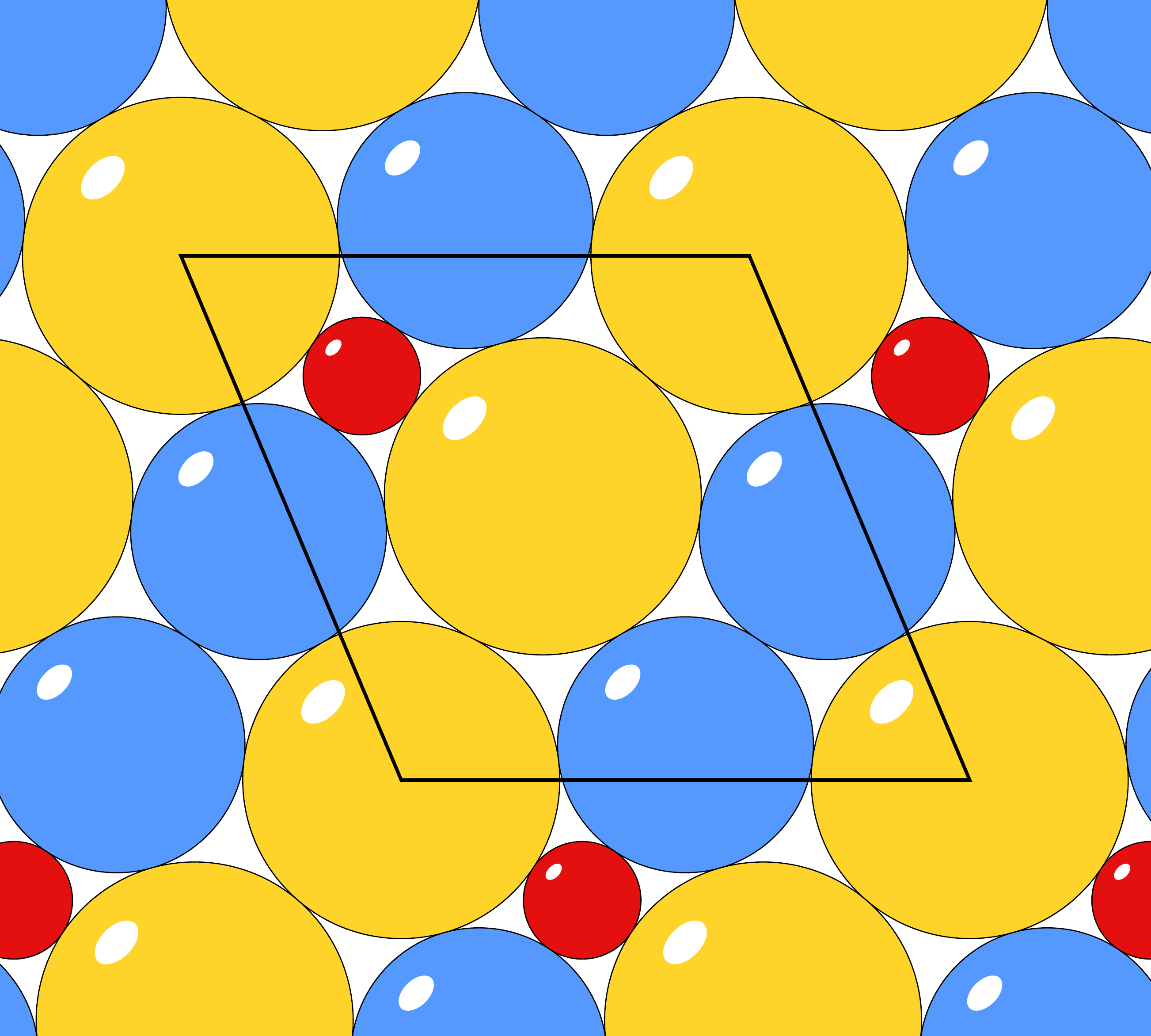}
\end{tabular}
\noindent
\begin{tabular}{lll}
  82 (L)\hfill 1r1r / 111s1s & 83 (E)\hfill 1r1r / 11r1s & 84 (H)\hfill 1r1r / 11rr1s\\
  \includegraphics[width=0.3\textwidth]{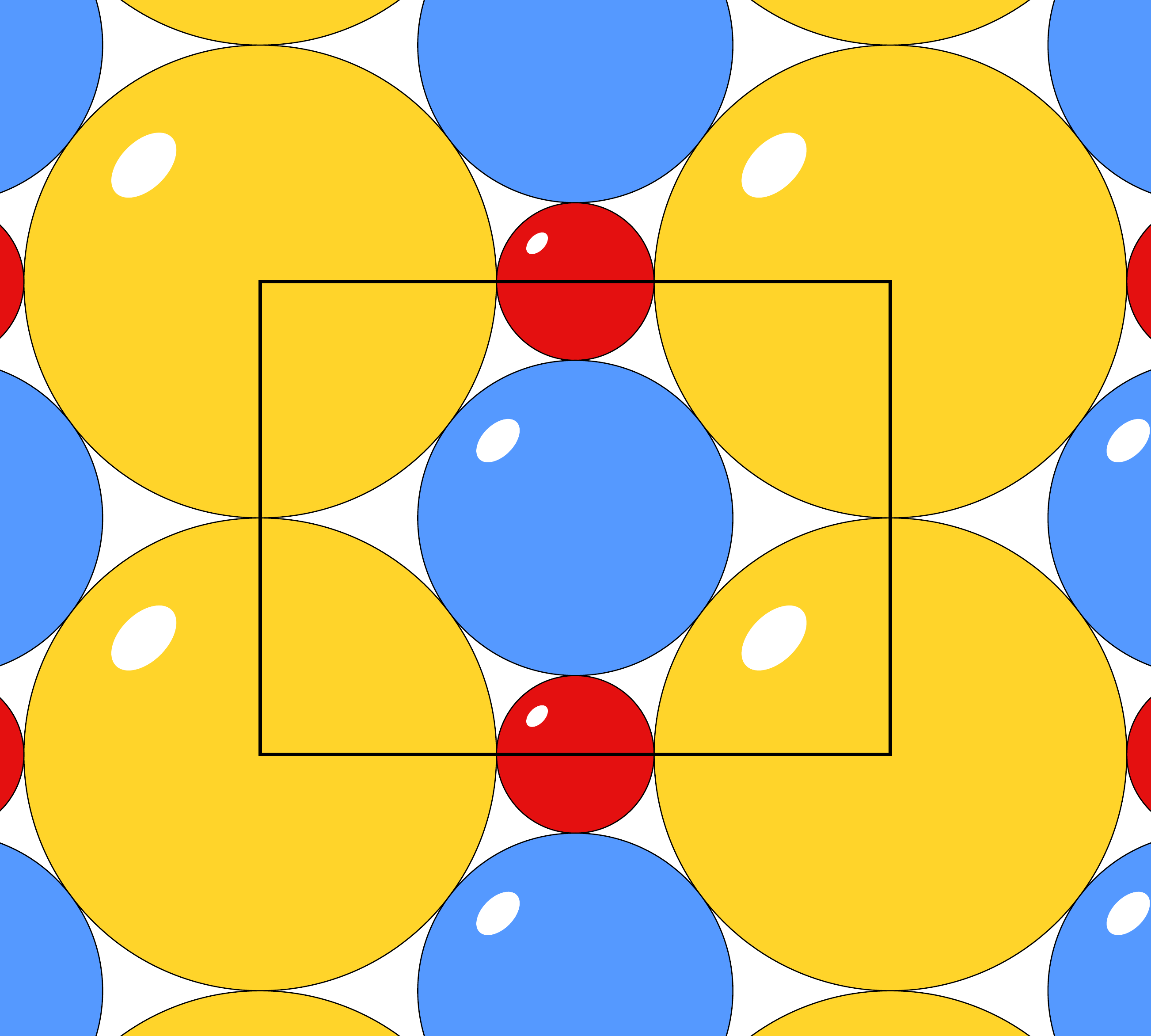} &
  \includegraphics[width=0.3\textwidth]{packing_1r1r_11r1s.pdf} &
  \includegraphics[width=0.3\textwidth]{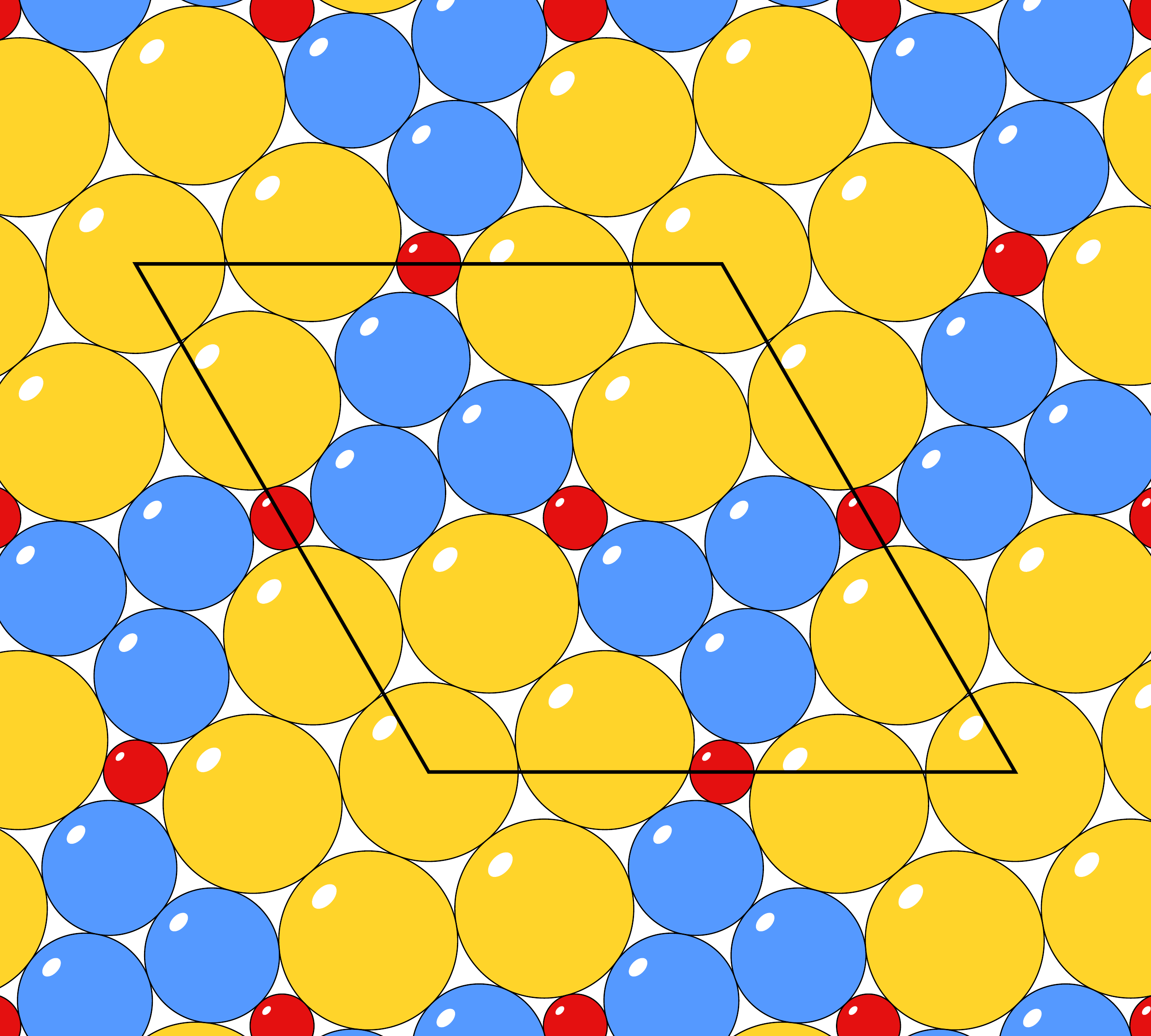}
\end{tabular}
\noindent
\begin{tabular}{lll}
  85 (L)\hfill 1r1r / 11s1s & 86 (L)\hfill 1r1r / 11s1s1s & 87 (L)\hfill 1r1r / 1r1r1s\\
  \includegraphics[width=0.3\textwidth]{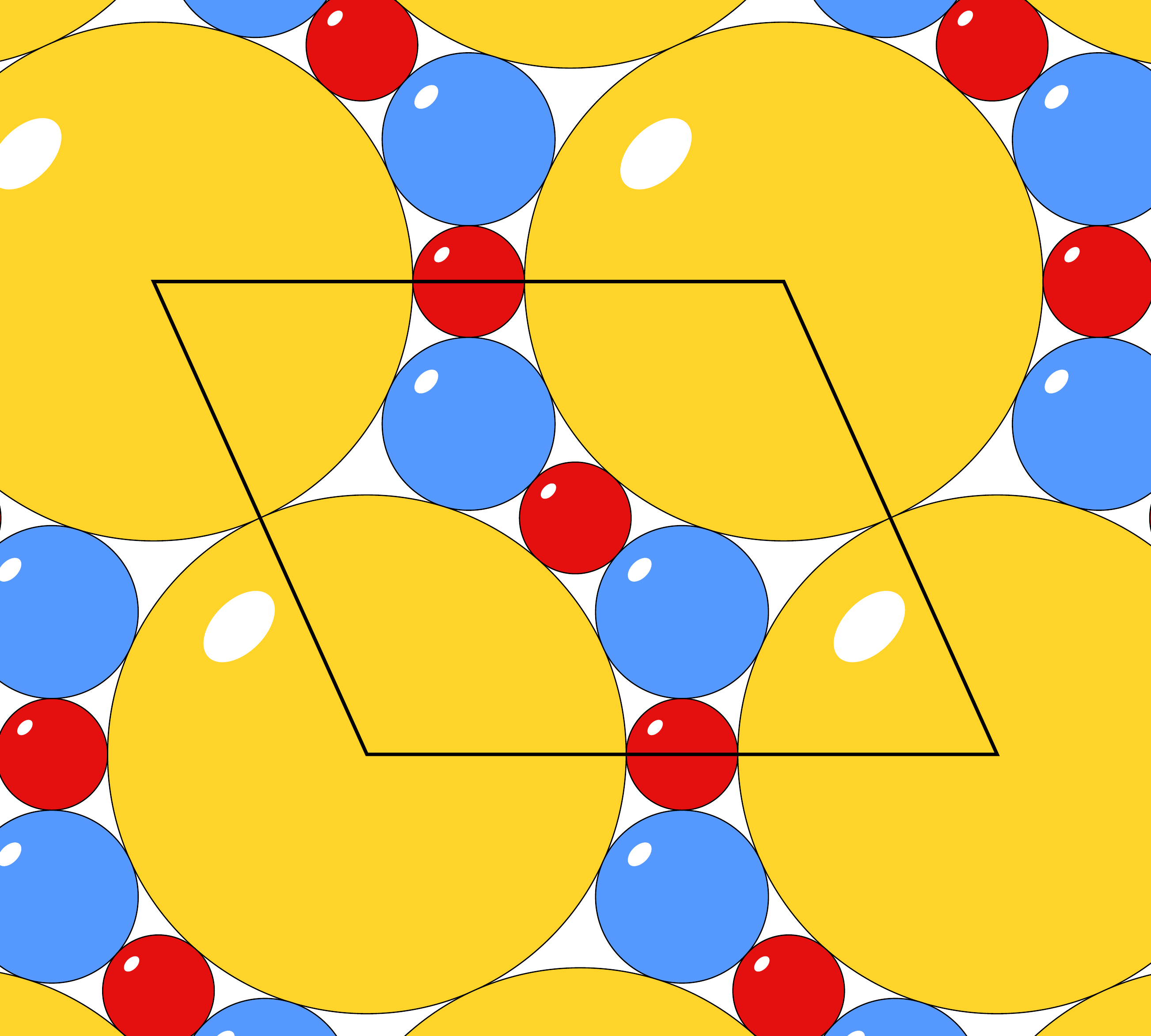} &
  \includegraphics[width=0.3\textwidth]{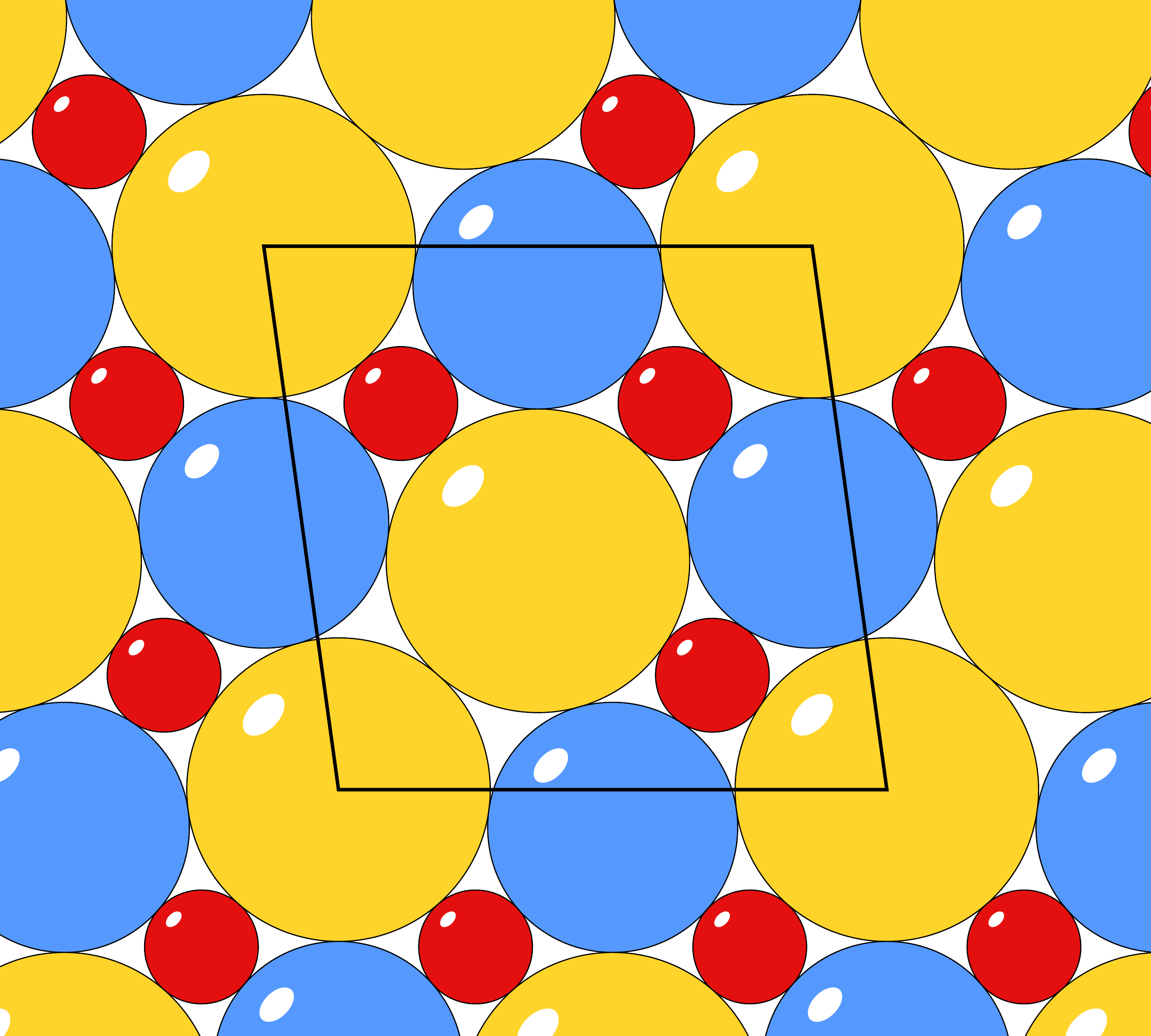} &
  \includegraphics[width=0.3\textwidth]{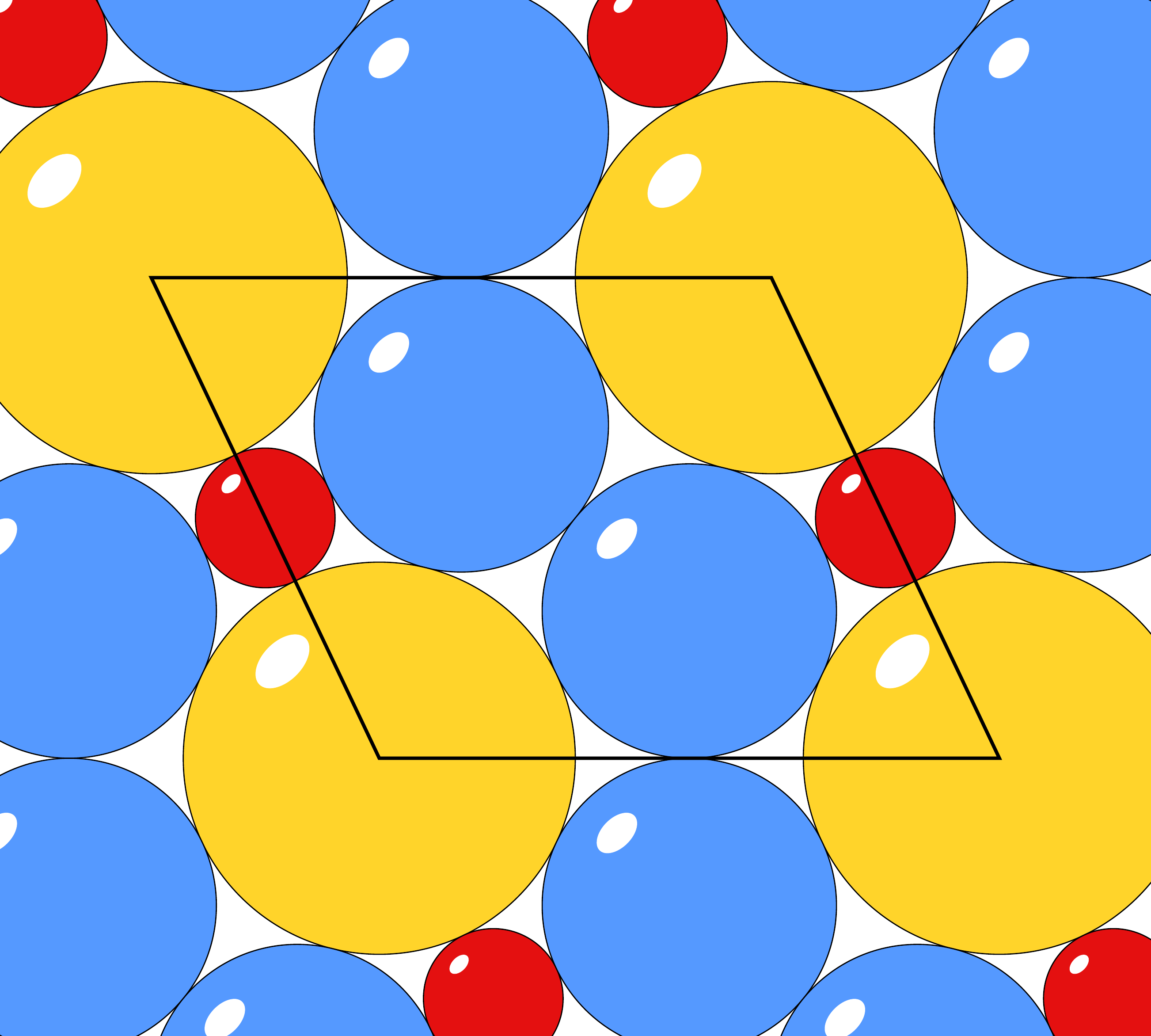}
\end{tabular}
\noindent
\begin{tabular}{lll}
  88 (L)\hfill 1r1r / 1r1s1s & 89 (H)\hfill 1r1r / 1rr1s & 90 (H)\hfill 1r1r / 1rrr1s\\
  \includegraphics[width=0.3\textwidth]{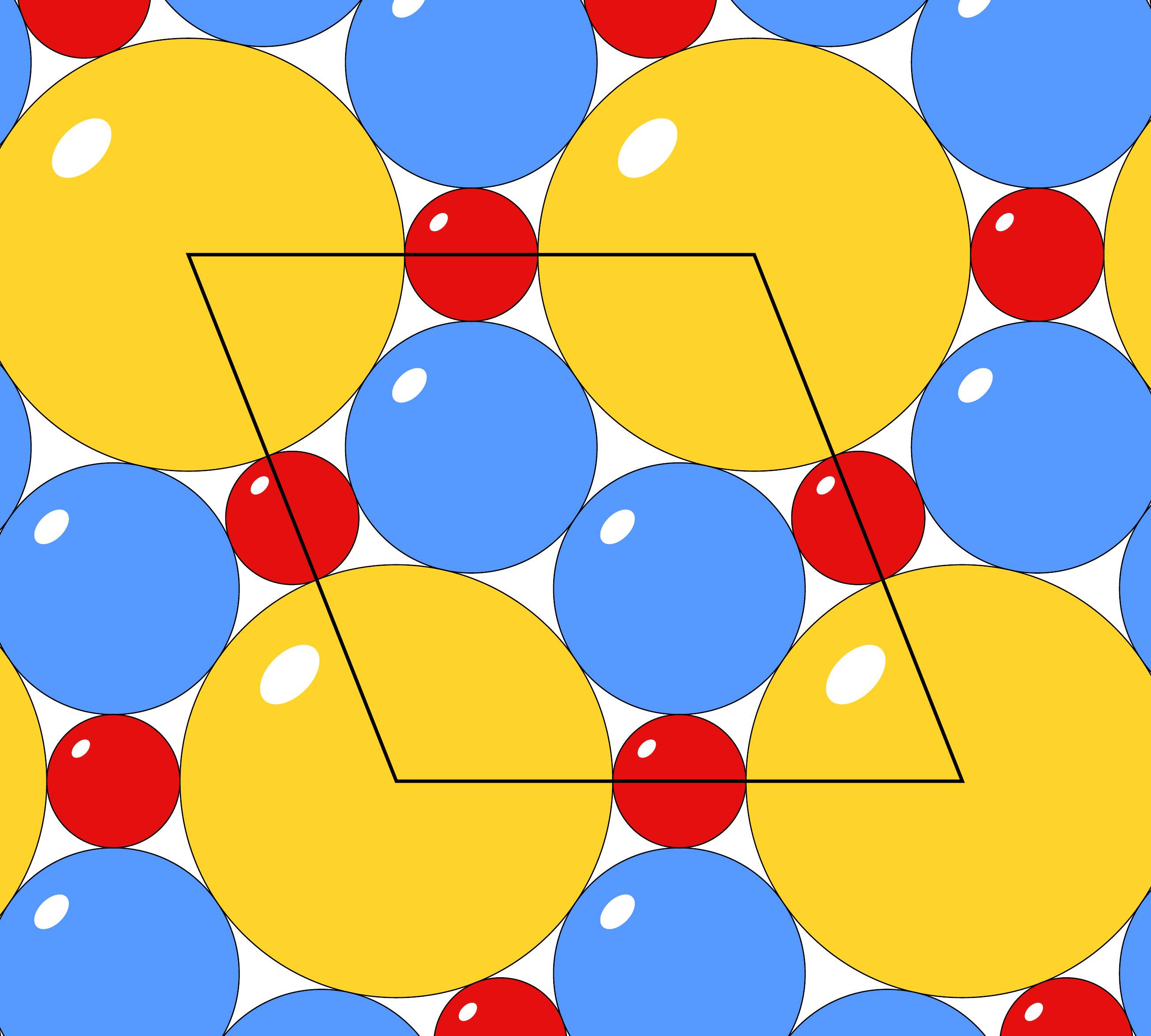} &
  \includegraphics[width=0.3\textwidth]{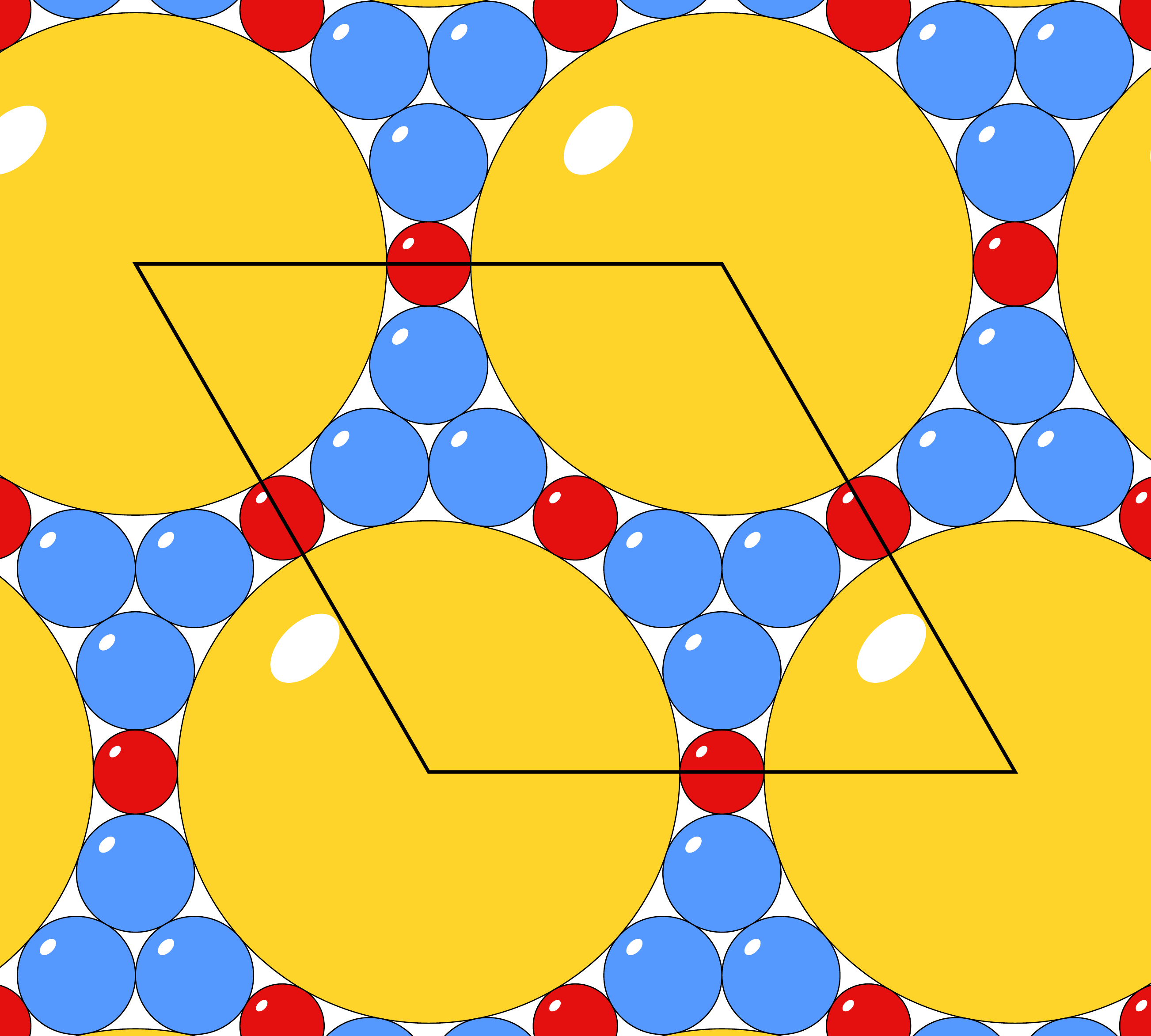} &
  \includegraphics[width=0.3\textwidth]{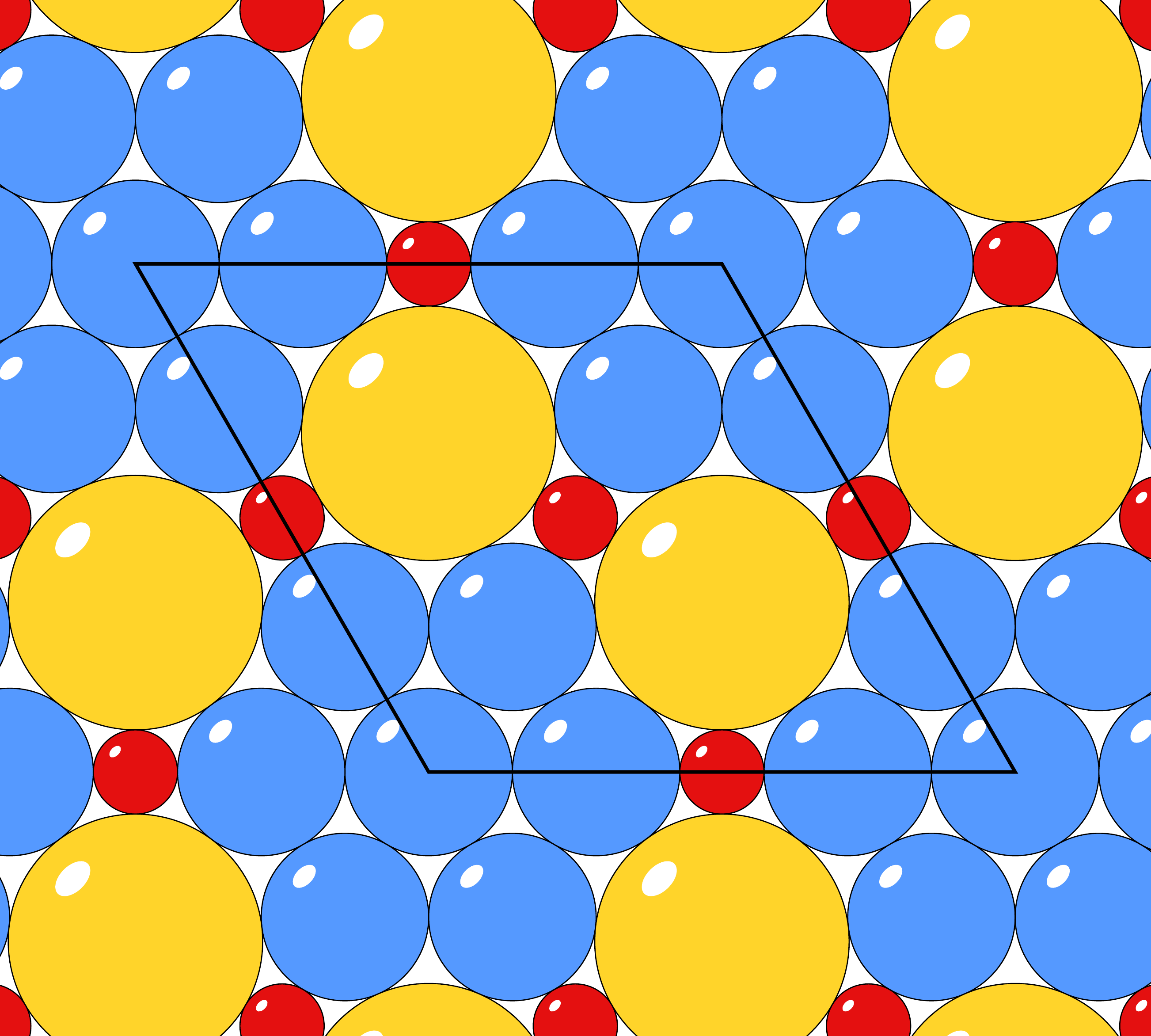}
\end{tabular}
\noindent
\begin{tabular}{lll}
  91 (H)\hfill 1r1r / 1s1s1s & 92 (E)\hfill 1r1r / 1srsrs & 93 (H)\hfill 1r1rr / 1s1srs\\
  \includegraphics[width=0.3\textwidth]{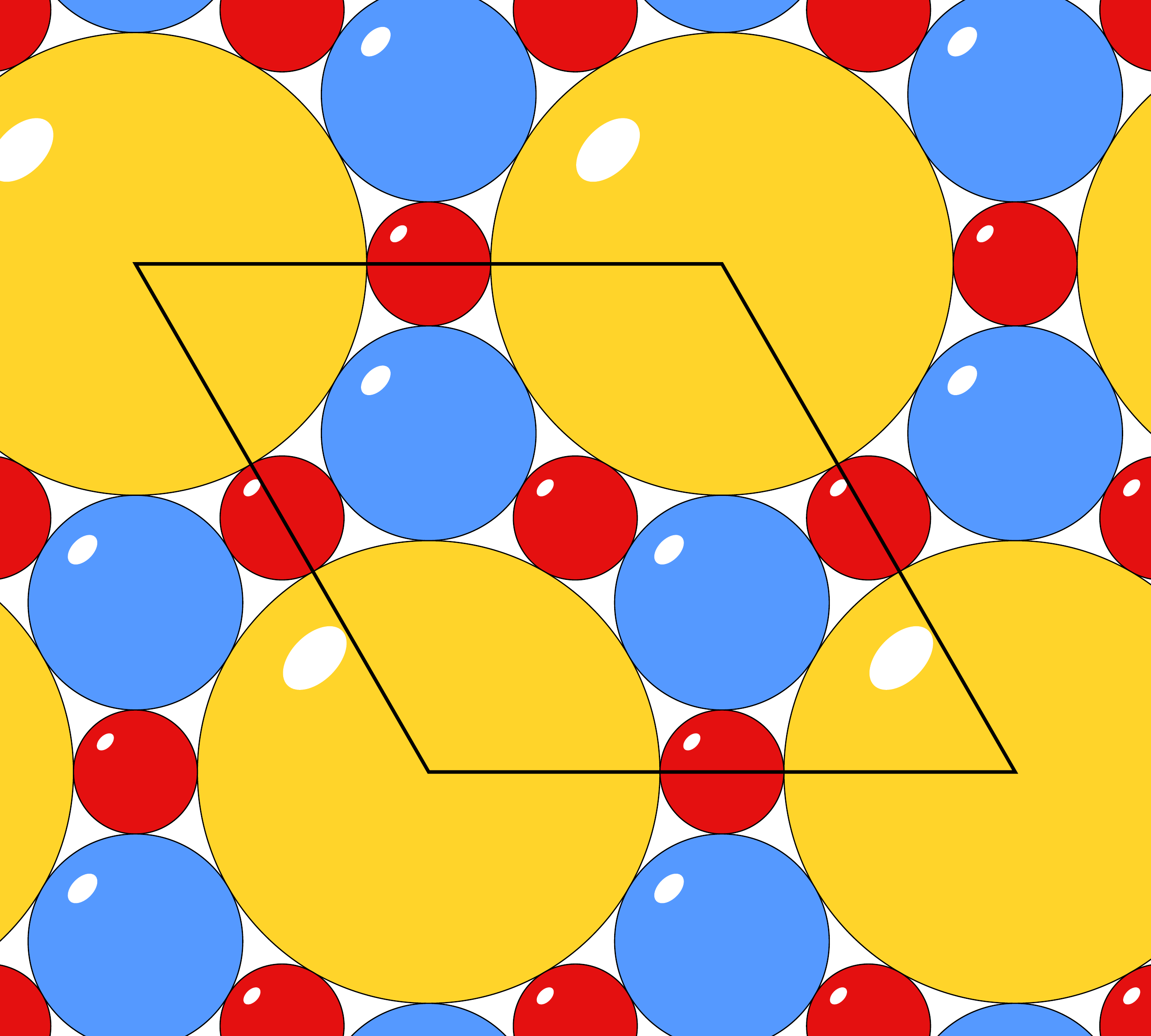} &
  \includegraphics[width=0.3\textwidth]{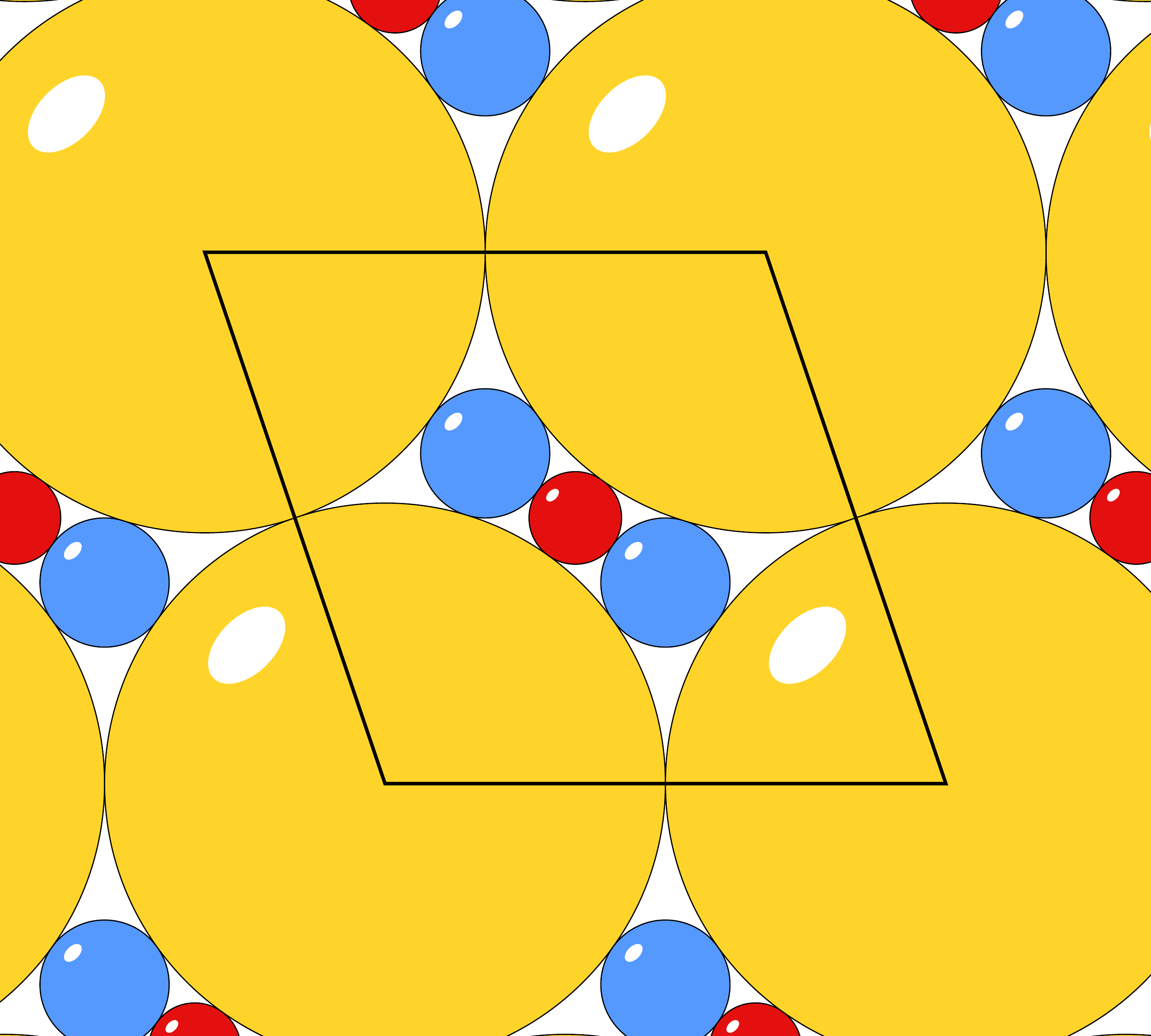} &
  \includegraphics[width=0.3\textwidth]{packing_1r1rr_1s1srs.pdf}
\end{tabular}
\noindent
\begin{tabular}{lll}
  94 (L)\hfill 1r1s / 1111s & 95 (L)\hfill 1r1s / 111r1s & 96 (L)\hfill 1r1s / 111s1s\\
  \includegraphics[width=0.3\textwidth]{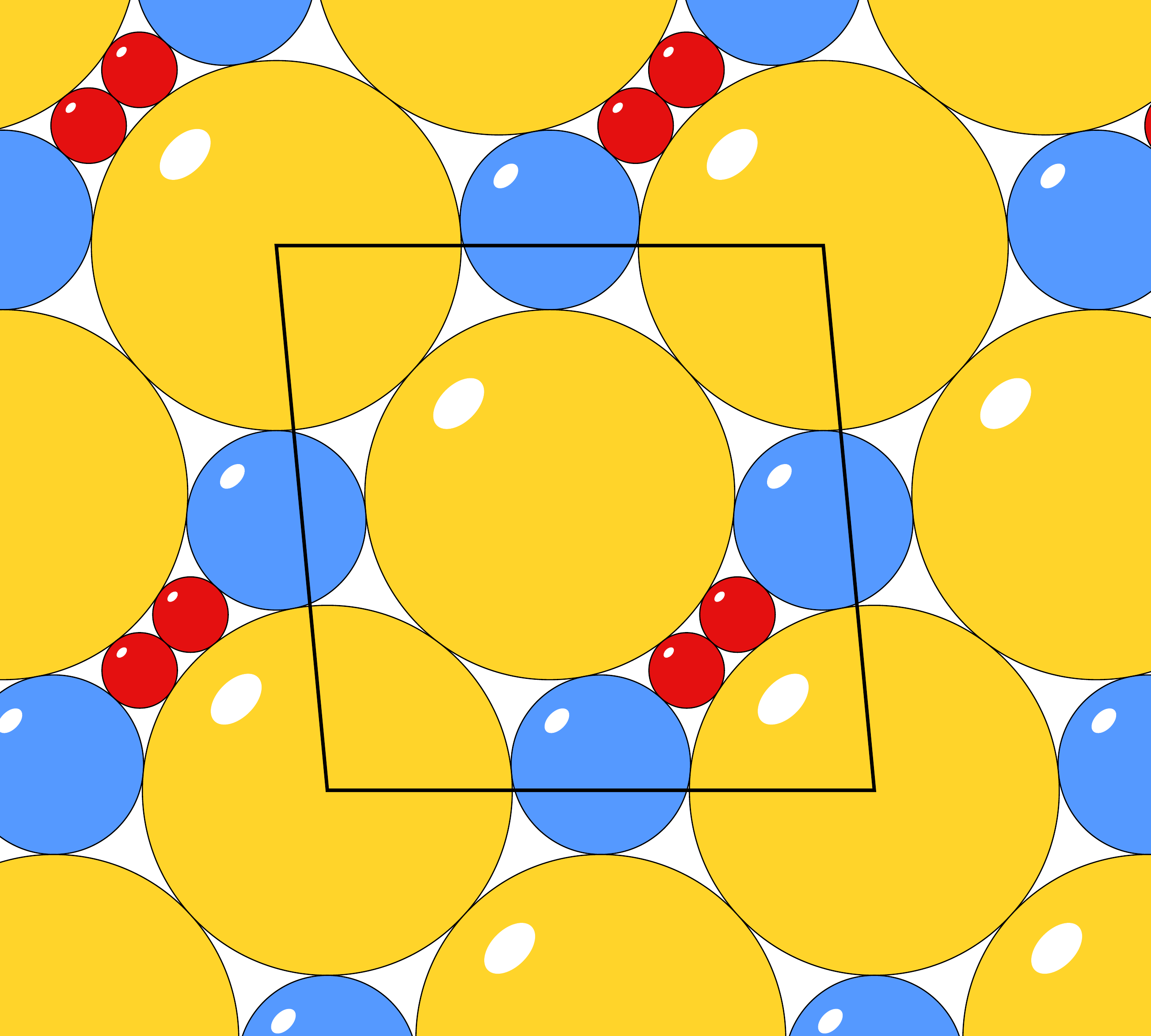} &
  \includegraphics[width=0.3\textwidth]{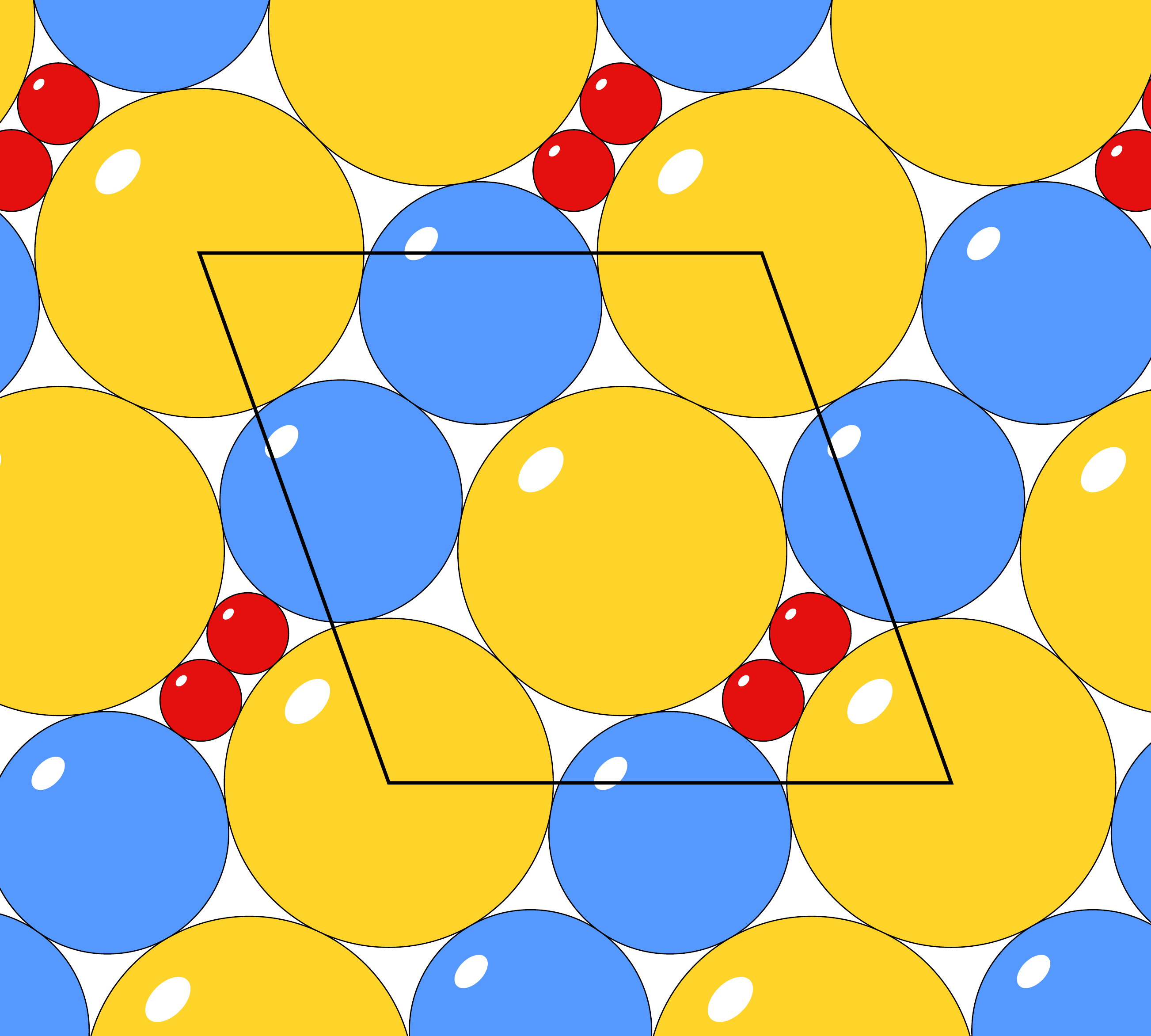} &
  \includegraphics[width=0.3\textwidth]{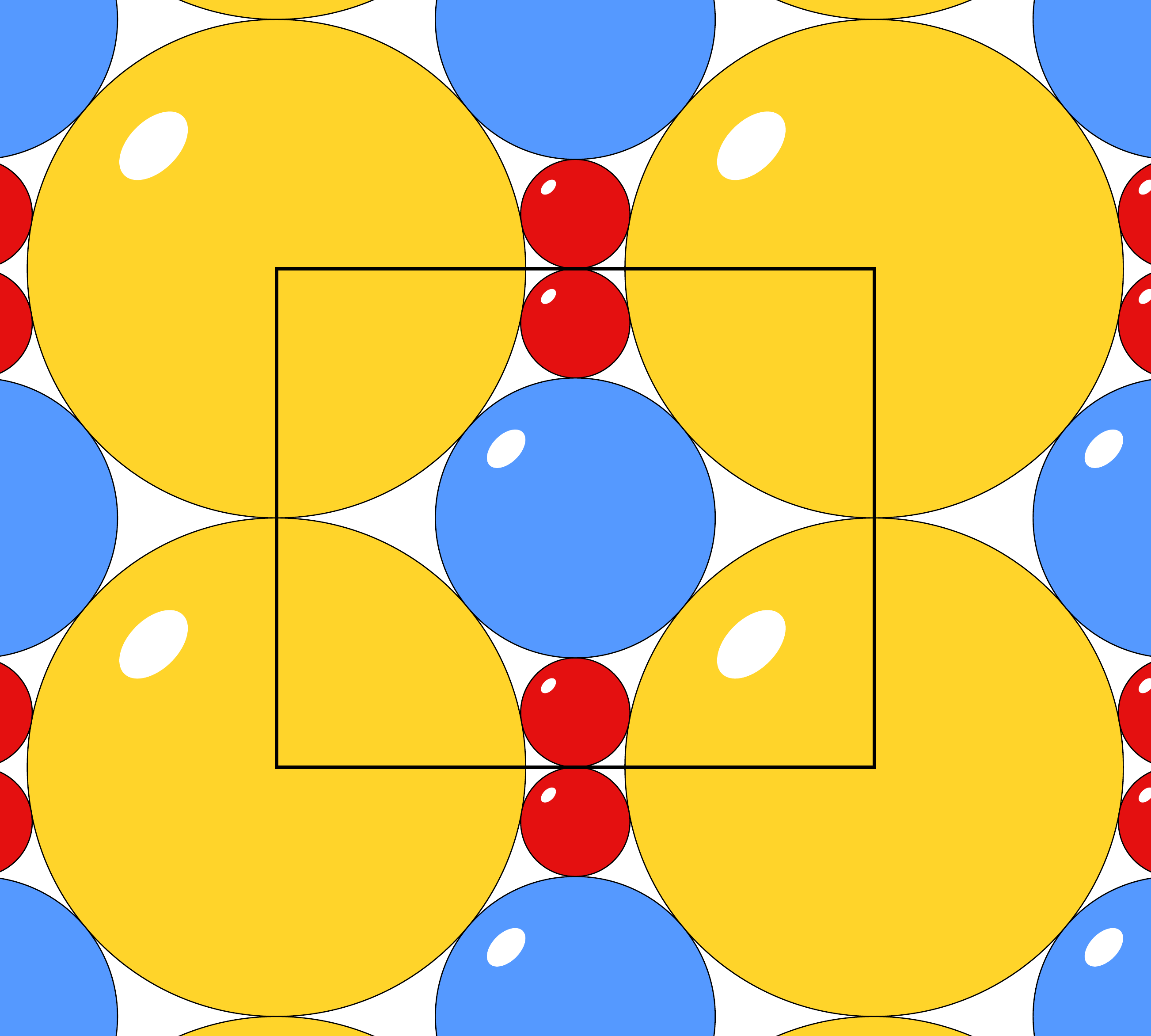}
\end{tabular}
\noindent
\begin{tabular}{lll}
  97 (L)\hfill 1r1s / 11r1s & 98 (L)\hfill 1r1s / 11r1s1s & 99 (H)\hfill 1r1s / 11rr1s\\
  \includegraphics[width=0.3\textwidth]{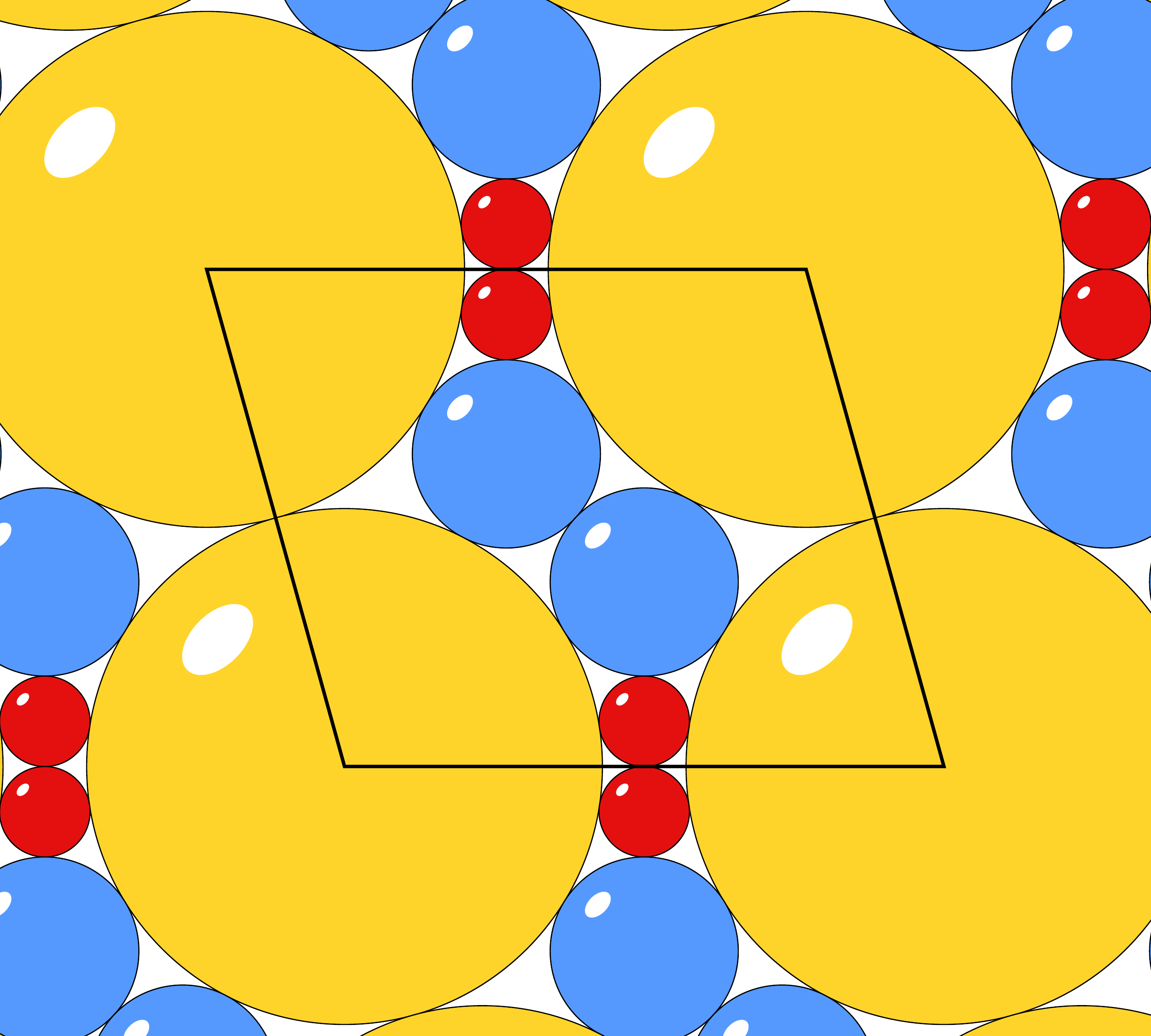} &
  \includegraphics[width=0.3\textwidth]{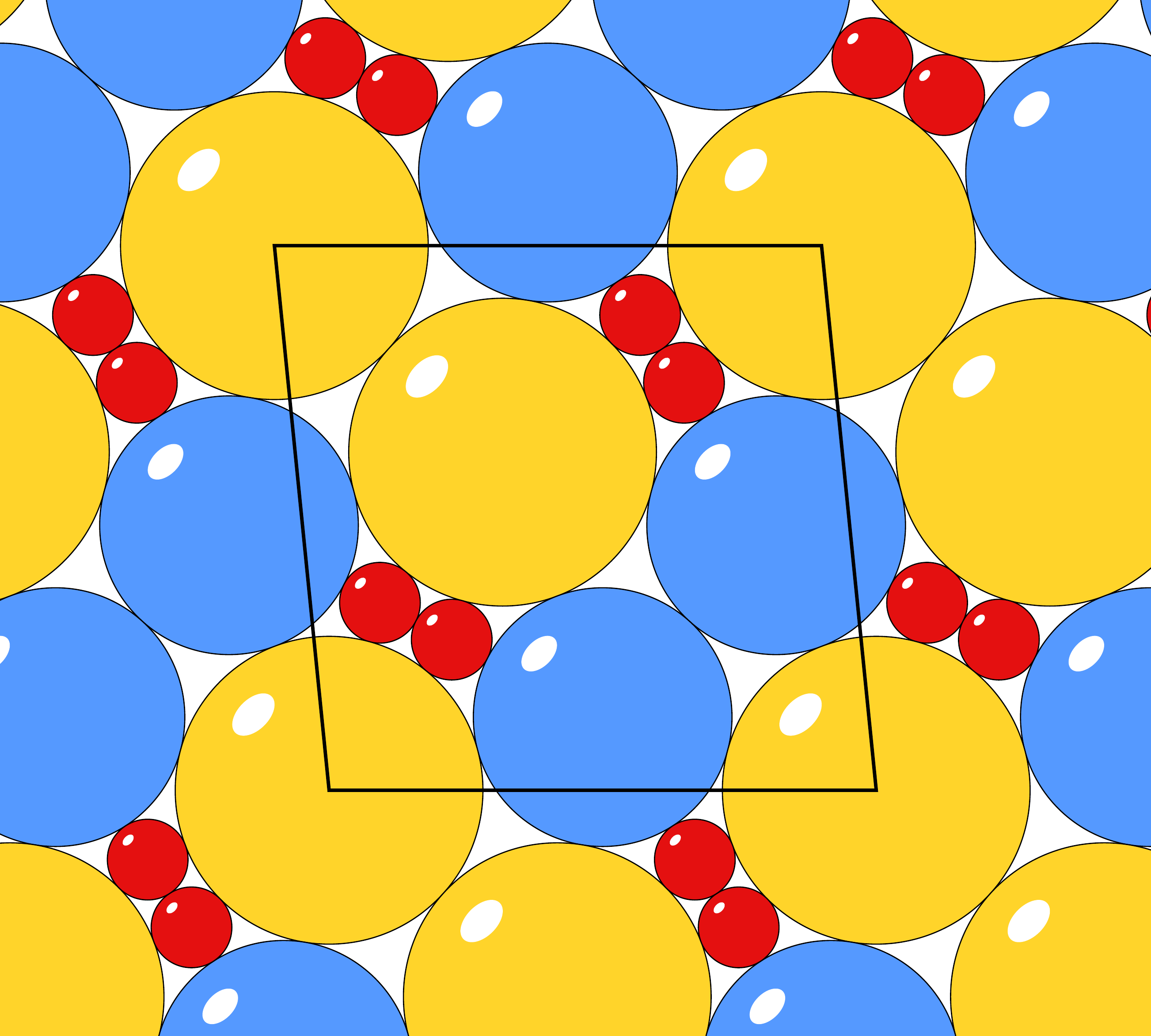} &
  \includegraphics[width=0.3\textwidth]{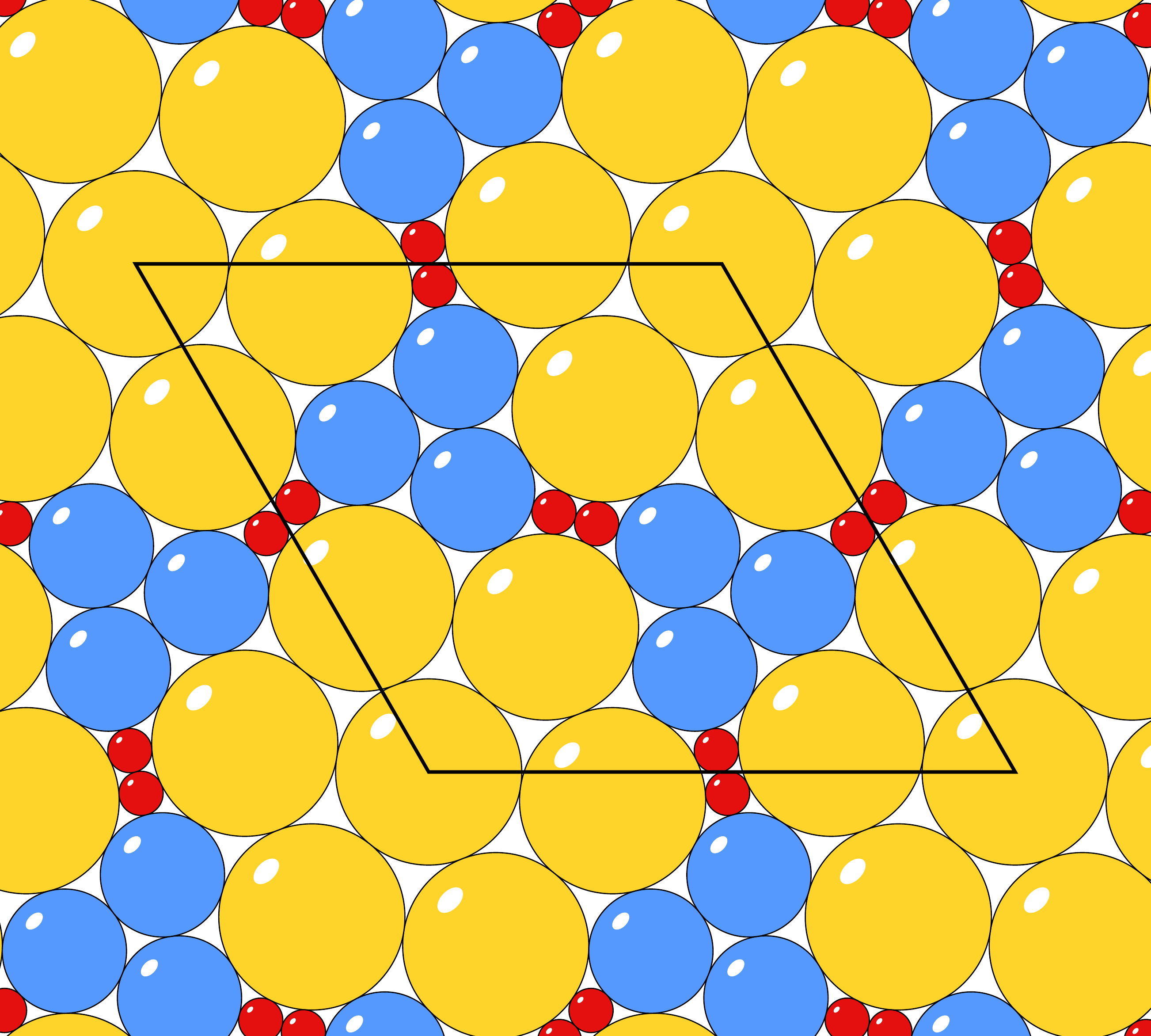}
\end{tabular}
\noindent
\begin{tabular}{lll}
  100 (L)\hfill 1r1s / 11s1s & 101 (L)\hfill 1r1s / 11s1s1s & 102 (L)\hfill 1r1s / 1r1r1s\\
  \includegraphics[width=0.3\textwidth]{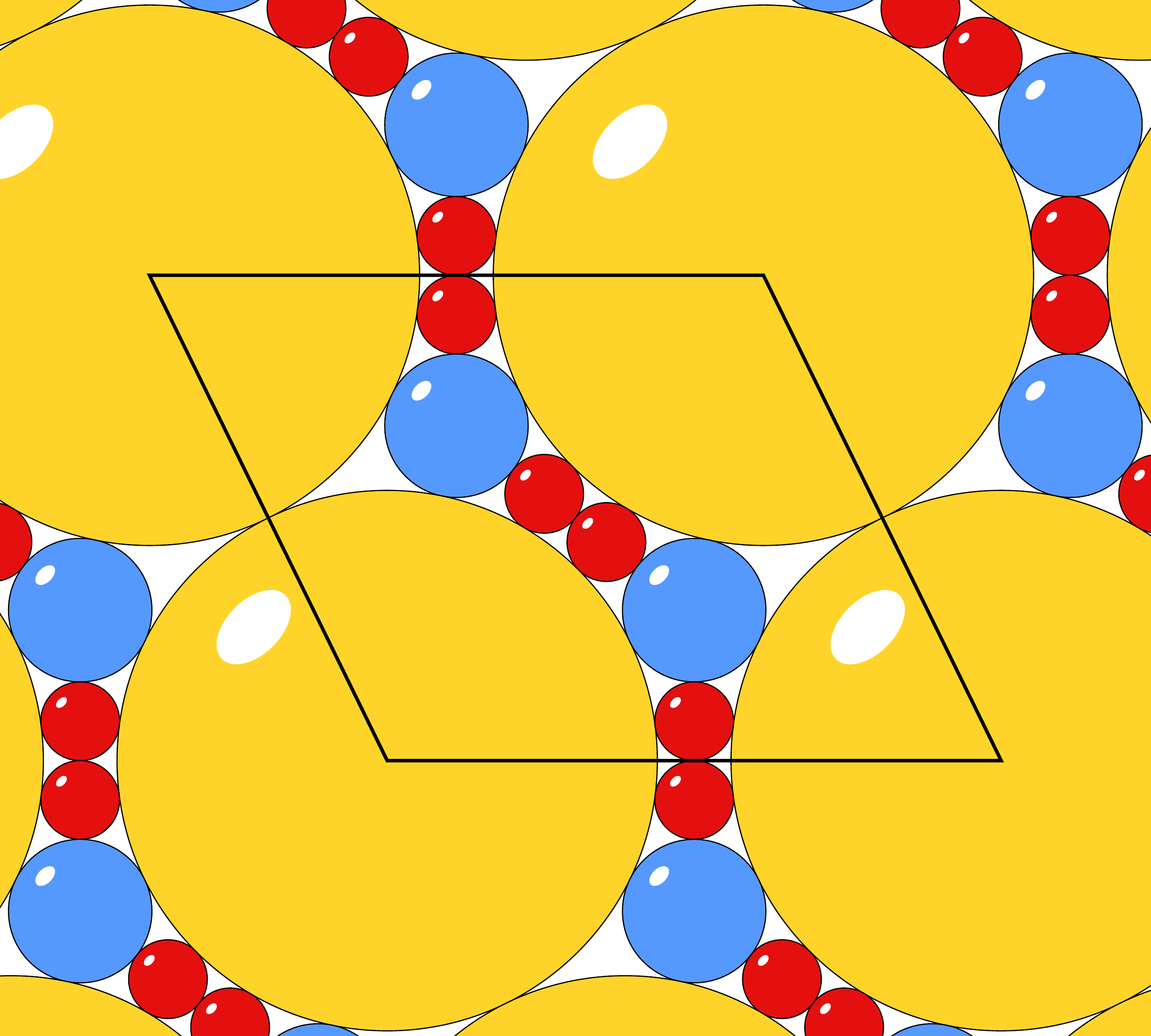} &
  \includegraphics[width=0.3\textwidth]{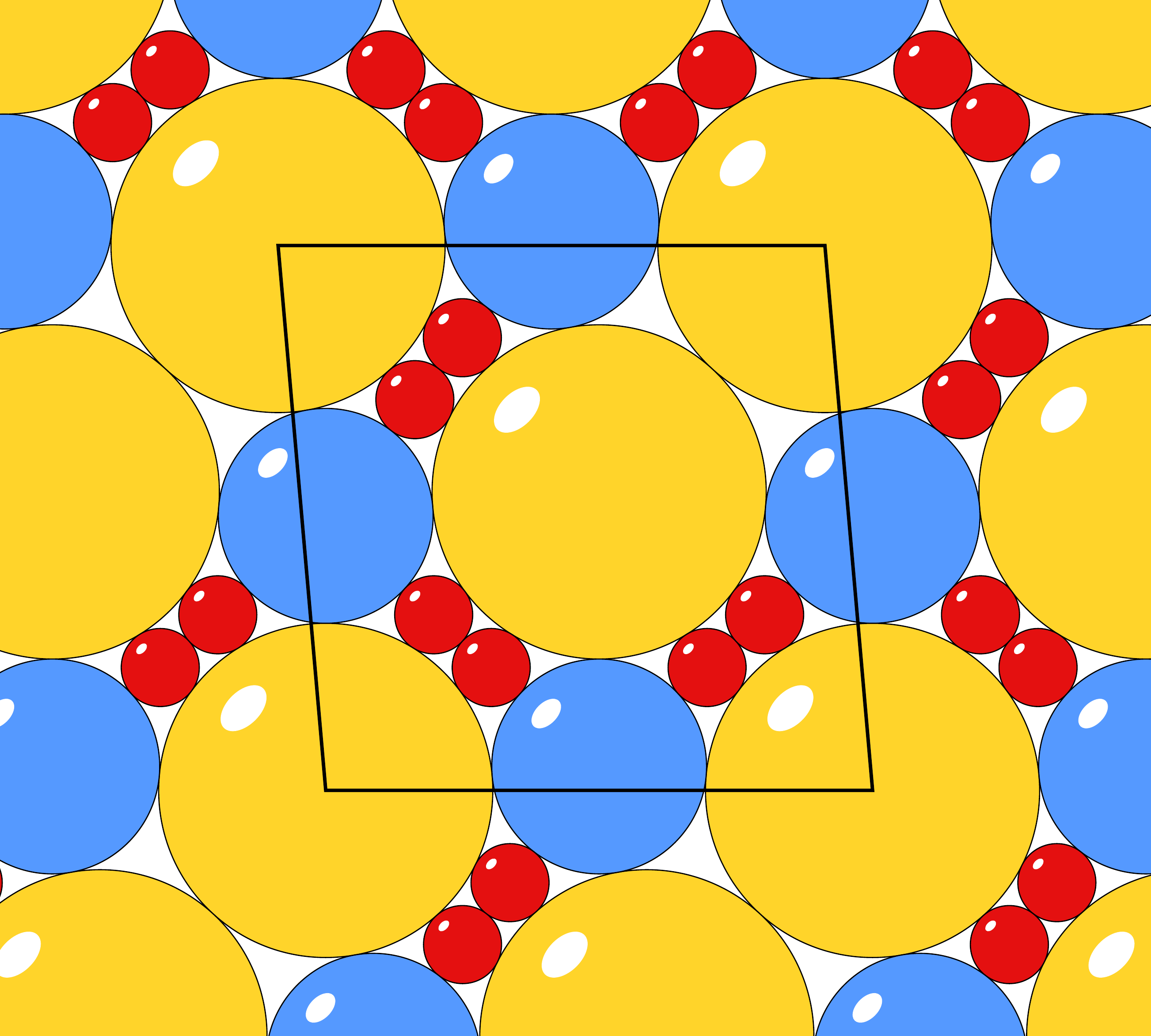} &
  \includegraphics[width=0.3\textwidth]{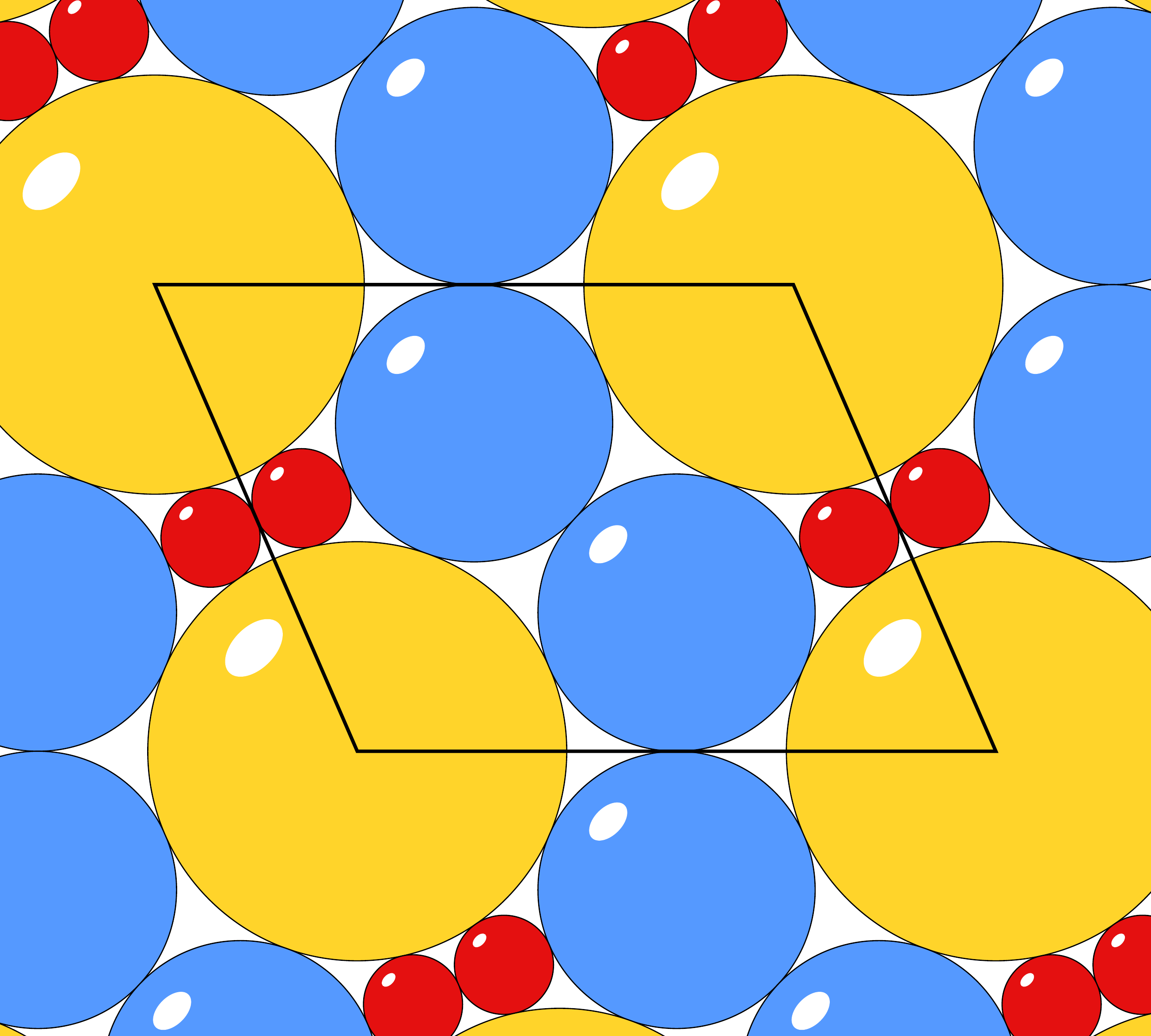}
\end{tabular}
\noindent
\begin{tabular}{lll}
  103 (L)\hfill 1r1s / 1r1s1s & 104 (L)\hfill 1r1s / 1r1s1s1s & 105 (H)\hfill 1r1s / 1rr1s\\
  \includegraphics[width=0.3\textwidth]{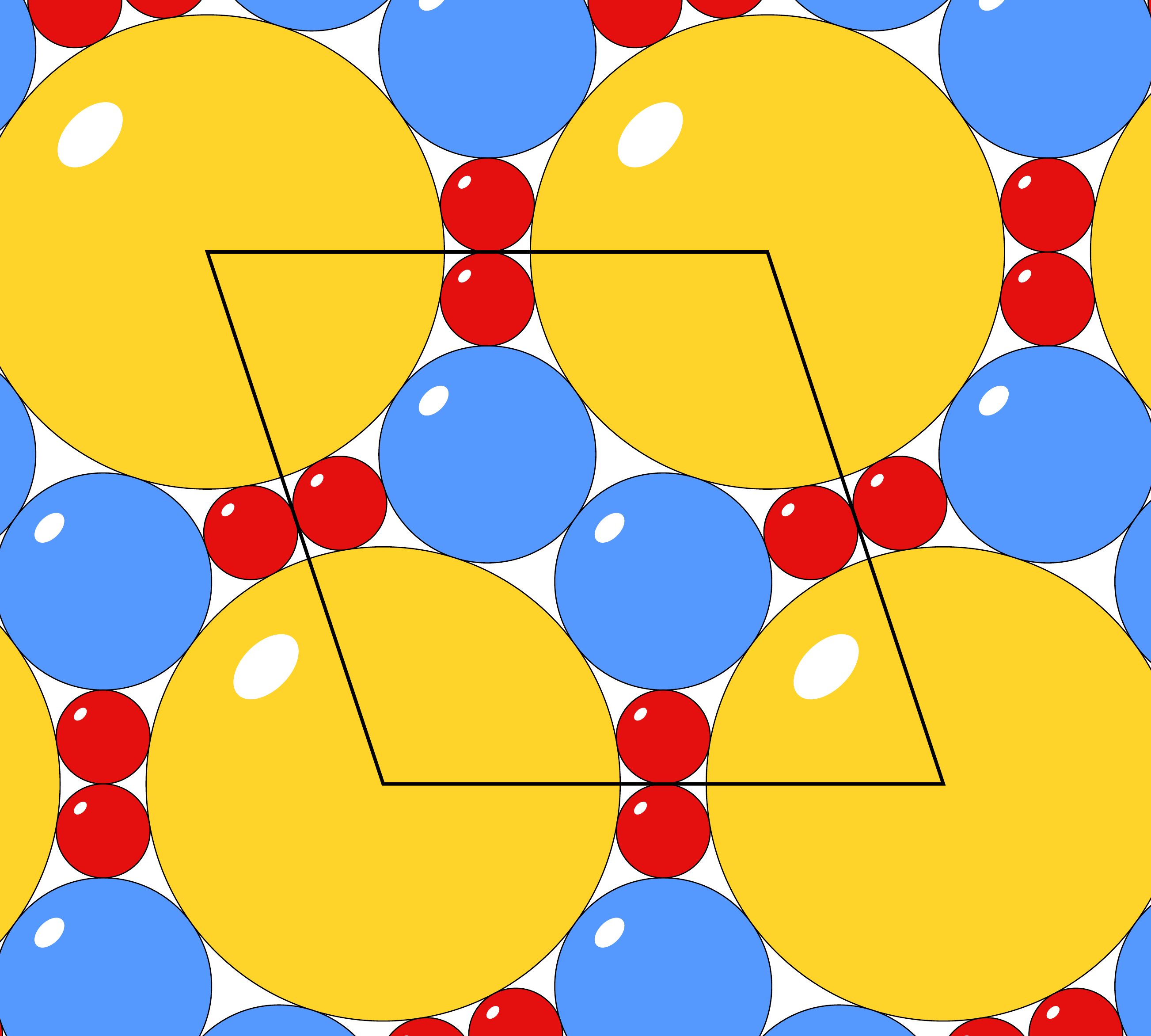} &
  \includegraphics[width=0.3\textwidth]{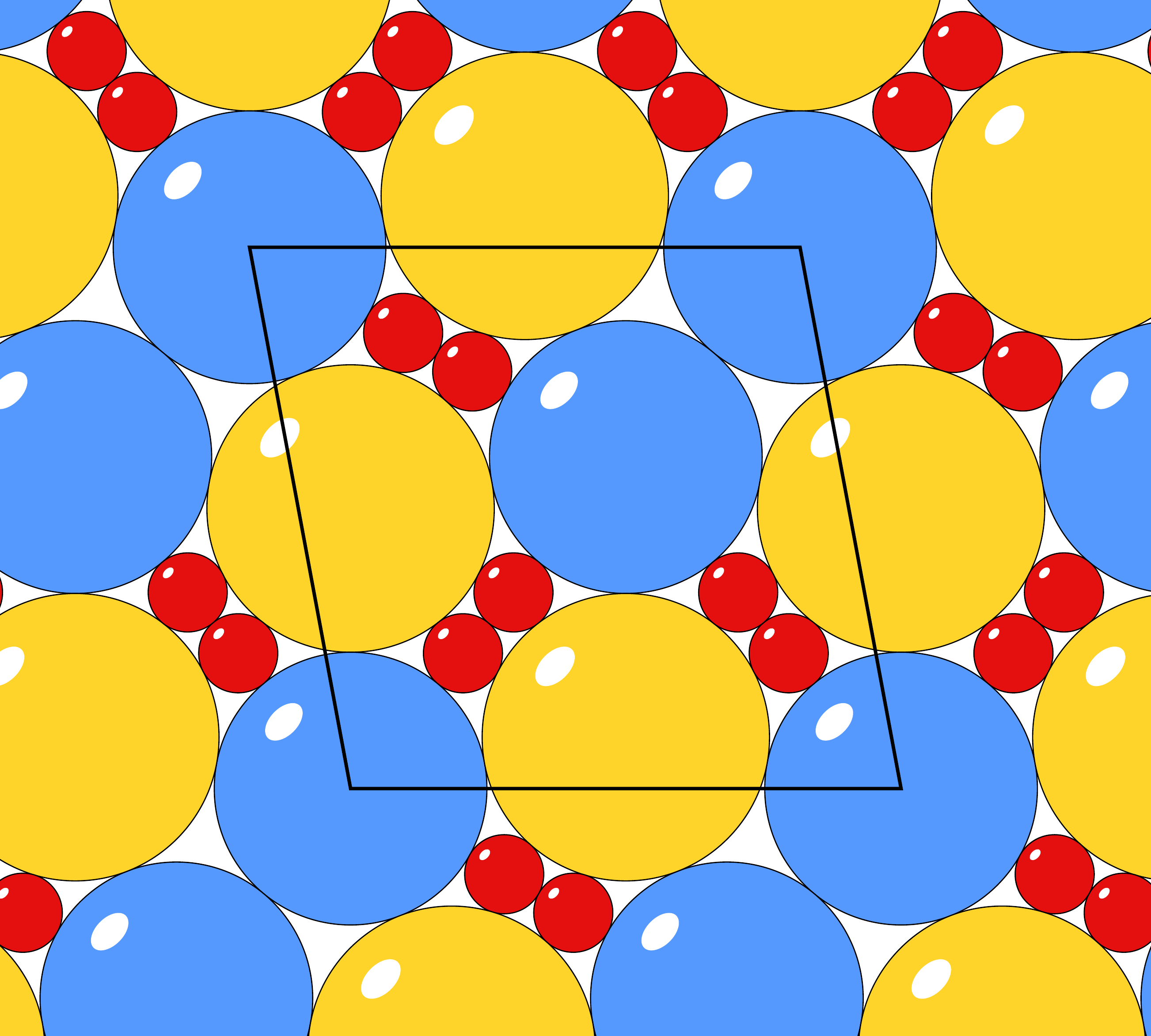} &
  \includegraphics[width=0.3\textwidth]{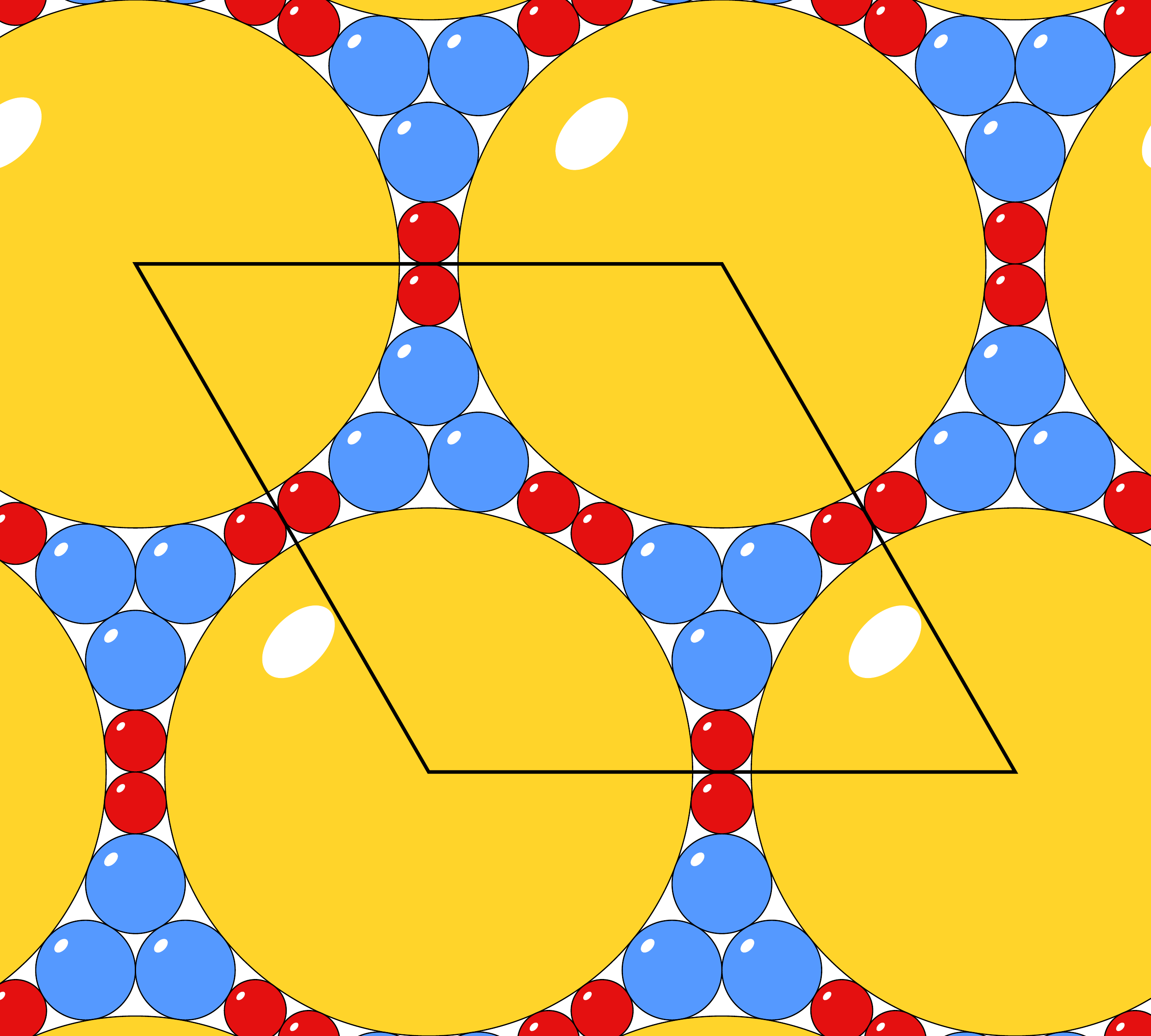}
\end{tabular}
\noindent
\begin{tabular}{lll}
  106 (H)\hfill 1r1s / 1rrr1s & 107 (H)\hfill 1r1s / 1s1s1s & 108 (H)\hfill 1r1s / 1s1s1s1s\\
  \includegraphics[width=0.3\textwidth]{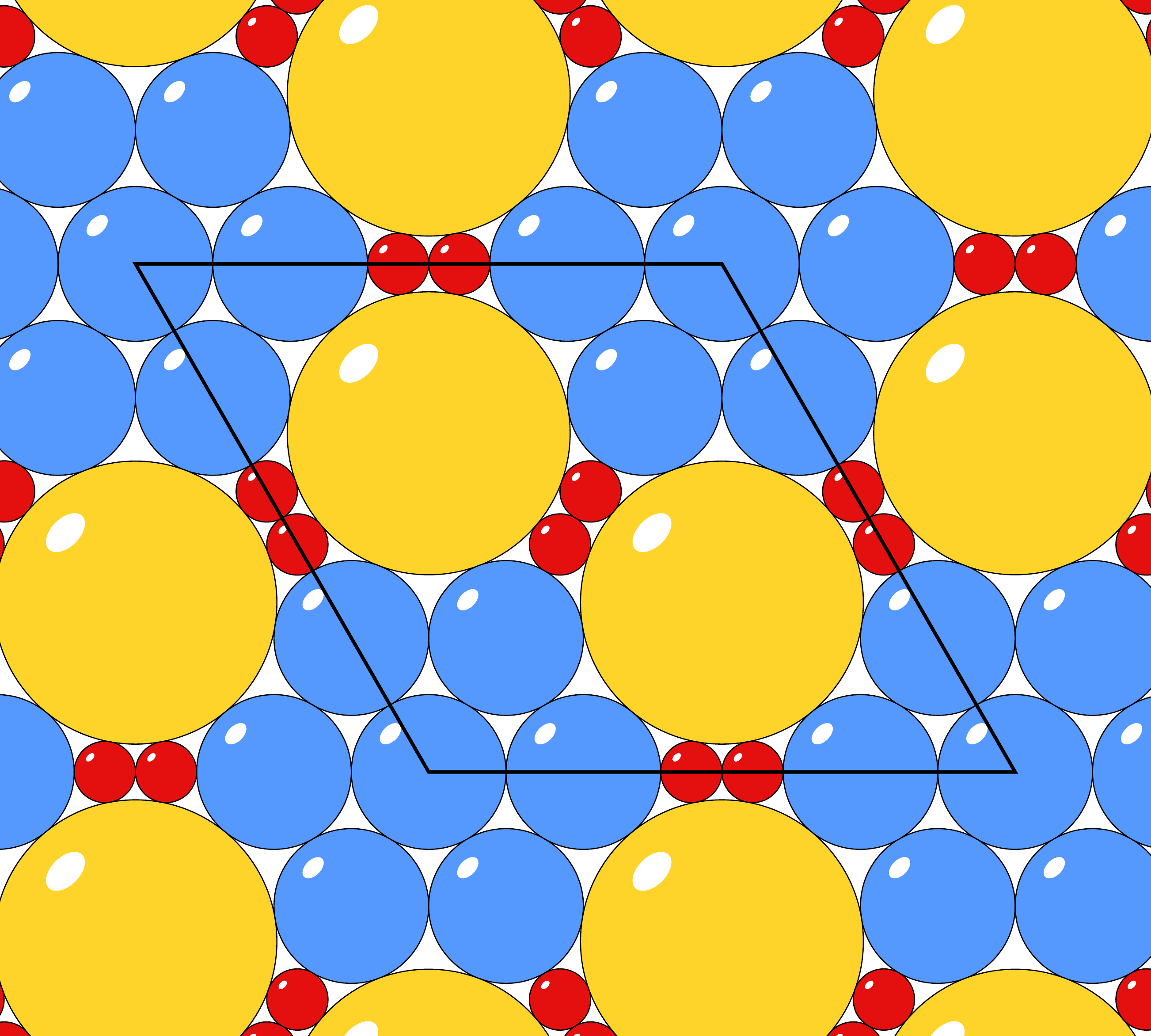} &
  \includegraphics[width=0.3\textwidth]{packing_1r1s_1s1s1s.pdf} &
  \includegraphics[width=0.3\textwidth]{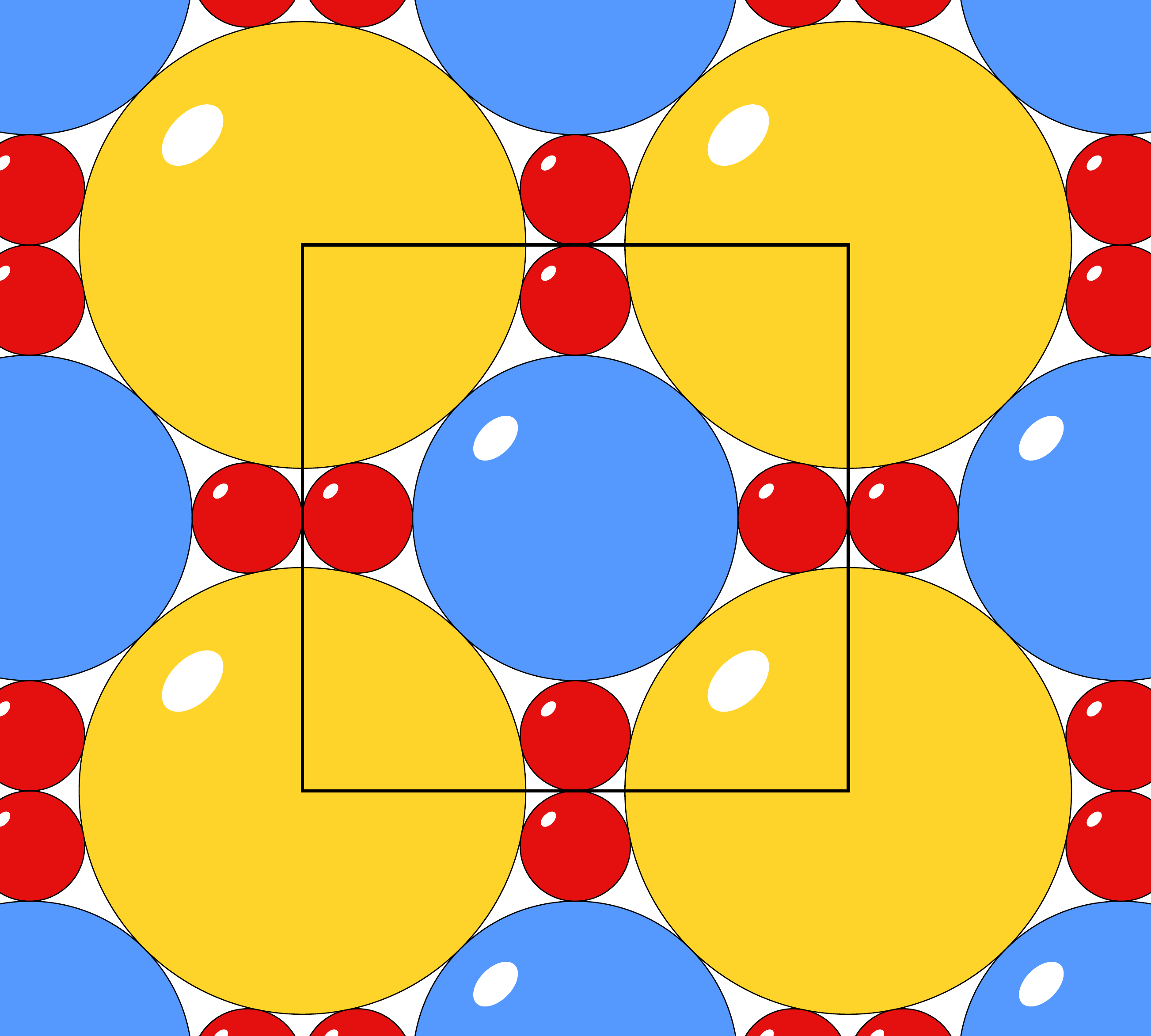}
\end{tabular}
\noindent
\begin{tabular}{lll}
  109 (E)\hfill 1r1s / 1s1sss & 110 (H)\hfill 1r1ss / 111s1s & 111 (H)\hfill 1r1ss / 11r1s\\
  \includegraphics[width=0.3\textwidth]{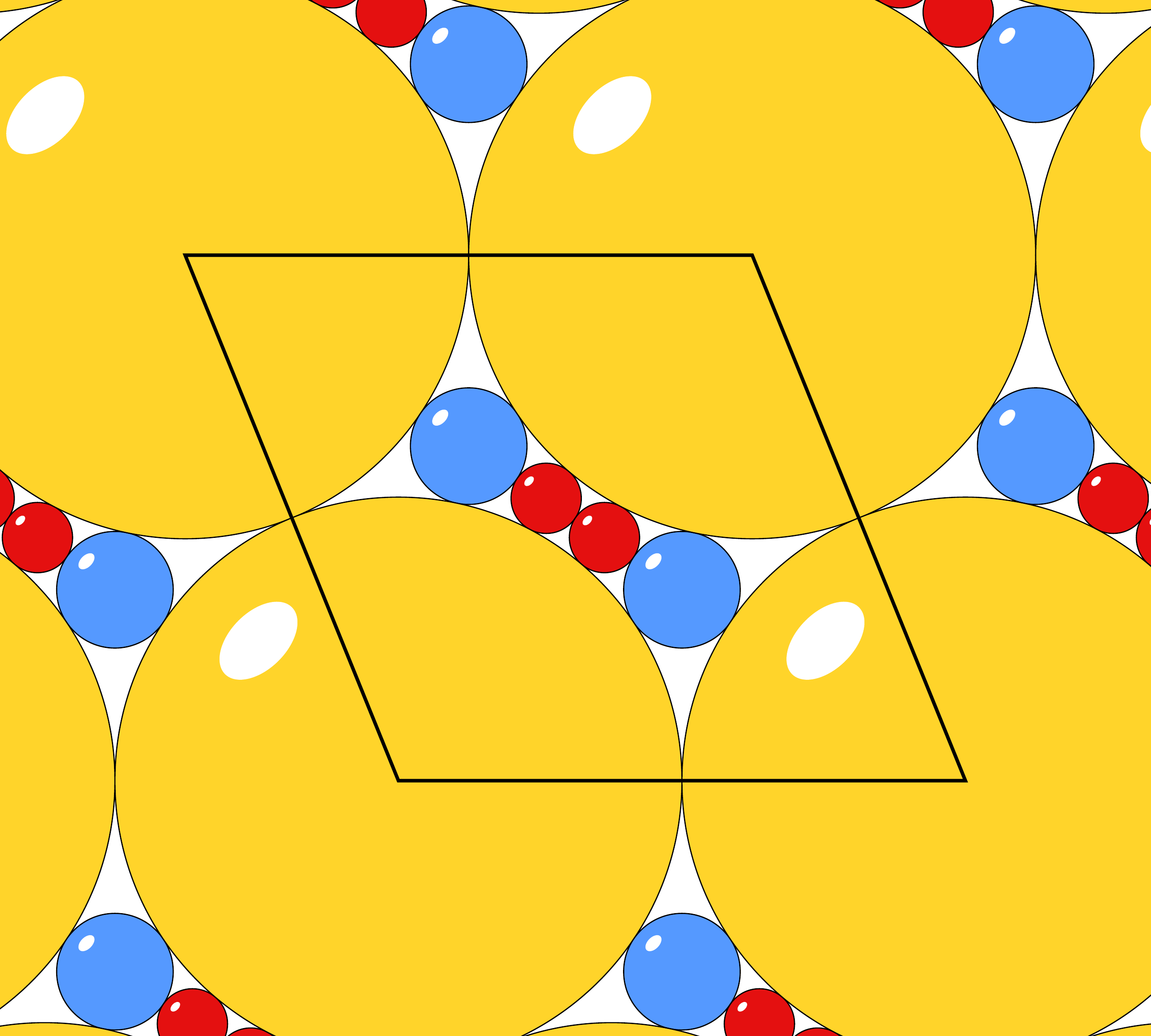} &
  \includegraphics[width=0.3\textwidth]{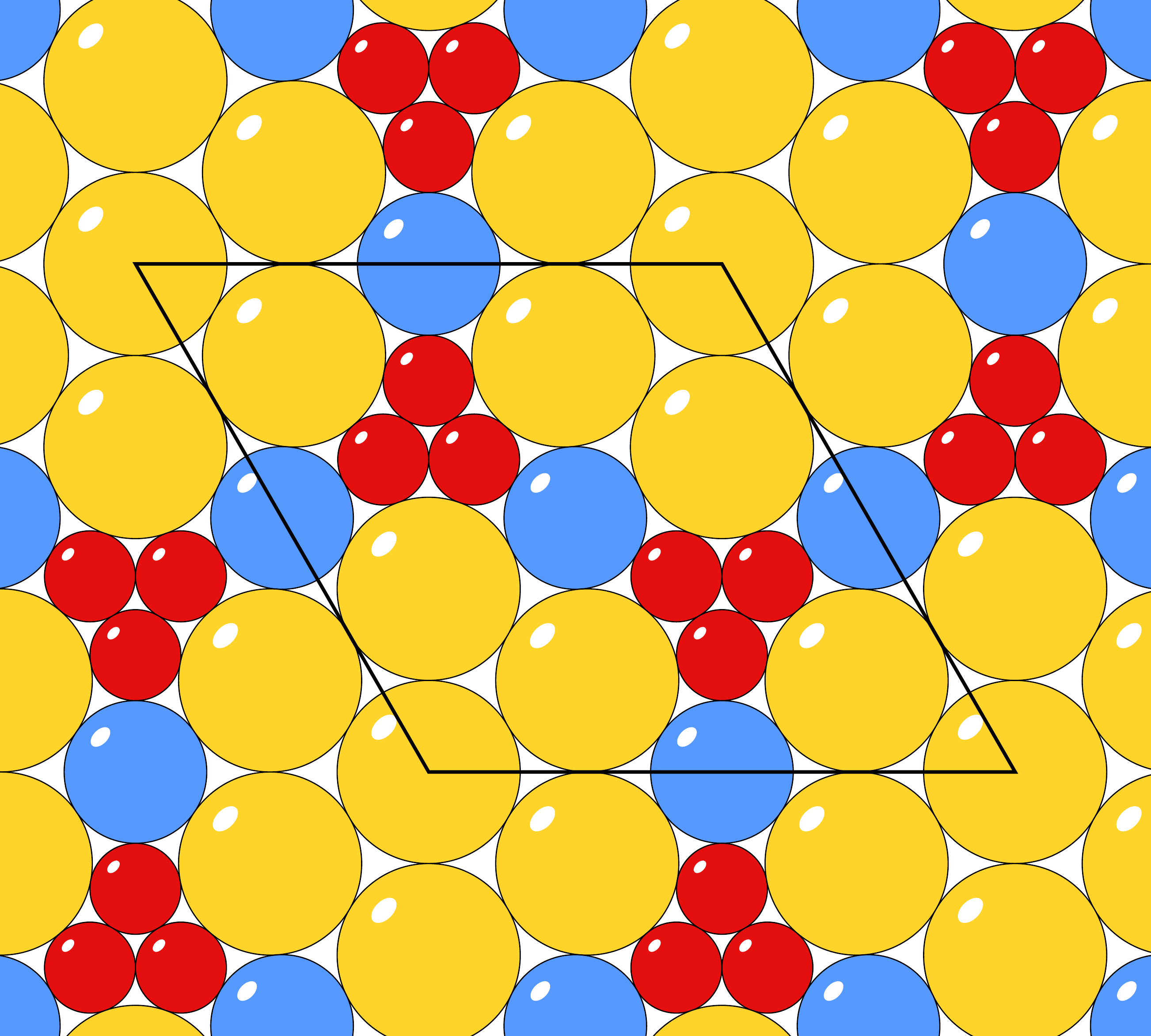} &
  \includegraphics[width=0.3\textwidth]{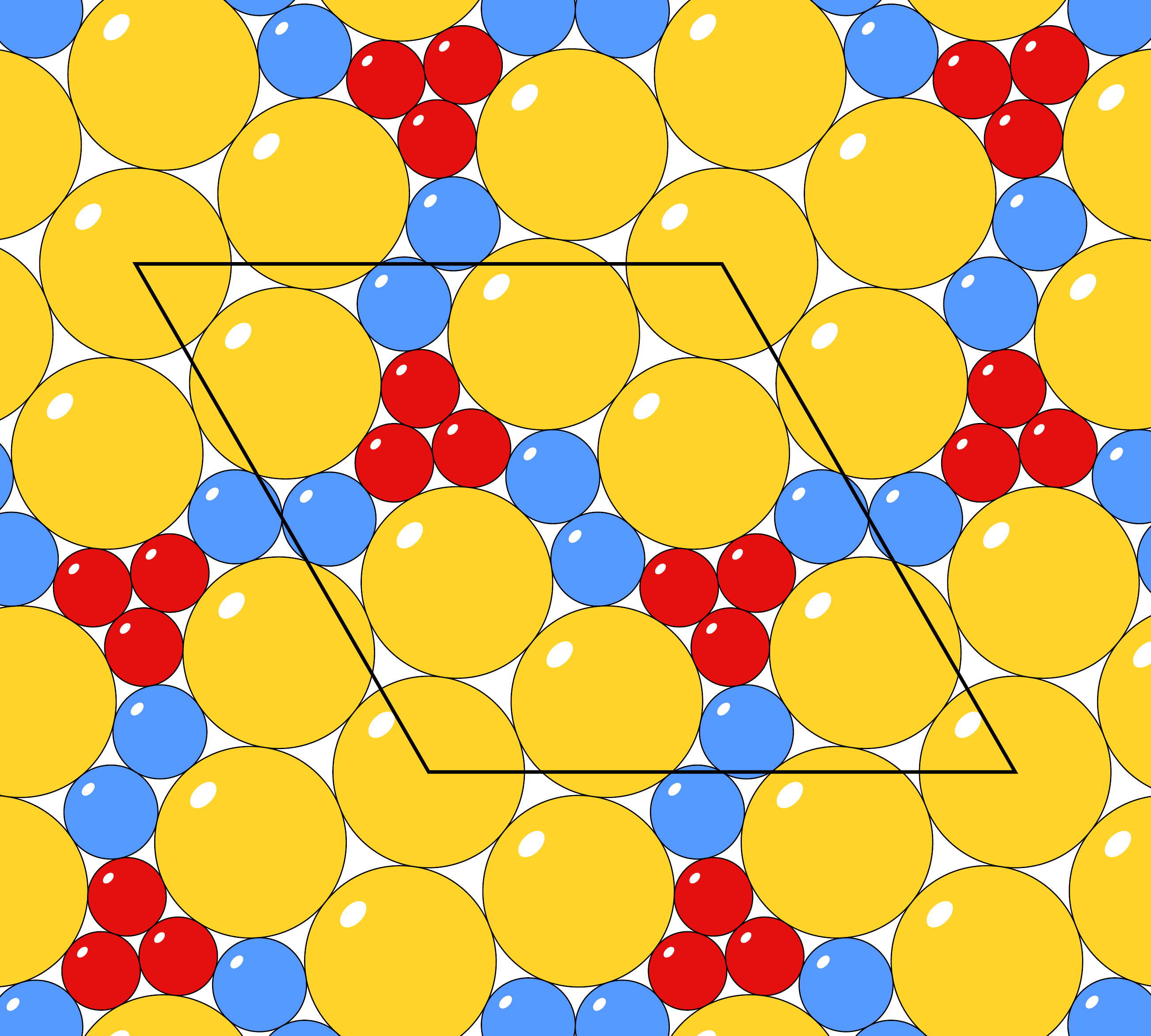}
\end{tabular}
\noindent
\begin{tabular}{lll}
  112 (H)\hfill 1r1ss / 11rr1s & 113 (H)\hfill 1r1ss / 11s1s & 114 (H)\hfill 1r1ss / 1rrr1s\\
  \includegraphics[width=0.3\textwidth]{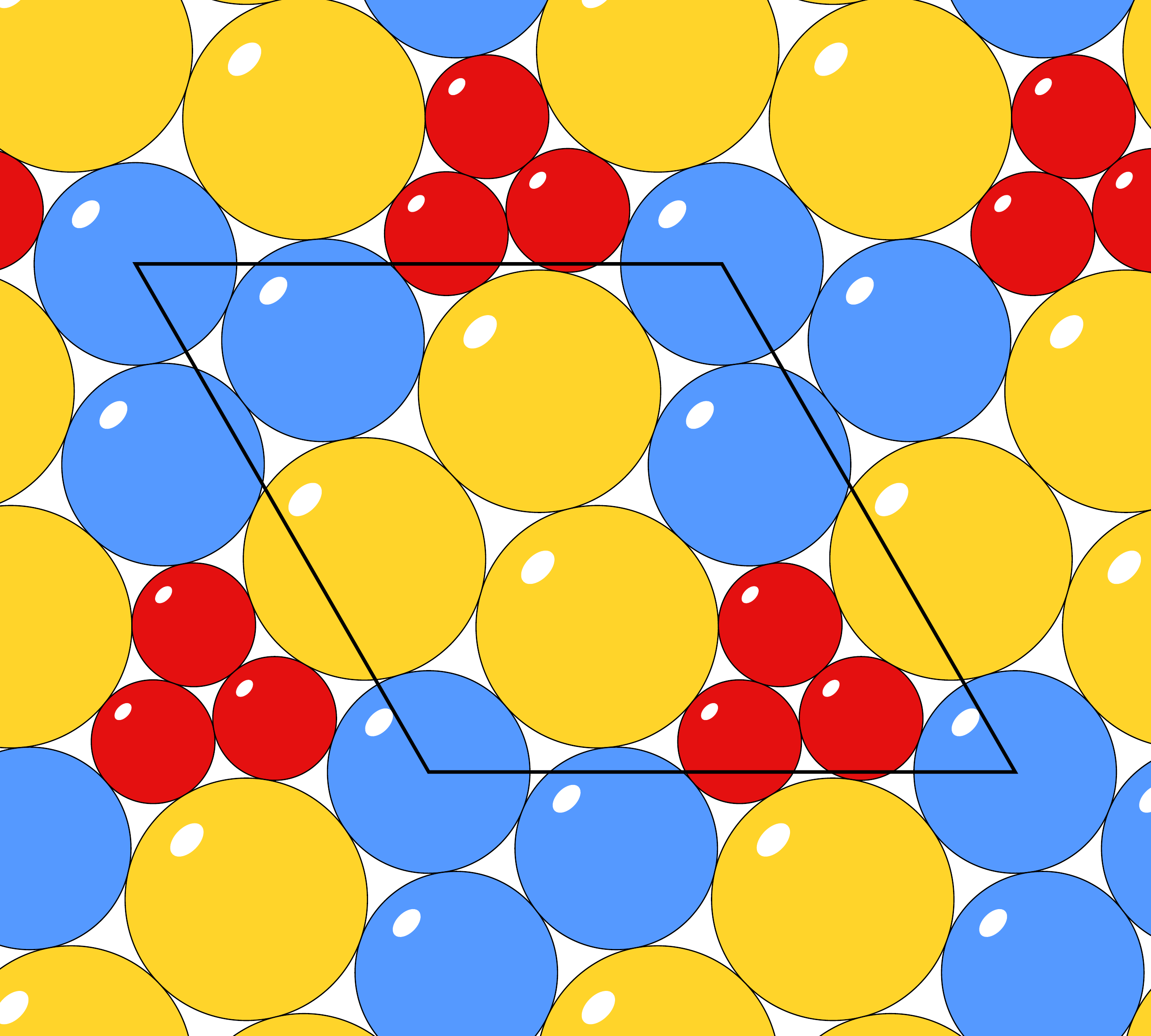} &
  \includegraphics[width=0.3\textwidth]{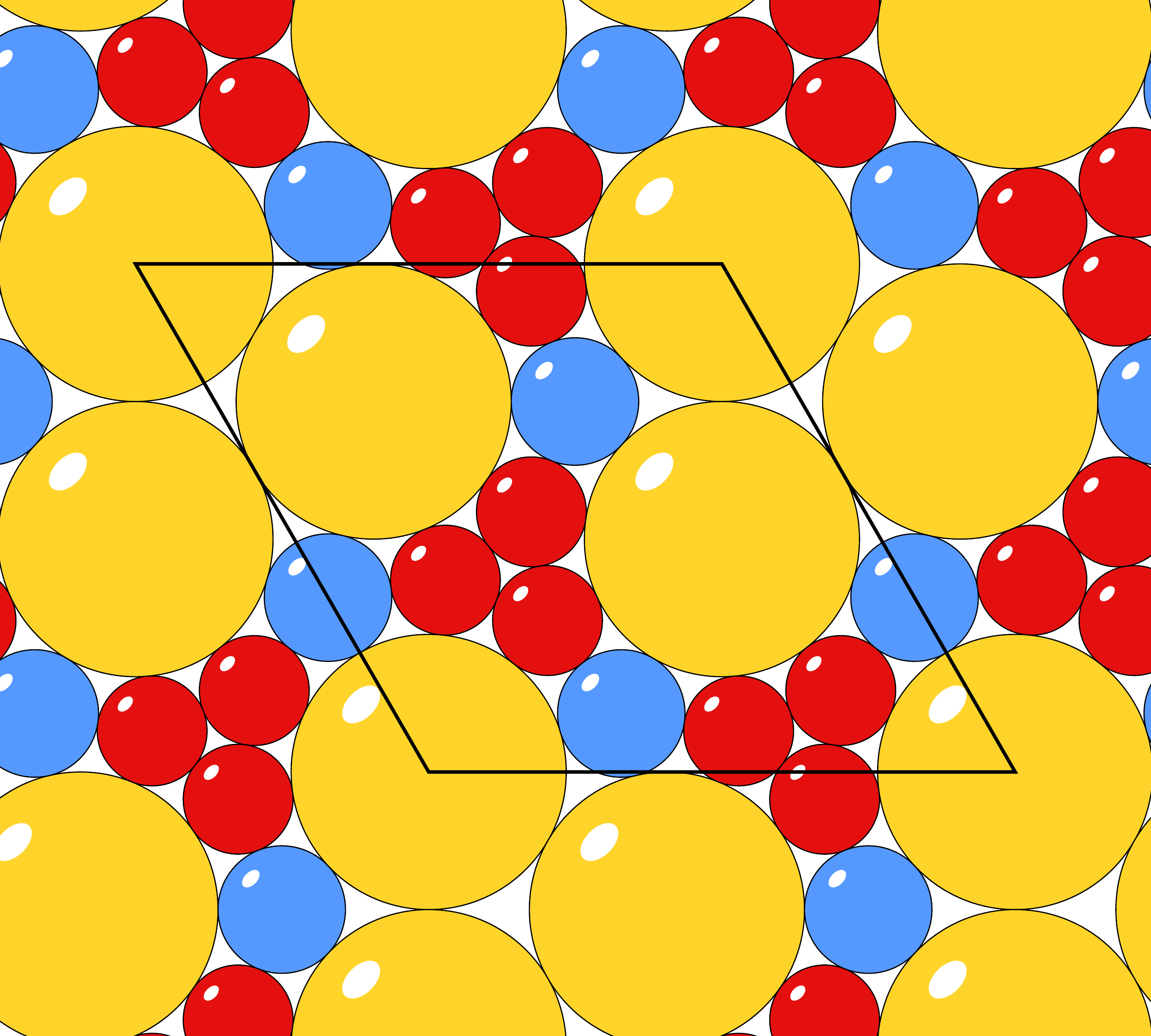} &
  \includegraphics[width=0.3\textwidth]{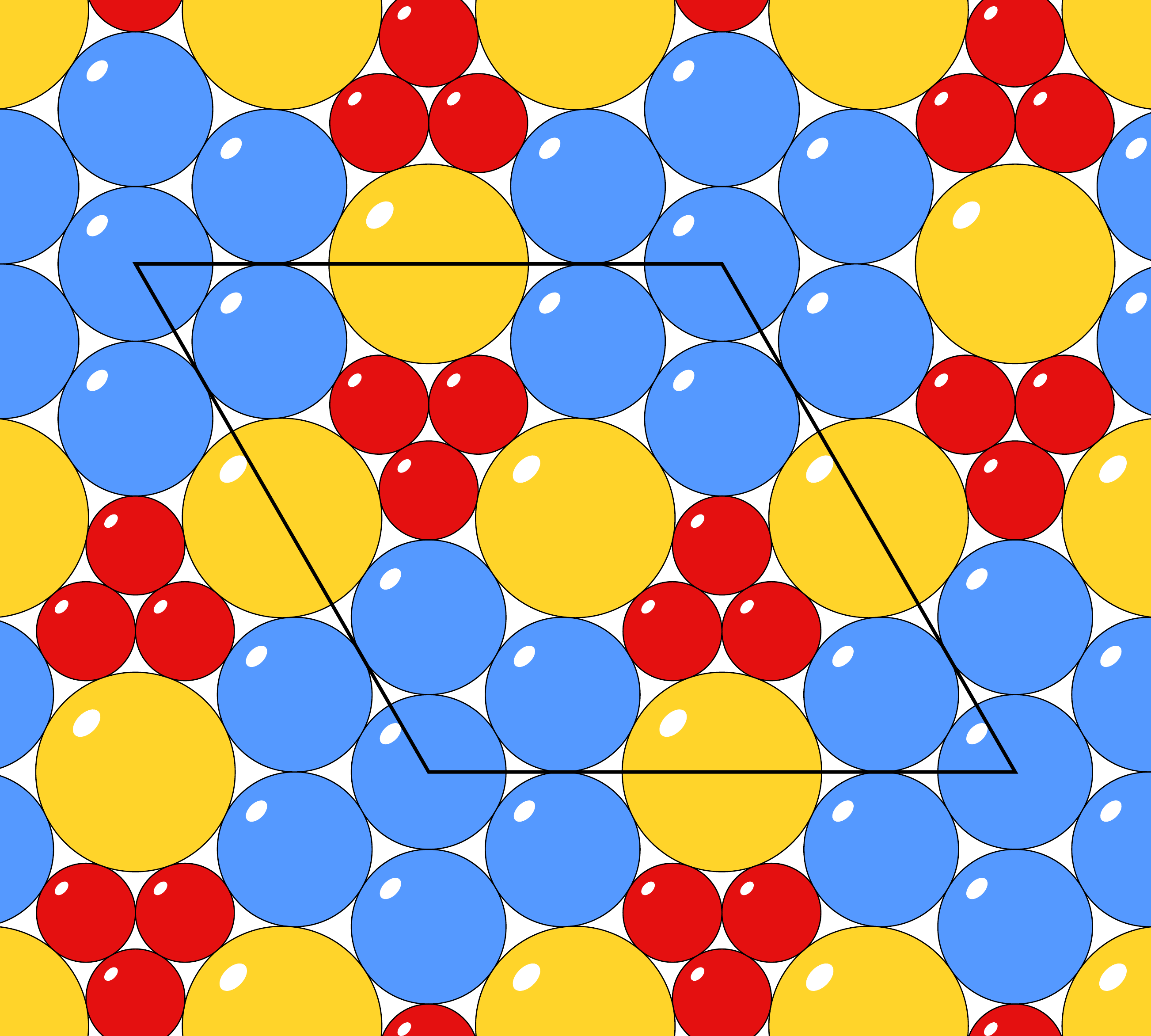}
\end{tabular}
\noindent
\begin{tabular}{lll}
  115 (H)\hfill 1r1ss / 1s1s1s & 116 (H)\hfill 1rr1s / 111srs & 117 (H)\hfill 1rr1s / 11srrs\\
  \includegraphics[width=0.3\textwidth]{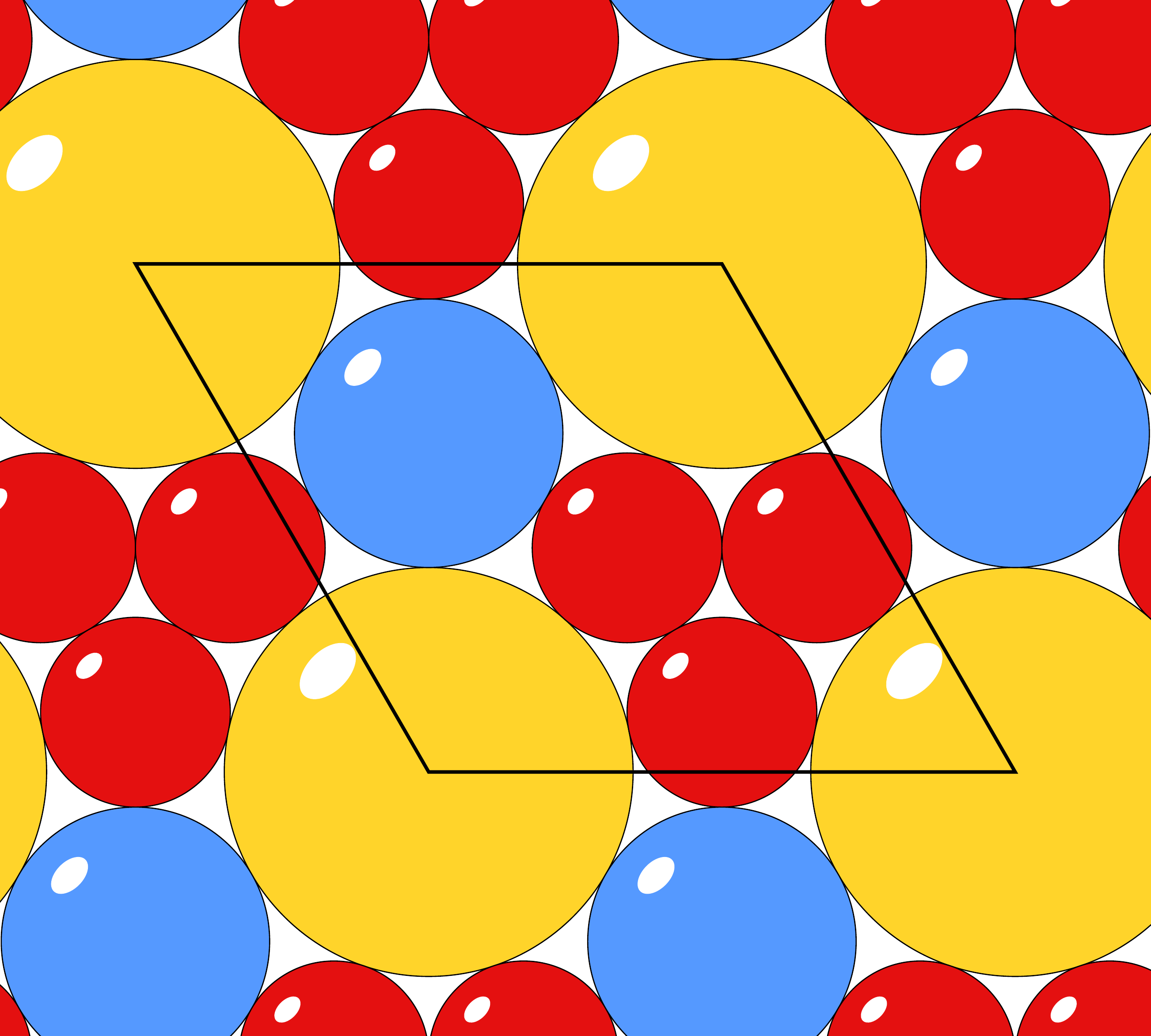} &
  \includegraphics[width=0.3\textwidth]{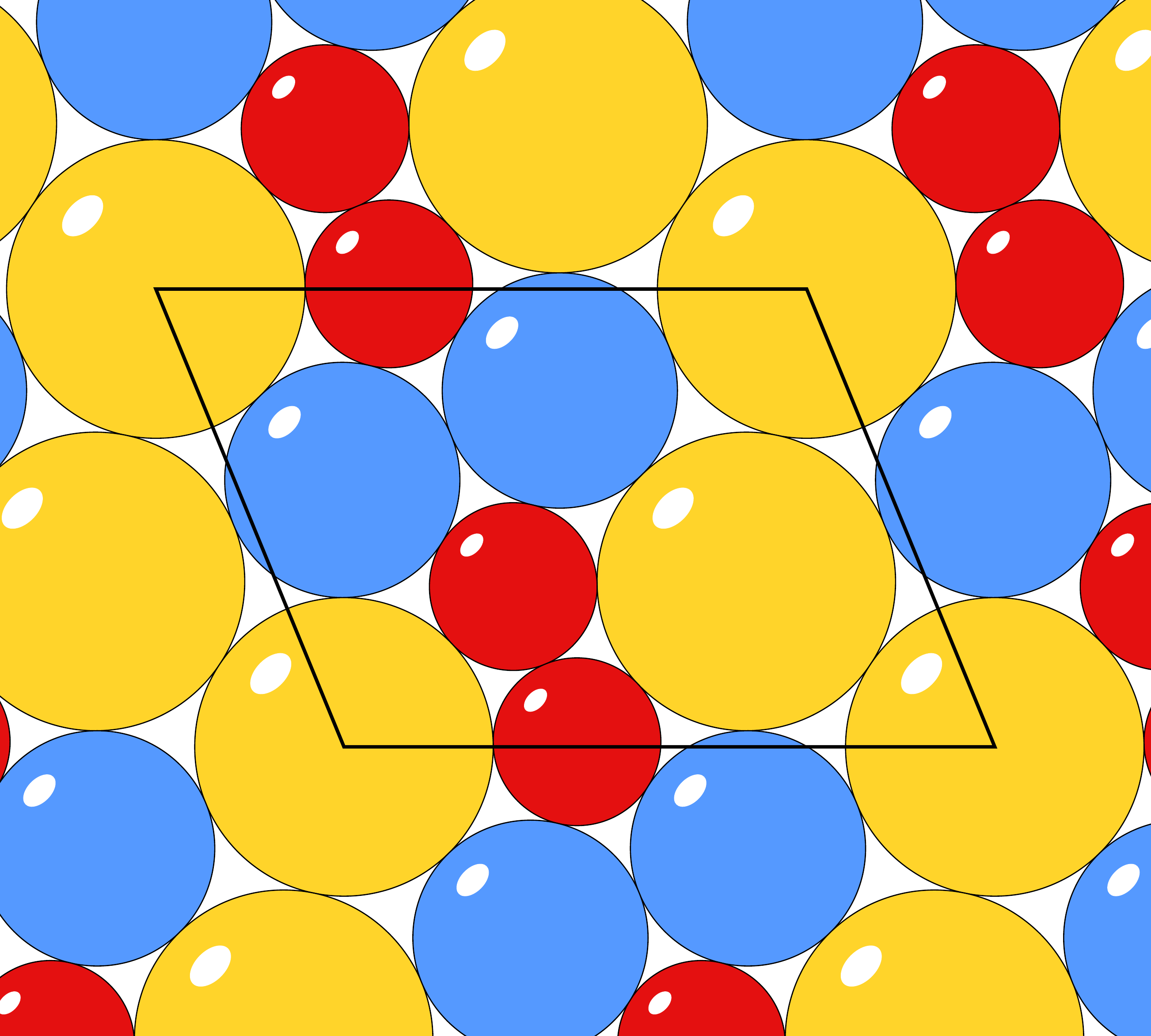} &
  \includegraphics[width=0.3\textwidth]{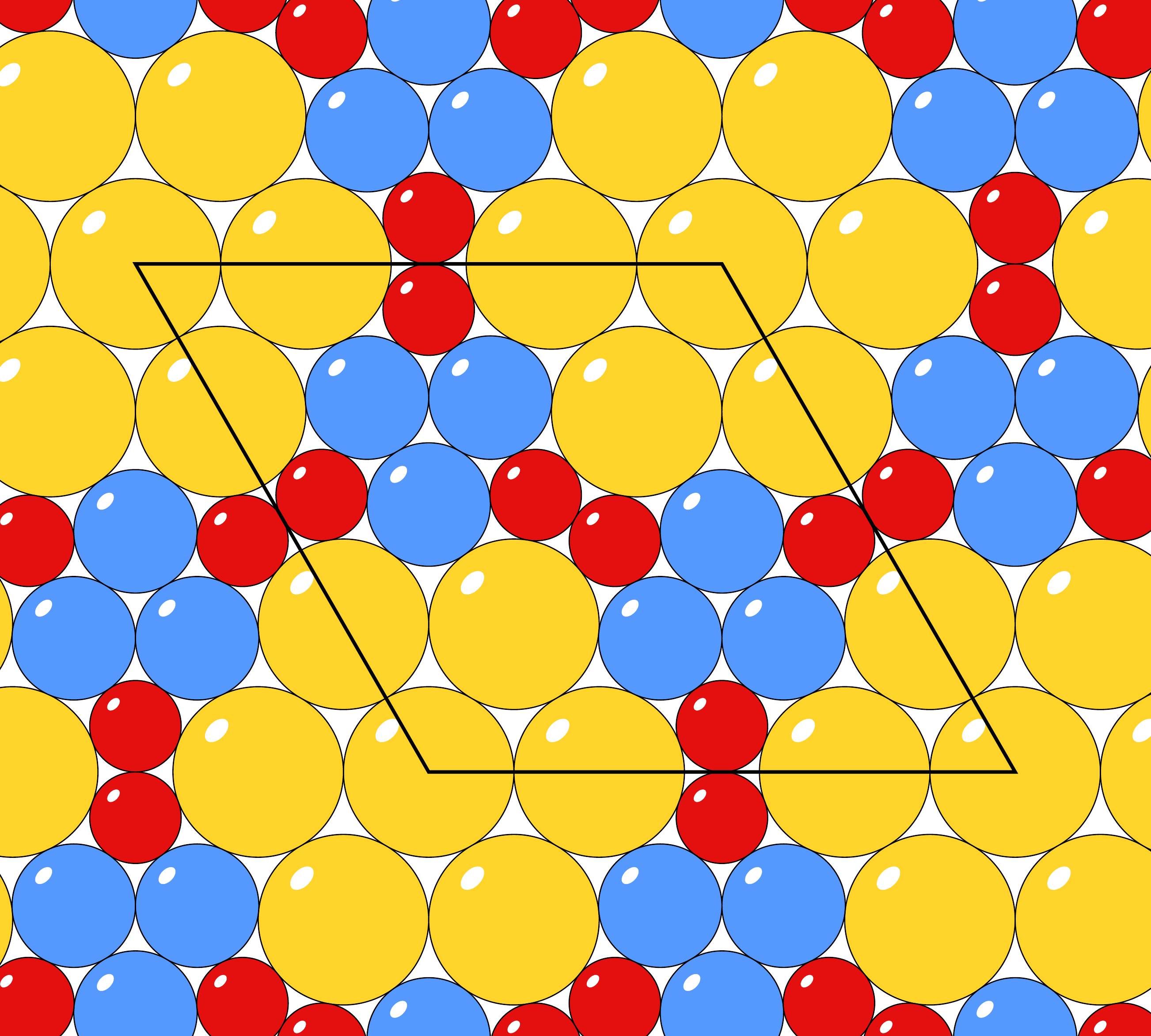}
\end{tabular}
\noindent
\begin{tabular}{lll}
  118 (L)\hfill 1rr1s / 11srs & 119 (L)\hfill 1rr1s / 1rrrrs & 120 (H)\hfill 1rr1s / 1srrrs\\
  \includegraphics[width=0.3\textwidth]{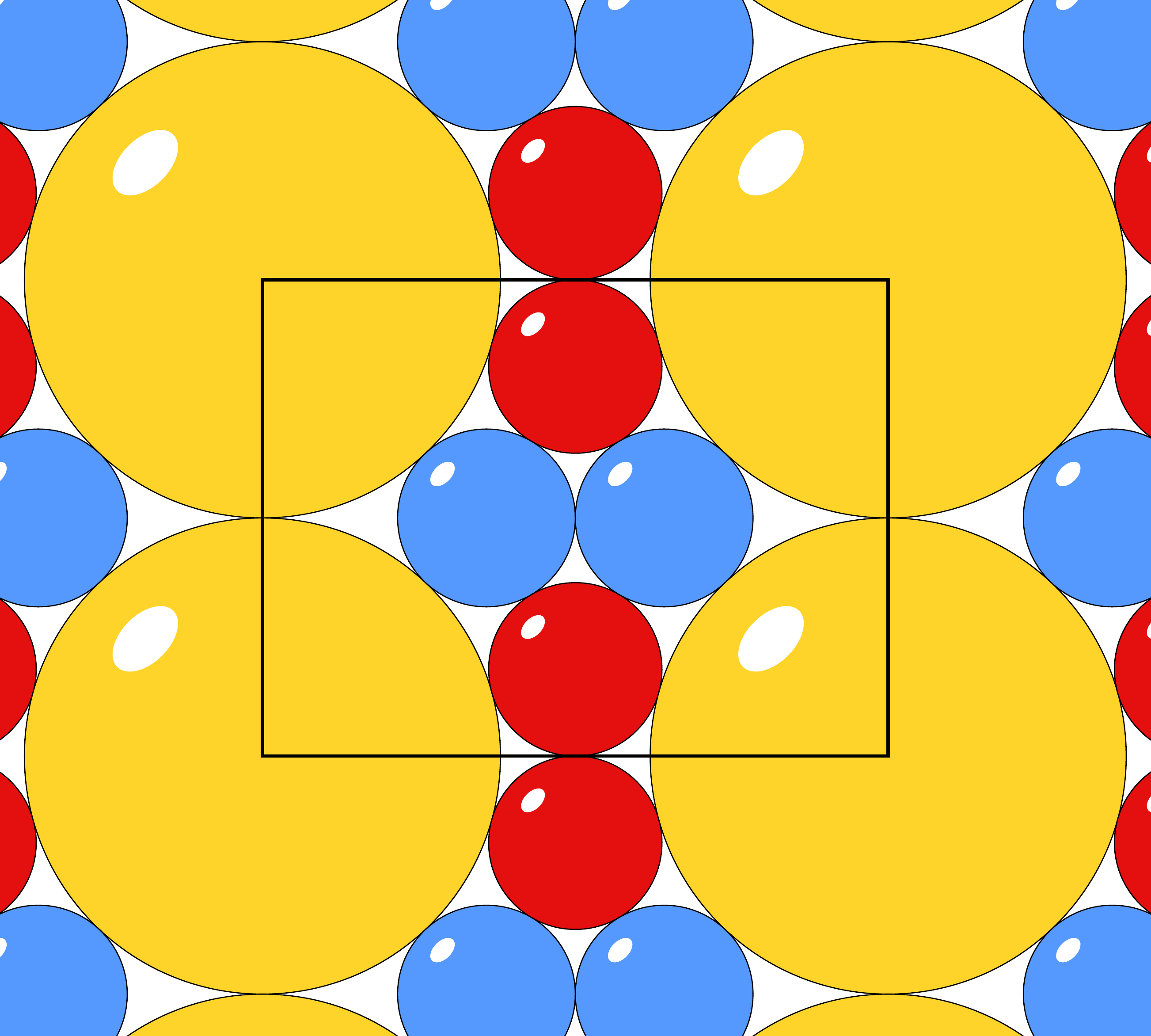} &
  \includegraphics[width=0.3\textwidth]{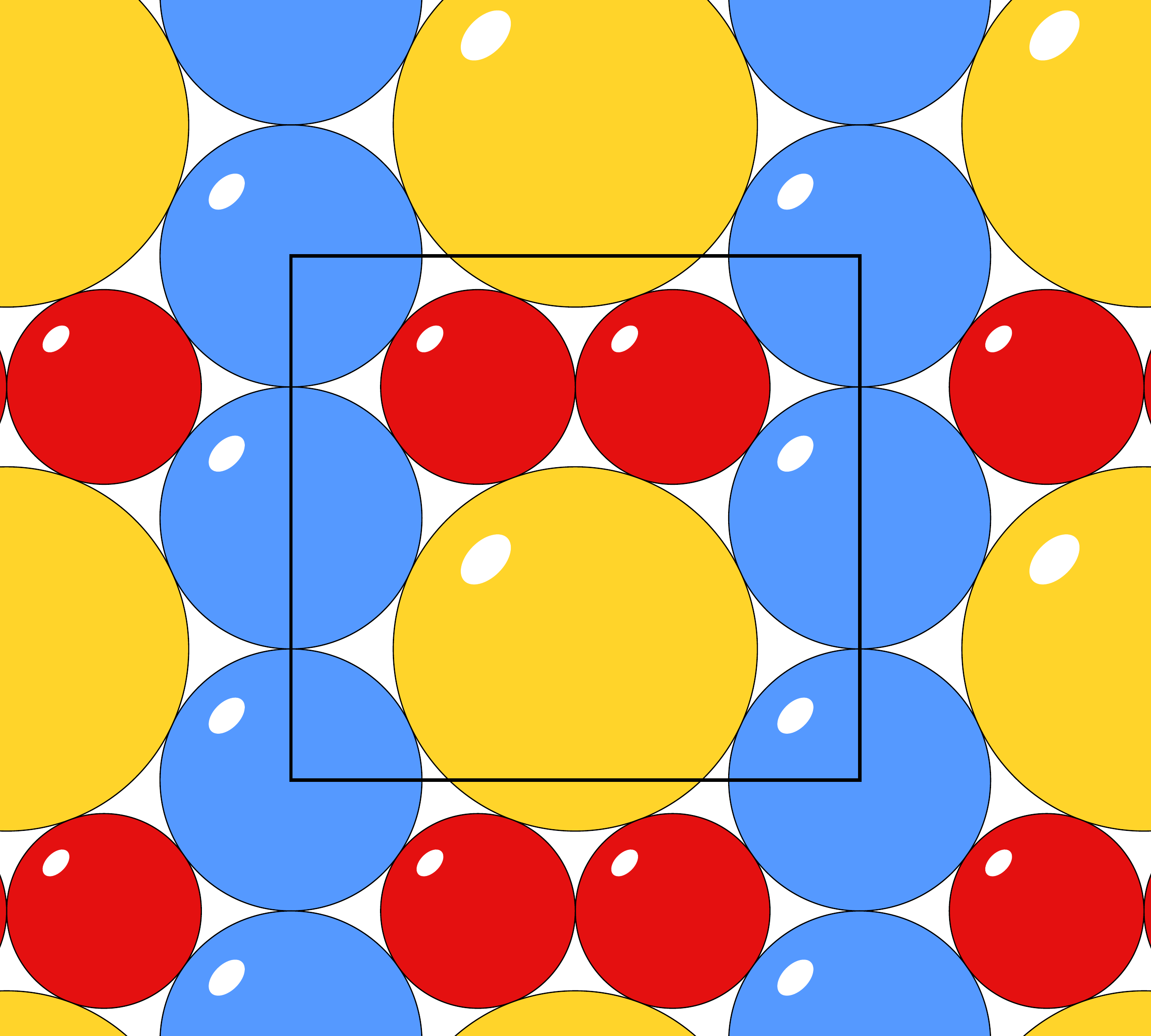} &
  \includegraphics[width=0.3\textwidth]{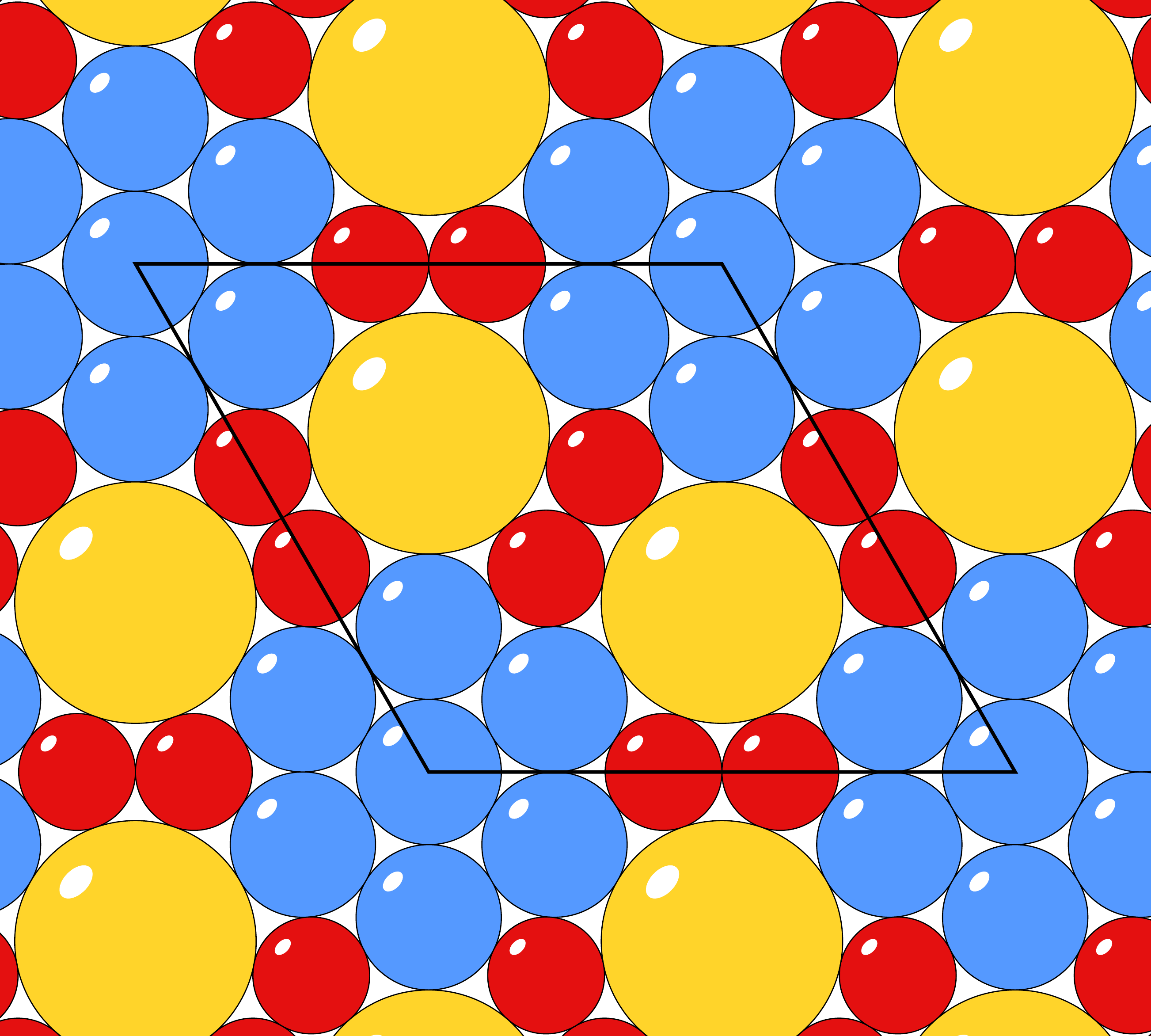}
\end{tabular}
\noindent
\begin{tabular}{lll}
  121 (L)\hfill 1rrr / 11srsrs & 122 (H)\hfill 1rrr / 1srrsrs & 123 (H)\hfill 1rrrr / 11rsrs\\
  \includegraphics[width=0.3\textwidth]{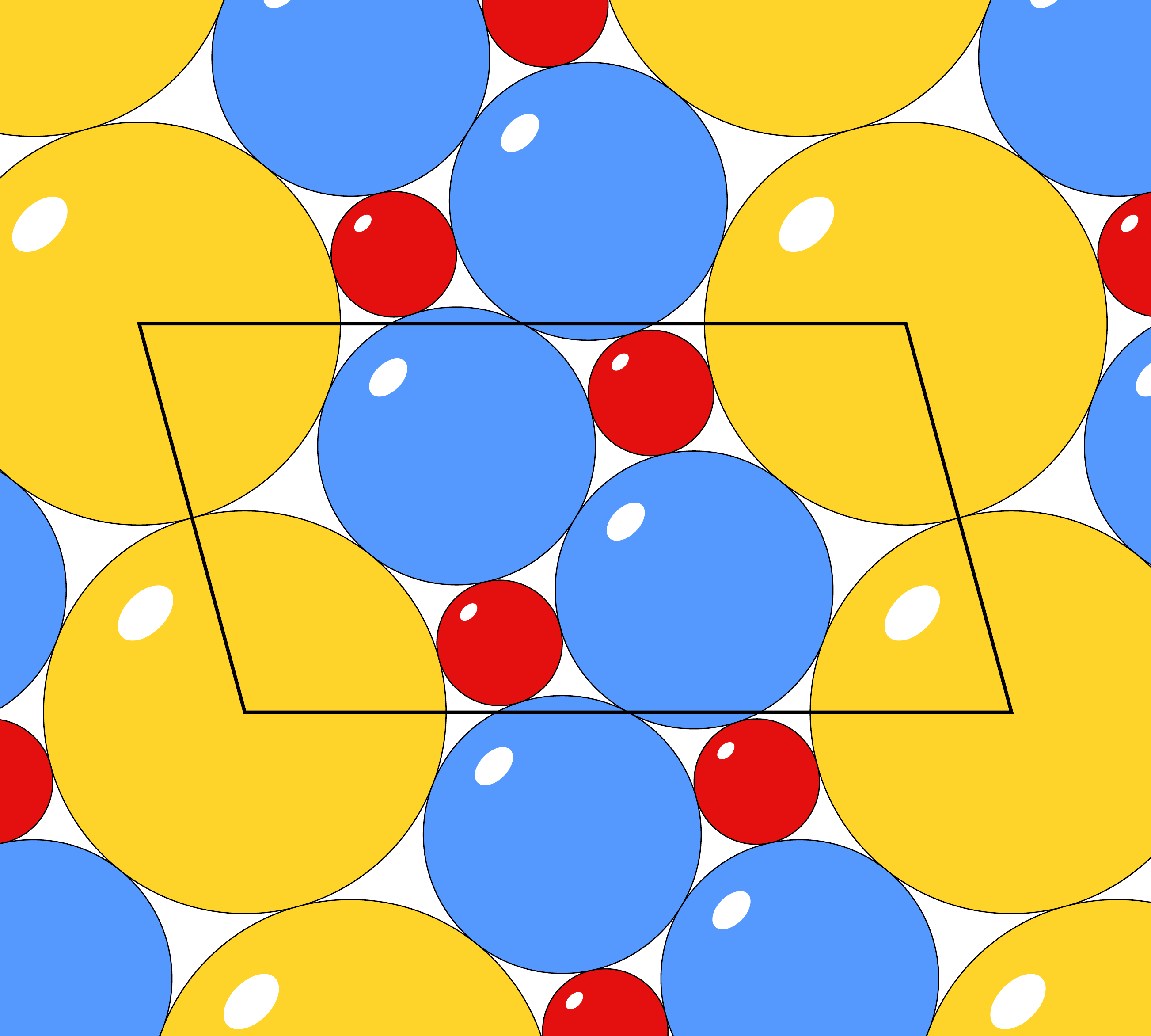} &
  \includegraphics[width=0.3\textwidth]{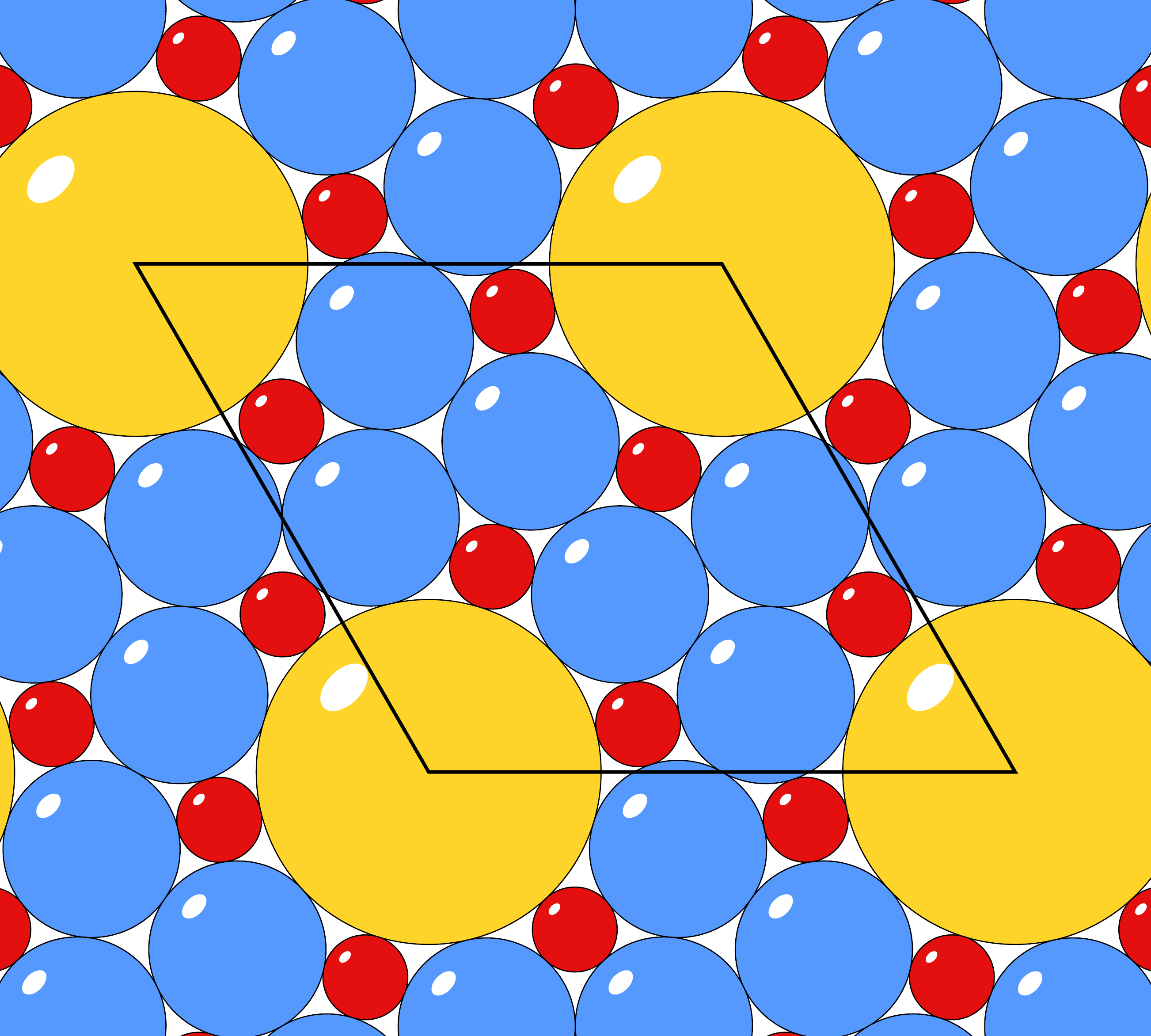} &
  \includegraphics[width=0.3\textwidth]{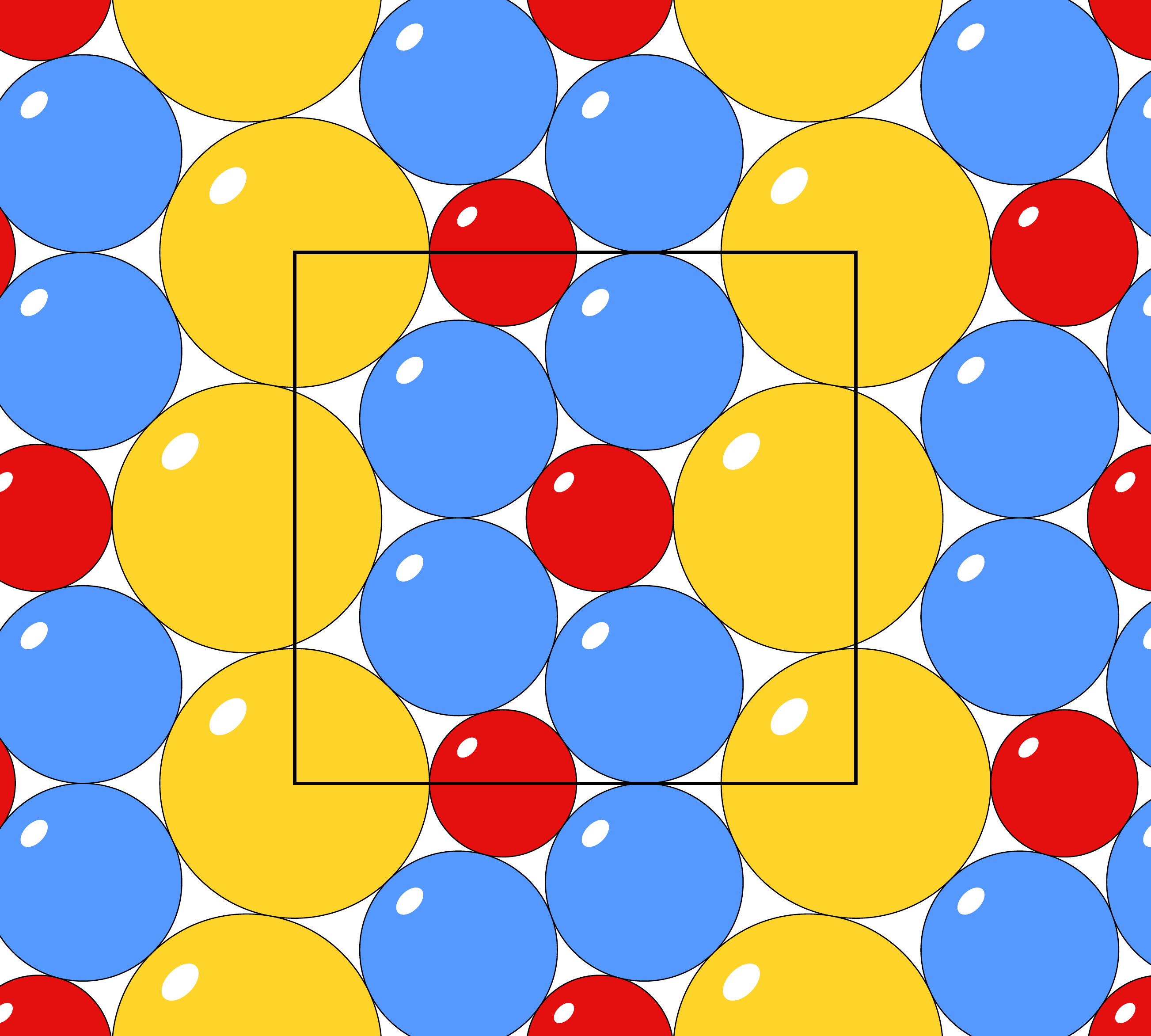}
\end{tabular}
\noindent
\begin{tabular}{lll}
  124 (H)\hfill 1rrrr / 1rrsrs & 125 (L)\hfill 1rrs / 11srsrss & 126 (H)\hfill 1rrs / 1srsrrss\\
  \includegraphics[width=0.3\textwidth]{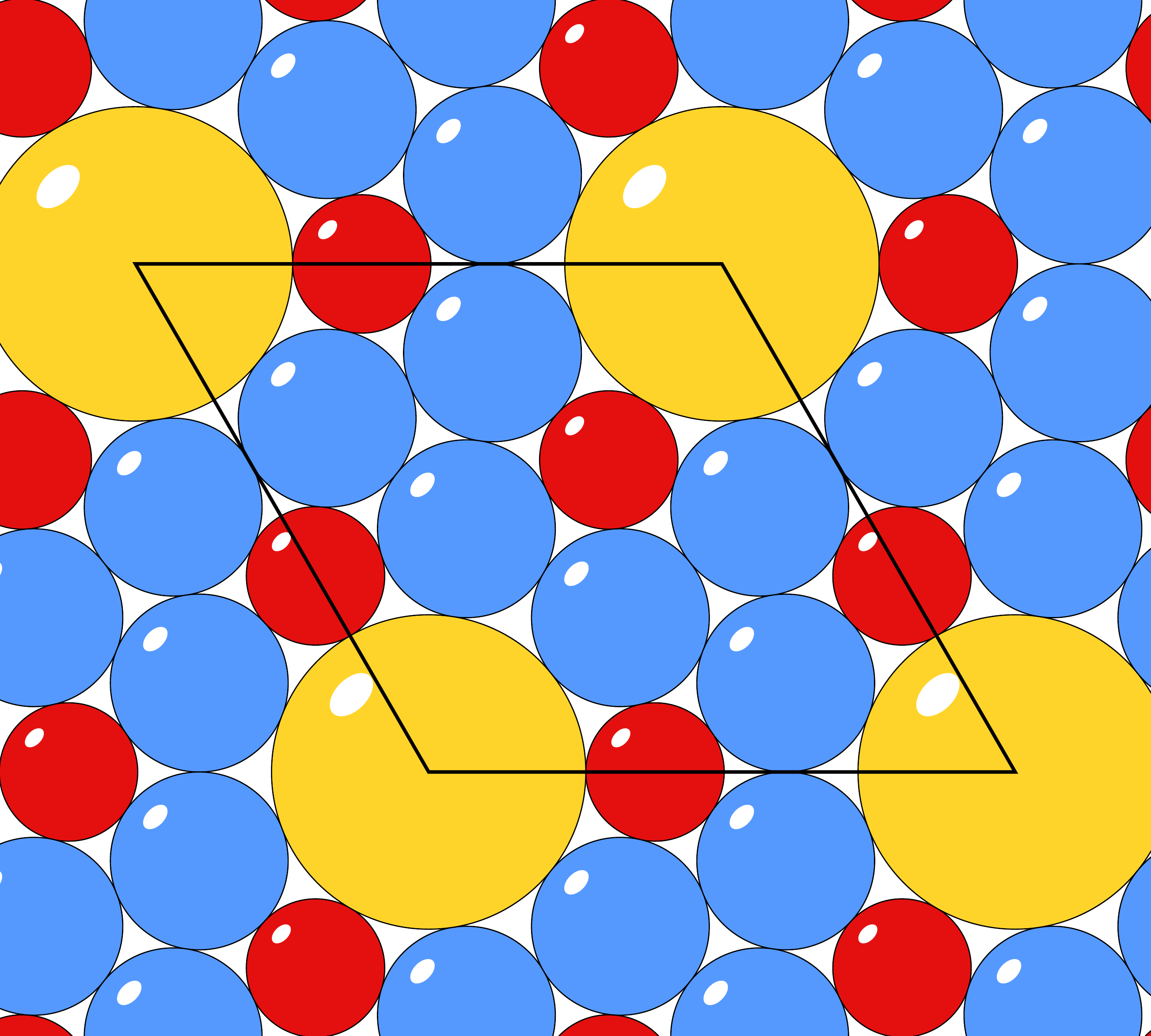} &
  \includegraphics[width=0.3\textwidth]{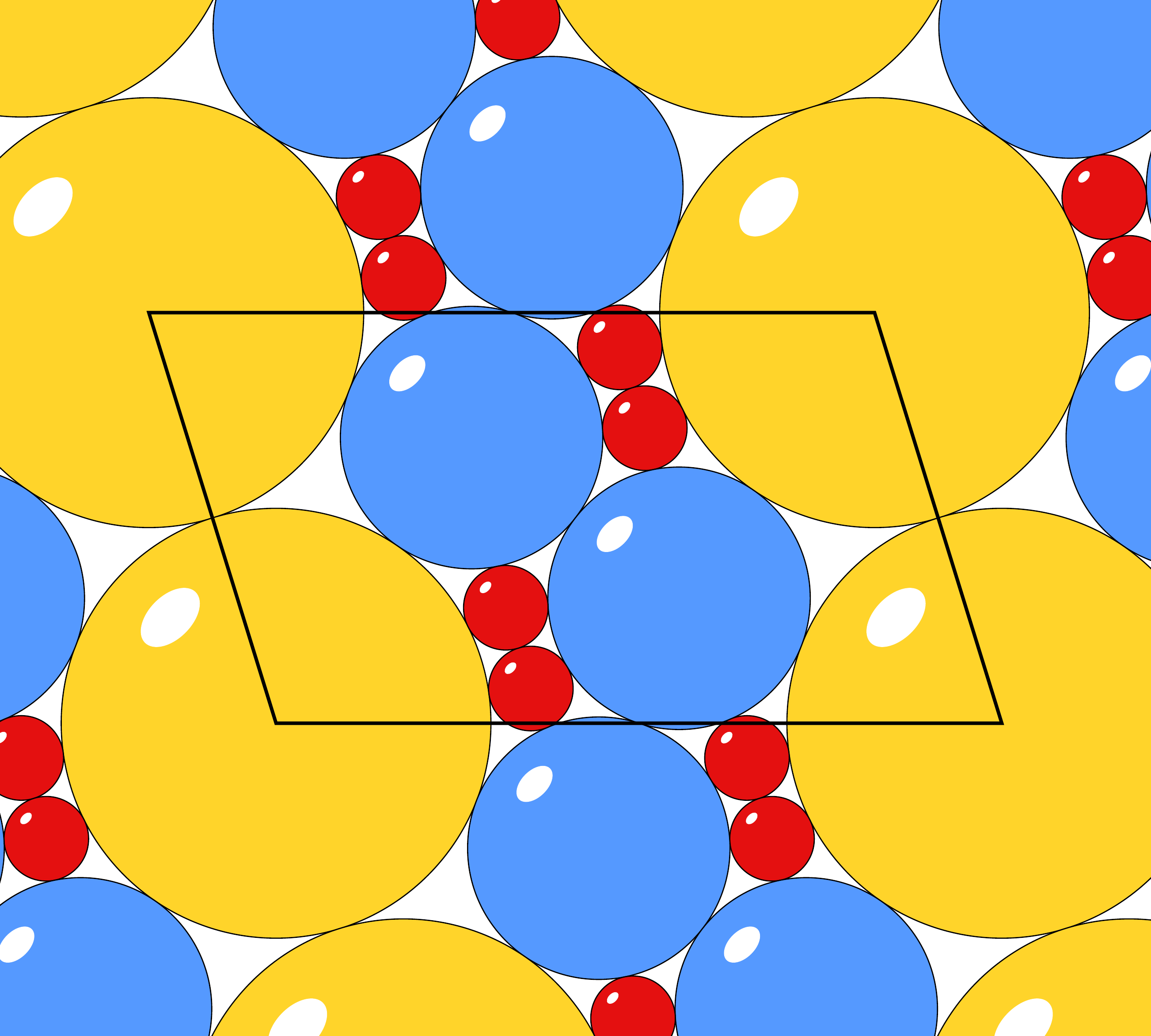} &
  \includegraphics[width=0.3\textwidth]{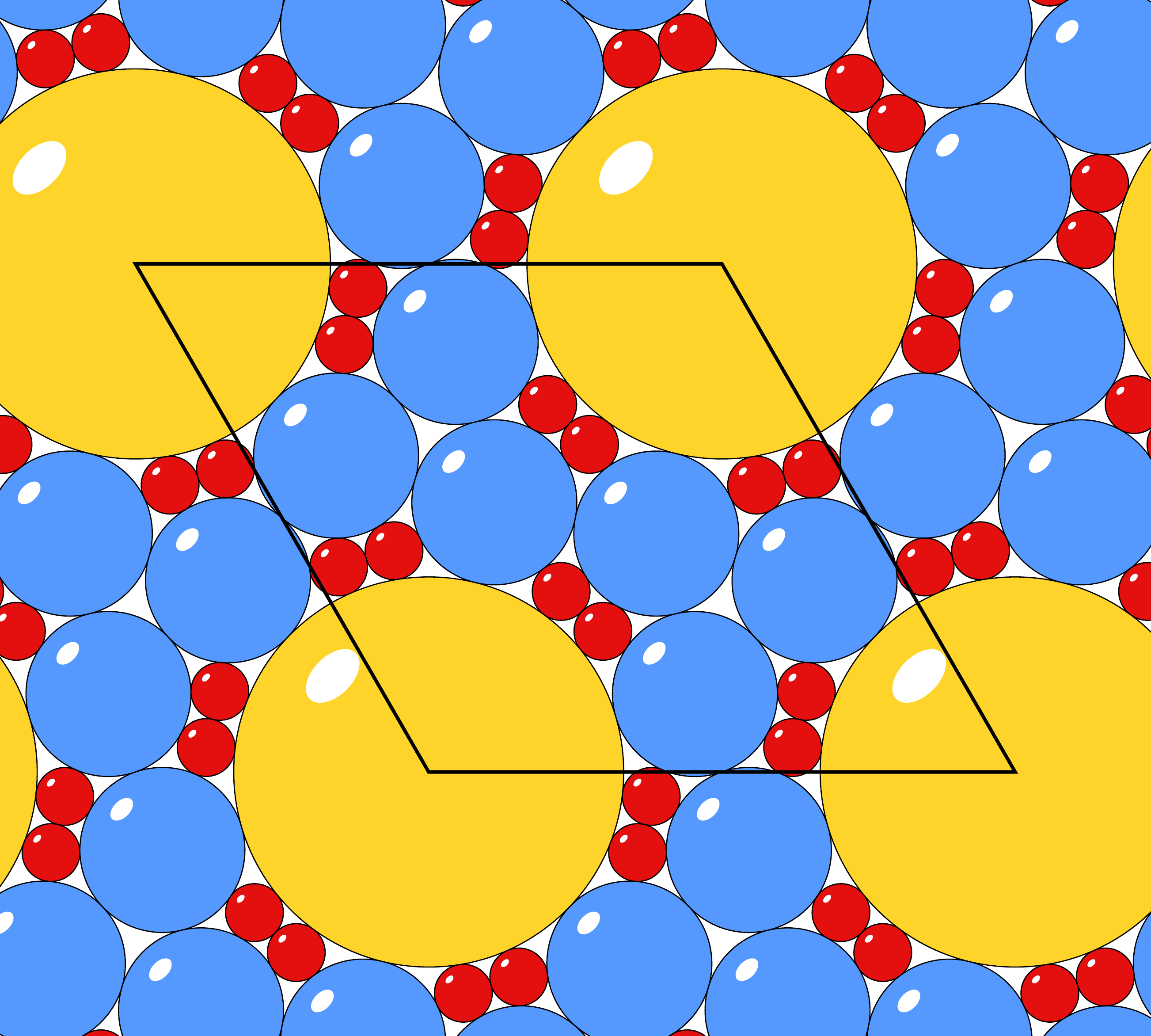}
\end{tabular}
\noindent
\begin{tabular}{lll}
  127 (L)\hfill 1rrsr / 11srss & 128 (H)\hfill 1rrsr / 1srrss & 129 (H)\hfill 1rs1s / 111ss\\
  \includegraphics[width=0.3\textwidth]{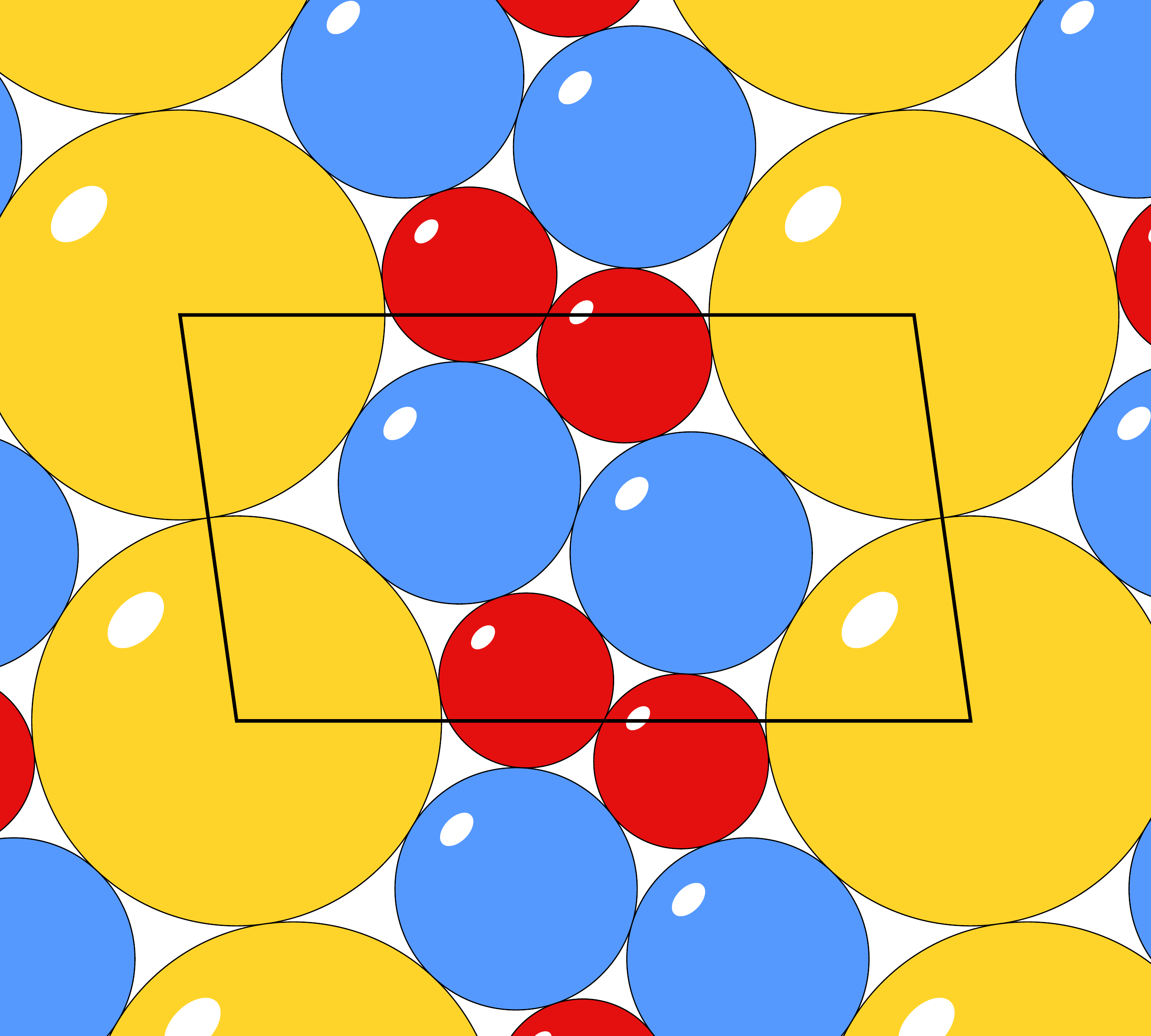} &
  \includegraphics[width=0.3\textwidth]{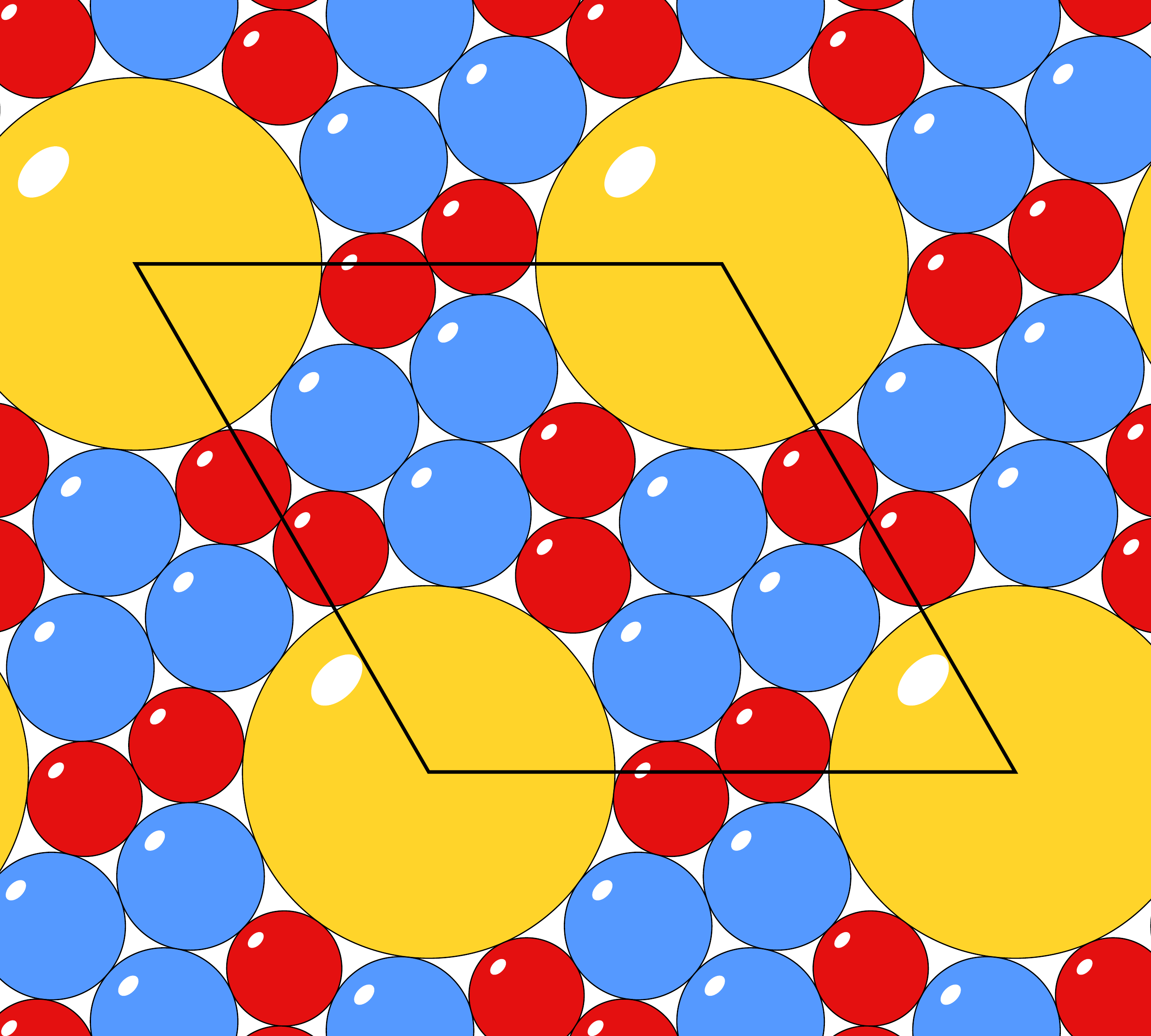} &
  \includegraphics[width=0.3\textwidth]{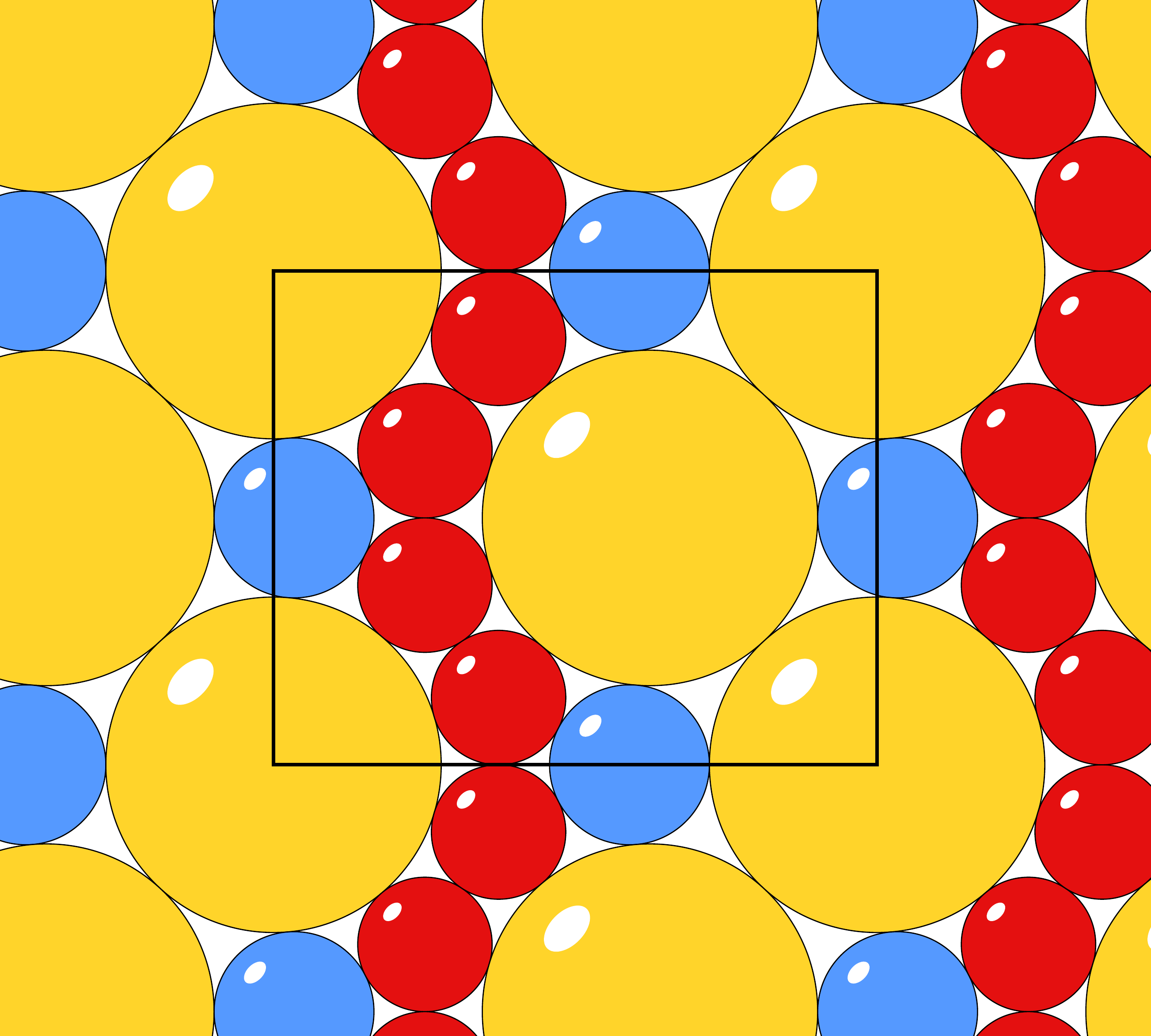}
\end{tabular}
\noindent
\begin{tabular}{lll}
  130 (L)\hfill 1rs1s / 11r1ss & 131 (L)\hfill 1rs1s / 1s1sss & 132 (L)\hfill 1rsr / 1111ss\\
  \includegraphics[width=0.3\textwidth]{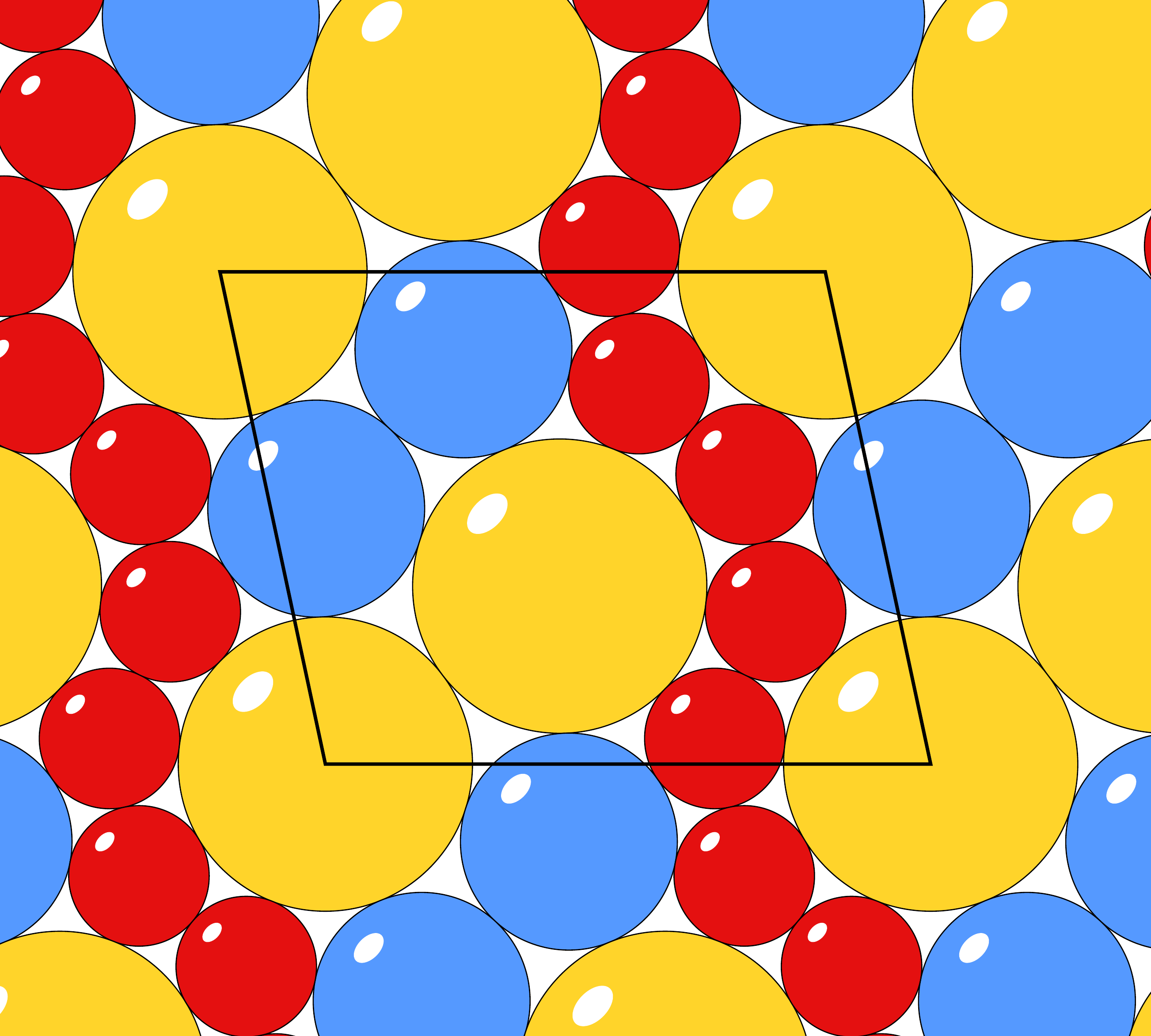} &
  \includegraphics[width=0.3\textwidth]{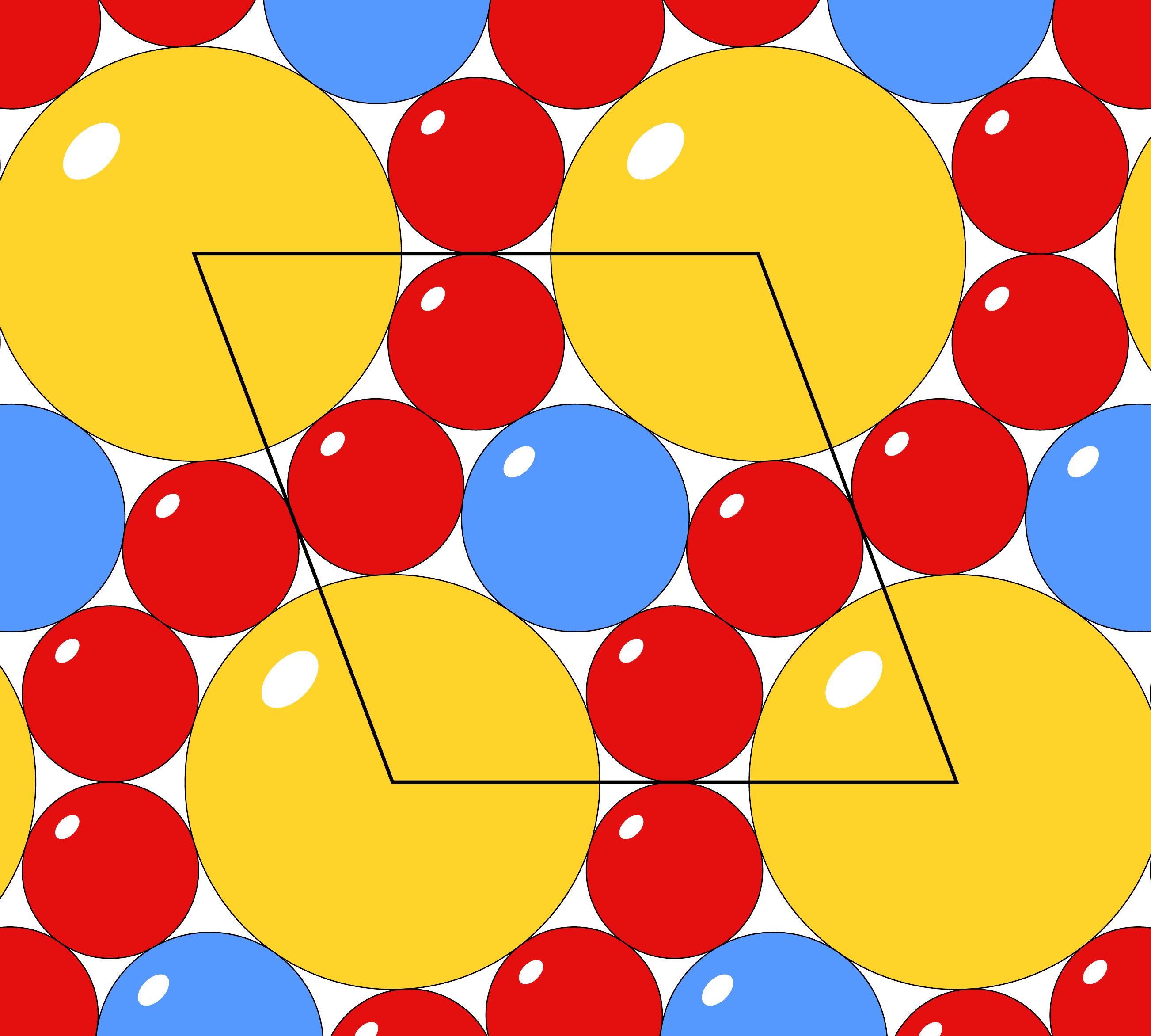} &
  \includegraphics[width=0.3\textwidth]{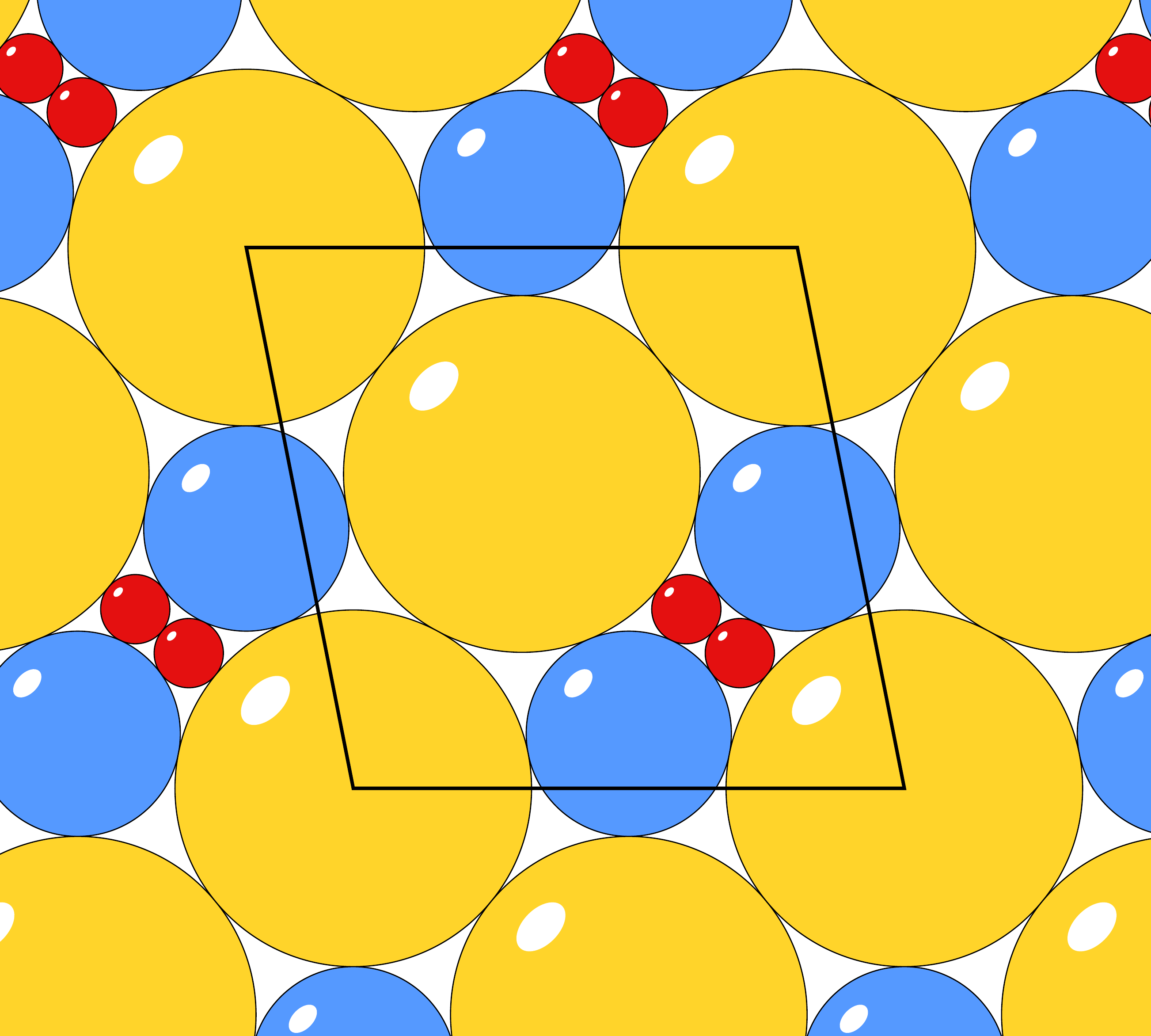}
\end{tabular}
\noindent
\begin{tabular}{lll}
  133 (L)\hfill 1rsr / 111r1ss & 134 (L)\hfill 1rsr / 111s1sss & 135 (E)\hfill 1rsr / 111ss\\
  \includegraphics[width=0.3\textwidth]{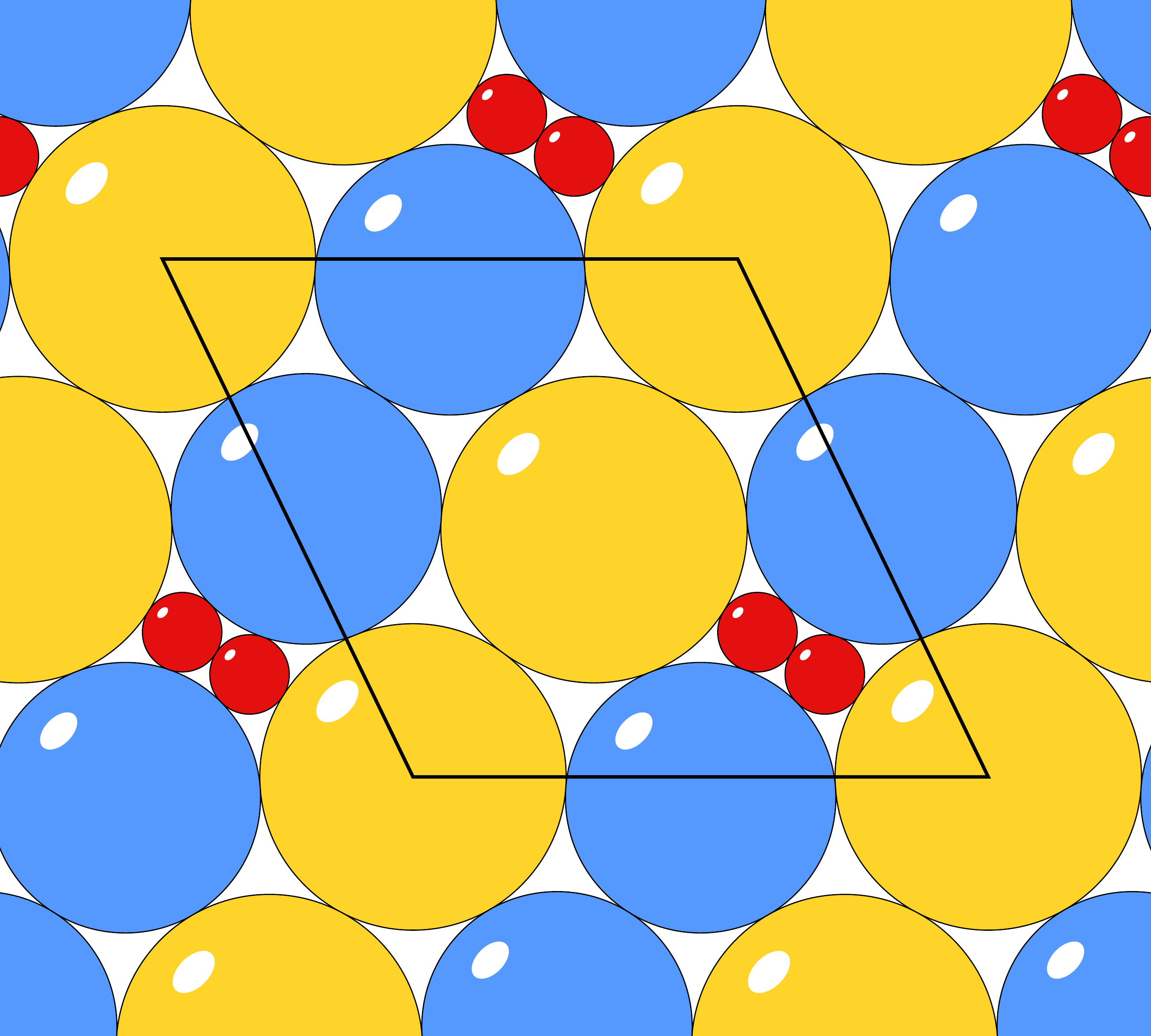} &
  \includegraphics[width=0.3\textwidth]{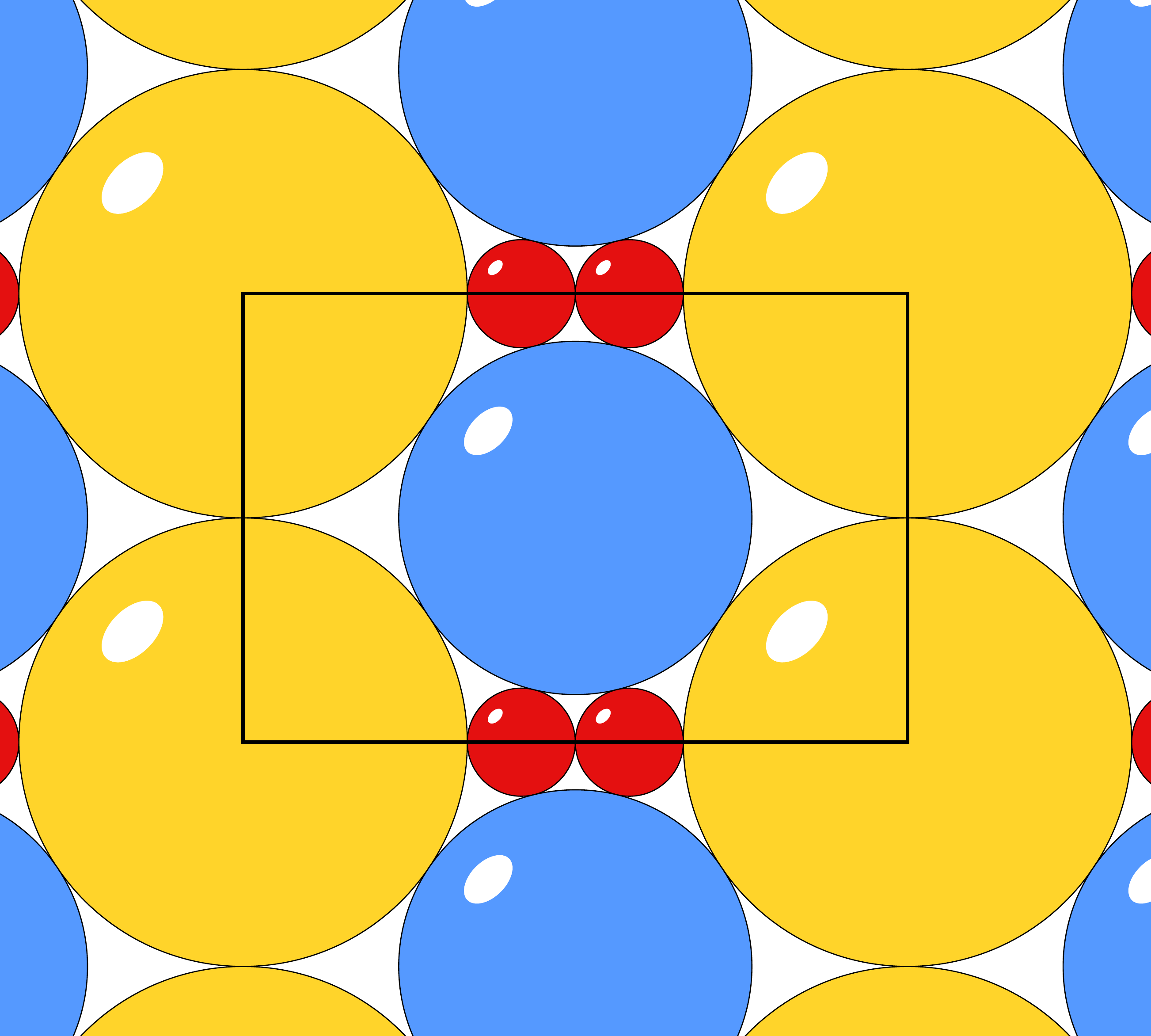} &
  \includegraphics[width=0.3\textwidth]{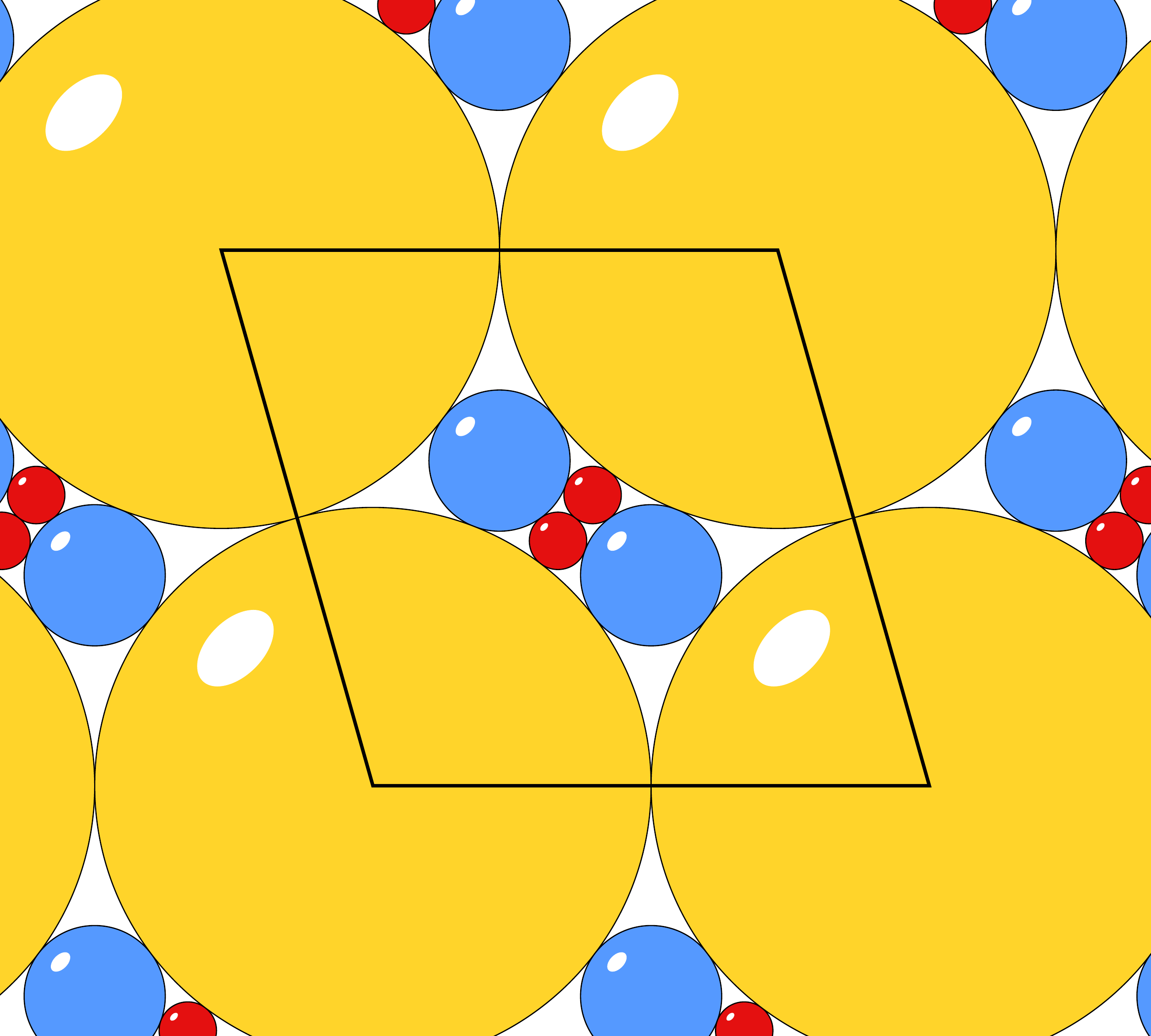}
\end{tabular}
\noindent
\begin{tabular}{lll}
  136 (L)\hfill 1rsr / 11r1ss & 137 (H)\hfill 1rsr / 11rr1ss & 138 (L)\hfill 1rsr / 11s1sss\\
  \includegraphics[width=0.3\textwidth]{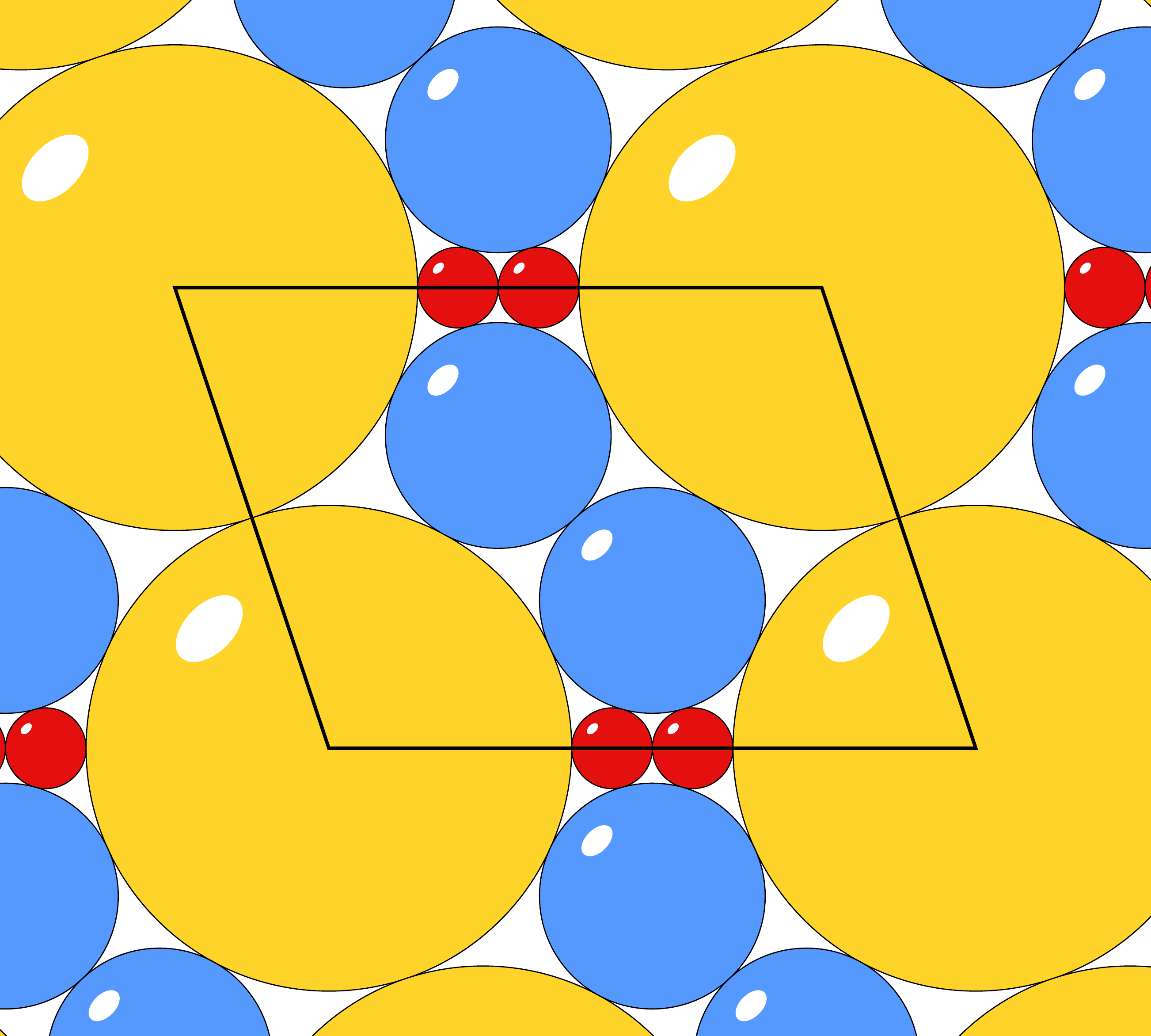} &
  \includegraphics[width=0.3\textwidth]{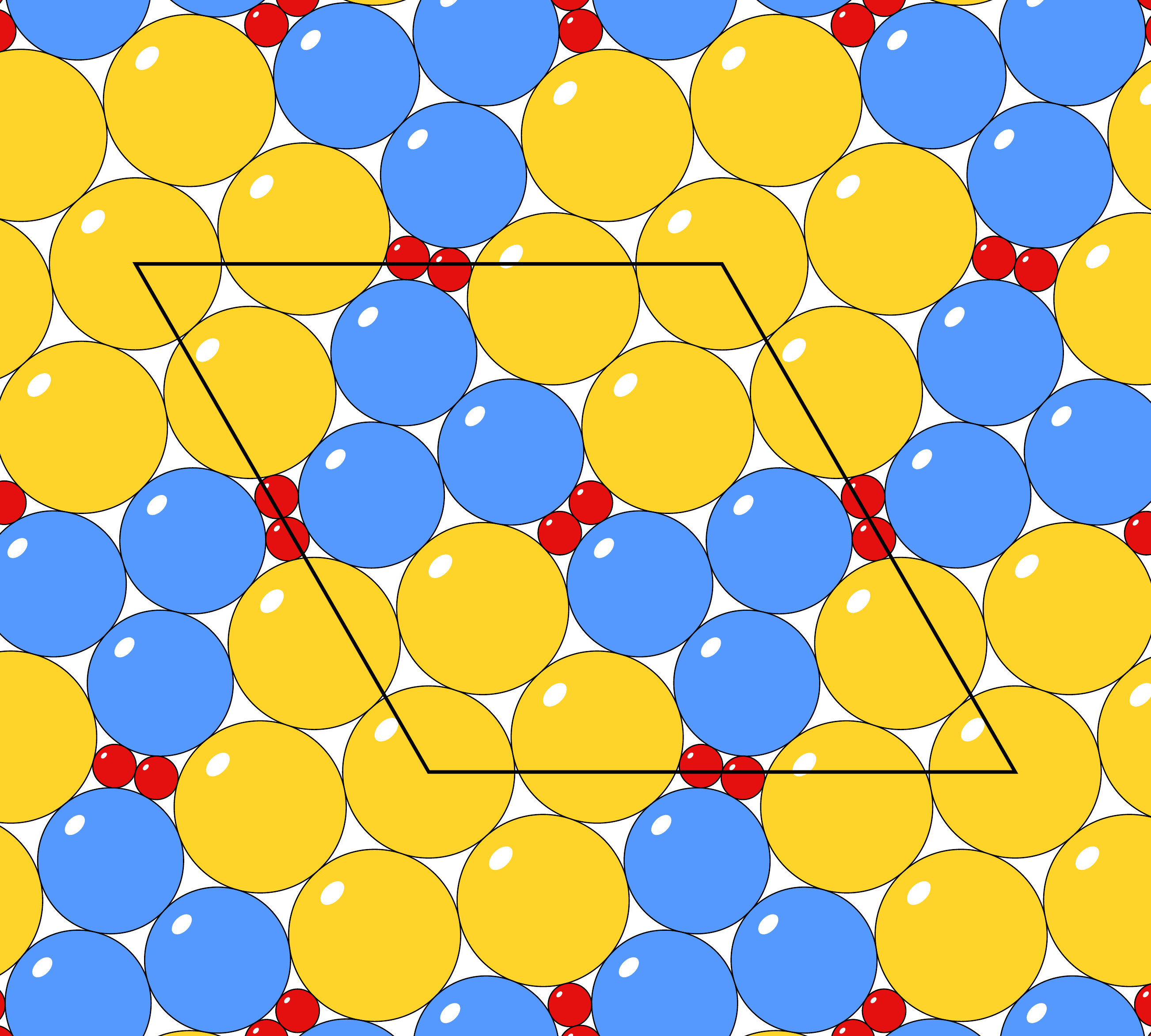} &
  \includegraphics[width=0.3\textwidth]{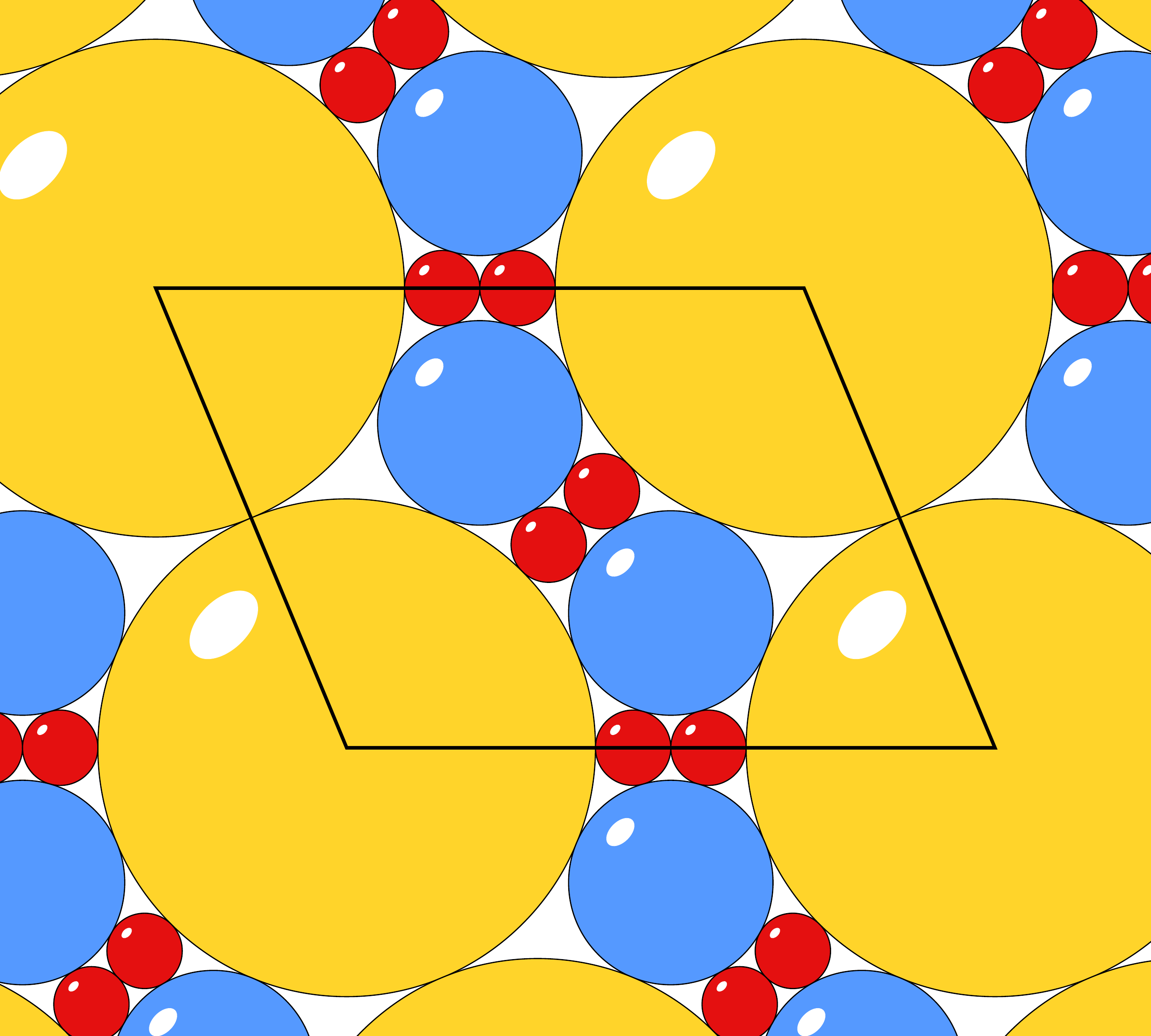}
\end{tabular}
\noindent
\begin{tabular}{lll}
  139 (L)\hfill 1rsr / 1r1r1ss & 140 (L)\hfill 1rsr / 1r1s1sss & 141 (H)\hfill 1rsr / 1rr1ss\\
  \includegraphics[width=0.3\textwidth]{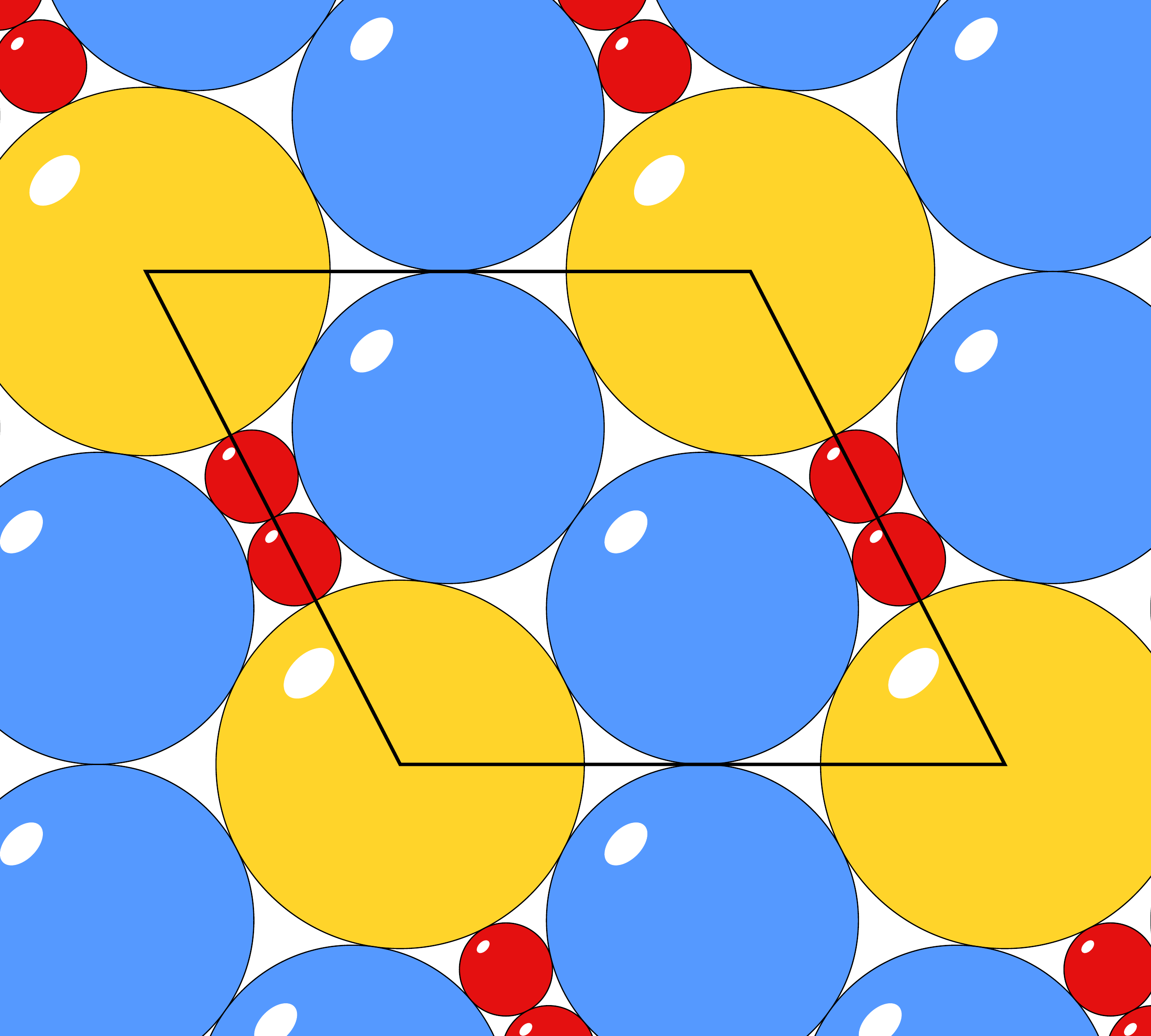} &
  \includegraphics[width=0.3\textwidth]{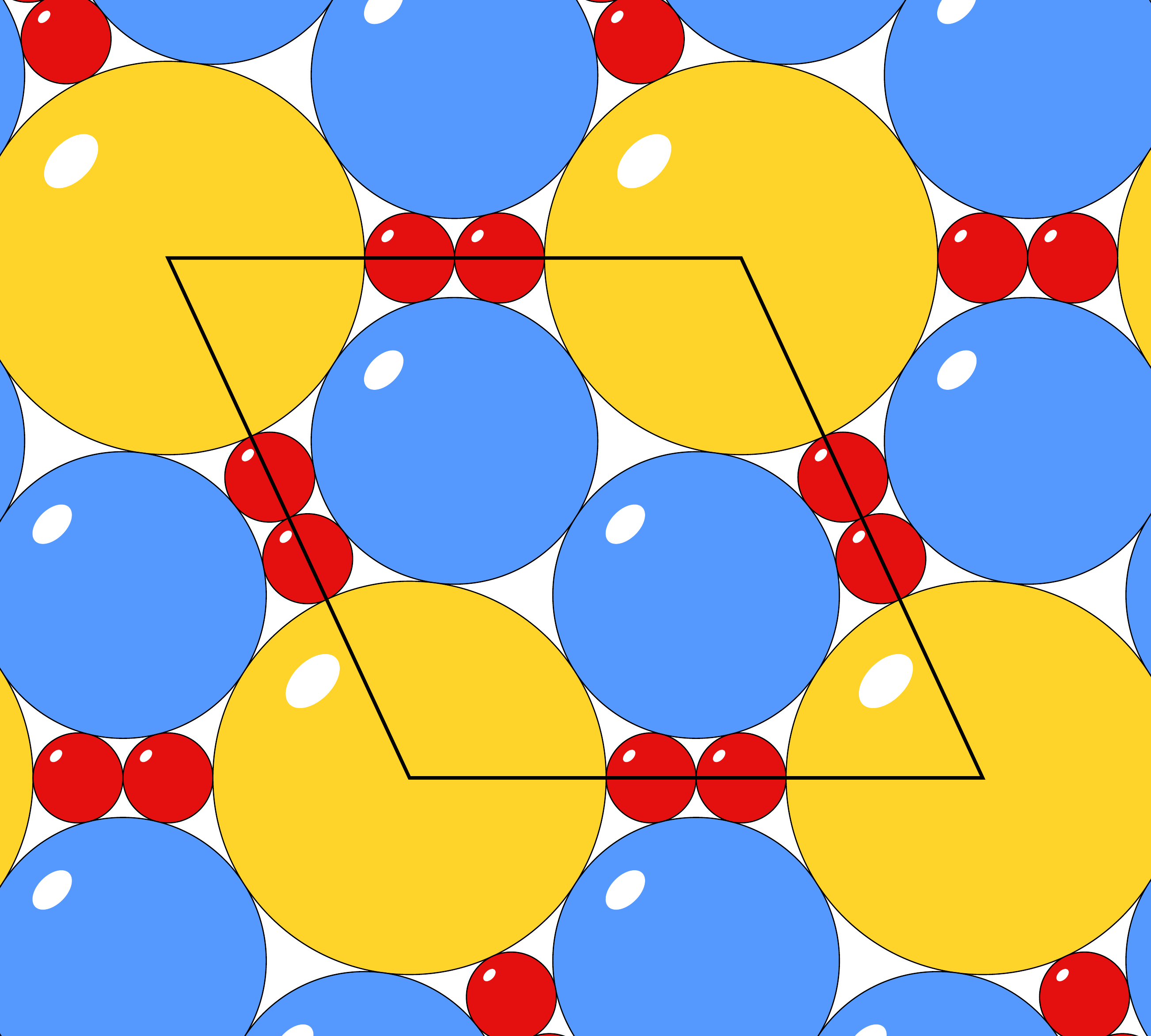} &
  \includegraphics[width=0.3\textwidth]{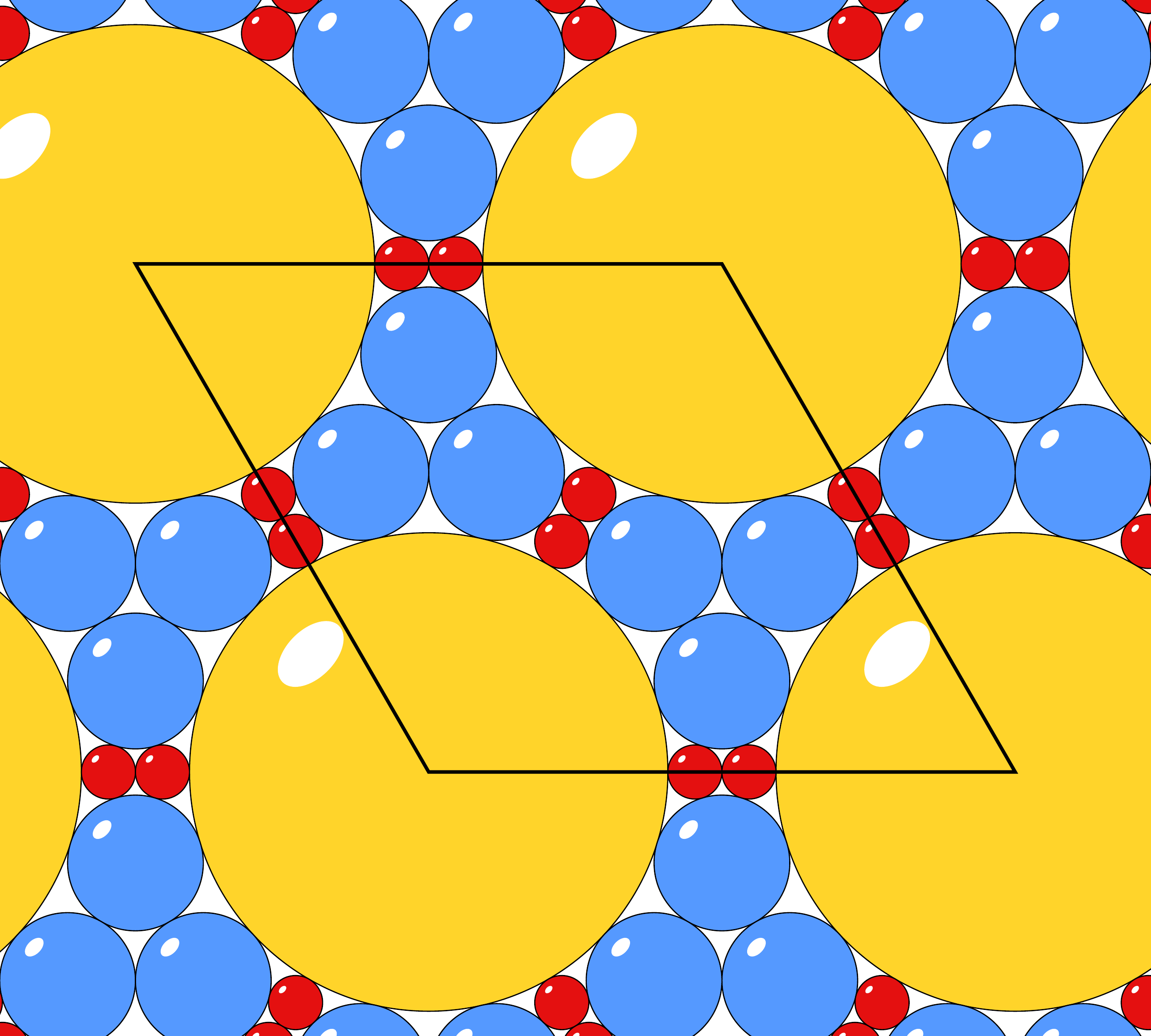}
\end{tabular}
\noindent
\begin{tabular}{lll}
  142 (H)\hfill 1rsr / 1rrr1ss & 143 (H)\hfill 1rsr / 1s1s1ssss & 144 (H)\hfill 1rsrs / 111ssss\\
  \includegraphics[width=0.3\textwidth]{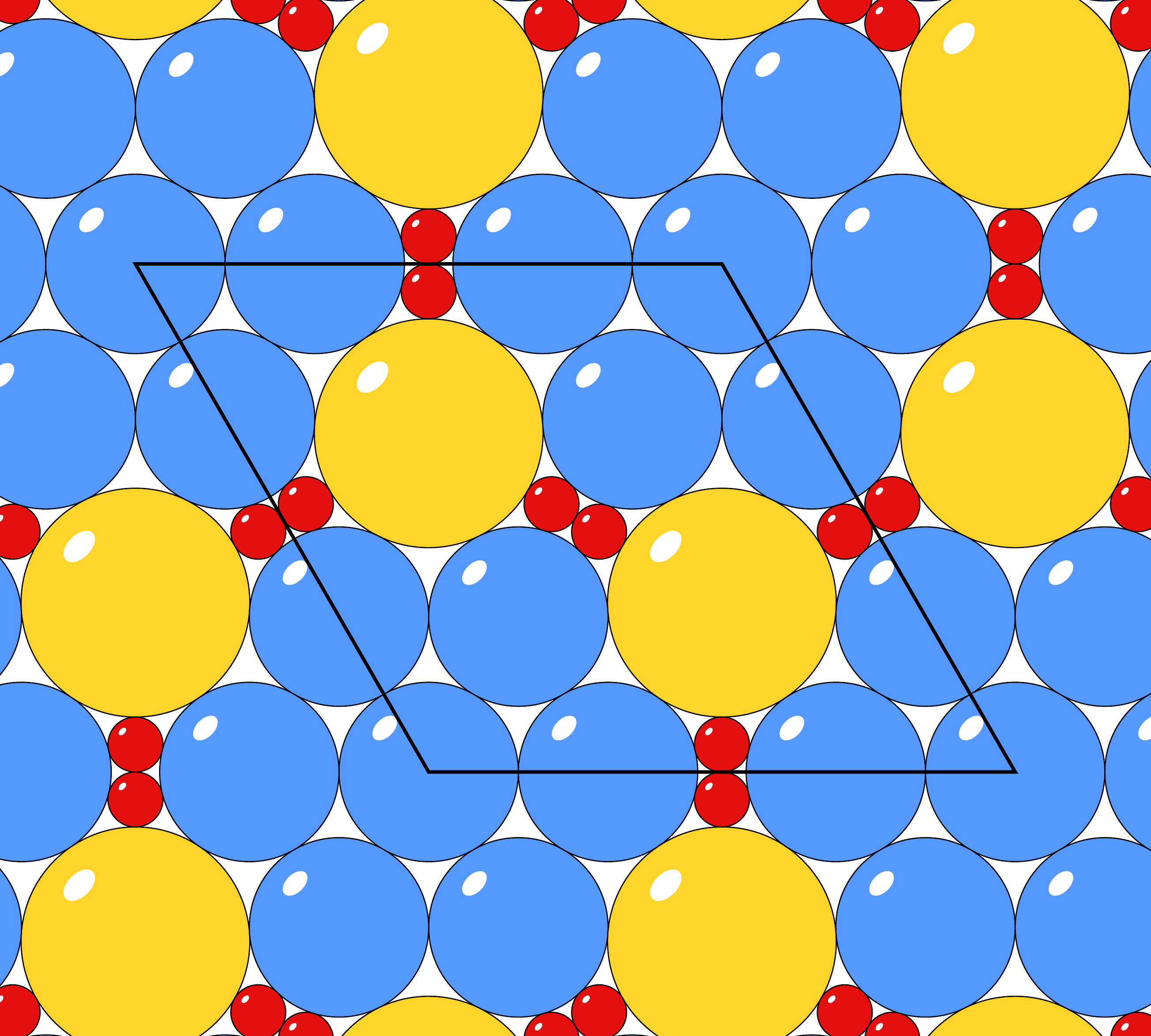} &
  \includegraphics[width=0.3\textwidth]{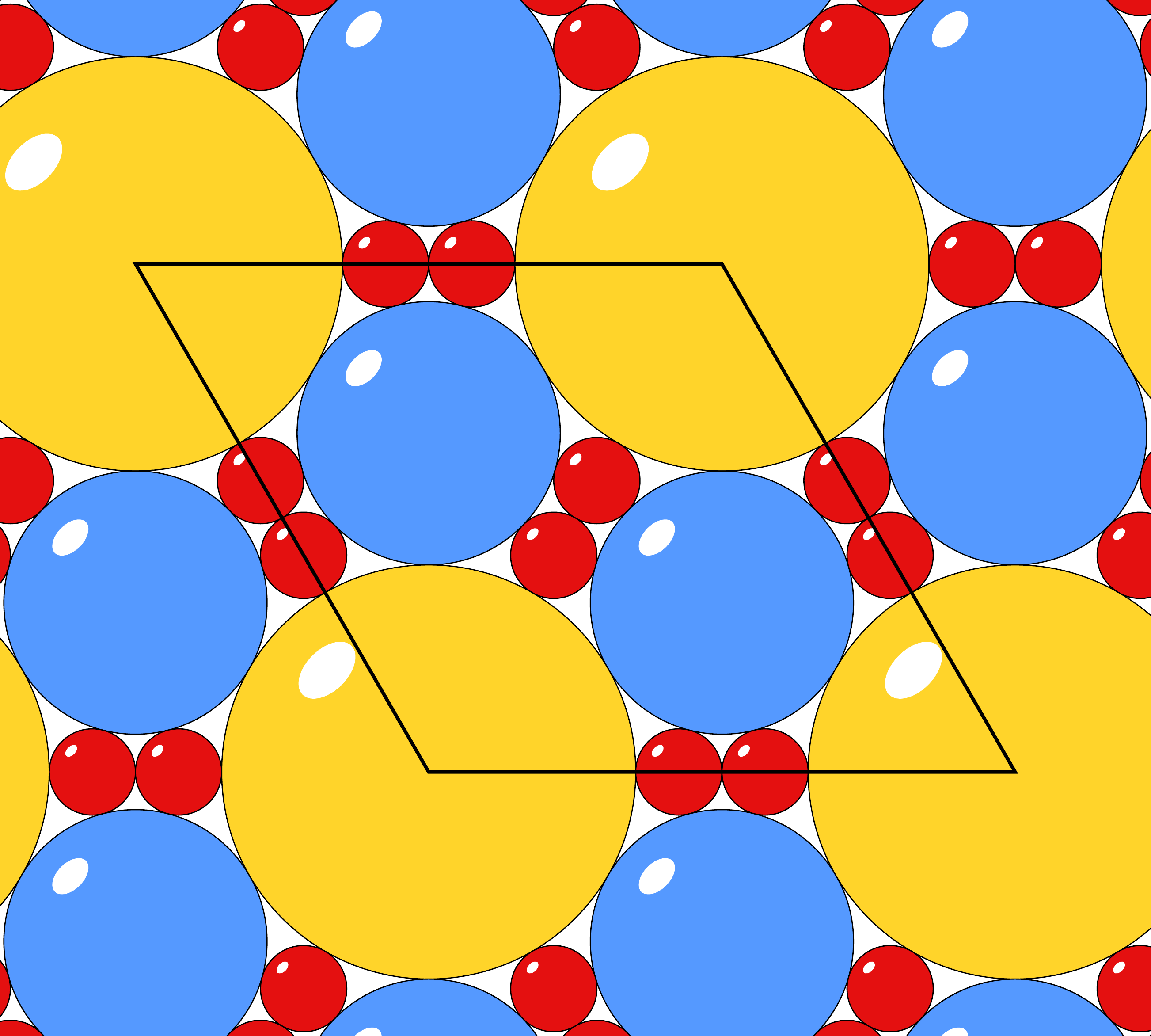} &
  \includegraphics[width=0.3\textwidth]{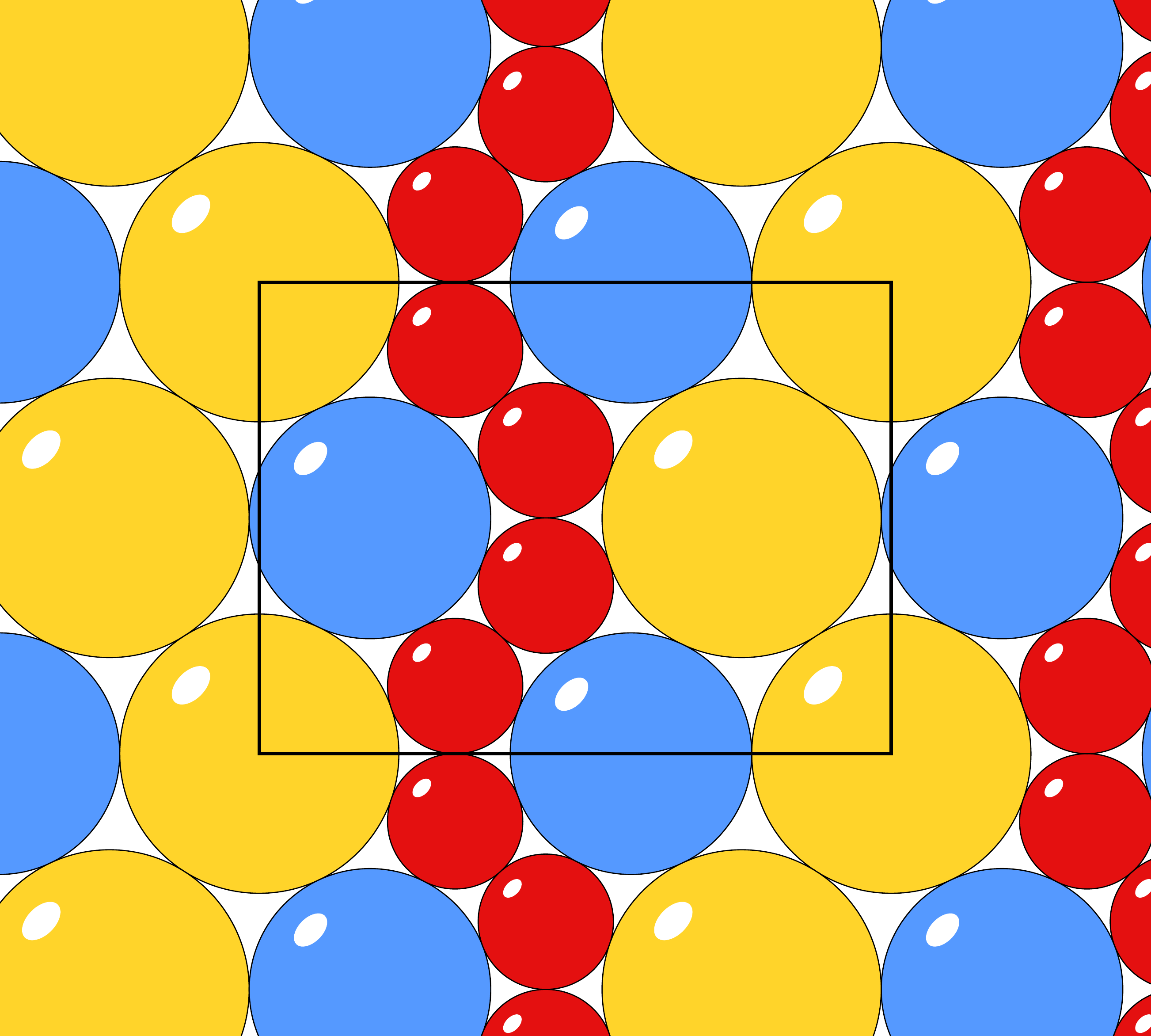}
\end{tabular}
\noindent
\begin{tabular}{lll}
  145 (L)\hfill 1rsrs / 11ssss & 146 (H)\hfill 1rsrs / 1r1ssss & 147 (H)\hfill 1rssr / 111ss\\
  \includegraphics[width=0.3\textwidth]{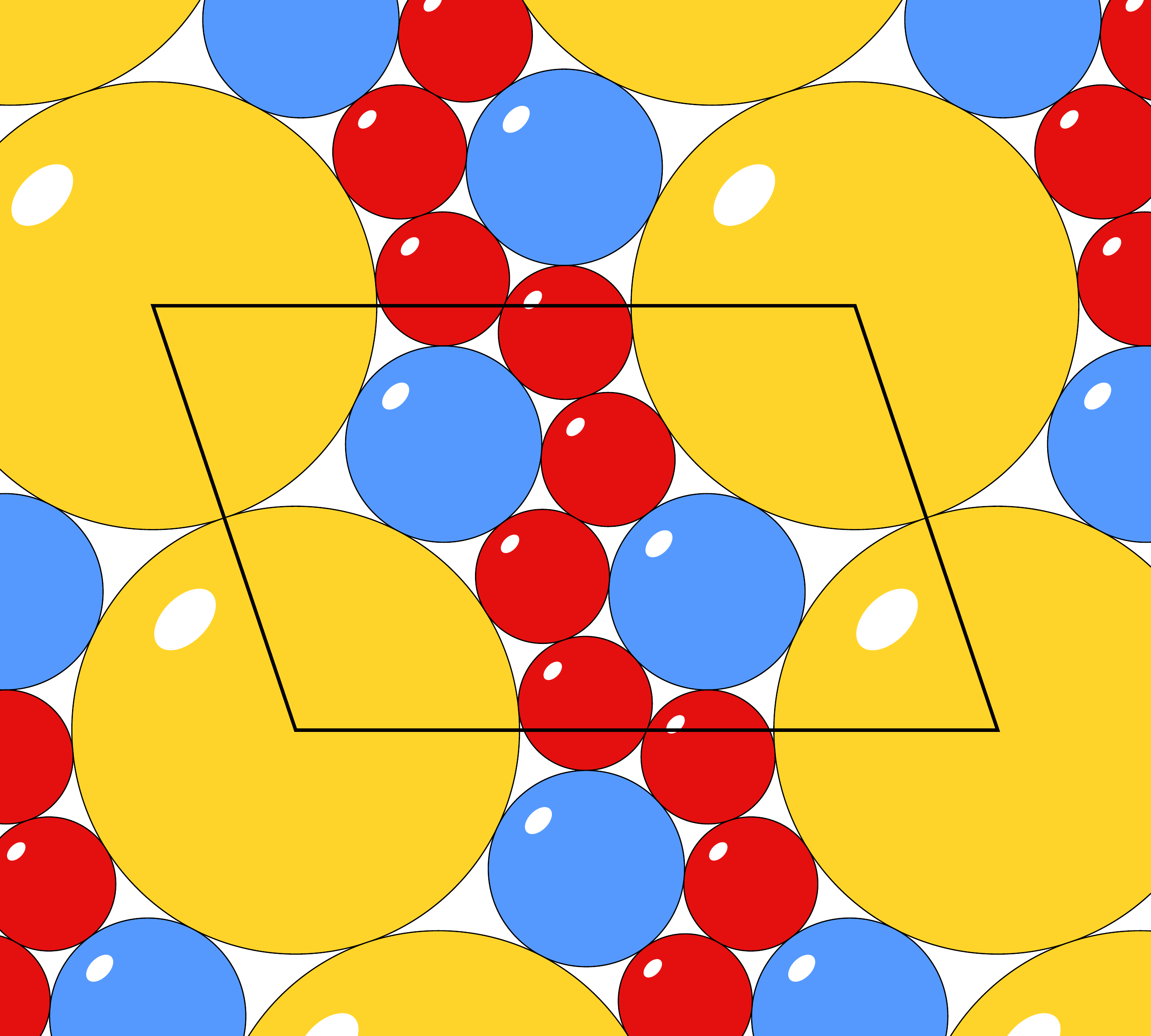} &
  \includegraphics[width=0.3\textwidth]{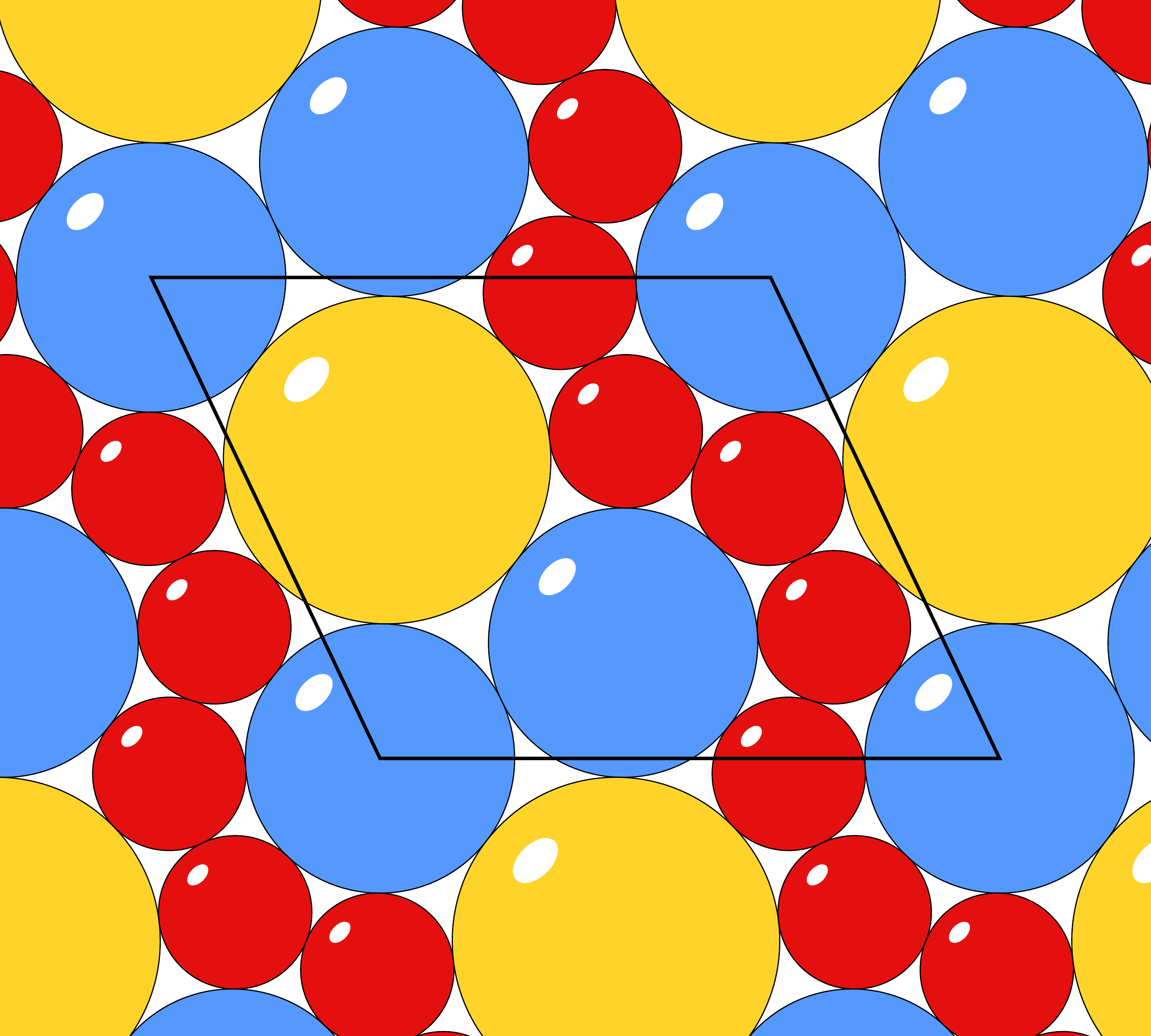} &
  \includegraphics[width=0.3\textwidth]{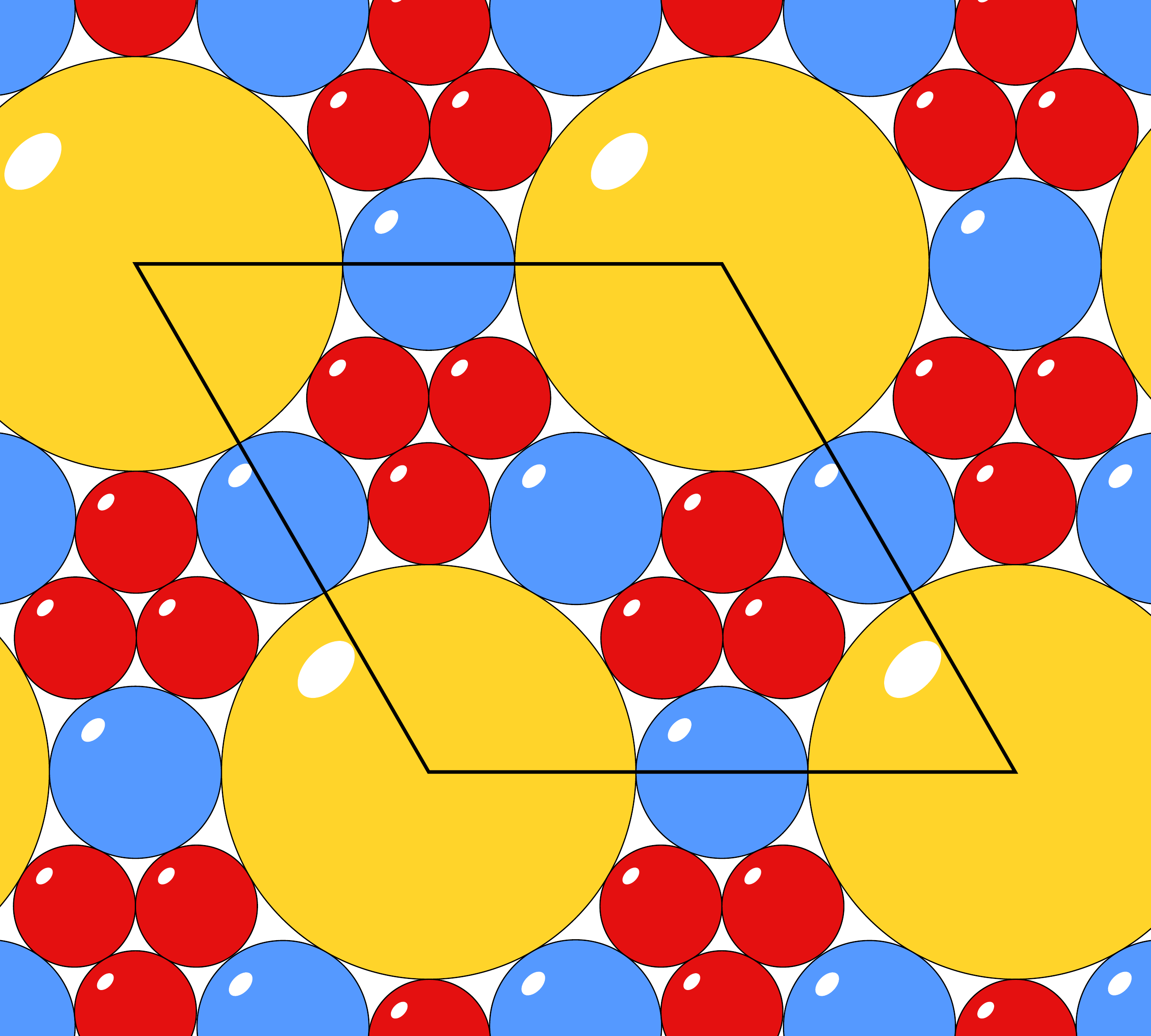}
\end{tabular}
\noindent
\begin{tabular}{lll}
  148 (H)\hfill 1rssr / 11r1ss & 149 (H)\hfill 1rssr / 11s1sss & 150 (H)\hfill 1rssr / 1r1ss\\
  \includegraphics[width=0.3\textwidth]{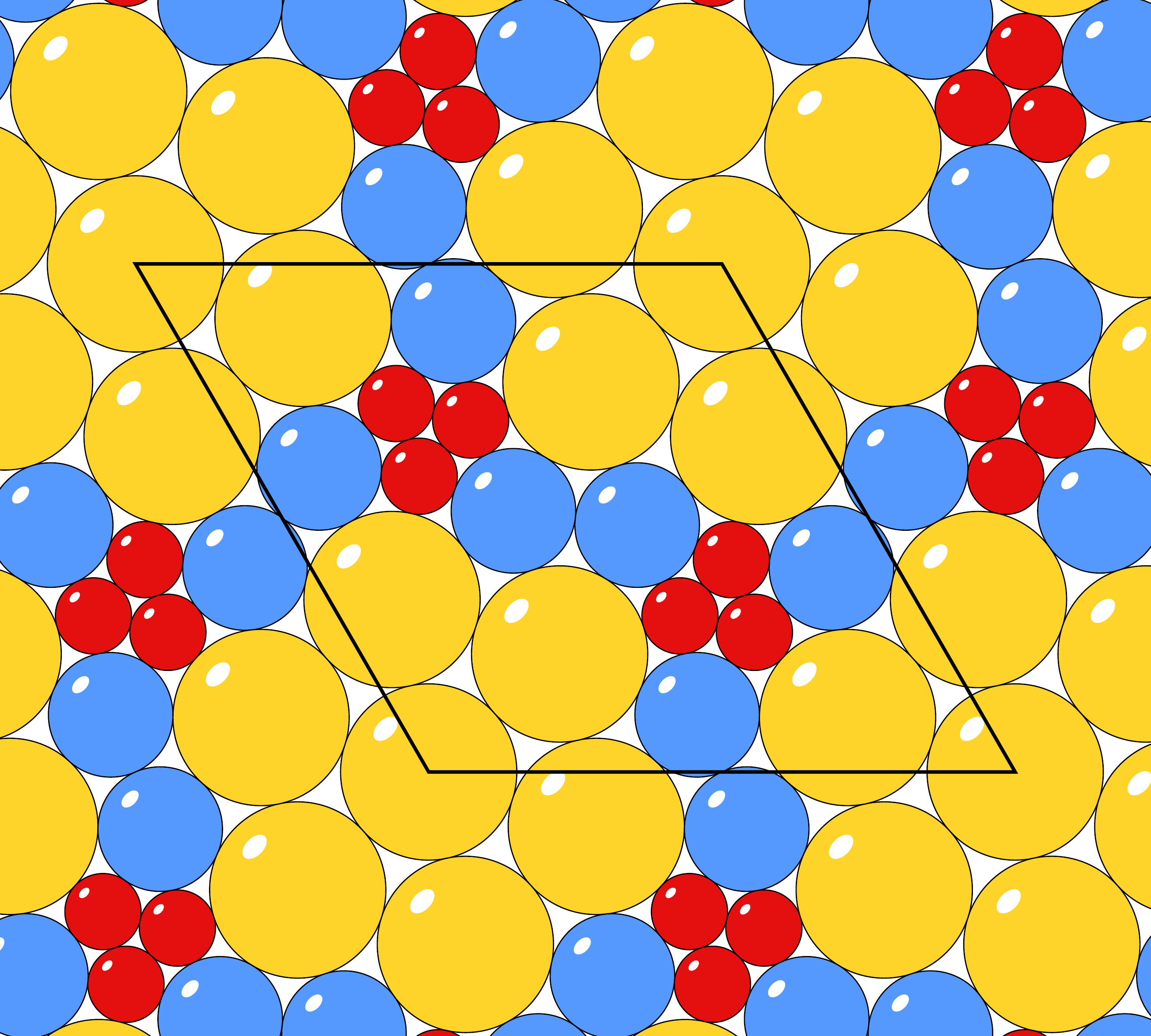} &
  \includegraphics[width=0.3\textwidth]{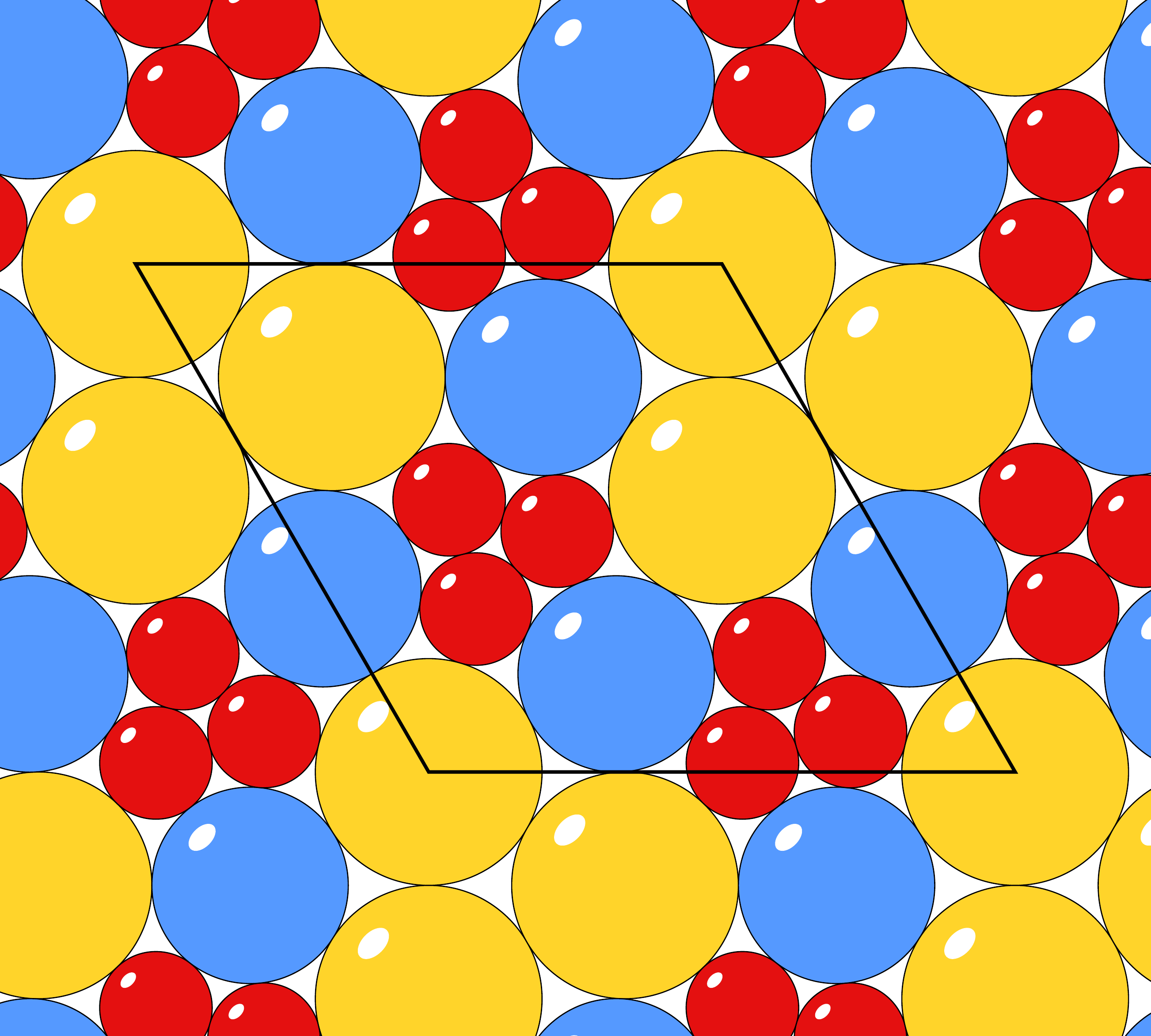} &
  \includegraphics[width=0.3\textwidth]{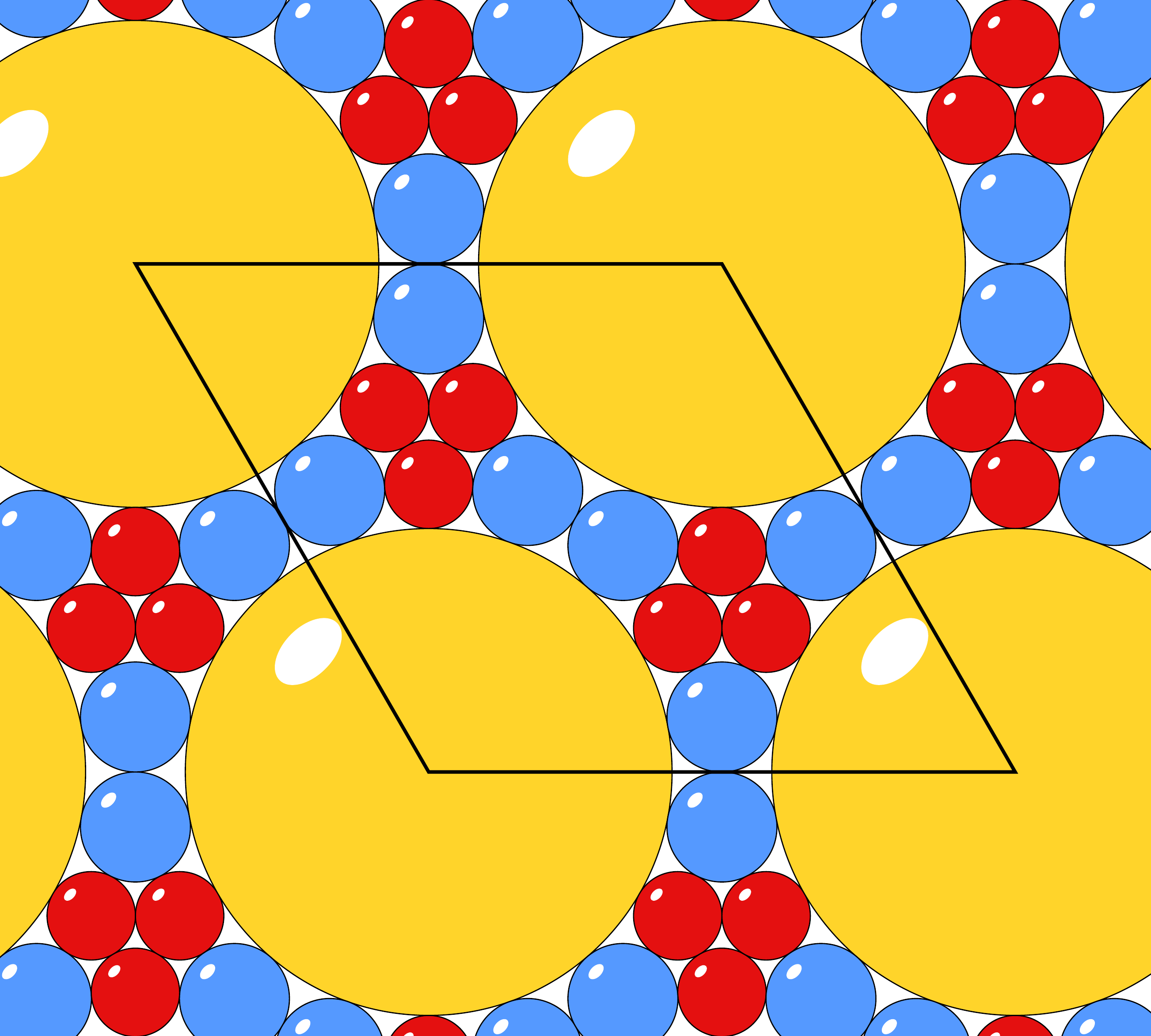}
\end{tabular}
\noindent
\begin{tabular}{lll}
  151 (H)\hfill 1rssr / 1rr1ss & 152 (S)\hfill 1rssr / 1s1sss & 153 (S)\hfill 1rsss / 111ss\\
  \includegraphics[width=0.3\textwidth]{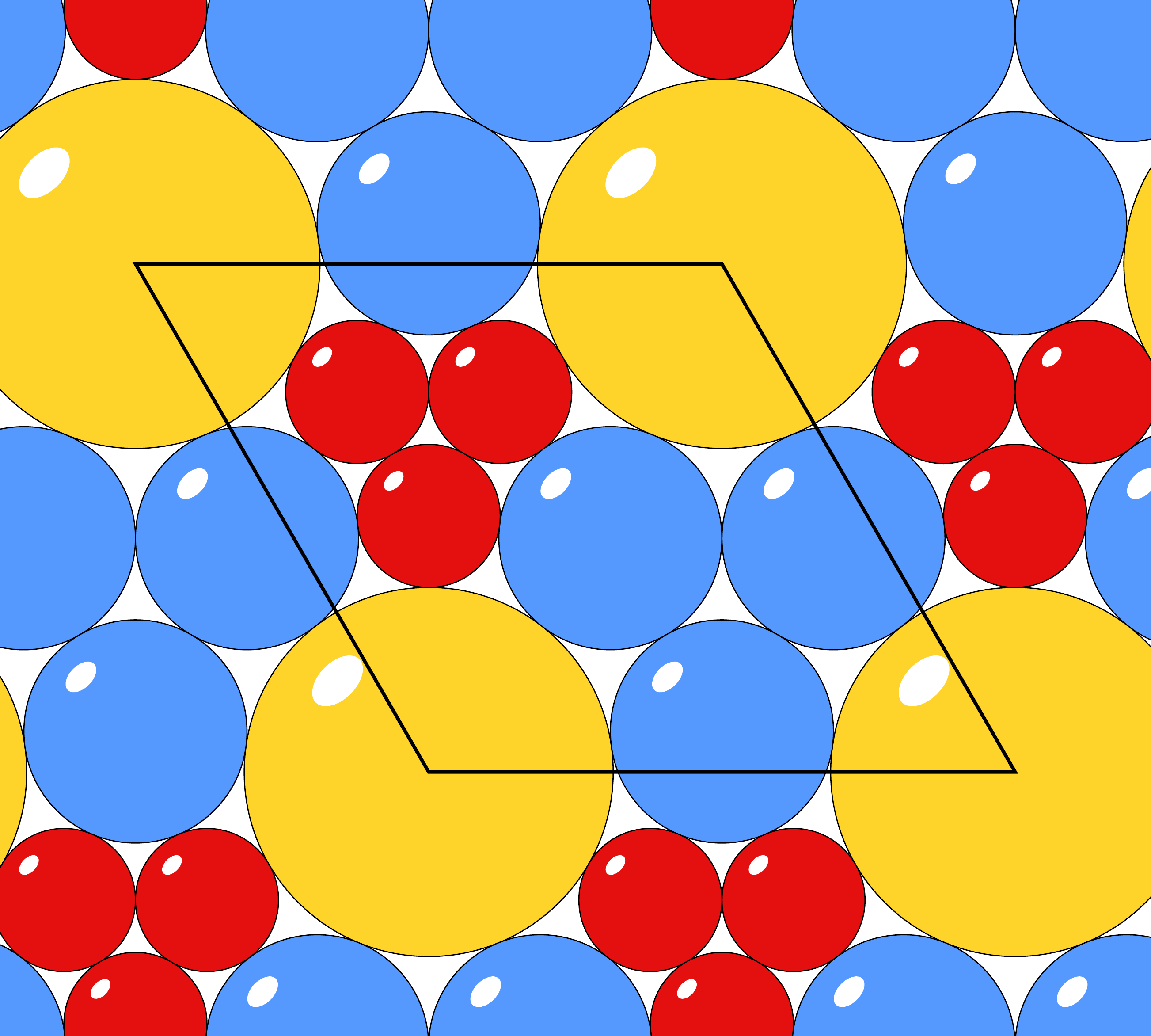} &
  \includegraphics[width=0.3\textwidth]{packing_1rssr_1s1sss.pdf} &
  \includegraphics[width=0.3\textwidth]{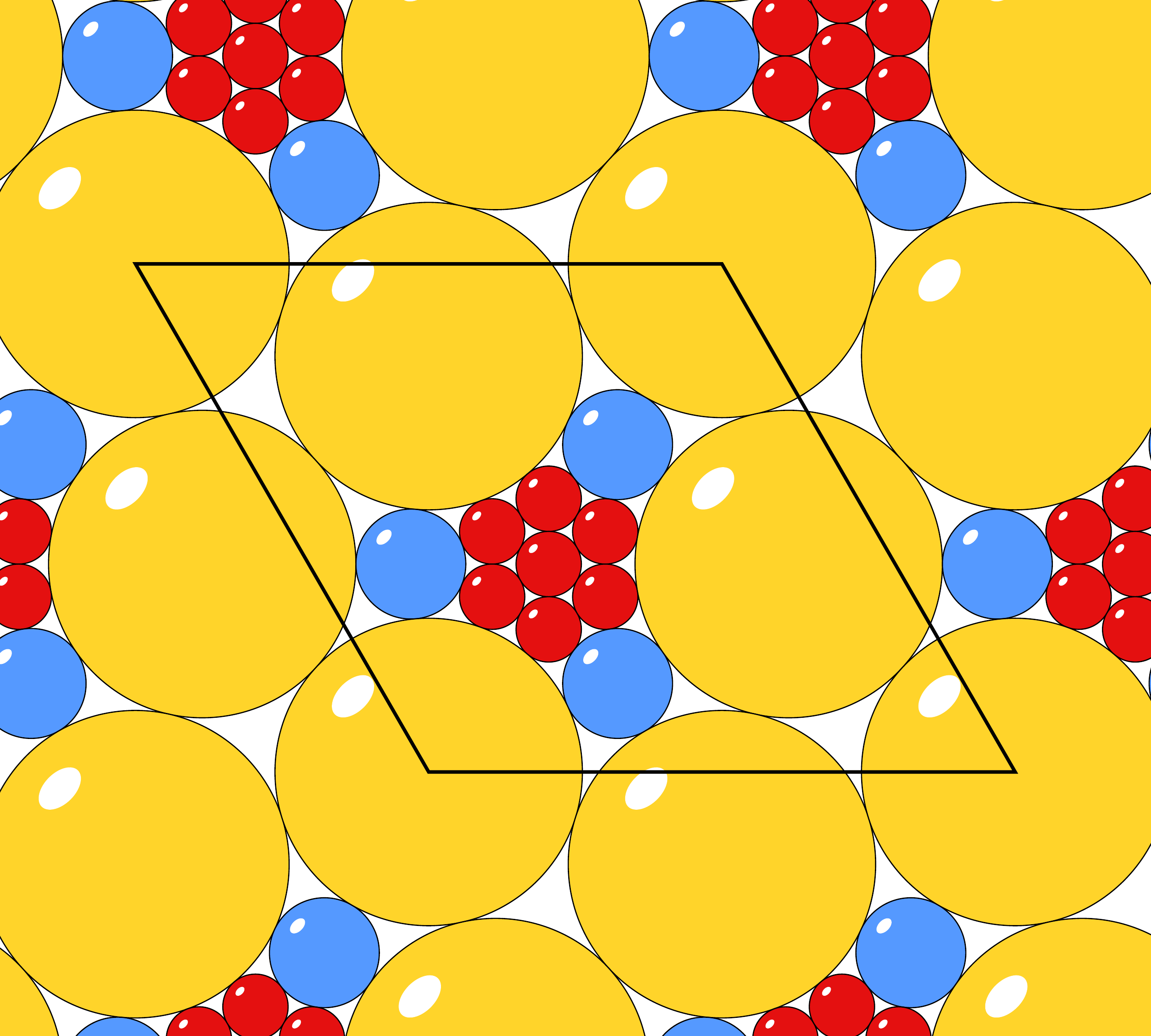}
\end{tabular}
\noindent
\begin{tabular}{lll}
  154 (H)\hfill 1rsss / 11r1ss & 155 (H)\hfill 1rsss / 11s1sss & 156 (H)\hfill 1rsss / 1r1ss\\
  \includegraphics[width=0.3\textwidth]{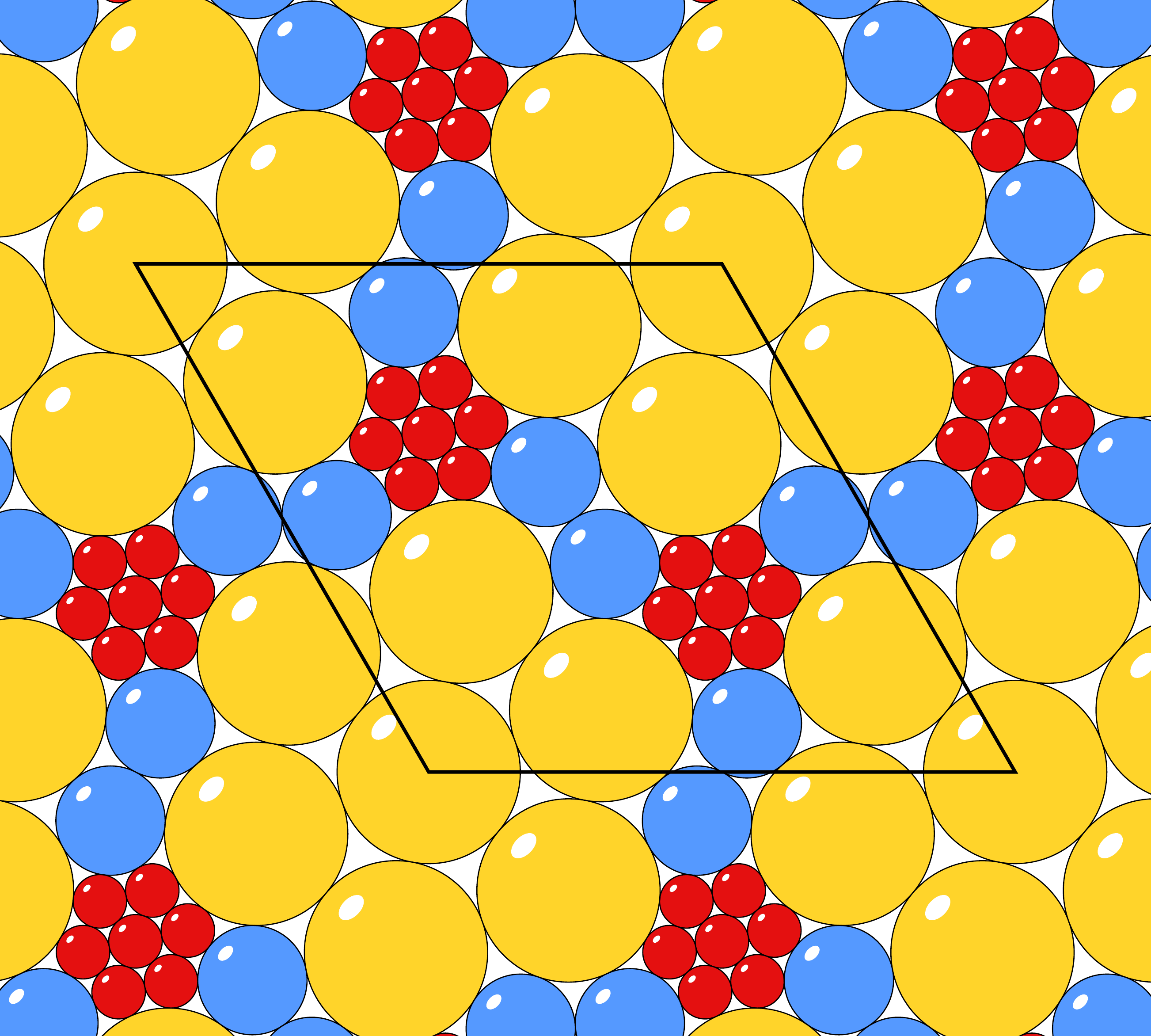} &
  \includegraphics[width=0.3\textwidth]{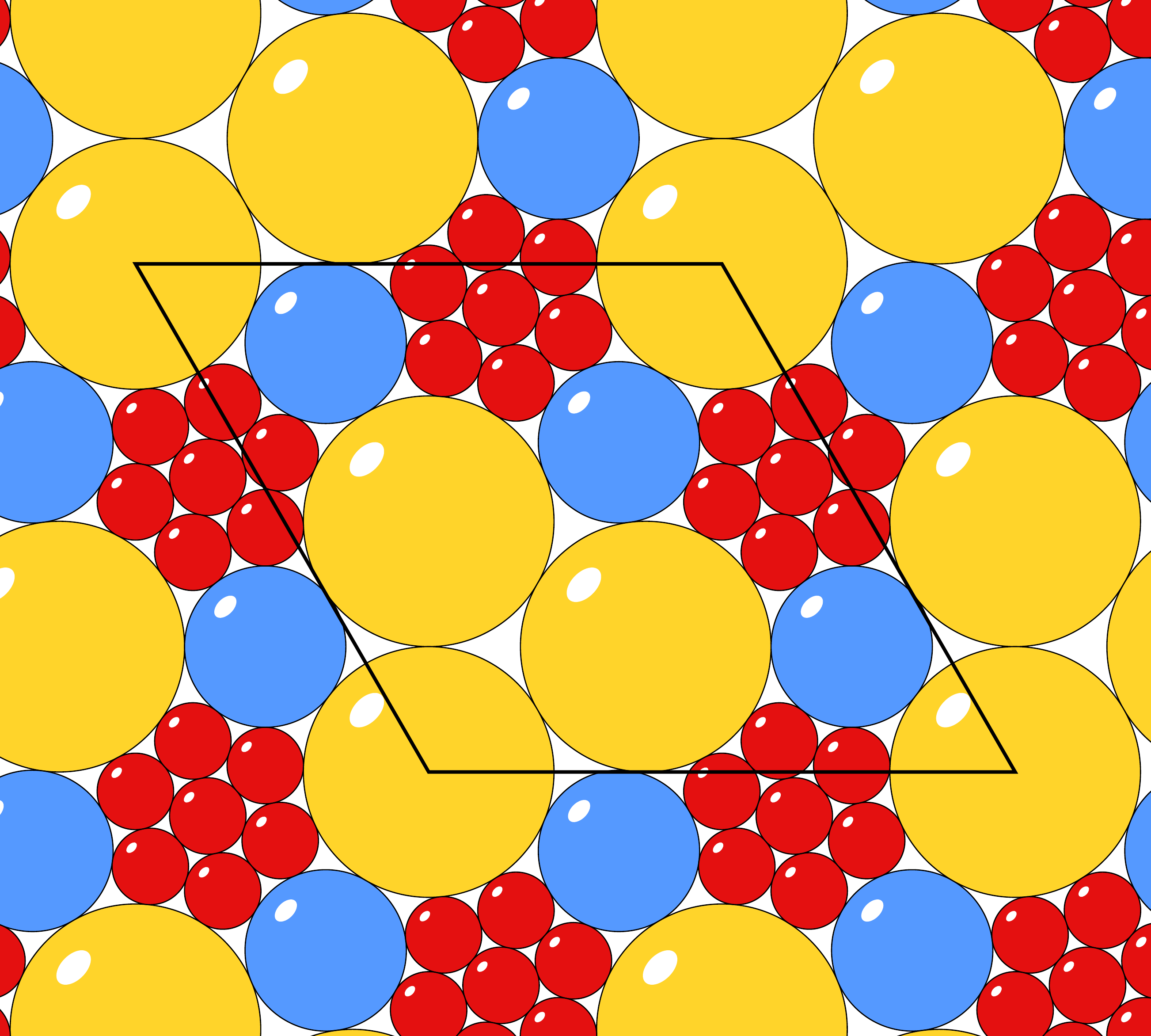} &
  \includegraphics[width=0.3\textwidth]{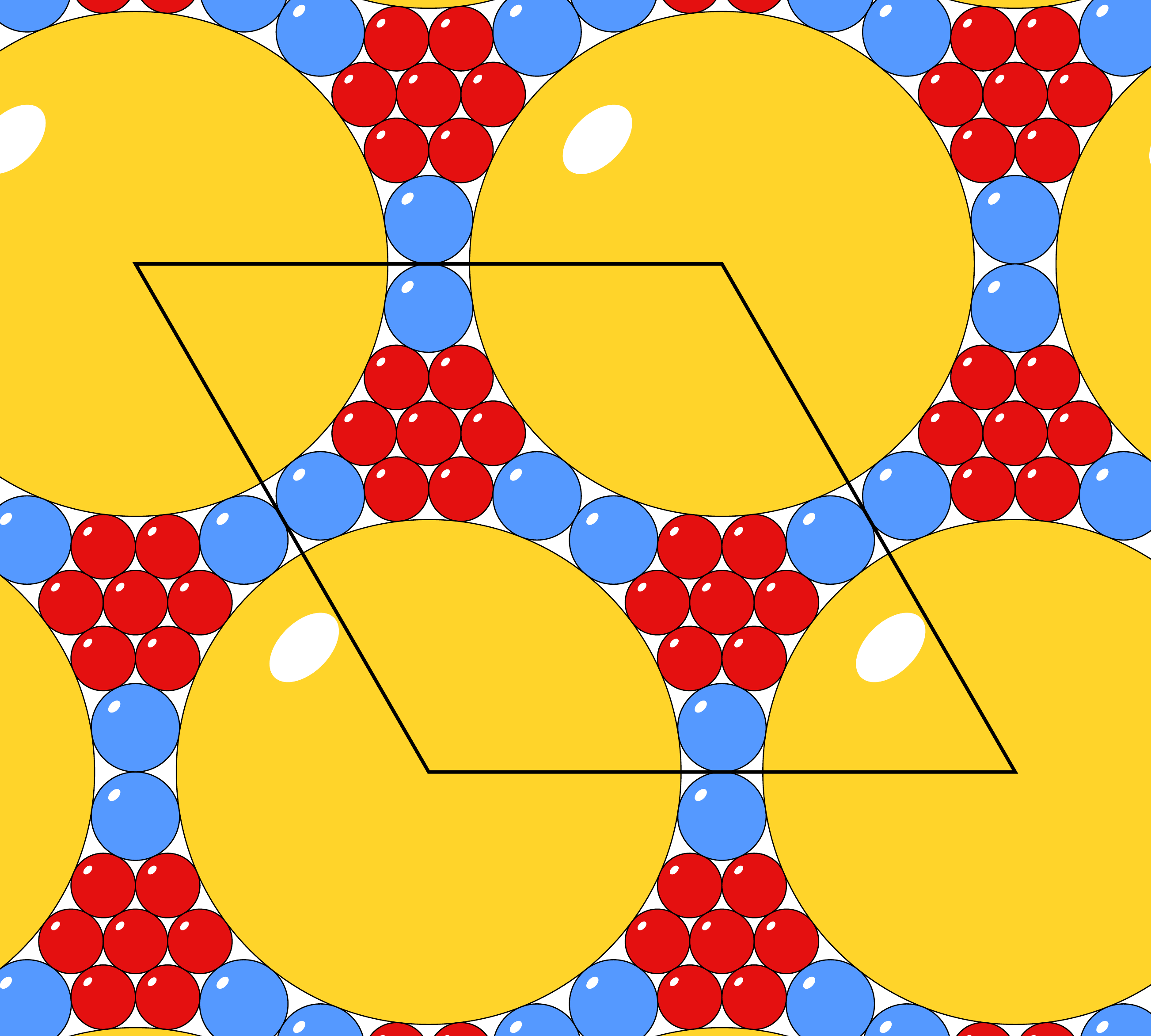}
\end{tabular}
\noindent
\begin{tabular}{lll}
  157 (H)\hfill 1rsss / 1rr1ss & 158 (H)\hfill 1rsss / 1s1sss & 159 (H)\hfill rrrr / 11rrsr\\
  \includegraphics[width=0.3\textwidth]{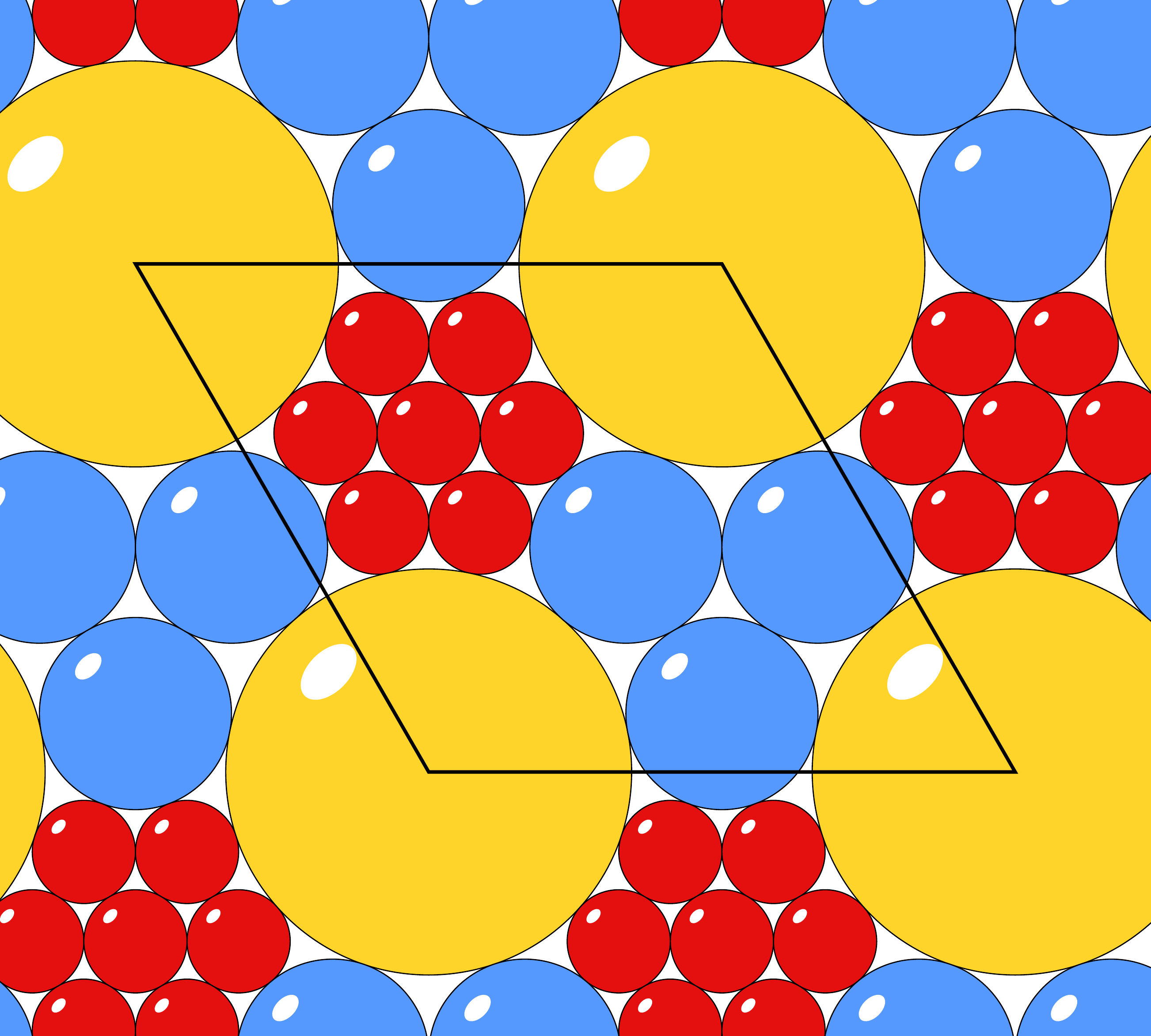} &
  \includegraphics[width=0.3\textwidth]{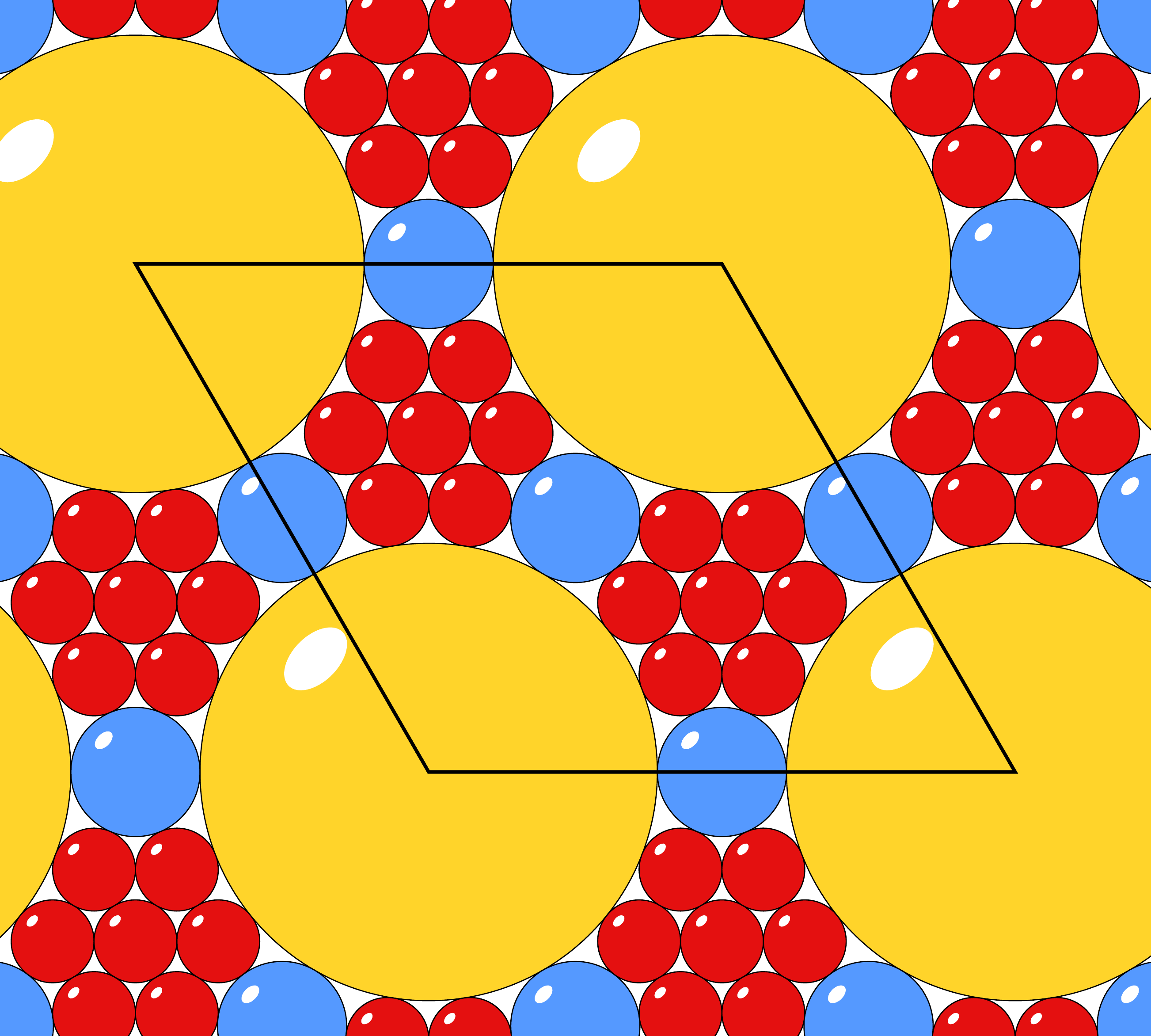} &
  \includegraphics[width=0.3\textwidth]{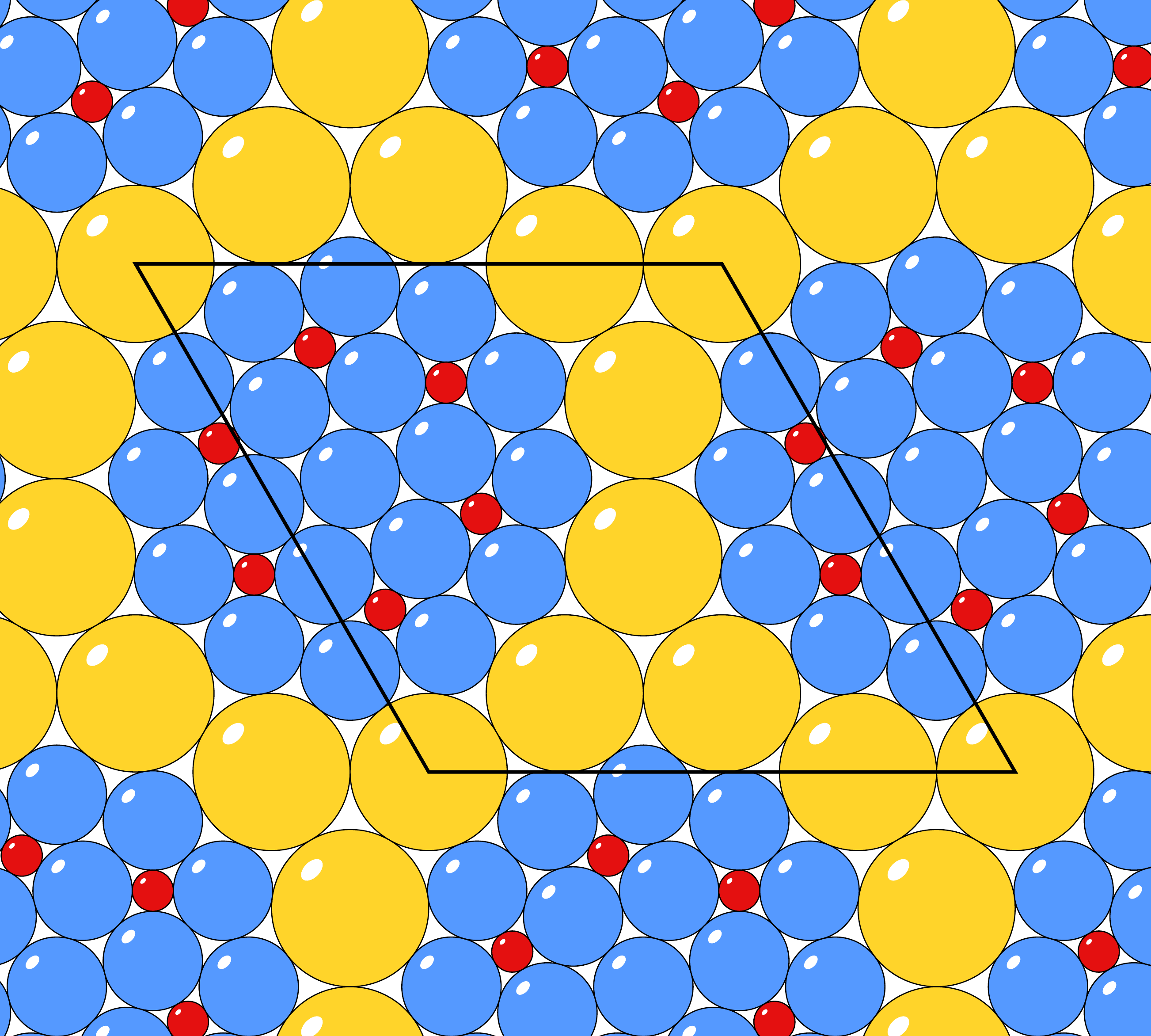}
\end{tabular}
\noindent
\begin{tabular}{lll}
  160 (E)\hfill rrrr / 11rsr & 161 (E)\hfill rrrr / 1r1rsr & 162 (E)\hfill rrrr / 1rrrsr\\
  \includegraphics[width=0.3\textwidth]{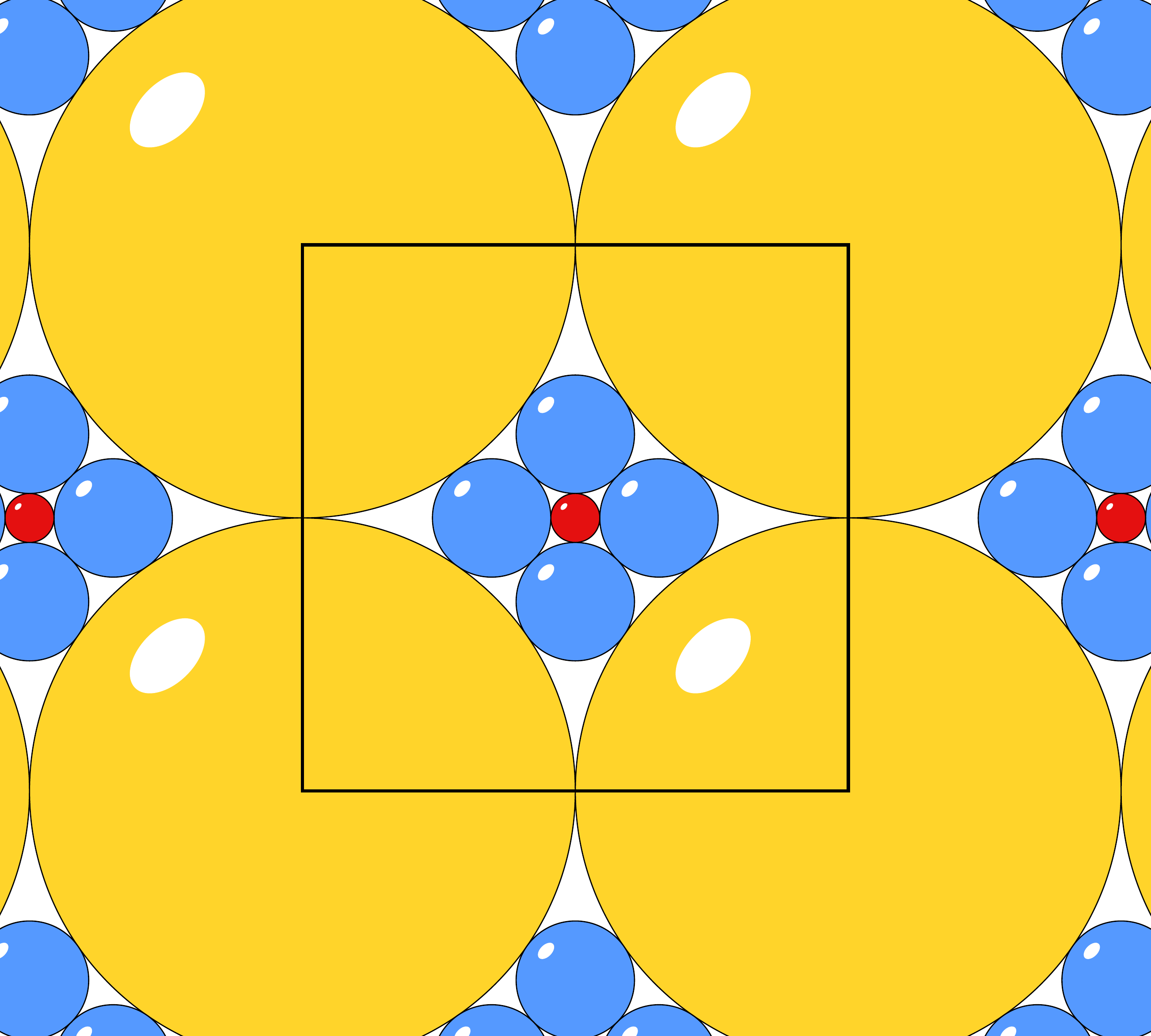} &
  \includegraphics[width=0.3\textwidth]{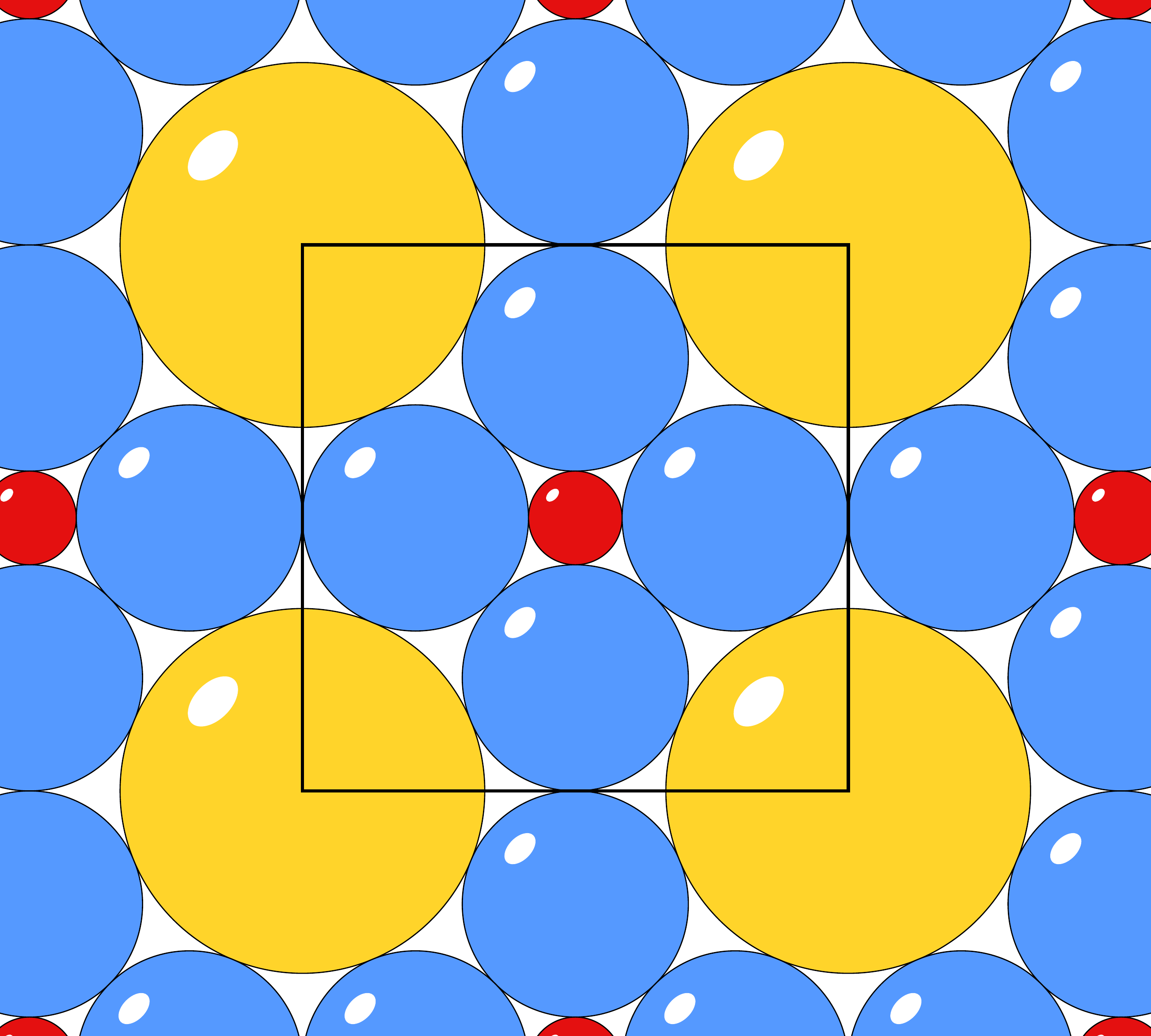} &
  \includegraphics[width=0.3\textwidth]{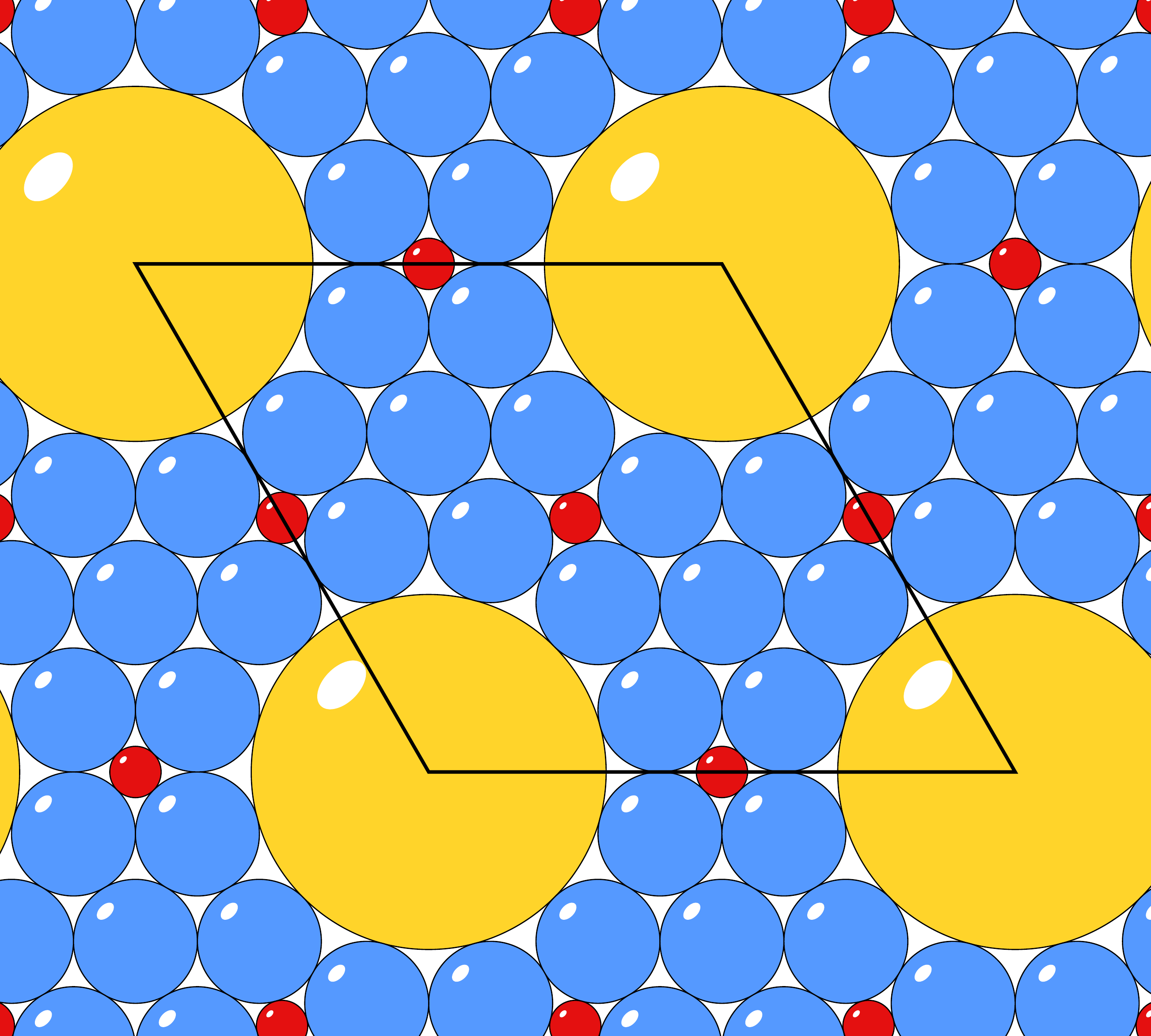}
\end{tabular}
\noindent
\begin{tabular}{lll}
  163 (L)\hfill rrsrs / 1rssssr & 164 (H)\hfill rrsss / 1rrsssr & \\
  \includegraphics[width=0.3\textwidth]{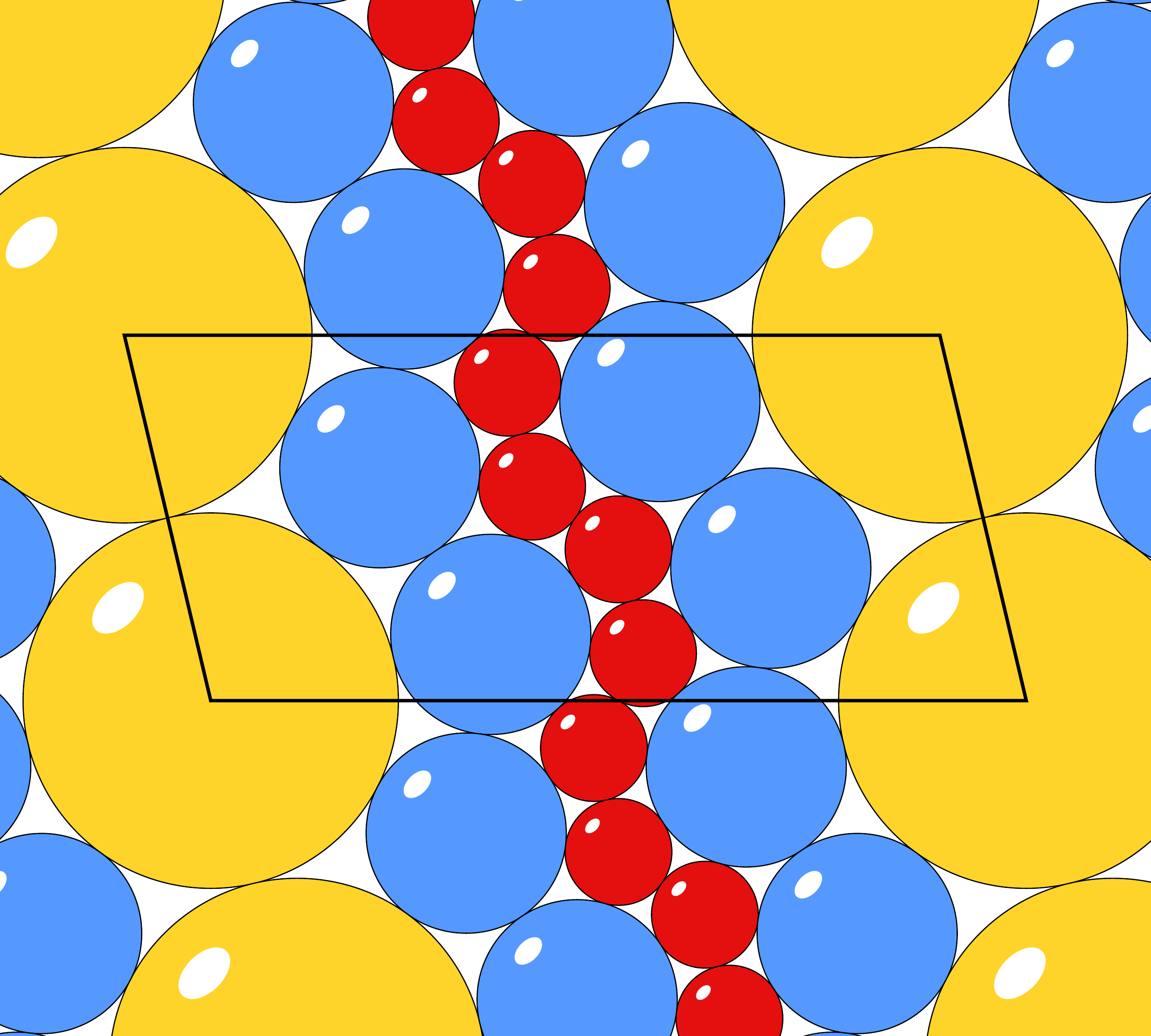} &
  \includegraphics[width=0.3\textwidth]{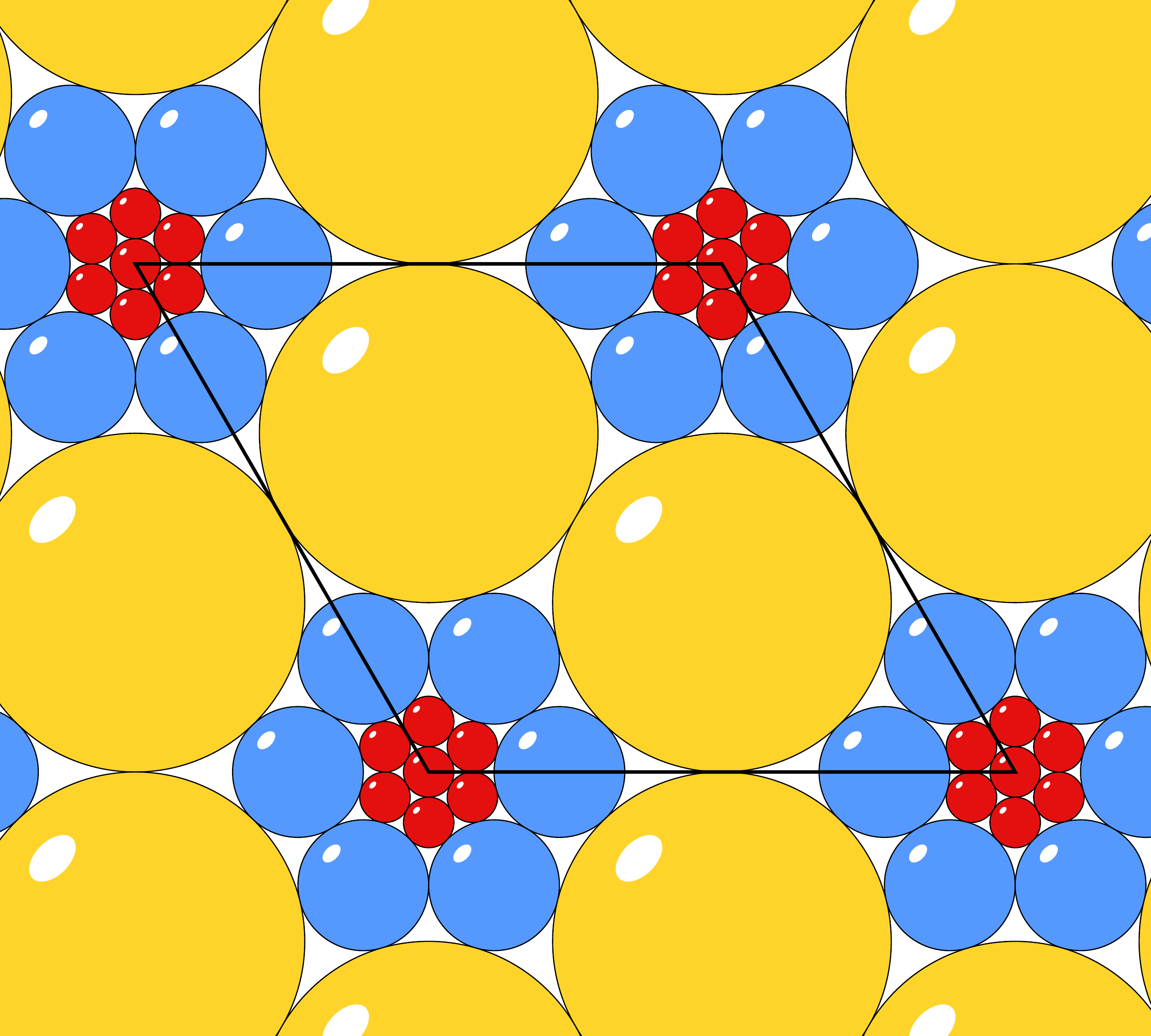} &
\end{tabular}

\section{Classification}
\label{sec:classification}

Appendix~\ref{sec:examples} gives an example of compact packing for each of the $164$ pairs $(r,s)$ which allows a compact packing by discs of size $s<r<1$.
However, many pairs allow not only one but a whole set of compact packings, namely a {\em tiling space} in the terminology of \cite{Rob04} (which extends symbolic dynamical systems to tilings).
In order to give an idea of the variety of possible packings, we assign to each case a {\em type} (letter {\bf H}, {\bf L}, {\bf S} or {\bf E} in brackets in App.~\ref{sec:examples}, Fig.~\ref{fig:2packings} and Fig.~\ref{fig:3packings}).
We distinguish four basic types with possible refinements in the following sense.
A packing set $Y$ is said to be a {\em refinement} of a paking set $X$ if there is a "local recoding" (a surjective continuous map which commutes with isometries) which maps $Y$ onto $X$.
Roughly, the local recoding simply removes the flourish.
In the terminology of dynamical system \cite{Rob04}, $Y$ is said to {\em factor} on $X$ and the local recoding is called a {\em factor map}.
The case c9, for example, is a refinement of the hexagonal compact packing with one size of discs: the local recoding removes the small discs between large discs.
The same holds for c5, with the local recoding replacing the clusters of 7 small discs by a large discs.
Two sets which are mutual refinements are said to be {\em conjugated}: they are pretty much the same ({\em e.g.}, c8 and c9 or c4 and $160$).
This allows to focus more on the very structure of packing sets.

\paragraph{Periodic packings (H).}
This is the simplest type: the disc sizes allow only finitely many compact packings with two independent periodic directions.  
This includes the hexagonal compact packing with one size of discs (whence the letter H), the compact packing with two sizes of discs labelled c6 in Fig.~\ref{fig:2packings} and $52$ cases with three sizes of discs.
One checks that $10$ out of the $17$ wallpaper groups appear as symmetry groups of these periodic compact packings (Tab.~\ref{tab:periodic}).
Refinements include the cases c5, c8, c9 and $22$ cases with three sizes of discs.

\begin{table}[hbtp]
\centering
\begin{tabular}{|cccccccccc|}
\hline
p6m & p6 & p31m & pmg & p3m1 & cmm & p3 & p4g & p4m & pgg\\
78 & 49 & 113 & 47 & 115 & 116 & 112 & 93 & 108 & 53\\
\hline
\end{tabular}
\caption{
Examples of periodic compact packings with three sizes of discs (numbers refer to App.~\ref{sec:examples}) for each possible symmetry group.}
\label{tab:periodic}
\end{table}

\paragraph{Laminated packings (L).}
The disc sizes allow only compact packings with exactly one periodic direction (and maybe finitely many degenerated cases).
This includes c1 and c3, already described in \cite{Ken06}, as well as $54$ cases with three sizes of discs ($7$ of which are refinements).
Fig.~\ref{fig:laminated} gives a typical example.

\begin{figure}[hbtp]
\centering
\includegraphics[width=\textwidth]{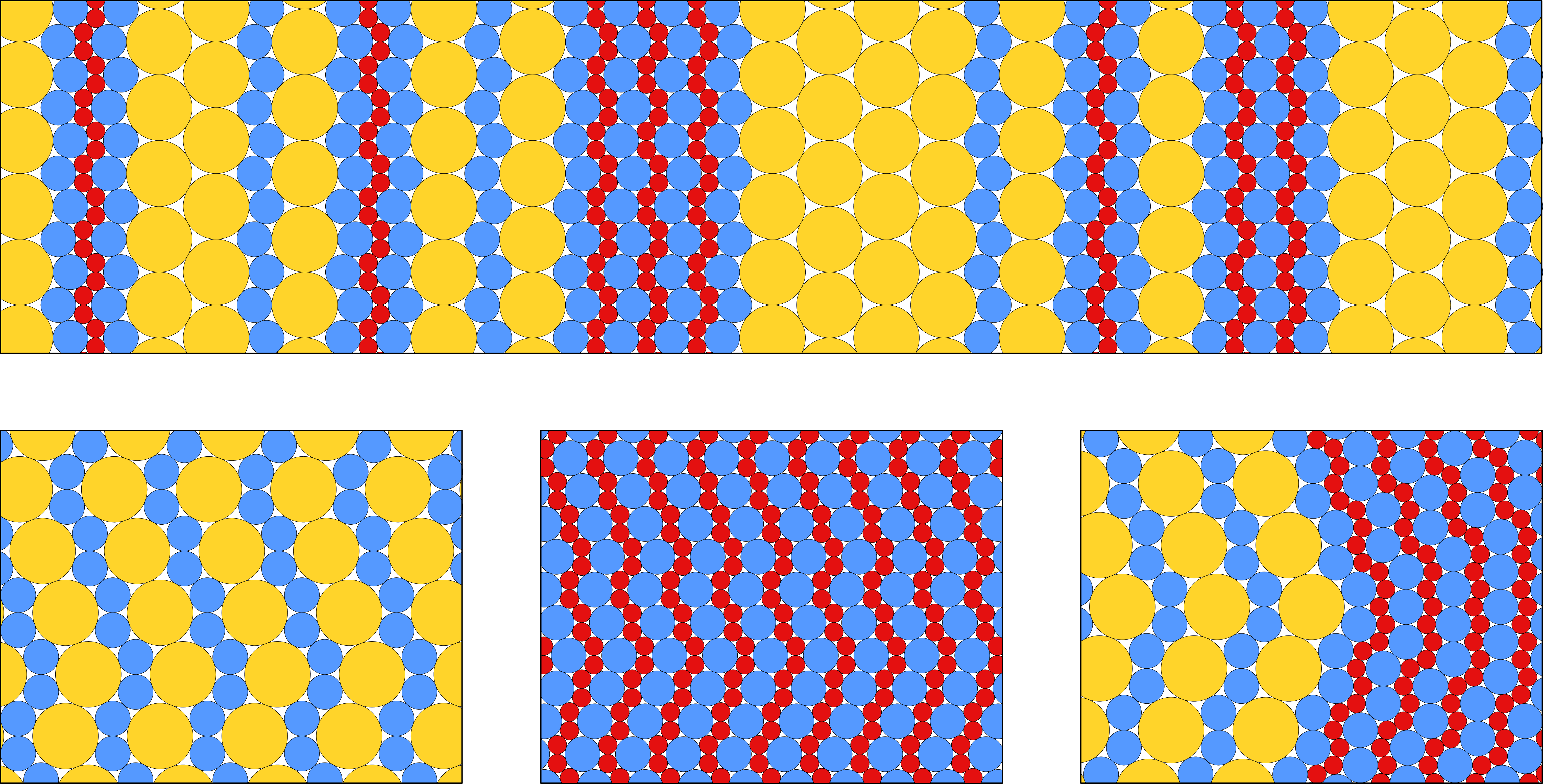}
\caption{
The typical packings of case $163$ alternate lines of large, medium and small discs, such that there is always a line of medium discs between a line of large discs and a line of small discs, and only large discs can form two consecutive lines (top).
There are also the laminated packings of case c3, where the lines of large and medium discs can be bended
 (bottom-left and bottom-center) and a single (up to isometry) degenerated case (bottom-right).
}
\label{fig:laminated}
\end{figure}

\paragraph{Shield packings (S).}
The disc sizes allow compact packings which can be seen as tilings by an equilateral triangle and a {\em shield}, that is, a convex hexagon with two different angles (one obtuse and one acute) which alternate.
This case includes c2, as well as $13$ cases with three sizes of discs ($7$ of which are refinements).
Fig.~\ref{fig:shield} describe these packings (\cite{Ken06} gives only two examples).

\begin{figure}[hbtp]
\centering
\includegraphics[width=\textwidth]{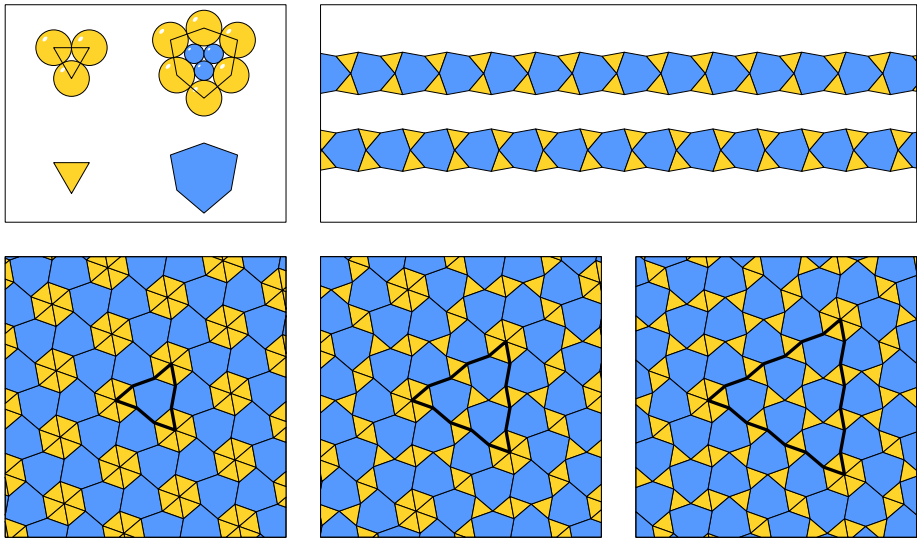}
\caption{
The compact packings which can be seen as tilings by a triangle and a shield (top-left) can form laminated packings (top-right: the two stripes can freely alternate) and a family of periodic packings looking like triangular grids of arbitrarily large size (bottom, the three first grids).
}
\label{fig:shield}
\end{figure}

\paragraph{Positive entropy packings (E).}
The set of discs of a packing which intersect a ball of radius $r$ forms what is called an {\em $r$-pattern}.
A packing set is said to have {\em zero entropy} if the number of different $r$-patterns (up to an isometry) grows subexponentially with $r^2$ (the volume of the ball).
The notion of entropy comes from dynamical systems, where it is used to measure the "complexity" of a system (in particular, to distinguish non-conjugated systems).
In pratice for our classification, zero-entropy means that the set of possible packings is rather easy to describe.
Periodic, laminated or shield packings do have zero entropy.
Not their refinements, because flourish can be added or not independently at each position, but this does not affect the very structure of packing sets which are still easy to describe.
The $23$ remaining cases, however, do not have zero entropy nor are refinements of zero entropy cases.
They are thus somehow more complicated to describe.

Actually, most of them can be seen as tilings by a square and a regular triangle, known in statistical mechanics as {\em square-triangle tilings} (Fig.~\ref{fig:square_triangle}).
This includes c4 and (up to a refinement) $8$ cases with three sizes of discs (numbered 1, 11, 16, 27, 61, 160, 161 and 162 in App.~\ref{sec:examples}).
This also includes, up to a shear of the square into a rhombus which does not modify the combinatorics, c7 and (up to a refinement) $11$ cases with three sizes of discs (numbered 2, 9, 12, 21, 30, 68, 71, 75, 92, 109, 135).

\begin{figure}[hbtp]
\centering
\includegraphics[width=\textwidth]{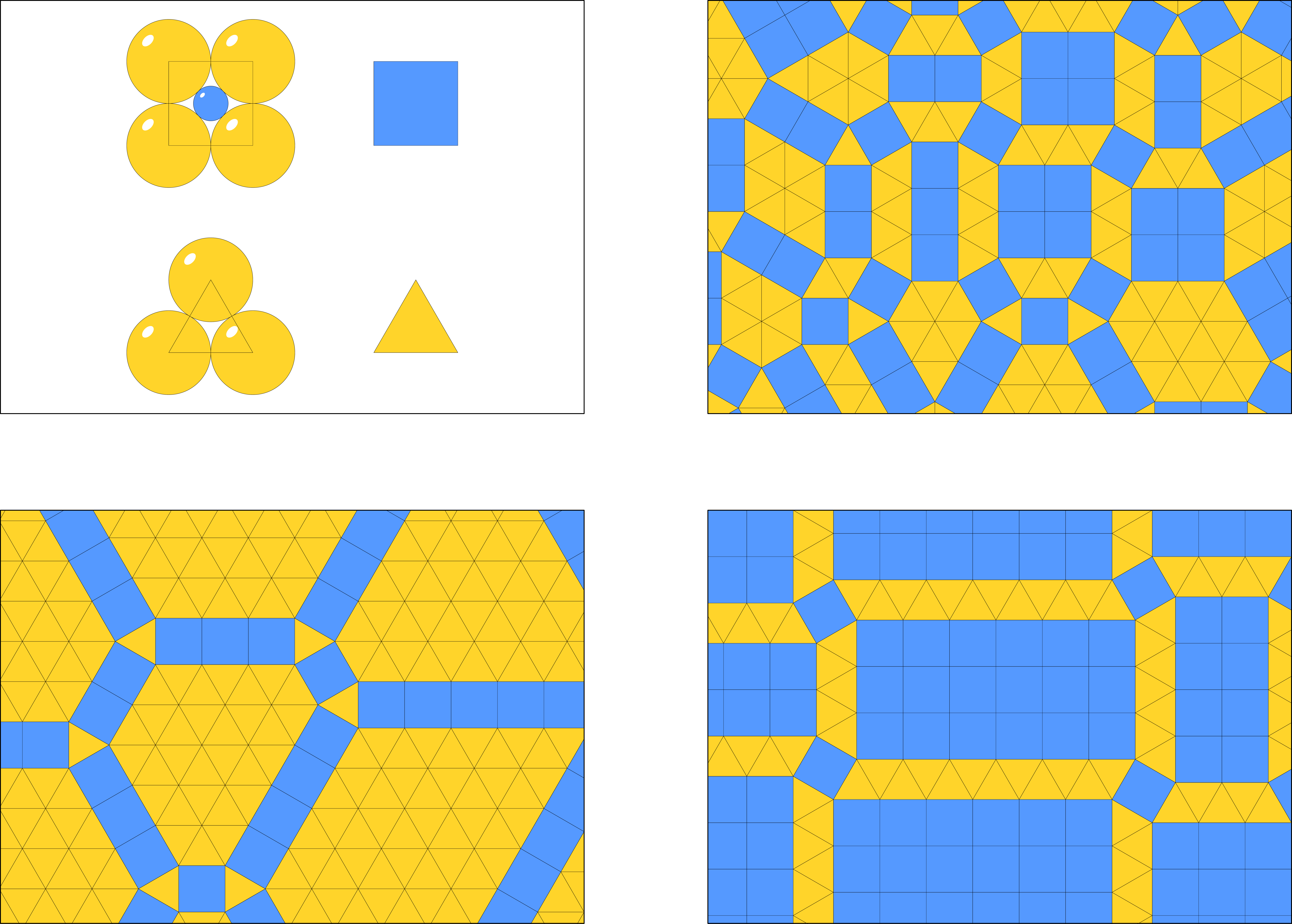}
\caption{
The compact packings which can be seen as tilings by a square and a regular triangle (top-left) can form a wide range of different packings, more or less random and with various proportions of tiles.
}
\label{fig:square_triangle}
\end{figure}

Besides square-triangle tilings, there are three cases which are a sort of mix of two different compact packings with two sizes of discs, namely those numbered 6, 7 and 8 which respectively mix c3/c4, c4/c7 and c3/c7.

Last but not least, the case $83$.
It is unique in the sense that it is neither a refinement of another case nor conversely.
The compact packings can form rather complicated curves which alternate a small and two medium discs (Fig.~\ref{fig:83}, left).
These packings can be seen as the tilings by a square, a regular triangle and an irregular one (Fig.~\ref{fig:83}, center and right).
The edges of the irregular triangle have length $2$, $2+2s$ and $2\sqrt{1+2r}$, where $r=\sqrt{2}-1$ and $s\simeq 0.249$ is root of $X^4 + 4X - 1$.
The smallest angle is $\tfrac{\pi}{4}$ and the largest one is $\arccos(1-\tfrac{1}{\sqrt{2}})$.

\begin{figure}[hbtp]
\centering
\includegraphics[width=\textwidth]{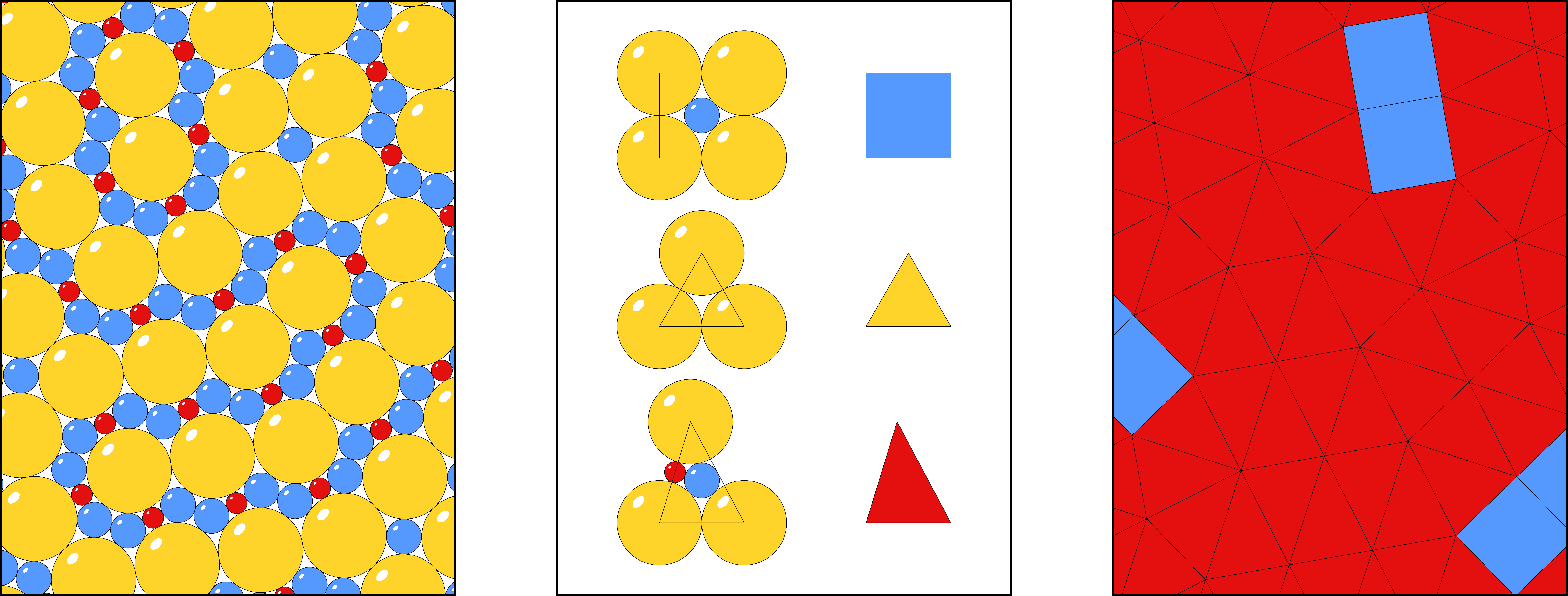}
\caption{
Case $83$: a compact packing and the corresponding tiling.
}
\label{fig:83}
\end{figure}

\section{Code}
\label{sec:code}

Computations and case checking have been done with Python and SageMath.
The full commented code is provided in supplementary materials.
It is organized in five programs here briefly described (the numbers in brackets give the number of code lines, comments included):

$\bullet$ \verb+coronas.sage+ (191) contains functions to convert vector angles to sequence and conversely, to find all the possible small and medium coronas (Sections~\ref{sec:small} and \ref{sec:medium}), to find the coronas compatible with interval values of $r$ and $s$.

$\bullet$ \verb+equations.sage+ (126) contains functions to compute the polynomial associated with a corona (Section~\ref{sec:polynomial}) and to check exactly whether given algebraic values of $r$ and $s$ are compatible with a given corona.

$\bullet$ \verb+two_phases.sage+ (28) deals with the large separated packings (Section~\ref{sec:two_phases}).

$\bullet$ \verb+two_small_coronas.sage+ (127) deals with the packings with two different s-coronas (Section~\ref{sec:two_smalls}).
It implements the hidden variable method, then apply interval arithmetic and exact filtering.

$\bullet$ \verb+one_small_coronas.sage+ (254) deals with the packings with only one s-coronas (Section~\ref{sec:one_small}).
It implements the (pre)cover condition and the hidden variable method, then apply interval arithmetic and exact filtering.

\paragraph{Acknowledgments.}
We thank T. Kennedy for pointing us reference \cite{Mes20}, hopefully after we completed our proof so that our approach has not been influenced.
We thank Thierry Monteil for answering various questions about SageMath, as well as Bruno Salvy for discussions on Gröbner basis.
We thank the referees of a short conference version of this paper \cite{FHS19}, as well as the referees of this long version.


\bibliographystyle{alpha}
\bibliography{three_discs}

\newcommand{\etalchar}[1]{$^{#1}$}
\begin{thebibliography}{CKM{\etalchar{+}}17}

\bibitem[CKM{\etalchar{+}}17]{CKM17}
H.~Cohn, A.~Kumar, S.~Miller, D.~Radchenko, and M.~Viazovska.
\newblock The sphere packing problem in dimension $24$.
\newblock {\em Annals of Mathematics}, 185:1017--1033, 2017.

\bibitem[CLO05]{CLO05}
D.~Cox, J.~Little, and D.~O'Shea.
\newblock {\em Using Algebraic Geometry}.
\newblock Number 185 in Graduate Texts in Mathematics. Springer, New York,
  2005.

\bibitem[CS99]{CS99}
J.~Conway and N.~Sloane.
\newblock {\em Sphere Packings, Lattices and Groups}.
\newblock Springer, 1999.

\bibitem[Dev16]{sage}
The~Sage Developers.
\newblock {\em {S}age {M}athematics {S}oftware ({V}ersion 8.2)}, 2016.
\newblock {\tt http://www.sagemath.org}.

\bibitem[DGPS16]{singular}
W.~Decker, G.~Greuel, G.~Pfister, and H.~Schönemann.
\newblock {\em Singular 4-0-3 — A computer algebra system for polynomial
  computations.}, 2016.
\newblock {\tt http://www.singular.uni-kl.de}.

\bibitem[Fau10]{Fau10}
J.-C. Faugère.
\newblock {FGb: A Library for Computing Gröbner Bases}.
\newblock In K.~Fukuda, J.~Hoeven, M.~Joswig, and N.~Takayama, editors, {\em
  {Mathematical Software - ICMS 2010}}, volume 6327 of {\em Lecture Notes in
  Computer Science}, pages 84--87. Springer Berlin / Heidelberg, 2010.

\bibitem[Fer19]{Fer18}
Th. Fernique.
\newblock Compact packings of the space with two spheres.
\newblock Discrete and Computational Geometry, 2019.

\bibitem[FHS19]{FHS19}
Th. Fernique, A.~Hashemi, and O.~Sizova.
\newblock Compact packings of the plane with three sizes of discs.
\newblock In {\em Discrete Geometry for Computer Imagery - 21st {IAPR}
  International Conference, {DGCI} 2019, Marne-la-Vall{\'{e}}e, France, March
  26-28, 2019, Proceedings}, pages 420--431, 2019.

\bibitem[FT43]{FT43}
L.~Fejes~T{\'o}th.
\newblock Über die dichteste {K}ugellagerung.
\newblock {\em Mathematische Zeitschrift}, 48:676--684, 1943.

\bibitem[FT64]{FT64}
L.~Fejes~T{\'o}th.
\newblock {\em Regular figures}.
\newblock International series of monographs in pure and applied mathematics.
  Macmillan, 1964.

\bibitem[Hal05]{Hal05}
Th. Hales.
\newblock A proof of the {K}epler conjecture.
\newblock {\em Annals of Mathematics}, 162:1065--1185, 2005.

\bibitem[Hep00]{Hep00}
A~Heppes.
\newblock On the densest packing of discs of radius $1$ and $\sqrt{2}-1$.
\newblock {\em Studia Scientiarum Mathematicarum Hungarica}, 36:433--454, 2000.

\bibitem[Hep03]{Hep03}
A.~Heppes.
\newblock Some densest two-size disc packings in the plane.
\newblock {\em Discrete and Computational Geometry}, 30:241--262, 2003.

\bibitem[HST12]{HST12}
A.~Hopkins, F.~Stillinger, and S.~Torquato.
\newblock Densest binary sphere packings.
\newblock {\em Phys. Rev. E}, 85:021130, 2012.

\bibitem[Ken04]{Ken04}
T.~Kennedy.
\newblock A densest compact planar packing with two sizes of discs.
\newblock preprint, \href{https://arxiv.org/abs/math/0412418}{arxiv:0412418},
  2004.

\bibitem[Ken06]{Ken06}
T.~Kennedy.
\newblock Compact packings of the plane with two sizes of discs.
\newblock {\em Discrete and Computational Geometry}, 35:255--267, 2006.

\bibitem[LH93]{LH93}
C.~Likos and C.~Henley.
\newblock Complex alloy phases for binary hard-disc mixtures.
\newblock {\em Philosophical Magazine B}, 68:85--113, 1993.

\bibitem[Mes20]{Mes20}
M.~Messerschmidt.
\newblock On compact packings of the plane with circles of three radii.
\newblock {\em Computational Geometry}, 86:101564, 2020.

\bibitem[OH11]{OH11}
P.~O’Toole and T.~Hudson.
\newblock New high-density packings of similarly sized binary spheres.
\newblock {\em J. Physical Chemistry {C}}, 115:19037--19040, 2011.

\bibitem[Pen78]{Pen78}
R.~Penrose.
\newblock Pentaplexity: a class of non-periodic tilings of the plane.
\newblock {\em Eureka}, 39:16--22, 1978.

\bibitem[Rob71]{Rob71}
R.~M. Robinson.
\newblock Undecidability and nonperiodicity for tilings of the plane.
\newblock {\em Inventiones mathematicae}, 12:177--209, 1971.

\bibitem[Rob04]{Rob04}
E.~A. Robinson.
\newblock Symbolic dynamics and tilings of $\mathbb{R}^d$.
\newblock {\em Symbolic dynamics and its applications}, 60:81--119, 2004.

\bibitem[Via17]{Via17}
M.~Viazovska.
\newblock The sphere packing problem in dimension $8$.
\newblock {\em Annals of Mathematics}, 185:991--1015, 2017.

\bibitem[Z{\etalchar{+}}13]{sage2}
P.~Zimmerman et~al.
\newblock {\em Calcul Mathématique avec Sage}.
\newblock CreateSpace Independent Publishing Platform, 2013.

\end{thebibliography}

\newpage
\end{document}